\documentclass{SIAMbook2019}

\makeatletter
\newcommand{\preprintfooter}{%
  \scriptsize A preprint version of a book to be published by Cambridge University Press.
  \copyright\ Sanz-Alonso and Al-Ghattas, 2027%
}

\let\SIAMoriginalheadings\ps@headings
\renewcommand{\ps@headings}{%
  \SIAMoriginalheadings
  \def\@oddfoot{\hfil\preprintfooter\hfil}%
  \let\@evenfoot\@oddfoot
}
\makeatother

\AtBeginDocument{\pagestyle{headings}}

\usepackage{epsfig}
\usepackage{graphicx}
\usepackage{placeins}
\usepackage{makeidx}
\usepackage{multicol}
\usepackage{algcompatible}

\usepackage{imakeidx}
\makeindex[columns=2, title=Alphabetical Index,intoc]

\usepackage{tikz}
\usetikzlibrary{fit,positioning,arrows,automata,calc}

\usepackage{times}
\usepackage[T1]{fontenc}

\title{\Huge A First Course in \\ Monte Carlo Methods}
\author{\Large Daniel Sanz-Alonso and Omar Al-Ghattas}
\date{\large University of Chicago}

\renewcommand{\phi}{\varphi}
\newcommand{\eps}{\varepsilon}

\newcommand{\R}{\mathbb{R}}

\newcommand{\N}{\mathbb{N}}

\definecolor{mygreen}{rgb}{0.1,0.75,0.2}

\newcommand{\nc}{\normalcolor}

\usepackage{amsmath,amsgen,amscd,amsfonts,amssymb} 

\usepackage{mathtools}

\newcommand{\qedhere}{$\square$}

\newtheorem{remark}[theorem]{Remark}
\newtheorem{assumption}[theorem]{Assumption}

\newcommand{\X}{X}

\newcommand{\Z}{Z}
\newcommand{\XX}{\mathbb{X}}

\newcommand{\Prob}{\operatorname{\mathbb{P}}}
\newcommand{\Expect}{\operatorname{\mathbb{E}}}
\newcommand{\V}{\operatorname{\mathbb{V}}}

\newcommand{\T}{\mathcal{T}}

\newcommand{\tg}{f}
\newcommand{\pr}{g}

\newcommand{\qrwmh}{q_{\mbox {\tiny{\rm RWM}}}}
\newcommand{\pind}{p_{\mbox {\tiny{\rm Indep}}}}

\newcommand{\pmh}{p_{\mbox {\tiny{\rm MH}}}}
\newcommand{\pdugs}{p_{\mbox {\tiny{\rm DUGS}}}}
\newcommand{\qlgv}{q_{\mbox {\tiny{\rm Lgv}}}}

\newcommand{\qind}{q_{\mbox {\tiny{\rm Indep}}}}
\newcommand{\pgs}{p_{\mbox {\tiny{\rm GS}}}}

\newcommand{\dkl}{d_{\mbox {\tiny{\rm KL}}}}

\newcommand{\dtv}{d_{\mbox {\tiny{\rm TV}}}}
\newcommand{\dchi}{d_{\mbox {\tiny{$ \chi^2$}}}}

\newcommand{\Nc}{\mathcal{N}}

\newcommand{\bt}{x}
\newcommand{\bti}{x_i}
\newcommand{\btj}{x_j}
\newcommand{\btmi}{x_{-i}}

\newcommand{\ELBO}{\textsc{elbo}}

\newcommand{\EEq}[1]{\mathbb{E}_{g}\left[#1\right]}
\newcommand{\EEqi}[1]{\mathbb{E}_{g_i}\left[#1\right]}
\newcommand{\EEqj}[1]{\mathbb{E}_{g_j}\left[#1\right]}
\newcommand{\EEqmi}[1]{\mathbb{E}_{g_{-i}}\left[#1\right]}

\usepackage{algorithm}

\usepackage{ragged2e}

\definecolor{dark-gray}{gray}{0.3}
\definecolor{dkgray}{rgb}{.4,.4,.4}
\definecolor{dkblue}{rgb}{0,0,.5}
\definecolor{medblue}{rgb}{0,0,.75}
\definecolor{rust}{rgb}{0.5,0.1,0.1}

\usepackage{url}
\usepackage[colorlinks=true]{hyperref}
\hypersetup{linkcolor=dkblue}    
\hypersetup{citecolor=rust}      
\hypersetup{urlcolor=rust}

\definecolor{midnight}  {rgb}{0,0,.5}

\usepackage{tcolorbox}
\tcbuselibrary{skins}
\newtcolorbox{mybox}[2][]{colframe=midnight,fonttitle=\bfseries, coltext=black,
colbacktitle=midnight,
coltitle=midnight,
enhanced,arc is angular,boxed title style={arc is angular},
detach title,before upper={\tcbtitle:~},
title=#2,#1,
before skip=6pt, after skip=6pt,
left=2pt,right=2pt,top=3pt,bottom=1pt}

\newcommand{\upperRomannumeral}[1]{\uppercase\expandafter{\romannumeral#1}}

\renewcommand{\hat}{\widehat}

\usepackage{enumerate}

\usepackage{verbatim}

\usepackage{mathrsfs}

\begin{document}

\maketitle
\frontmatter

\listoffigures
\listofalgorithms

\begin{thepreface}
\paragraph{Overview} This textbook provides a concise mathematical introduction to Monte Carlo methods, a rich family of algorithms with far-reaching applications in science and engineering. Monte Carlo methods are an exciting topic for computational mathematicians and statisticians: 
the design and analysis of modern algorithms 
are rooted in 
a broad mathematical toolbox
that includes 
ergodic theory of Markov chains, Hamiltonian dynamical systems, transport maps, stochastic differential equations, information theory, optimization, Riemannian geometry, and gradient flows, among many others. This concise textbook 
celebrates 
the breadth of mathematical ideas that have 
led to tangible 
advancements in Monte Carlo methods and their applications. 
To accommodate a broad audience, the level of mathematical rigor varies 
across chapters,
with the most technically demanding 
subjects receiving only an intuitive treatment.

The aim is not to be comprehensive or encyclopedic, but rather to illustrate some key principles in the design and analysis of Monte Carlo methods through a 
carefully-crafted choice of 
topics that emphasizes timeless over timely ideas. Algorithms are presented in a way that 
is conducive to 
conceptual understanding and mathematical analysis ---clarity and intuition are favored over state-of-the-art implementations that are harder to
comprehend 
or rely on \emph{ad-hoc} heuristics. To help readers navigate the expansive landscape of Monte Carlo methods, each algorithm is accompanied by 
a discussion of its strengths and limitations, as well as the types of problems for which it is most suitable.
The presentation is self-contained, and 
therefore adequate for self-guided learning and as a teaching resource. 
 Each chapter 
includes a section with 
bibliographic remarks 
to guide computational scientists and graduate students interested in research on Monte Carlo methods and their applications. \nc

\paragraph{Acknowledgements} This textbook developed as a result of several courses taught by Daniel Sanz-Alonso at the University of Chicago since 2019; Omar Al-Ghattas served as teaching assistant for two of these courses and contributed significantly to improve the presentation. The main target audience is advanced undergraduate and graduate students in computational mathematics, applied mathematics, and statistics. Our courses are also regularly taken by graduate students in financial mathematics, computer science, chemistry, physics, geophysics, economics, and public policy that are interested in using Monte Carlo methods in their research. A first draft was created by the students of the course Stochastic Simulation, STAT 31510, in Spring 2019. The students responsible for this first draft were: Weilin Chen, Fuheng Cui, Zhen Dai, Marlin Figgins, Kim Liu, Yuwei Luo, Shinpei Nakamura Sakai, David Noursi, Nick Rittler, Matthew Shin, Zihao Wang, Wen Yuan Yen, Joey Yoo, and Yanfei Zhou. We are grateful to these students, without whom this work would not exist, and to all the students that provided encouragement and constructive feedback over the years. 

We are particularly thankful to Zijian Wang for improving the material on Hamiltonian Monte Carlo and Gibbs samplers, to Shiv Agrawal and Yu-Chun Lai for first drafts on variational inference, to Zhaoming Li and Bobbi Shi for creating several figures, to Jiajun Bao, Haoyuan Mao and Aditya Raman for developing new exercises and figures, and to Ruiyi Yang for his feedback and for typesetting a first draft of the chapter on annealing strategies.
Lanran Fang and Felix Poirier also contributed to develop new exercises.  Finally, Deniz Akyildiz, Victor Elvira, Matthias Morzfeld, and Stephen Yin provided generous feedback that helped us improve the presentation. 

Daniel Sanz-Alonso is thankful to Adam Johansen and Gareth Roberts from whom he learned Monte Carlo methods as a graduate student at the University of Warwick. Their instructive courses and lecture notes \cite{johansen2010monte,Robertsunpublished} have deeply influenced our presentation. 

Last but not least, the authors are grateful to DOE, FBBVA, NGIA and NSF for their support, which has facilitated the completion of this work.

\flushright D. Sanz-Alonso and O. Al-Ghattas \\ University of Chicago

\end{thepreface}

\tableofcontents

\mainmatter

\chapter{Introduction}\label{ch:introduction}

This chapter provides a high-level introduction to Monte Carlo methods and to some overarching themes 
that appear throughout this textbook.
We begin with a motivating example in Section \ref{sec:motivatingexch1},
illustrating
how a simple Monte Carlo method can be used to approximate the area $\pi$ of the unit circle. Next, 
we introduce  two guiding areas of application in Section \ref{sec:motivatingapplications}: 
Bayesian statistics and statistical mechanics.
 Both of these application domains illustrate the need to develop Monte Carlo methods that scale well to high dimensions and that are implementable when the normalizing constant of the target distribution is unknown. 
Section \ref{sec:challenges}  
overviews these and other challenges that play an important role in the design of Monte Carlo methods. Section \ref{sec:themes} summarizes 
several
key principles 
and ideas 
that underlie 
the development of 
many Monte Carlo 
methods and sampling algorithms. 
 Section \ref{sec:coursestructure} describes the structure of this textbook and its use in the course for which it was developed.  
The chapter 
closes 
in Section \ref{sec:biblioch1} with historical and bibliographical remarks.

\section{Motivating Example}\label{sec:motivatingexch1}
Monte Carlo methods represent a quantity of interest as a parameter of a distribution and use a random sample from that distribution to estimate the parameter. 
As an illustrative example, this section describes a Monte Carlo method for approximating $\pi,$ 
the area of the unit circle $C$. Let $(X_1,X_2)\sim$ Unif$([-1,1]\times[-1,1])$ be a random vector 
uniformly distributed on the square
$S= [-1,1]\times[-1,1] $. 
Then, 
$$p:=\Prob\Bigl( (X_1,X_2) \in C \Bigr) = \Expect \biggl[{\bf{1}}_{\bigl\{(X_1,X_2) \in C \bigr\}} \biggr] =  \frac{\text{Area C}}{\text{Area S} } = \frac{\pi}{4}.$$
Hence, 
$\pi = 4 p, $
and  we have expressed 
the quantity of interest $\pi$ in terms of 
the probability $p.$
Now, we can estimate $p$ 
---and hence $\pi$--- using a sample $\bigl(X_1^{(1)},X_2^{(1)}\bigr), \ldots, \bigl(X_1^{(N)},X_2^{(N)}\bigr)$ of $N$ independent random vectors uniformly distributed on $S.$ 
Precisely, define 
$$B_N = \sum_{n=1}^N {\bf{1}}_{\Bigl\{\bigl(X_1^{(n)},X_2^{(n)}\bigr) \in C \Bigr\}}. $$
Then, $B_N \sim \text{Binomial}(N,p)$ and we can estimate $p$ by 
$\hat{p}_N = B_N/N$
and $\pi$ by $\hat{\pi}_N = 4\hat{p}_N.$
For example, as shown in Figure \ref{fig:estimatepi}, if $N= 10^4$ and we observe $B_{10000} = 7854,$ 
then
$$\hat{\pi}_{10000} = 4 \times \frac{7854}{10000} = 3.1416.$$ 
When $N$ is large, 
the probability ---\emph{before the random sample is generated}--- that our 
estimate is far from the true value is small.
For $\alpha \in (0,1),$ 
let $z_{\alpha/2}$ denote the value satisfying $$\Prob(- z_{\alpha/2} < Z <z_{\alpha/2}) = 1-\alpha, \quad Z \sim \Nc(0,1).$$
Then,
$$\biggl( \hat{p}_N - z_{\alpha/2} \sqrt{\frac{\hat{p}_N(1-\hat{p}_N)}{N}},  \hat{p}_N + z_{\alpha/2} \sqrt{\frac{\hat{p}_N(1-\hat{p}_N)}{N}} \, \biggr) $$
is an approximate
$1-\alpha$ confidence interval\index{confidence interval} for $p,$ and therefore 
$$\biggl( \hat{\pi}_N - z_{\alpha/2} \sqrt{\frac{\hat{\pi}_N(4-\hat{\pi}_N)}{N}},  \hat{\pi}_N + z_{\alpha/2} \sqrt{\frac{\hat{\pi}_N(4-\hat{\pi}_N)}{N}} \, \biggr) $$
is an approximate $1 - \alpha$ confidence interval for $\pi.$ In the example shown in Figure \ref{fig:estimatepi} where $N = 10^4$ and $\hat{\pi}_{10000} = 3.1416$,
using the value $z_{0.025}=1.96$ corresponding to a 95\% confidence level yields the confidence interval $(3.1094, 3.1738)$.

\bigskip

Let us summarize:
\begin{enumerate}
\item We expressed
the quantity of interest (in our case $\pi$) as an expectation, using 
that for any event $A$, $\mathbb{P}(A) = \Expect [{\bf{1}}_A]$. 
 In this textbook, the 
quantities
we seek to approximate via Monte Carlo methods will often take the form of expected values.
\item  We constructed an estimator of the expectation 
based on a random sample.
\item We used statistical theory to assess the quality of the estimator. The law of large numbers guarantees that the sample approximation converges  to the quantity of interest. Furthermore, the central limit theorem quantifies the speed of 
this convergence and can be used to 
construct confidence intervals. 
\end{enumerate}

\begin{figure}
    \centering
    \includegraphics[width=0.5\linewidth]{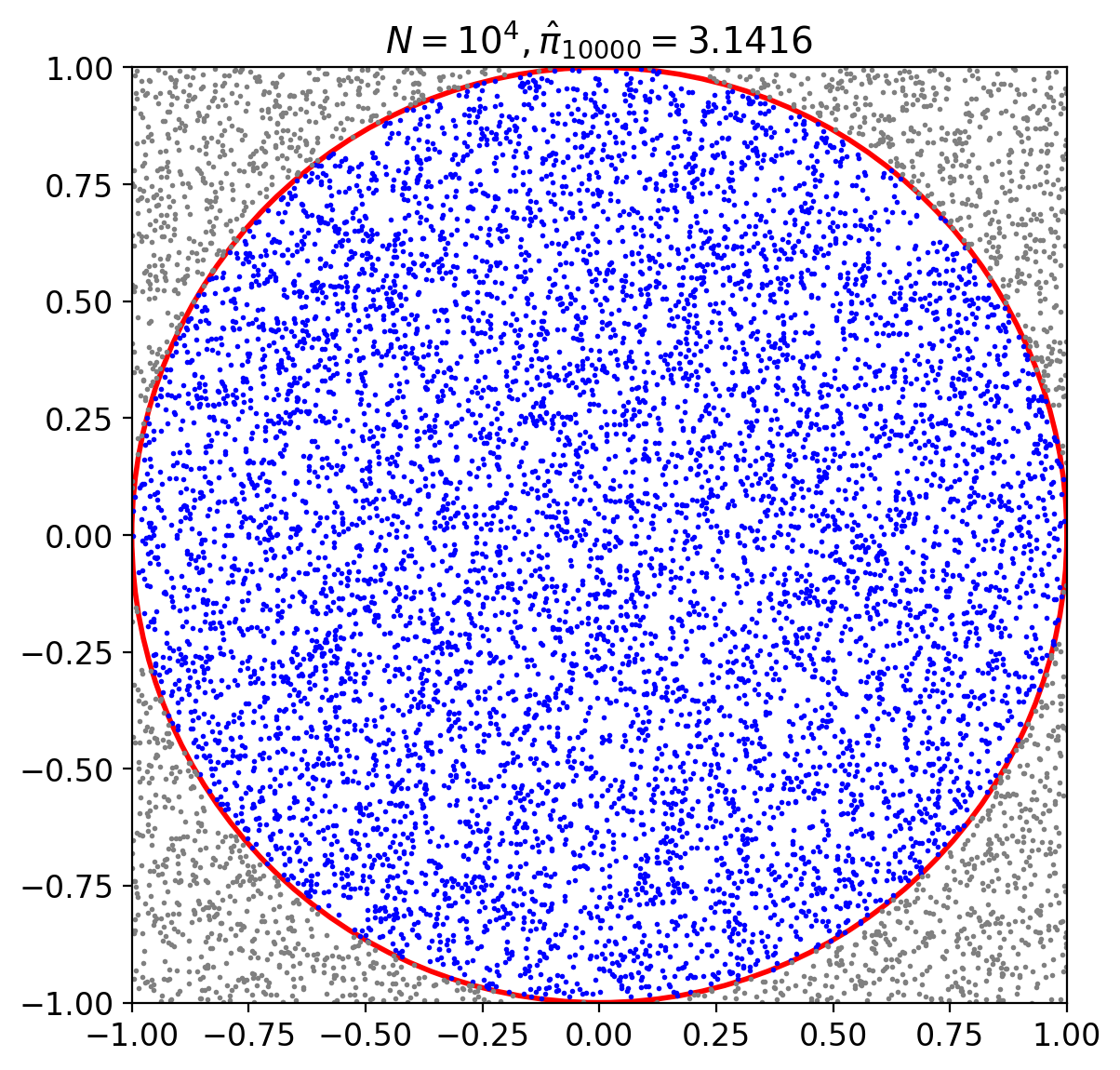}
    \caption{Estimation of $\pi$ with sample size $N=10^4$. 
    Here,  $B_{10000} = 7854$ draws fell within the unit circle, leading to an estimate $\hat{\pi}_{10000} = 3.1416.$}
    \label{fig:estimatepi}
\end{figure}

\section{Guiding Applications}\label{sec:motivatingapplications}
This  textbook  is primarily concerned with Monte Carlo methods for computing
integrals of the form
\begin{equation}\label{eq:genericintegral}
    \mathcal{I}_f[h]:=\int_{E}h(x)f(x) \, dx = \mathbb{E}_{X \sim f} \bigl[h(X)\bigr],
\end{equation}
where $f$ is a given \emph{target distribution} supported on a set $E \subseteq \R^d$ and $h : E \to \mathbb{R}$ is a given \emph{test function}. 
For instance, taking $h(x) = {\bf{1}}_A(x)$ to be the indicator function of a subset $A \subset E,$ we can then compute the probability $\Prob_{X \sim f}(X \in A).$ 
In this section, we describe two application domains where computing integrals of the form \eqref{eq:genericintegral} is important: 
Bayesian statistics and statistical mechanics. In both cases, the dimension $d$ can be exceedingly large and the target distribution $f$ is 
known only up to a normalizing constant, 
rendering deterministic quadrature methods infeasible and posing significant challenges for the design of Monte Carlo algorithms.

\paragraph{Bayesian Statistics\index{Bayesian statistics}}
Given a likelihood \index{likelihood} model $f(y| \theta)$ and a prior \index{prior} density $f(\theta),$ the posterior \index{posterior} density of the parameter $\theta$ given the data $y$ is 
\begin{equation}\label{eq:bayesianintro}
f(\theta| y) = \frac{1}{c} f(y|\theta) f(\theta),
\end{equation}
where $c = \int f(y|\theta) f(\theta) \, d\theta$ is a normalizing constant \index{normalizing constant} 
ensuring that
$f(\theta|y)$ integrates to 
one. 
Bayesian inference 
often requires computing expectations of the form
$$\mathcal{I}_f[h] = \int h(\theta) f(\theta|y) \, d\theta,$$
where the posterior $f(\theta |y)$ is the target distribution and $h$ is a test function. 
The most important choice of test function is arguably 
$h(\theta) = \theta,$ which 
gives 
the posterior mean estimator of $\theta$.  Other choices of test function enable computation of credible intervals ---regions of the parameter space with prescribed posterior probability--- 
as well as 
higher posterior moments
used for uncertainty quantification. 

Before the 
advent 
of Monte Carlo methods, Bayesian inference was 
computationally feasible 
only
in low-dimensional settings
or under 
the restrictive model assumption of conjugate priors: 
given a likelihood, choosing a conjugate prior ensures that the posterior belongs to the same family of distributions as the prior, whereby point estimators and credible intervals may 
be computed in closed form.
In modern applications, however, 
non-conjugate models are common
and the parameter $\theta$ can be high-dimensional; think, for instance, that 
$\theta$ may represent a vectorized high-resolution image that we wish
to denoise. Then, 
the normalizing constant \index{normalizing constant} $c$ in \eqref{eq:bayesianintro} 
---known as the \emph{marginal likelihood} of the data--- 
is typically unknown and hard to compute, as it is itself defined as an integral 
over
a high-dimensional parameter space. 
 This textbook therefore emphasizes the importance of designing Monte Carlo methods that scale well to high dimension and that are applicable when the normalizing constant of the target distribution is unknown. 

The development of Monte Carlo methods 
goes hand in hand with the popularization of 
large-scale Bayesian statistics.  The Bayesian framework facilitates the formulation of hierarchical models and enables sequential inference through the updating of beliefs as new data arrive.  
Gibbs samplers and particle filters ---two classes of Monte Carlo methods studied in this textbook--- 
dovetail  naturally  with hierarchical and sequential inference, respectively.

\paragraph{Statistical Mechanics\index{statistical mechanics}}
Another important source of challenging sampling problems 
arises in
statistical mechanics, particularly in
the simulation of molecular dynamics.\index{molecular dynamics} Boltzmann postulated that the positions $q$ and momenta $p$ of the atoms in a molecular system of  constant size, occupying a constant volume, and in contact with a heat bath (at constant temperature), are distributed according to
\begin{equation*}\label{eq:boltzmann}
f(q, p) = \frac{1}{c} \exp\Bigl( -\beta \bigl(V(q)+K(p)\bigr)\Bigr), 
\end{equation*}
where $c$ is a normalizing constant \index{normalizing constant} known as the \emph{partition function}, $\beta$ 
denotes 
the inverse temperature,  $V$ is a potential energy describing 
interactions between the particles in the system, and $K$ represents the kinetic energy of the system. The potential $V$ is often a 
highly irregular
function with many local minima. This feature, along with other characteristics typical of molecular systems, makes it difficult to compute averages with respect to $f(q,p).$ Nevertheless, 
many quantities of scientific
interest are defined as expectations under the 
Boltzmann distribution, and 
computing integrals of the form
$$\mathcal{I}_f[h] = \int h(q,p) f(q,p) \, dq\, dp$$
is therefore of central importance.
As in Bayesian statistics, computing the normalizing constant $c$ of the Boltzmann distribution (also known as the Gibbs distribution) is challenging, which motivates again the need to develop Monte Carlo methods for target densities that are only known up to a normalizing constant.

\section{Six Challenges}\label{sec:challenges}
To understand the design and motivation behind different Monte Carlo methods, it is essential to first appreciate the 
main challenges they aim to address.
This section provides a high-level overview of some of these challenges.
 Clearly, no single algorithm can effectively address all of them simultaneously;
otherwise, there would be no need for the 
many types of Monte Carlo methods presented in this textbook! 
Throughout, each new algorithm we introduce will be accompanied by a discussion of its 
pros and cons in relation to addressing different challenges, thus helping to clarify its role within the vast landscape of Monte Carlo methods.

\paragraph{Challenge 1: High Dimensions}
Monte Carlo integration methods remain effective even in high-dimensional settings, provided that one can obtain a sample from the target distribution.
 Unfortunately, sampling from a generic target distribution in high dimensions is often difficult.
On the bright side, there are important classes of target distributions for which scalable sampling algorithms are available. For instance,  log-concave 
distributions can be 
sampled efficiently
using 
algorithms closely related to those developed for convex optimization (Chapter \ref{chap:diffusions}).
As another example, if one can find a 
\emph{proposal distribution} that is 
both close to the target and easy to sample from, then samples from the proposal can be 
transformed into samples from the target (Chapters \ref{chapter1} and \ref{chap:MCintegration}). 

\paragraph{Challenge 2: Evaluating the Target and its Gradient}
 The guiding applications to Bayesian statistics and statistical mechanics in Section \ref{sec:motivatingapplications} illustrate that the target distribution is often known only up to a normalizing constant. In many Bayesian inference problems, evaluating the likelihood function can 
also be challenging, which has motivated 
the development of approximate Bayesian computation,
likelihood-free methods, and Monte Carlo algorithms 
tailored to large datasets
(Chapter \ref{chapter1}). 
Relatedly, some Monte Carlo methods leverage the gradient of the logarithm of the target density to improve the scalability to high dimensions (Chapters \ref{chap:diffusions} and \ref{chap:HMC}), but evaluating these gradients can  be computationally expensive.

\paragraph{Challenge 3: Multi-Modality\index{distribution!multi-modal}}
Many sampling algorithms studied in this textbook are \emph{local}, in the sense that each new draw typically lies in a small neighborhood of the previous one (Chapters \ref{chap:MCMC}, \ref{chap:diffusions}, and \ref{chap:HMC}). For 
these methods, sampling multi-modal target distributions poses a significant challenge, 
especially 
when regions of high target density are separated by regions of low density. In this setting, local 
methods may over-sample around one mode without adequately exploring others.
A variety of strategies have been developed to address this issue. For example, one can introduce
flattened or \emph{tempered} ancillary target distributions to 
 accelerate exploration (Chapter \ref{chap:annealing}), or augment the state space with auxiliary variables  
to define an extended target in a higher-dimensional space with a more benign landscape (Chapters \ref{chap:annealing} and \ref{chap:HMC}). 

\paragraph{Challenge 4: Computing Rare Events\index{rare event}}
A different set of challenges arises when the test function $h(x) = {\bf{1}}_A(x)$ is the indicator function of a small probability event $A.$ In 
this case, a random sample from the target distribution will typically contain no 
draws in $A$, causing the 
Monte Carlo estimate of $\Prob(A)$ to be 
zero. To obtain an estimate with low \emph{relative error},  it is then 
essential to sample from a suitable proposal distribution rather than the target (Chapters \ref{chap:MCintegration} and \ref{chap:particlefilters}). 
How should one select the proposal distribution to minimize the variance of Monte Carlo estimates of rare events? You will find out in Chapter \ref{chap:MCintegration}.

\paragraph{Challenge 5: Choice of Variables and Conditioning of the Target}
For some algorithms, a key challenge is sampling joint distributions with strongly correlated variables (Chapter \ref{chap:gibbs}) or variables with different scales (Chapters \ref{chap:MCMC}, \ref{chap:diffusions}, and \ref{chap:HMC}). 
To mitigate these issues, it is often beneficial to
reparameterize the problem before applying Monte Carlo methods. Other strategies include 
blocking together correlated variables, 
appropriately preconditioning sampling algorithms, and, relatedly, using affine-invariant methods.

\paragraph{Challenge 6: Assessing Convergence\index{Markov Chain Monte Carlo (MCMC)!assessing convergence}}
Assessing the convergence of Monte Carlo methods can be challenging.  Diagnostics based on hypothesis tests can help rule out convergence, but cannot guarantee convergence (Chapter \ref{chap:MCMC}).   
This challenge is particularly acute for multi-modal target distributions, where an algorithm may appear to have converged after thoroughly exploring one mode while failing to visit others.

\section{Ten Ideas}\label{sec:themes}
This section summarizes some key ideas that underlie the design of many Monte Carlo methods. 

\paragraph{Idea 1: Deterministic Transformations} To sample from a given target distribution, deterministic transformation methods rely on a transport map to convert draws from an easy-to-sample proposal distribution into draws from the target (Chapter \ref{chapter1}). Under mild assumptions on 
the target and proposal, many such transport maps exist. Deterministic transformation methods 
 effectively  replace the problem of sampling from the target with the problem of parameterizing and learning a transport map between the proposal and target distributions. 

\paragraph{Idea 2: Probabilistic Accept/Reject Mechanisms}
In their simplest form, probabilistic accept/reject methods rely on a probabilistic mechanism to convert draws from a proposal distribution into draws from a given target (Chapter \ref{chapter1}). A related idea underlies Markov chain Monte Carlo methods (Chapter \ref{chap:MCMC}), where 
a probabilistic accept/reject mechanism 
transforms samples from a proposal Markov kernel into samples from a kernel that satisfies detailed balance with respect to the target. The unifying 
idea is to leverage a probabilistic rule to turn 
easy-to-generate samples into approximate samples from the target distribution.

\paragraph{Idea 3: Weighting Samples} Rather than deterministically transforming or probabilistically accepting/rejecting samples, some Monte Carlo methods assign weights to draws from a proposal distribution to approximate expected values with respect to the target (Chapters \ref{chap:MCintegration} and \ref{chap:particlefilters}).  
Although this approach avoids constructing transport maps or rejecting samples, it comes with an important limitation: the variance of the weights can become large when target and proposal are far apart. 
In such cases, a small number of samples may receive disproportionately large weights, leading to low effective sample size.
This phenomenon, known as weight degeneracy, is especially pronounced in high-dimensional settings.

\paragraph{Idea 4: Markov Chains and Local Moves} Some of the most widely used Monte Carlo methods rely on random samples that are not independent, but that instead form a Markov chain: the distribution of each new random draw depends on the value of the previous one (Chapters \ref{chap:MCMC}, \ref{chap:gibbs}, \ref{chap:diffusions}, and \ref{chap:HMC}).
Positive correlation between samples increases the variance of Monte Carlo estimates (Appendix \ref{chap:markovchains}), but the flexibility of the Metropolis--Hastings framework for Markov chain Monte Carlo makes up for this caveat (Chapter \ref{chap:MCMC}). 
In particular, proposal Markov kernels that exploit 
information about the target distribution and its gradient can accelerate convergence in high-dimensional settings. 

\paragraph{Idea 5: Discretization of (Stochastic) Differential Equations} 
A key idea in designing proposal kernels for Markov chain Monte Carlo algorithms is to leverage 
Langevin stochastic differential equations and Hamiltonian differential equations that preserve the target (Chapters \ref{chap:diffusions} and \ref{chap:HMC}). Time discretizations of these 
continuous-time dynamics yield
natural proposal kernels, which can be, if desired, combined with an accept/reject correction to ensure convergence to the target.

\paragraph{Idea 6: Optimization} Optimization \index{optimization} and sampling share 
deep conceptual similarities, and their interplay
can be leveraged for algorithmic design. 
First, methods for sampling log-concave target distributions are closely related to convex optimization algorithms:  
while gradient descent optimization methods can be viewed as time-discretizations of gradient systems in parameter space, Langevin Monte Carlo  sampling methods can be viewed as discretizations of gradient flows in the space of probability densities (Chapter \ref{chap:diffusions}). 
Second, annealing algorithms for global optimization 
operate by sampling from a sequence of target distributions that become increasingly concentrated
around their shared mode, thus leveraging sampling for optimization (Chapter \ref{chap:annealing}). Finally, variational inference methods seek the closest tractable distribution to a given target, thus solving an optimization problem over densities to approximate integrals with respect to the target (Chapter \ref{chap:optimization}). 

\paragraph{Idea 7: Annealing and Tempering}
To facilitate sampling from
multi-modal distributions, a powerful idea is to introduce auxiliary tempered distributions. A popular strategy in this direction is parallel tempering,
where we run several Markov chains at different temperatures and occasionally swap their states (Chapter \ref{chap:annealing}). Tempering is useful for instance in Bayesian inference, where instead of targeting the posterior with the full data directly, a natural idea is to consider a sequence of intermediate tempered posteriors based on smaller batches of data.

\paragraph{Idea 8: Auxiliary Variables} Many Monte Carlo methods accelerate convergence by introducing auxiliary variables. In this textbook,  
you will study several algorithms that build on this idea, including simulated tempering (Chapter \ref{chap:annealing}), Hamiltonian Monte Carlo (Chapter \ref{chap:HMC}), and the slice sampler (Chapter \ref{chap:gibbs}). Auxiliary variables also play a central role in sequential Monte Carlo methods (Chapter \ref{chap:particlefilters}). 

\paragraph{Idea 9: Coordinate-Wise Updates}
The Gibbs sampler and the coordinate ascent variational inference algorithm are coordinate-wise sampling and optimization algorithms, where joint distributions are iteratively sampled or approximated by updating one coordinate at a time (Chapters \ref{chap:gibbs} and \ref{chap:optimization}). This coordinate-wise structure makes 
these algorithms well suited for 
high-dimensional problems.  
Gibbs samplers are extremely effective for sampling from intractable joint distributions with easy-to-sample conditionals. Coordinate ascent variational inference is particularly appealing in exponential family models, where closed form formulas can be used to update the parameters.

\paragraph{Idea 10: Minorization-Maximization}
This  textbook also covers the expectation maximization algorithm (Chapter \ref{chap:optimization}), 
a prototypical minorization-maximization method. It plays a role in maximum likelihood estimation analogous to that of the Gibbs sampler in posterior sampling. In
particular, the expectation maximization 
algorithm can be used to provide effective initializations for Gibbs samplers.

\section{Course Structure}\label{sec:coursestructure}
 This textbook grew out of a 9-week course on Monte Carlo methods that has been taught regularly at the University of Chicago since 2019. 
The prerequisites for the course are multivariate calculus,  linear algebra, a first course in probability and statistics (including Markov chains), and elementary knowledge of ordinary differential equations.
The course follows the structure of this textbook, covering the material chapter by chapter. To review Markov chain theory in preparation, students are encouraged to read Appendix \ref{chap:markovchains} and consult the standard references on Markov chains and stochastic processes cited there.
\nc 

 Course assessment is based on weekly homework assignments, an independent project, and a final exam. Exercises representative of the weekly assignments appear at the end of each chapter in the published version of this book. In addition, students complete a 10-page independent project, which is intended to
 encourage students to explore a topic of their choice more deeply, with the flexibility to focus on theory, methodology, or applications. The final exam 
 provides students with an opportunity to review the core concepts and algorithms presented in this textbook. 

Numerous variations of this course structure and grading schemes are possible.
For instance, a semester-long course could include a review of Markov chains before introducing the Metropolis--Hastings algorithm in Chapter \ref{chap:MCMC}, and 
the assessment could incorporate hands-on group projects and oral presentations.

\section{Discussion and Bibliography}\label{sec:biblioch1}
One of the earliest probabilistic experiments resembling modern Monte Carlo methods is Buffon’s needle experiment, proposed and solved by Georges-Louis Leclerc, Comte de Buffon, in the eighteenth century \cite{comte1770histoire}. Later work \cite{marquis1812theorie} observed that 
repeated trials of this experiment yield a stochastic procedure for approximating $\pi.$ Estimating $\pi$ by sampling points uniformly in a square, as in Section \ref{sec:motivatingexch1}, is a standard introductory example, used in many expositions of Monte Carlo methods; see, for instance, the textbook \cite{kalos2008monte} and the lecture notes \cite{johansen2010monte}.

Historically, a major limitation of Monte Carlo methods was the cost of carrying out random experiments and processing their results. Physical randomization procedures were difficult to perform repeatedly, and the numerical calculations required to analyze their outcomes were often laborious. As an early attempt to streamline this process, Enrico Fermi studied neutron diffusion using a mechanical device known as the \textit{Fermiac}, which simulated random neutron trajectories through matter.

The situation changed fundamentally with the advent of the digital computer. Among the first to realize its potential for stochastic simulation were John von Neumann and Stanislaw Ulam, who were then working  (along with Fermi)  for the Manhattan Project at Los Alamos National Laboratory; see \cite{eckhardt1987stan}. The development of Monte Carlo simulation required reliable methods for generating random numbers computationally. Early work by von Neumann introduced simple pseudo-random number generators \cite{von195113}, illustrating how randomness could be produced algorithmically rather than through physical processes. 
In 1947 Ulam and von Neumann proposed using computer simulations to study neutron diffusion in fissionable materials. The 1949 paper by Nicholas Metropolis and Ulam \cite{metropolis1949monte} introduced the first unclassified, formal description of the Monte Carlo method. 

At that time ---and for several decades thereafter--- programming was often regarded as routine technical work and was frequently performed by women. Klára Dán von Neumann, the wife of John von Neumann, is credited with writing the code for one of the first computer implementations of a Monte Carlo simulation. 
Interestingly, John and Klára von Neumann first met in a casino in Monte Carlo. However, the name ``Monte Carlo'' is generally attributed to Ulam, who reportedly drew inspiration from his interest in games of chance and from an uncle who frequently borrowed money to gamble in Monte Carlo (unpublished remarks by Ulam, 1983).

Important methodological advances during this early period included the development of variance reduction techniques and importance sampling, which significantly improved the efficiency of Monte Carlo estimators. A major breakthrough occurred in 1953 when Metropolis and collaborators calculated the equation of state via modified Monte Carlo integration over configuration space \cite{metropolis1953equation}. This work established the basic idea of constructing a Markov chain whose stationary distribution coincides with the target distribution of interest. The method was later generalized by Wilfred Keith Hastings in 1970, resulting in what is now known as the Metropolis--Hastings algorithm \cite{hastings1970monte}.

Despite these advances, Monte Carlo methods were not widely used in mainstream statistics until the 1980s and 1990s. A major turning point came with the introduction of the Gibbs sampling method by Stuart Geman and Donald Geman \cite{geman1987stochastic}. Soon afterward, work by Alan Gelfand and Adrian Smith \cite{gelfand1990sampling} demonstrated how Markov chain Monte Carlo methods could be used to perform Bayesian inference in complex statistical models, thereby triggering a rapid expansion of Monte Carlo methodology within statistics.

Since then, the field has continued to develop in several directions. Sequential Monte Carlo methods, or particle filters, were introduced for inference in dynamical systems and state-space models \cite{del2004feynman,chopin2020introduction}. In parallel, Quasi-Monte Carlo methods were developed to improve numerical integration by replacing random samples with carefully constructed low-discrepancy sequences \cite{niederreiter1992random}, 
 and Monte Carlo algorithms based on Hamiltonian and Langevin dynamics grew in importance \cite{duane1987hybrid,neal2012bayesian,besag1994comments,roberts1996exponential}.   
More recently, the core ideas underlying Monte Carlo methods have found applications in generative modeling and modern machine learning, where stochastic simulation and probabilistic inference play a central role.

Monte Carlo and related approximation methods are now implemented in a wide range of software packages across different programming languages. Commonly used examples include Stan, PyMC, JAGS, BUGS/OpenBUGS, NIMBLE, Turing.jl, Pyro, NumPyro, TensorFlow Probability, particles, SMCTC, LibBi, and emcee. These packages support, to varying degrees,  Markov chain Monte Carlo, sequential Monte Carlo, and variational inference algorithms studied in this textbook.

\chapter{Transformation and Accept/Reject Sampling Methods}\label{chapter1}
This chapter is concerned with sampling\index{sampling} methods: algorithms for generating random draws from a given distribution. 
To be concrete, 
suppose that the \emph{target distribution}  
from which we wish
to sample from has p.d.f. $\tg$ on a subset $E \subseteq \R^d.$ The goal is to obtain a \emph{sample} $\X^{(1)},\ldots,\X^{(N)}\in E$ 
such that the draws 
$\X^{(n)}\in E,$ $1\le n \le N,$ are independent and  distributed 
according to the 
p.d.f. $\tg.$ Intuitively, a histogram of the sample should accurately approximate the target p.d.f. provided that the sample size $N$ is sufficiently large.
We focus on transformation and accept/reject sampling methods to showcase two important overarching themes in the development of Monte Carlo methods: the use of deterministic transformations and probabilistic accept/reject mechanisms to transform samples from a tractable distribution into samples from the target distribution. 

Our starting point will be a ``random'' number generator\index{random number generator} to produce independent and uniformly distributed draws $U^{(1)},\ldots,U^{(N)} \sim \text{Unif}(0,1).$ In practice, random number generators are deterministic algorithms;  each \emph{pseudo-random} number\index{pseudo-random number} $U^{(n)}$ is obtained by iteratively applying a deterministic transformation to a deterministically chosen initial seed. While the procedure is deterministic and the resulting sample is not truly random, it imitates the behavior of a uniform sample in that statistical tests for departure from independence and the uniform distribution are not rejected more often than would be expected by chance. Similarly, the sampling methods considered in this chapter are not truly random, but imitate the statistical behavior of a sample drawn from the target distribution.
Throughout this textbook, we assume 
that a \emph{random number generator} is available, that is, computer code that outputs values $U^{(1)},\ldots,U^{(N)}$ that are statistically indistinguishable from an i.i.d. sample from the Unif$(0,1)$ distribution.

When the target distribution is not uniform and does not belong to a standard parametric model (e.g.~normal, exponential, geometric, etc.), sampling can be challenging, especially in high dimensions. In this chapter, we consider two general approaches for sampling: transformation methods (Section \ref{sec:transformation}) and accept/reject methods (Section \ref{sec:acceptreject}).
Both approaches first produce a sample from a \emph{reference}\index{distribution!reference} or \emph{proposal distribution}\index{distribution!proposal} that is tractable in the sense that it is easy to sample from. Transformation methods find a map 
that turns 
draws from the reference distribution into draws from the target distribution;\index{distribution!target} accept/reject methods rely on a probabilistic mechanism to 
achieve the same goal.
The chapter closes in Section \ref{sec: bibliochapter1} with bibliographic remarks. For simplicity, we assume throughout that the target distribution is
either
discrete or continuous, and characterized by a p.m.f. or p.d.f. denoted by $\tg.$ Similarly, the reference distribution will be characterized by a p.m.f. or p.d.f. denoted by $\pr.$

\section{Transformation Methods\index{transformation method}}\label{sec:transformation}
Let $\X \sim f$ be a random variable taking values 
in
a set $E \subseteq \R^d$ and let $U \sim \text{Unif}(0,1)^d$ be 
uniformly distributed on the $d$-dimensional unit cube $(0,1)^d$. In this section, 
we seek a map $\T : (0,1)^d \to E$ such that
the random variable $\T(U)$ has the same distribution as $\X,$ written $\X \stackrel{d}{=} \T(U).$  
Once such a \emph{push-forward} or \emph{transport} map $\mathcal{T}$ is found,  
we can obtain a sample from the target distribution $\tg$ by:
 \begin{enumerate}
     \item Sampling $U^{(1)}, \ldots, U^{(N)} \stackrel{\text{i.i.d.}}{\sim} \text{Unif}(0,1)^d;$ and
     \item Setting $\X^{(n)} := \T(U^{(n)})$ for $1 \le n \le N.$
 \end{enumerate}
A natural question is whether 
 such a transport map\index{transport map} $\T$ 
 exists.
 We will provide an affirmative answer 
 along with an 
 explicit construction  
 under the assumption that the target admits a positive p.d.f. For 
 expository purposes, we consider the scalar case $d=1$
 in Subsection \ref{sssec:inversetransformation}, and then generalize to the $d$-dimensional case in Subsection \ref{sssec:RKrearrangement}.

\subsection{Inverse Transformation Method\index{inverse transformation method}}\label{sssec:inversetransformation}
Let $\X \sim f$ be a real-valued random variable
with c.d.f. $F$.
The inverse transformation method relies on the (generalized) inverse c.d.f. of $X,$ given by 
$$ F^-(u):= \inf \{x: F(x) \ge u\},$$
to transform a uniform sample into a sample from $\tg.$ The following result shows that the inverse c.d.f.
is a valid transport map. 
\begin{theorem}[Transport: Uniform to Target]\label{thm:inversetransf}
Let $X$ be a real-valued random variable with inverse c.d.f. $F^-$ and let $U\sim \emph{Unif}(0,1).$ It holds that $F^-(U) \stackrel{d}{=} \X.$
\end{theorem}
\begin{proof}
We 
prove here only the 
case where $F$ is a bijection from $\R$ to $(0,1).$ In such a case, it holds that $F^-=F^{-1},$ as illustrated in Figure \ref{fig:transformationmethod}.
Then, for any $x \in \R,$
$$\Prob \bigl(F^-(U)\leq x\bigr)=\Prob \bigl(F^{-1}(U)\leq x\bigr)=\Prob\bigl(U\leq F(x)\bigr)=F(x),$$
which shows that $F$ is the c.d.f. of the random variable $F^-(U).$ 
Exercise~2.2 of the published version of this book investigates the case where $F$ is not invertible.
\hfill \qedhere
\end{proof}

  \FloatBarrier
    \begin{figure}[htp]
      \centering
      \includegraphics[width=0.75\columnwidth]{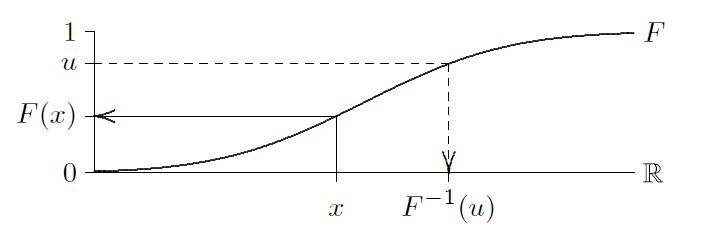}
      \caption{A strictly increasing c.d.f. and associated inverse c.d.f.}
      \label{fig:transformationmethod}
    \end{figure}
    \FloatBarrier

The inverse transformation method produces a sample from the target distribution as follows: 
\FloatBarrier
\begin{algorithm}
\caption{\label{alg:general} Inverse Transformation Method}
\begin{algorithmic}[1]
\STATE {\bf Input}: Target distribution $f$ with inverse c.d.f. $F^-,$ sample size $N.$
\STATE Set $\X^{(n)} := F^-(U^{(n)}),\quad 1\leq n\leq N,$ \quad 
where $ U^{(1)}, \ldots, U^{(N)} \stackrel{\text{i.i.d.}}{\sim} \text{Unif}(0,1)$.\STATE {\bf Output}: Sample $ \{\X^{(n)}\}_{n=1}^N \stackrel{\text{i.i.d.}}{\sim} f.$
\end{algorithmic}
\end{algorithm}

\begin{mybox}[colback=white]{Pros and Cons}
The inverse transformation method is very efficient whenever $F^-$ 
admits
a closed form expression that can be evaluated cheaply.
However, Algorithm \ref{alg:general} is only applicable 
for sampling
real-valued random variables. 
\end{mybox}

\begin{example}[Inverse Transformation Method: Exponential Distribution]
Suppose that $\X\sim \text{Exponential}(\lambda)$, so that $F(x)=1-e^{-\lambda x}, \, x\geq 0$. Let $U \sim \text{Unif}(0,1).$ Then,
$$ -\frac{1}{\lambda}\log(1-U)\sim \text{Exponential}(\lambda).$$
Using that $1-U \stackrel{d}{=} U$, we also have that
$$ -\frac{1}{\lambda}\log(U)\sim \text{Exponential}(\lambda).$$
Figure \ref{fig:inverse_transformation_exponential} shows a histogram of $N = 10^5$ draws generated from an Exponential(1) distribution using the inverse transformation method. \hfill \qedhere
\end{example}

\FloatBarrier

\begin{figure}[!htb]
\minipage{0.5\textwidth}
  \includegraphics[width=0.9\linewidth,height = 4.65cm]{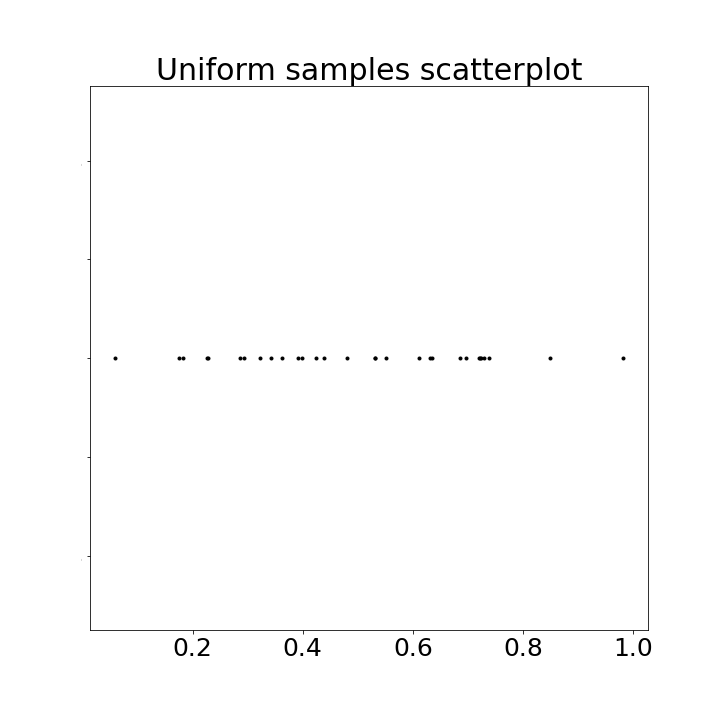}
\endminipage\hfill
\minipage{0.5\textwidth}
  \includegraphics[width=0.9\linewidth,height = 4.65cm]{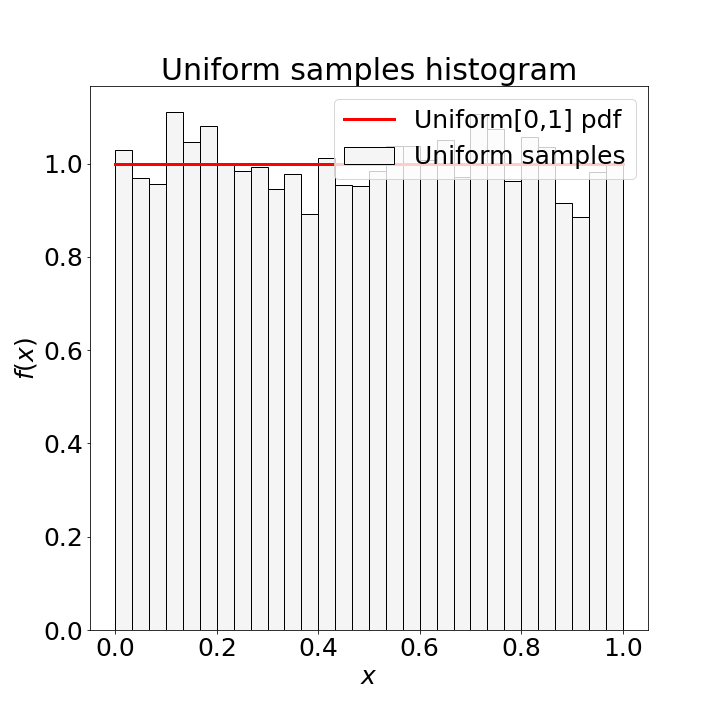}
\endminipage\hfill
\minipage{0.5\textwidth}%
  \includegraphics[width=0.9\linewidth,height = 4.65cm]{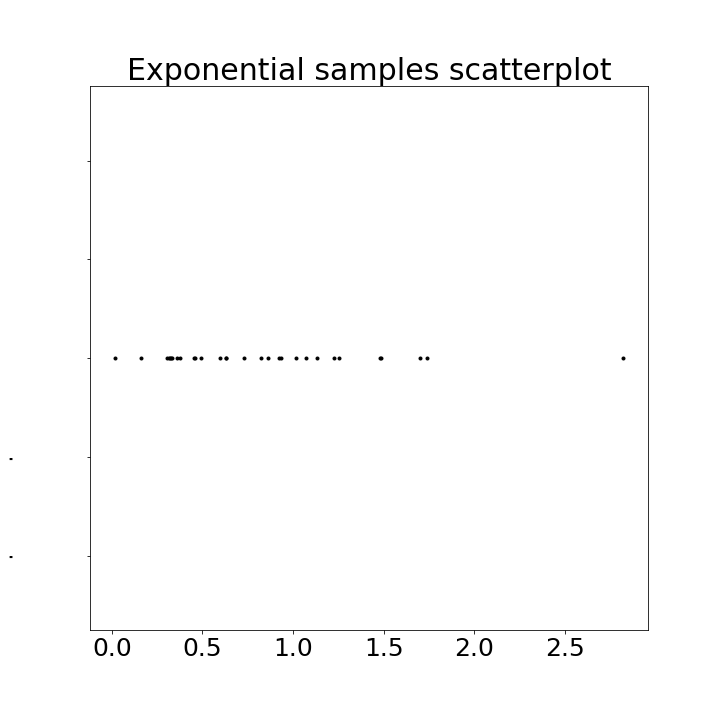}
\endminipage\hfill
\minipage{0.5\textwidth}%
  \includegraphics[width=0.9\linewidth,height = 4.65cm]{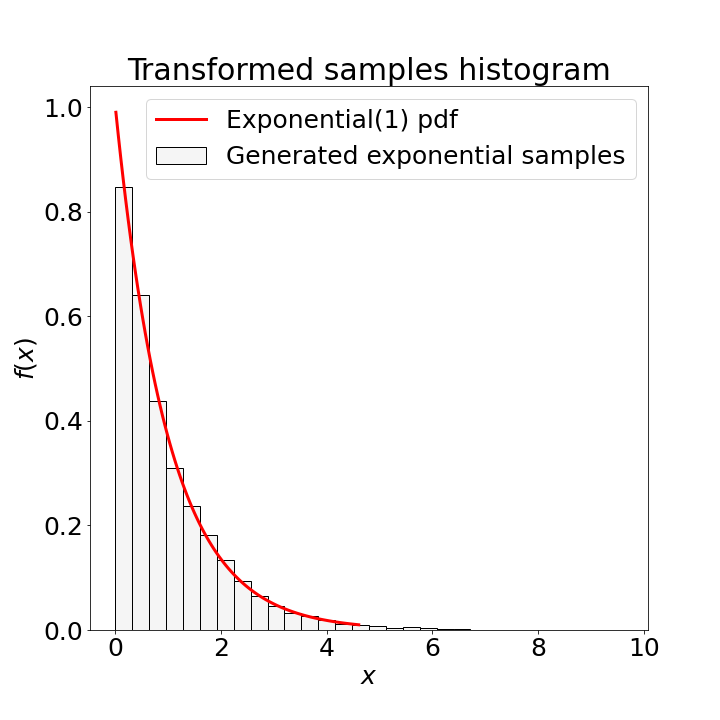}
\endminipage\hfill
\caption{Inverse transformation method for sampling from an Exponential(1) distribution. Uniform samples (first row) are transformed into samples from the target distribution (second row).}
\label{fig:inverse_transformation_exponential}
\end{figure}

\FloatBarrier

So far, we have discussed how to transform a uniform sample into a sample from a given target distribution. In some applications, it is important to use a proposal distribution that is not uniform. The following theorem provides a transport map from proposal to target.

\begin{theorem}[Transport: Proposal to Target]\label{th:transportproposal}
  Let $Z$ be a real-valued random variable with continuous and strictly increasing c.d.f. $F_Z$ and let $X$ be a real-valued random variable with inverse c.d.f. $F_X^-.$   Then, it holds that $F_X^-\circ F_Z(Z)  \stackrel{\text{d}}{=} X.$ 
\end{theorem}
\begin{proof}
By Theorem \ref{thm:inversetransf}, it suffices to show that $F_Z(Z) \sim \text{Unif}(0,1).$
To that end, note that, for any $z \in (0,1),$
\begin{equation*}
  \Prob \bigl(F_Z(Z) \le z \bigr) = \Prob \bigl(Z \le F_Z^{-1}(z) \bigr) = F_Z \bigl(F_Z^{-1}(z)\bigr) = z,
\end{equation*}
as desired. \hfill \qedhere
\end{proof}

Let $X \sim f$ and let $Z \sim g,$ where $g$ is a strictly positive p.d.f. In view of Theorem \ref{th:transportproposal}, we can obtain a draw from $f$ by (1) drawing $Z^{(n)} \sim g;$ and (2) setting $X^{(n)} : = F_X^-\circ F_Z(Z^{(n)}).$

\subsection{Knothe-Rosenblatt Rearrangement}\label{sssec:RKrearrangement}
We now generalize the inverse transformation method to target distributions on $\R^d,$  $d\in \N.$ 
 For ease of exposition, 
 we first consider 
 the case $d = 2.$ Let $X = (X_1,X_2)\sim f$ and let $U = (U_1, U_2) \sim \text{Unif}(0,1)^2.$ The idea is to \emph{disintegrate} the target $f$  as the product of the marginal of $X_1$ 
and the 
 conditional of $X_2$ given $X_1$:
$$f(x) = f_{X_1}(x_1) f_{X_2|X_1 = x_1} (x_2 |x_1), \quad x = (x_1,x_2) \in \R^2.$$

We then sequentially apply the inverse transformation method to the one-dimensional targets $f_{X_1}$ and $f_{X_2 |X_1=x_1}$ as 
detailed next.
 
Denote by $F_{X_1}$ the marginal c.d.f. of $X_1.$ By the inverse transformation method, we know that $F_{X_1}^- (U_1) \stackrel{d}{=} X_1.$ Now, for any $x_1 
\in \R$ denote by $F_{X_2|X_1 = x_1}$ the c.d.f. of $X_2|X_1 = x_1.$ By the inverse transformation method, we know that $F_{X_2|X_1 = x_1}^- (U_2) \stackrel{d}{=} X_2|X_1=x_1.$ To summarize,
define recursively 
maps $\T^1$ and $\T^2$ by
\begin{alignat*}{2}
    \T^1: (0,1)  &\to \R,           \hspace{3.5cm}   \T^2:(0,1)^2 &&\to \R, \\ 
    u_1 &\mapsto F_{X_1}^-(u_1),    \hspace{3cm}  (u_1,u_2) &&\mapsto F_{X_2|X_1 = \T^1(u_1)}^-(u_2).
\end{alignat*}
Then, the triangular map $\T : (0,1)^2 \to \R^2$ given by
    \begin{equation}\label{eq:KR2d}
    \T(u) = 
    \begin{bmatrix*}[l]
    \,\,\T^1(u_1 ) \\
    \,\,\T^2 (u_1 , u_2  ) 
    \end{bmatrix*},\quad \quad u = (u_1,u_2)\in \R^2,
\end{equation}
satisfies
$\T(U) \stackrel{d}{=} X.$ The map $\T$ defined in \eqref{eq:KR2d} is called the \emph{Knothe-Rosenblatt rearrangement}\index{Knothe-Rosenblatt (KR) rearrangement} from $\text{Unif}(0,1)^2$ to $f.$ The construction can be generalized to higher dimensions:

\begin{theorem}[Transport via Triangular Maps]\label{th:transport}
Let $f$ be a strictly positive p.d.f. on $\R^d.$ Let $X\sim f$ and let $U \sim \emph{Unif}(0,1)^d.$ 
For $1 \le i \le d,$ define maps $\T^i$ as follows:\footnote{For vector $u = (u_1, \ldots, u_d) \in \R^d$ and $1\le i \le d,$ we denote by $u_{1:i}:= (u_1, \ldots, u_i) \in \R^i$ the vector defined by the first $i$ coordinates of $u.$ Likewise, for triangular map $\T:\R^d \to \R^d,$ we denote by $\T^{1:i}: \R^i \to \R^i$ the map defined by the first $i$ components of $\T.$} 
\begin{align}
   &\T^1(u_1) = F_{X_1}^-(u_1), \\
   &\T^{i}(u_{1:i}) := F^-_{X_{i}| X_{1:i-1}=\T^{1:i-1}(u_{1:i-1})}(u_i), \quad \quad  2 \le i \le d. 
\end{align}
Then, the triangular map\index{triangular map}
    \begin{equation}
    \T(u) = 
    \begin{bmatrix*}[l]
    \,\,\T^1(u_1 ) \\
    \,\,\T^2 (u_1 , u_2  ) \\
   \quad  \vdots \\
    \,\,\T^d (u_1,\ldots, u_d )
    \end{bmatrix*}, \quad \quad u = \bigl(u_1, \ldots, u_d \bigr) \in \R^d,
\end{equation}
satisfies that $X \stackrel{\text{d}}{=} \T(U).$ Moreover, for all $1\le i \le d,$ the components $\T^i : (0,1)^d \to \R$ are monotonically increasing in the $i$-th variable. 
\end{theorem}

The map $\T$ defined in Theorem \ref{th:transport} is called the Knothe-Rosenblatt (KR) rearrangement from Unif$(0,1)^d$ to $f.$ The map $\T$ is referred to as a \textit{triangular} map since its Jacobian is a (lower) triangular matrix. The KR rearrangement can be used for sampling, as summarized in the following algorithm:

\FloatBarrier
\begin{algorithm}
\caption{\label{alg:KR} Sampling via the KR Rearrangement}
\begin{algorithmic}[1]
\STATE {\bf Input}: KR rearrangement $\T$ from Unif$(0,1)^d$ to $f,$ sample size $N.$
\STATE Set $\X^{(n)} := \T(U^{(n)}),\quad 1\leq n\leq N,$ \quad 
where $U_1, \ldots, U^{(N)}\stackrel{\text{i.i.d.}}{\sim} \text{Unif}(0,1)^d$.\STATE {\bf Output}: Sample $\{\X^{(n)}\}_{n=1}^N \stackrel{\text{i.i.d.}}{\sim}f.$
\end{algorithmic}
\end{algorithm}
\FloatBarrier

\FloatBarrier
\begin{mybox}[colback=white]{Pros and Cons}
As with other transformation methods, sampling via the KR rearrangement is very efficient if we can find and evaluate the KR map.  However, computing the KR rearrangement can be expensive in high dimensions. Conditional independence of the target-reference pair leads to sparsity in the KR map that can be exploited when computing or estimating the KR map. We refer to Section \ref{sec: bibliochapter1} for further discussion on this topic.  
\end{mybox}
\bigskip

Using Theorem \ref{th:transportproposal}, it is possible to generalize the construction of the KR rearrangement to transport a strictly positive proposal p.d.f. $g$ on $\R^d$ into a strictly positive target p.d.f. $f.$ Exercise~2.4 of the published version of this book investigates this generalization.

\section{Accept/Reject Methods\index{accept/reject method}}\label{sec:acceptreject}
In this section, we study rejection sampling\index{rejection sampling} as a prototypical example of an \emph{accept/reject} method. The idea is to sample from a \emph{proposal} (also known as an instrumental or reference)
distribution and discard draws that are unlikely under the \emph{target} distribution of interest.
We denote the target distribution by $\tg$ and the proposal distribution by  $\pr$. Typically, the target is expensive or impossible to sample from directly, but the proposal is easy to sample from. We present rejection sampling in Subsection \ref{ssec:rejectionsampling} and a related accept/reject approach for Approximate Bayesian Computation (ABC) in Subsection \ref{ssec:ABC}.

\subsection{Rejection Sampling}\label{ssec:rejectionsampling}
Before introducing the rejection sampling algorithm, we start with a simple identity and an example, both of which will correspond to a uniform proposal distribution $\pr$ in the subsequent general formulation of  rejection sampling.
 
First, the identity
\begin{equation}\label{eq:identitych1}
    \tg(x)=\int_{0}^{\tg(x)}1 \, du
\end{equation}
shows that $\tg(x)$ is the marginal of a random vector that is uniformly distributed on the area under the density $\tg(x)$, that is, on the region $\bigl\{(x,u):0\leq u\leq \tg(x) \bigr\}.$ In the next example, we use the identity \eqref{eq:identitych1} to develop an algorithm to sample from a beta distribution.

\FloatBarrier

\begin{figure}[!htb]
\minipage{0.5\textwidth}
  \includegraphics[width=\linewidth]{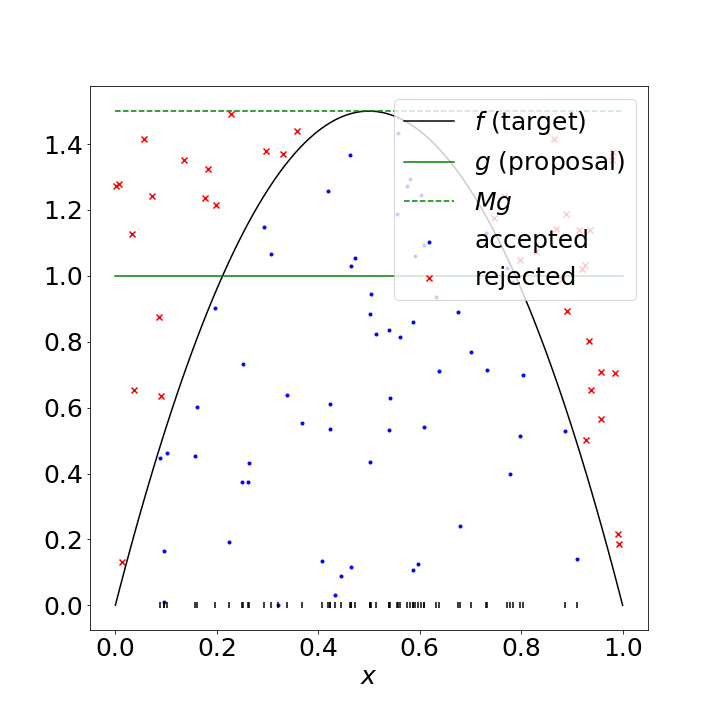}
\endminipage\hfill
\minipage{0.5\textwidth}
  \includegraphics[width=\linewidth]{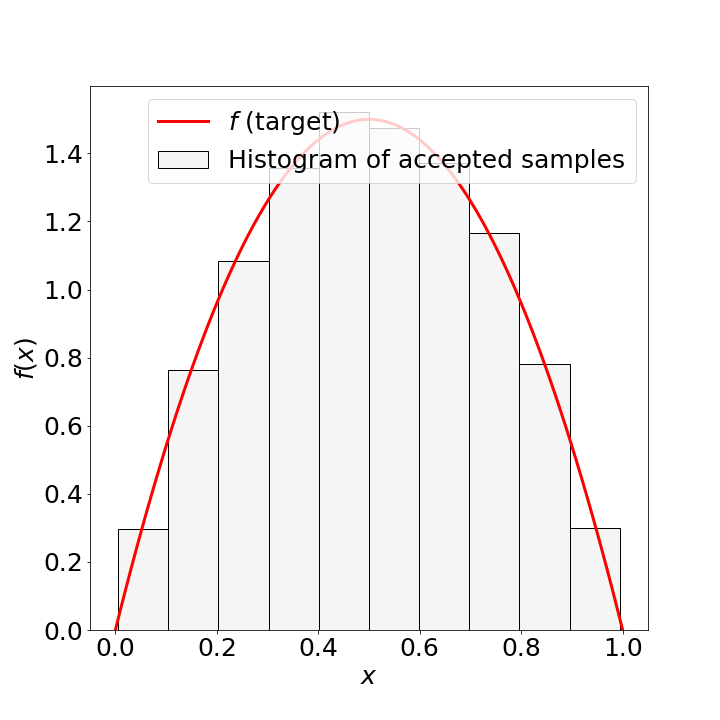}
\endminipage\hfill
\caption{Rejection sampling for a $\text{Beta}(2,2)$ distribution, see Example \ref{ex:example}. The target p.d.f. is given by $\tg(x) = 6x(1-x) \ \text{for} \ x \in (0,1).$ Proposed samples are obtained by uniformly sampling in the rectangle with base $1$ and height $M.$ We accept those samples that lie below the curve of $f.$ The first coordinate of each accepted sample is distributed according to the target.}
\label{fig:rejectionsampling}
\end{figure}

\FloatBarrier
  
\begin{example}[Rejection Sampling: Beta Distribution]\label{ex:example}
Let  $\tg=\text{Beta}(\alpha,\beta)$ so that $f(x) \propto x^{\alpha -1}(1-x)^{\beta -1}$ for $x\in (0,1),$ and suppose that $\alpha,\beta>1.$ The identity \eqref{eq:identitych1} suggests that to obtain a draw $\X^{(n)}\sim f$, we could sample uniformly in the rectangle with base $1$ and height $M$ until we get a draw under the curve $\tg$ (see Figure \ref{fig:rejectionsampling}), and then keep the first coordinate.
The following pseudocode formalizes the procedure: 
\FloatBarrier
\begin{algorithm}
  \begin{algorithmic}[1]
  \STATE For $n=1, \ldots, N$ do:
    \STATE{ \textbf{Step 1:} Generate independently $\Z^{(n)}\sim \text{Unif}(0,1)$ and  $U^{(n)}\sim \text{Unif}(0,1)$, so $(\Z^{(n)},MU^{(n)})$ is uniform on the rectangle $[0,1] \times [0,M]$.}
    \STATE{ \textbf{Step 2:} Set $\X^{(n)}=\Z^{(n)}$ if $\tg(\Z^{(n)})\geq MU^{(n)}$ (that is, if $(\Z^{(n)},MU^{(n)})$ lies under the curve $f$). Otherwise, go back to step 1.}
  \end{algorithmic}
\end{algorithm}
\FloatBarrier
 Any $\X^{(n)}$ defined this way has p.d.f. $f$ as desired.
Note that the conditional probability that $(\Z^{(n)},MU^{(n)})$ is accepted if $\Z^{(n)}=z$ is
$$\Prob\Bigl(MU^{(n)}<\tg(\Z^{(n)}) \Big|\Z^{(n)}=z\Bigr)=\Prob\biggl(U^{(n)}<\frac{\tg(z)}{M}\biggr)=\frac{\tg(z)}{M}.$$
Since $U^{(n)}$ and $\Z^{(n)}$ are independent, we can rewrite the above procedure as follows:

\FloatBarrier
\begin{algorithm}
  \begin{algorithmic}[1]
  \STATE For $n=1, \ldots, N$ do:
    \STATE{\textbf{Step 1:} Draw $\Z^{(n)}\sim \text{Unif}(0,1).$}
    \STATE{\textbf{Step 2:} Set $\X^{(n)}:=\Z^{(n)}$ with probability $\frac{\tg(\Z^{(n)})}{M}$. Otherwise, go back to step 1.}
  \end{algorithmic}
\end{algorithm}
\FloatBarrier

\begin{remark}\label{rem:independence}
Notice that here and throughout this textbook, we  will on occasion omit independence requirements for draws whenever it is clear from context. In particular, in the procedure above the draws $Z^{(n)}$ are independent. 
\end{remark}
\hfill \qedhere
\end{example}

Example \ref{ex:example} uses 
the fact
that the density $\text{Beta}(\alpha,\beta), \alpha,\beta>1$ is bounded within a (finite) rectangle. The proposal distribution can then be taken to be uniform. We now generalize the procedure to cases where the range or the domain of $\tg$ 
is
unbounded. To that end, we use a proposal $\pr$ satisfying the following conditions:
\begin{itemize}
\item $\pr$ can be sampled from and evaluated.
\item $\pr$ has compatible support with $\tg$ (i.e. $\pr(x)>0$ when $\tg(x)>0$).
\item There is $M>0$ such that, for all  $x\in E,$ $\tg(x)\leq M \pr(x).$
\end{itemize}
Notice that for the target $\tg = \text{Beta}(2,2)$ in the example in Figure \ref{fig:rejectionsampling}, the choices $\pr = \text{Unif}(0,1)$ and $M = 3/2$ satisfy these conditions. We are ready to introduce the rejection sampling algorithm: 

\FloatBarrier
\begin{algorithm}
  \caption{Rejection Sampling\label{algo:rejection}}
  \begin{algorithmic}[1]
  \STATEx{ \textbf{Input:} Target $\tg,$  proposal $\pr,$  constant $M>0$ such that, for all $x \in E,$ $\tg(x)\leq M \pr(x)$, sample size $N.$}
  \STATEx For $n = 1, \ldots, N$ do:
  \STATE{\textbf{Step 1:} Draw $\Z^{(n)} \sim g$ independently of all other draws.}
  \STATE{\textbf{Step 2:} Set $\X^{(n)}:=\Z^{(n)}$ with probability $\frac{\tg(\Z^{(n)})}{M\pr(\Z^{(n)})}$. Otherwise, go to step 1.}
  \STATEx{\textbf{Output:} Sample  $\{X^{(n)}\}_{n=1}^N \stackrel{\text{i.i.d.}}{\sim} f.$}
    \end{algorithmic}
\end{algorithm}
\FloatBarrier

For a convenient implementation of steps 1 and 2, we may proceed as follows:
\FloatBarrier
\begin{algorithm}
  \begin{algorithmic}[1]
    \STATE{\textbf{Step 1:} Draw $\Z^{(n)}\sim g$ and $U^{(n)}\sim \text{Unif}(0,1)$ independently.}
    \STATE{\textbf{Step 2:} Accept $\X^{(n)}=\Z^{(n)}$ if $U^{(n)}\leq \frac{\tg(\Z^{(n)})}{M\pr(\Z^{(n)})}$. Otherwise, go to step 1.}
  \end{algorithmic}
\end{algorithm}
\FloatBarrier

\FloatBarrier
\begin{mybox}[colback=white]{Pros and Cons}
 Rejection sampling can be used when $\tg$ and $\pr$ are only known up to a normalizing constant (see Exercise~2.5 of the published version of this book). The method is very versatile and can be applied 
 on a
 general state space $E,$ beyond the Euclidean setting considered here.
Potential disadvantages include that (i) draws from the proposal that are not accepted are thrown away; (ii) the time to obtain the sample $\{X^{(n)}\}_{n=1}^N$ is random; and (iii) the method suffers from the curse of 
dimension, 
as will be discussed below. 
\end{mybox}
\FloatBarrier

The next result shows that the rejection sampling algorithm produces draws distributed according to the target. It also characterizes the acceptance rate under the assumption that $\tg$ and $\pr$ are normalized densities; otherwise, the second and third assertions need to be adjusted (see Exercise~2.5 of the published version of this book).

\bigskip

\begin{theorem}[Rejection Sampling Acceptance Rate]\label{th:rejectionsampling} \leavevmode
  \begin{enumerate}
    \item The draws $\X^{(n)}$ generated by the rejection sampling algorithm have distribution $\tg$ and are independent.
    \item The probability that a draw $Z^{(n)}$ generated in step 1 is accepted in step 2 is $1/M$.
    \item The algorithm will 
    terminate
    in finite time (with probability one). The expected number of iterations to accept each draw $\X^{(n)}$ is $M$.
  \end{enumerate}
\end{theorem}
\begin{proof}
First, the independence of the $X^{(n)}$ follows from the fact that the $Z^{(n)}$ are drawn independently. 
Let now $A\subseteq E$. Note that
\begin{align*}
  \Prob(\Z^{(n)}\in A\ \& \ \Z^{(n)}\text{ is accepted})&=\int_{A}\frac{\tg(x)}{M\pr(x)}\pr(x) \, dx=\frac{\int_{A}\tg(x) \, dx}{M},
\end{align*}
where $\frac{\tg(x)}{M\pr(x)}$ is the probability that $\Z^{(n)}$ is accepted given that $\Z^{(n)}=x$ and $\pr$ is the density of $\Z^{(n)}$. Setting $A=E$ and noting that $\tg$ integrates to 
one  
gives
 \begin{equation*}
  \Prob(\Z^{(n)}\text{ is accepted})=\frac{1}{M},
\end{equation*}
proving the second claim in the statement. Now,
\begin{align*}
  \Prob(\X^{(n)}\in A)=\Prob(\Z^{(n)}\in A | \Z^{(n)}\text{ is accepted})&=\frac{\Prob(\Z^{(n)}\in A\ \& \ \Z^{(n)}\text{ is accepted})}{\Prob(\Z^{(n)}\text{ is accepted})}\\
  &=\frac{\frac{\int_{A}\tg(x) \, dx}{M}}{\frac{1}{M}}=\int_{A}\tg(x) \, dx,
\end{align*}
which proves the first claim. 
The third claim follows by noting that the number of iterations required to accept a draw is distributed according to a geometric random variable with parameter $1/M$, and hence has expectation $M$.
\end{proof}

\paragraph {Choice of Upper-Bound}
 Theorem \ref{th:rejectionsampling} shows that if $\tg$ and $\pr$ are p.d.f.s  such that $\tg(x)\leq M \pr(x)$ for every $x \in E$, then Algorithm \ref{algo:rejection} has acceptance rate $1/M$. Therefore, given $\tg$ and $\pr$, $M$ should be chosen as small as possible while still satisfying the requirement that $\tg(x)\leq M \pr(x)$ for every $x\in E$. 
 In this way,
 the acceptance rate is maximized. Note that it always holds that $M\ge 1,$ since $ 1 = \int_E f(x) \, dx \le M \int_E g(x) \, dx = M.$ 

\paragraph{Choice of Proposal}
We want our proposal to be (1) easy to sample from; and (2) close to the target, so that we can find a small constant $M\ge 1$  satisfying $\tg(x)\leq M \pr(x)$ for every $x\in E.$
These two goals are often in conflict,
and a compromise must be found. In the extreme case where $\tg=g$ one could choose $M=1$ (lowest possible bound), which is clearly not optimal if $\tg$ is extremely expensive (or impossible) to sample from directly! 
A common strategy is to look for proposals in some parametric family \{$\pr_\theta, \theta \in \Theta\}.$ If for each $\theta \in \Theta$ we can find $M_\theta$ such that $\tg(x)\leq M_\theta \pr_\theta(x)$ for every $x\in E,$ then it is recommended to use as proposal the density with parameter $\theta^\star$ that minimizes $M_\theta$ over $\theta \in \Theta.$

\begin{remark}
Note that if $E\subseteq \R^d$ is unbounded and $\tg(x)\leq M \pr(x)$ for every $x\in E,$ then $\pr$ necessarily has 
heavier
tails than $\tg$. Several Monte Carlo methods (e.g. importance sampling in Chapter \ref{chap:MCintegration}) share this principle that the proposal should have 
heavier 
tails than the target.
\end{remark}

\paragraph{Envelope Rejection Sampling\index{rejection sampling!envelope}}
If evaluating the target density $\tg$ is expensive, delaying the acceptance can improve the efficiency. Specifically, let $\ell$ be a function that is cheap to evaluate and such that $\ell(x) \le \tg(x)$ for all $x\in E.$ Then, we can modify the rejection sampling algorithm by delaying acceptance as follows:
\FloatBarrier
\begin{algorithm}
  \begin{algorithmic}[1]
\STATE{ \textbf{Step 1:} Sample $\Z^{(n)}\sim g$, $U^{(n)}\sim$ Unif$(0,1)$.}
\STATE{ \textbf{Step 2:} Accept  $\X^{(n)}=\Z^{(n)}$ if $U^{(n)}\leq \frac{\ell(\Z^{(n)})}{M\pr(\Z^{(n)})}$. Otherwise, go to step 3.}
\STATE{ \textbf{Step 3:} Accept  $\X^{(n)}=\Z^{(n)}$ if $U^{(n)}\leq \frac{\tg(\Z^{(n)})}{M\pr(\Z^{(n)})}$. 
Otherwise, go to step 1.}
  \end{algorithmic}
\end{algorithm}
\FloatBarrier

\paragraph{Rejection Sampling and the Curse of Dimension}
Suppose that, for a given target $\tg$ and proposal $\pr,$ we have found a constant $M$ with $\tg(x)\leq M \pr(x)$ for every $x\in E=\R$. We know that the probability of accepting a sample from $\pr$ in the rejection sampling algorithm is $\frac{1}{M}$. Now consider a $d$-dimensional i.i.d. setting, where $\tg_d$ and $\pr_d$ are p.d.f.s on $E=\R^d$ defined by
\begin{alignat*}{3}
      \tg_d(x ) :=\tg(x_1)\times \dotsb \times \tg(x_d),  \quad &&x  = (x_1, \ldots, x_d) \in \R^d,  \\
      \pr_d( x ) :=\pr(x_1)\times \dotsb \times \pr(x_d), \quad &&x  = (x_1, \ldots, x_d) \in \R^d.
\end{alignat*}
To implement rejection sampling with target $\tg_d$ and proposal $\pr_d$ we can only guarantee that, for every $x \in \R^d,$ the following upper bound holds:
\begin{equation*}
\frac{f_d(x)}{g_d(x )}=\frac{\tg(x_1)\times \dotsb \times \tg(x_d)}{\pr(x_1)\times \dotsb \times \pr(x_d)}\leq M^d.
\end{equation*}
Thus, Theorem \ref{th:rejectionsampling} shows that the probability of accepting each draw from $\pr_d$ decreases exponentially with $d$. In other words, the expected number of iterations needed to obtain an accepted draw  grows exponentially with $d$.

\subsection{Approximate Bayesian Computation} \label{ssec:ABC}
As discussed in Section \ref{sec:motivatingapplications}, Bayesian 
inference on a parameter $\theta$ given data $y$ is based on the posterior distribution $f(\theta|y).$ The posterior is obtained by combining a likelihood model $f(y| \theta)$ and a prior density $f(\theta)$ using Bayes' formula:
\begin{equation*}\label{eq:bayesian}
f(\theta| y) = \frac{1}{c} f(y|\theta) f(\theta),
\end{equation*}
where $c$ is a normalizing constant so that $f(\theta|y)$ integrates to one. Approximate Bayesian Computation (ABC)\index{Approximate Bayesian Computation (ABC)} refers to a family of methods designed to obtain \emph{approximate} posterior samples when the likelihood function is costly to evaluate but easy to sample from. 
ABC was first introduced in population genetics, where one is interested in recovering the history of a population from genetic data. The history is represented by a genealogy tree, which is easy to sample from given parameters such as mutation and growth rates, but the likelihood of the corresponding model is intractable. Likelihood-free methods, which
avoid evaluating the likelihood, are important in many applications beyond population genetics, including epidemiology and cosmology.

The key idea behind ABC is to replace the evaluation of likelihoods with a comparison between sampled and observed data. This comparison is based on a choice of distance and tolerance. In its original and most basic form, ABC is the following accept/reject method:
\FloatBarrier
\begin{algorithm}
  \caption{ABC Rejection Sampler\label{ABC}}
  \begin{algorithmic}[1]
   \STATEx{\textbf{Input:} Prior $\tg(\theta)$, likelihood $f(y|\theta)$, distance $d$, tolerance $\eps,$ data $y_o,$ sample size $N.$}
\STATEx   For $n = 1, \ldots, N$ do:
    \STATE{\textbf{Step 1:}} {Sample $\theta^{(n)}\sim \tg(\theta)$ from the prior.}
    \STATE{\textbf{Step 2:}} {Sample data $y^{(n)}\sim \tg( y | \theta^{(n)})$ from the likelihood}.
    \STATE{ If $d(y_o, y^{(n)})\leq \eps$, accept $\theta^{(n)}$.
    Otherwise, go to step 1.}
    \STATEx{\textbf{Output:}  $\{\theta^{(n)}\}_{n=1}^N \stackrel{\text{i.i.d.}}{\sim} \tg(\theta | d(y_o, y_\text{sim})\leq \eps)$, where $y_\text{sim} \sim f(y).$}
  \end{algorithmic}
\end{algorithm}
\FloatBarrier

\begin{mybox}[colback=white]{Pros and Cons}
ABC is useful when the likelihood is easy to sample from but expensive to evaluate. For some Bayesian inference problems, ABC and related likelihood-free algorithms are virtually the only 
feasible 
methods. Some caveats of ABC include that (i) it only produces approximate posterior samples; (ii) theoretical developments are scarce; and (iii) the tolerance and the distance are often chosen heuristically, and the results can be sensitive to these choices.
\end{mybox}

\paragraph{Choice of Distance}
It is often impractical to define a  distance $d(y_o, y^{(n)})$ between the full data-sets. Instead, it is preferable to use lower-dimensional sufficient summary statistics $S(y_o)$ and $S(y^{(n)})$ of the data-sets and define \begin{equation*}
  d(y_o, y^{(n)})=\tilde{d}\bigl(S(y_o), S(y^{(n)})\bigr),
\end{equation*}
where $\tilde{d}$ is a distance function defined on the summary statistic space.

\paragraph{Choice of Tolerance}
In choosing $\eps$, there is a trade-off between computational burden and approximation of the posterior. Smaller $\eps$ decreases the acceptance probability, but leads to samples whose distribution is closer to the posterior. This trade-off will be illustrated in the following example.

\begin{example}[ABC: Beta-Binomial Bayesian Model]
\label{ABC-Binomial}
Consider a random experiment which involves $n_{o}$ flips of a coin with unknown probability of heads $p$.
Letting $y$ denote the number of heads observed out of the $n_{o}$ flips, we have a $\text{Binomial}(n_{o},p)$ likelihood function
\begin{align*}
\tg(y | p) = {n_{o} \choose y} p^y (1-p)^{n_{o}-y}.
\end{align*}
Taking a $\text{Beta}(\alpha,\beta)$ prior for $p,$ with p.d.f. $f(p) \propto p^{\alpha - 1}(1-p)^{\beta - 1},$ $p \in (0,1),$
leads to a posterior distribution
\begin{align*}
\tg(p | y) &\propto \tg(y | p) f(p) \\
&\propto p^{\alpha + y - 1}(1-p)^{n_{o} + \beta - y -1}, \quad \quad 0 < p <1.
\end{align*}
That is, the posterior of $p$ given $y$ is a $ \text{Beta}(\alpha + y, n_{o} + \beta - y)$ distribution. For the purpose of this example, we choose $\alpha = \beta = 1$, which corresponds to a Unif(0,1) prior.

Next, for the sake of the example, assume that we do not know the true posterior distribution of $p$. We will approximate it using ABC. To that end, we simulate the model using a sample size of $n_{o} = 500$ with $p = 0.4$ and record the number $y_o$ of observed successes. We seek to approximate the posterior of $p$ given $y_o.$ We repeatedly sample values $p^{(n)}$ from the prior distribution $\text{Beta}(1,1)$, and simulate the number of successes $y^{(n)} \sim \tg( \, \cdot \, | p^{(n)})$. We next define a distance function,
\begin{align*}
d(y_o, y^{(n)}) := \frac{1}{n_{o}} |y^{(n)}-y_o|.
\end{align*}
Using three tolerance values $\eps = 0.01, 0.02, 0.05$, we accept the sample $p^{(n)}$ if $d(y^{(n)},y_o) < \eps$. For each choice of tolerance, we continue this procedure until we obtain $N = 1000$ accepted samples. Finally, for each tolerance, we plot a histogram of the accepted samples against the true posterior distribution, which is $\text{Beta}(1 + y_o, n_{o} + 1 - y_o)$. Figure \ref{fig:ABCtolerance} shows the results, along with the  number of samples generated in each case  to reach $N=1000$ accepted samples.  \hfill \qedhere
\end{example}

\begin{figure}[!htb]
\minipage{0.33\textwidth}
  \includegraphics[width=\linewidth]{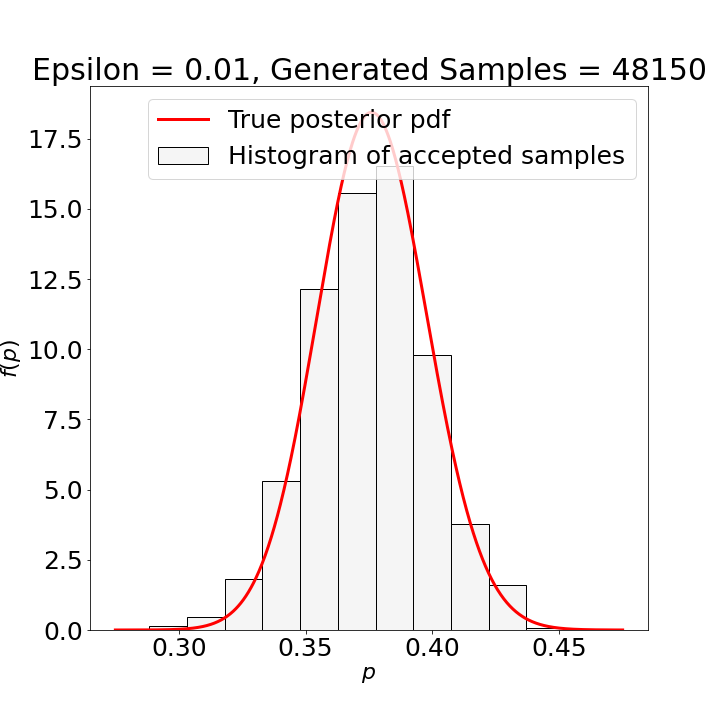}
\endminipage\hfill
\minipage{0.33\textwidth}
  \includegraphics[width=\linewidth]{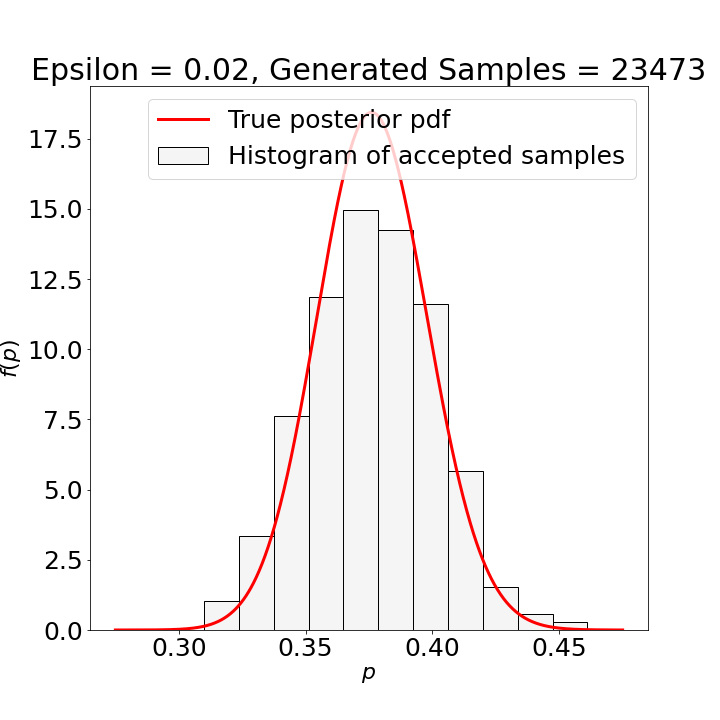}
\endminipage\hfill
\minipage{0.33\textwidth}%
  \includegraphics[width=\linewidth]{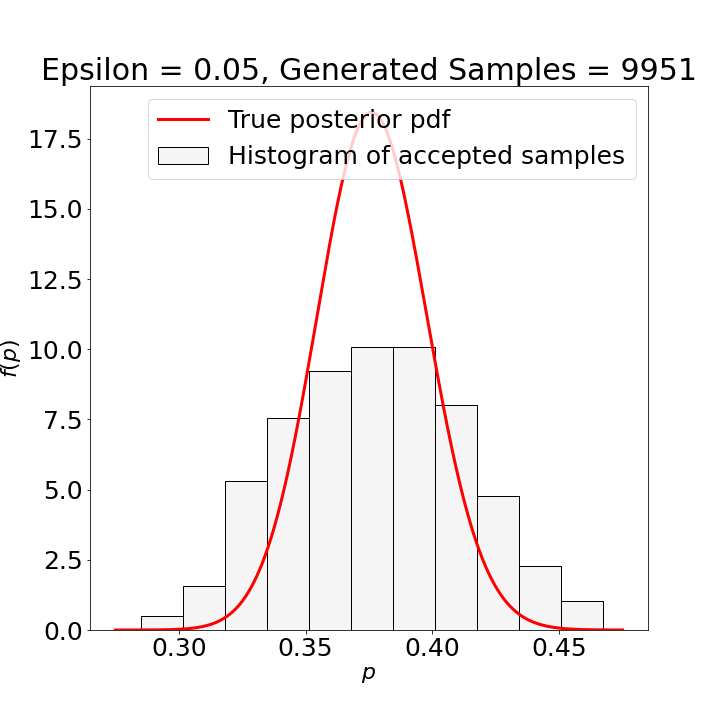}
\endminipage\hfill
\caption{\label{fig:ABCtolerance} ABC for a Beta-binominal model as described in Example \ref{ABC-Binomial}. The results illustrate the trade-off between cost and accuracy in the choice of tolerance.}
\end{figure}

\section{Discussion and Bibliography}\label{sec: bibliochapter1}
In this chapter, we took as our starting point a random number generator\index{random number generator} to produce uniform draws. The study of random number generators is a rich subject in 
its own right; 
we refer to  \cite{gentle2003random,knuth1997art,l2012random} for in-depth surveys of existing techniques and philosophical discussions on the use of deterministic algorithms for stochastic simulation.

Transformation methods and accept/reject methods for non-uniform sampling are reviewed in many textbooks, see e.g.  \cite{devroye2013non,robert2013monte}. The first algorithmic use of rejection sampling is often credited to Von Neumann \cite{von195113} although the idea 
already underlies Buffon's needle experiment.
The paper \cite{von195113} also considers inversion from a uniform proposal distribution. Another early example of a transformation method is the Box-Muller algorithm \cite{box1958note} to sample from a normal distribution; Exercise~2.3 of the published version of this book investigates its derivation. The KR rearrangement was introduced in \cite{knothe1957contributions,rosenblatt1952remarks}. We refer to \cite{ambrosio2008gradient} for further background on
disintegration of measures. More recently,  \cite{marzouk2016sampling} investigates the use of transport maps for sampling. Triangular transport maps also form the building blocks of many complex architectures for \emph{normalizing flows} in machine learning \cite{papamakarios2021normalizing}.

We refer to
\cite{turner2012tutorial} for an accessible introduction to ABC. The main ideas behind ABC can be traced back at least to \cite{rubin1984bayesianly}, but it was not until the late 90s that ABC methods received their name and gained popularity through a series of influential works, including \cite{tavare1997inferring,pritchard1999population,beaumont2002approximate}. The ABC rejection sampler in 
Algorithm \ref{ABC} was introduced in 
\cite{pritchard1999population}. Implementations of ABC using Markov chain Monte Carlo algorithms  and particle filters were introduced in \cite{marjoram2003markov,Sisson1760}.
ABC originates in the population genetics literature, where models are inherently complex and likelihoods are usually prohibitively costly to evaluate or impossible to compute. ABC techniques are also widely used in genetics \cite{genetics,genetics1,genetics2},  epidemiology \cite{epidemiology1}, human demographic history \cite{hamilton}, phylogeography \cite{becquet}, and many other applications. From a theoretical viewpoint, an influential paper is
\cite{wilkinson2013approximate}, which shows that ABC gives exact inference for a wrong model. ABC methods for Bayesian model choice have been extensively investigated, see for instance
\cite{robert2011lack} and references therein. An active area of research relies on conditional density estimation for likelihood-free inference, see \cite{papamakarios2016fast} and also \cite{papamakarios2017masked}. These techniques have the advantage of not relying on heuristically chosen tolerance levels. 

\chapter{Monte Carlo Integration and Importance Sampling}
\label{chap:MCintegration}

This chapter introduces Monte Carlo methods to compute integrals of the form
$$\mathcal{I}_f[h]:=\int_{E}h(x)f(x) \, dx = \mathbb{E}_{X \sim f} \bigl[h(X)\bigr],$$
where $f$ is a p.d.f. supported on a set $E \subseteq \R^d$ and $h : E \to \mathbb{R}$ is a real-valued \emph{test function}. For instance, taking 
$h(x) = {\bf{1}}_A(x),$
the indicator function of a subset $A \subset E,$ we 
can then compute 
the probability $\Prob_{X \sim f}(X \in A).$ The methods in this chapter also apply to discrete target distributions and to vector-valued test functions, and are thus useful, for example, to compute moments of discrete or continuous target distributions.

The main idea behind Monte Carlo integration is to compute $\mathcal{I}_f[h]$ using a random sample. 
We consider two strategies:
\begin{enumerate}
    \item \emph{Classical Monte Carlo integration} relies on a sample $X^{(1)}, \ldots, X^{(N)} \stackrel{\text{i.i.d.}}{\sim} f$ from the target distribution\index{distribution!target} to define 
    the
    estimator 
    $$ \mathcal{I}_f^{\text{\tiny MC}}[h]:= \frac{1}{N} \sum^{N}_{n=1} h(X^{(n)}) \approx \mathcal{I}_f[h]. $$
    \item \emph{Importance sampling} relies on a sample $X^{(1)}, \ldots, X^{(N)} \stackrel{\text{i.i.d.}}{\sim} g$ from a proposal distribution\index{distribution!proposal} $g \neq f$ to define 
    the
    estimator 
    \begin{align}\label{auximportanceestimator}
        \mathcal{I}_f^{\text{\tiny IS} }[h] := \frac{1}{N}  \sum_{n=1}^{N} w(X^{(n)})  h(X^{(n)})\approx\mathcal{I}_f[h].
    \end{align}
    Each draw $X^{(n)}\sim g$ is 
    assigned
    an \emph{importance weight} 
    $w(X^{(n)}) := f(X^{(n)})/ g(X^{(n)})$ determined by the ratio between target and proposal. 
\end{enumerate}

We will study classical Monte Carlo in Section \ref{ssec:classicalMC} and importance sampling in Section \ref{ssec:importancesampling}.  There are two main motivations 
for using
importance sampling: (i) it can reduce the variance of classical Monte Carlo; and (ii) it enables Monte Carlo integration in cases where sampling from $f$ is intractable. In addition to the importance sampling estimator in \eqref{auximportanceestimator}, 
we will 
introduce an \emph{autonormalized} estimator that is applicable when the target and proposal distributions can only be evaluated up to a normalizing constant. Section \ref{sec:variancereduction} discusses other variance reduction techniques and Section \ref{sec:bibliochapter2} closes with bibliographical remarks.

\section{Classical Monte Carlo}\label{ssec:classicalMC}
Classical Monte Carlo integration\index{classical Monte Carlo integration} presupposes that we can obtain $N$ independent draws from the target distribution. 
The integral of interest $\mathcal{I}_f[h]$ is then 
approximated by the average of the values of $h$ at those draws. We refer to the method, summarized in Algorithm \ref{alg:MCI} below, as \emph{classical} Monte Carlo to distinguish it from other Monte Carlo methods.  In the literature it is often simply called Monte Carlo, standard Monte Carlo, or vanilla Monte Carlo. 
\begin{algorithm}[H]
\caption{Classical Monte Carlo Integration}
\label{alg:MCI}
\begin{algorithmic}[1]
 \STATEx{ \textbf{Input:} Target distribution $f,$ test function $h,$ sample size $N.$}
\STATE Sample $X^{(1)}, \ldots, X^{(N)} \stackrel{\text{i.i.d.}}{\sim} f.$
			  \STATEx{\textbf{Output:} Monte Carlo estimator $\mathcal{I}_f^{\text{\tiny MC}}[h]:= \frac{1}{N} \sum^{N}_{n=1} h(X^{(n)}) \approx \mathcal{I}_f[h].$}
\end{algorithmic}
\end{algorithm}

The following result shows that the classical Monte Carlo estimator is unbiased and that its mean squared error,\index{mean squared error} which agrees with its variance, is of order $1/N.$

\begin{theorem}[Classical Monte Carlo Error]\label{thm2.1} The following holds:
\begin{enumerate}
\item $\mathbb{E} \bigl[\mathcal{I}_f^{\emph{\tiny MC}}[h] \bigr]=\mathcal{I}_f[h];$ 
\item $\mathbb{V} \bigl[\mathcal{I}_f^{\emph{\tiny MC}}[h] \bigr]=\frac{1}{N} \mathbb{V}_{X \sim f} \bigl[h(X)\bigr].$
\end{enumerate} 
\end{theorem}

\begin{proof}
For the first claim, linearity of expectation and the fact that $X^{(n)} \sim f$ give 
\begin{align*}
\Expect\bigl[\mathcal{I}_f^{\text{\tiny MC}}[h] \bigr]  &= \mathbb{E}\biggl[ \frac{1}{N}\sum_{n=1}^N  h\bigl(X^{(n)} \bigr)\biggr]  
=\frac{1}{N} \sum _{n=1}^{N} \mathbb{E} \bigl[h\bigl(X^{(n)} \bigr) \bigr]=\mathbb{E}_{X \sim f} \bigl[h(X)\bigr]=\mathcal{I}_f[h].
\end{align*}
For the second claim, note that since $X^{(n)} \stackrel{\text{i.i.d.}}{\sim}  f,$ 
\begin{equation*}
\mathbb{V} \bigl[\mathcal{I}_f^{\text{\tiny MC}}[h] \bigr] =\frac{1}{N^2}\sum_{n=1}^N \mathbb{V} \bigl[h\bigl(X^{(n)} \bigr)\bigr] 
=\frac{1}{N^2}N\mathbb{V}_{X \sim f} \bigl[h(X)\bigr]=\frac{1}{N}\mathbb{V}_{X \sim f}\bigl[h(X)\bigr].   \tag*{\qedhere}
\end{equation*}
\end{proof}

Theorem \ref{thm2.1} can be used to construct confidence intervals\index{confidence interval} for $\mathcal{I}_f[h]$ based on $\mathcal{I}_f^{\text{\tiny MC}}[h].$
To do so, note that by the central limit theorem, as $N \to \infty$,
\begin{equation}\label{eq:MonteCarloCLT}
    \sqrt{N} \, \frac{\mathcal{I}_f^{\text{\tiny MC}}[h]-\mathcal{I}_f[h]}{\sqrt{\mathbb{V}_{X \sim f}[h(X)]}} \Rightarrow \mathcal{N}(0,1).
\end{equation}
For $\alpha \in (0,1)$ define the $z$-score $z_{\alpha/2}$ 
by the requirement that 
$\Prob( - z_{\alpha/2} < Z < z_{\alpha/2}) = 1 - \alpha,$ where $Z \sim \Nc(0,1).$ Then, it follows from \eqref{eq:MonteCarloCLT} that
\begin{equation*}
   \biggl( \mathcal{I}_f^{\text{\tiny MC}}[h] - z_{\alpha/2} \sqrt{ \frac{\mathbb{V}_{X \sim f}[h(X)]}{N}}, \mathcal{I}_f^{\text{\tiny MC}}[h] + z_{\alpha/2} \sqrt{ \frac{\mathbb{V}_{X \sim f}[h(X)]}{N}} \biggr) 
\end{equation*}
is an approximate $1-\alpha$ confidence interval for $\mathcal{I}_f[h].$
The variance term can itself be approximated by Monte Carlo:
\begin{align}\label{eq:varianceMonteCarlo}
\begin{split}
\mathbb{V}_{X \sim f} \bigl[h(X)\bigr] =\mathbb{E}\Bigl[ \bigl(h(X)-\mathbb{E}_{X \sim f} \bigl[h(X) \bigr]\bigr)^2\, \Bigr] \\
\approx \frac{1}{N} \sum^N_{n=1}  \Bigl(h(X^{(n)})-\mathcal{I}_f^{\text{\tiny MC}}[h]\Bigr)^2.
\end{split}
\end{align}
The sample size $N$ is typically large in Monte Carlo simulations, and hence the large $N$ asymptotics required to invoke the central limit theorem \eqref{eq:MonteCarloCLT} and to ensure the convergence of the variance approximation \eqref{eq:varianceMonteCarlo} are well justified.

\begin{mybox}[colback=white]{Pros and Cons}
The standard error of the classical Monte Carlo estimator is of order $N^{-1/2}$, which implies a very slow rate of convergence compared to most deterministic quadrature methods:
increasing the sample size by a factor of $100$ only reduces the expected error by a factor of $10$.
Another drawback is that the error can only be characterized probabilistically, whereas for deterministic quadrature rules deterministic bounds can be derived.
On the bright side, Theorem \ref{thm2.1} does not make any assumption on the state space $E$ (or its dimension) and the only implicit assumption on $h$ is that $\mathbb{E}_{X \sim f} \bigl[h(X)^2 \bigr]<\infty.$  In contrast, classical error bounds for deterministic quadrature rules in $\R^d$ need the number $N$ of function evaluations to scale as $N^d$ as $d\to \infty.$
Moreover, differentiability conditions on $h$ often needed for deterministic methods play no role in Theorem \ref{thm2.1}. The insensitivity of classical Monte Carlo to dimension and 
regularity
of $f$ and $h$ makes it appealing in cases where deterministic methods struggle.  It is important to note, however, that the classical Monte Carlo method presupposes that sampling  from $f$ (without sampling error) is possible. When this is not true (which is often the case in applications) the error of Monte Carlo methods that rely on approximate 
samples from the target $f$ may become
sensitive to the dimension of the state space. 
\end{mybox}

\section{Importance Sampling\index{importance sampling}}\label{ssec:importancesampling}
Now we study how to approximate $\mathcal{I}_f[h]$ using draws from a proposal distribution $g\neq f$.
There are two distinct motivations to sample from a proposal distribution $g$, rather than from the target distribution $f.$ 

\paragraph{Motivation 1: Variance Reduction\index{variance reduction}}
If the proposal is cleverly chosen,
one can obtain an estimator with a relatively lower variance. From a historical perspective, variance reduction was the original motivation for introducing importance sampling; in particular, the design of proposal distributions for the simulation of rare events \index{rare event} has been extensively researched. 
A natural question is how best to choose the proposal in order to compute the probability of a rare event under the target. \index{rare event} In such a context, the test function $h$ would be the indicator of a small-probability set.  \hfill $\square$
\paragraph{Motivation 2:  Target Intractability\index{distribution!intractable}} In some applications, sampling directly from $f$ may not be possible, but one may have access to samples from a different distribution $g$. This scenario arises for instance in sequential signal estimation, where prior samples obtained by a stochastic dynamics model are used to compute expectations with respect to a posterior distribution which incorporates observed data. 
A natural question is:
given a proposal and a target, what is the worst-case error of importance sampling over a given class of test functions? \hfill $\square$

\begin{figure}[!htb]
\begin{center}
  \includegraphics[width=0.45\linewidth]{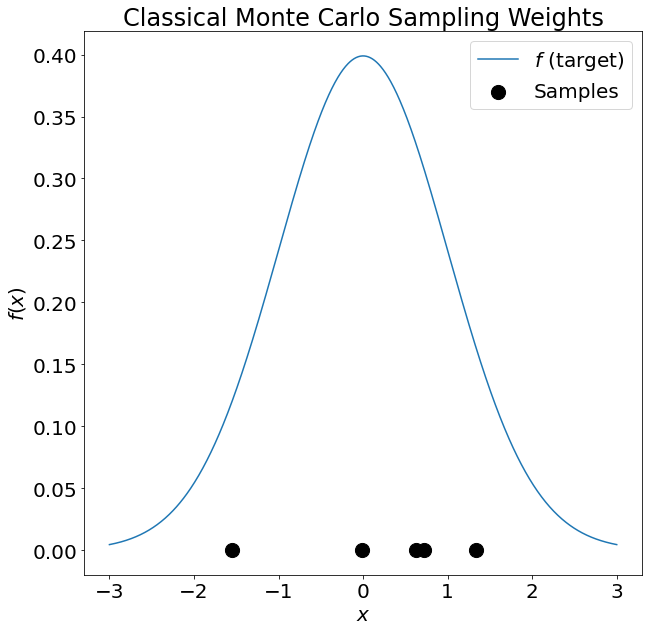}
  \includegraphics[width=0.45\linewidth]{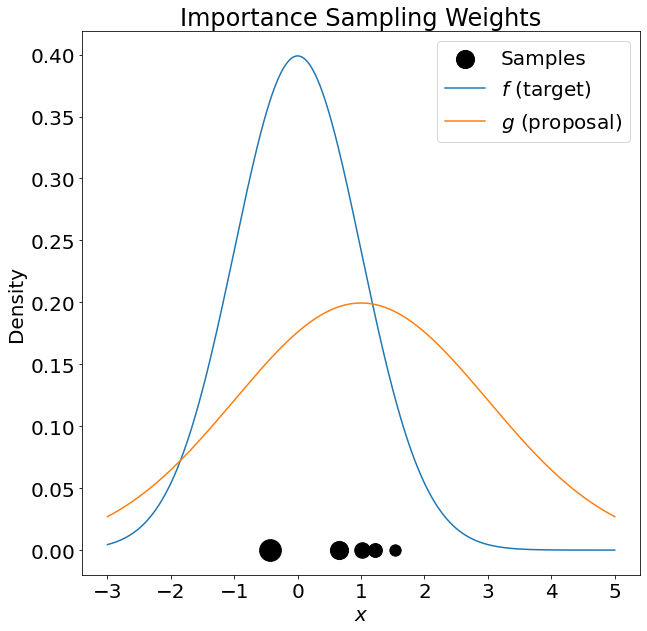}
 \caption{In classical Monte Carlo (left), samples from $f$ are given equal weights. In importance sampling (right), samples from $g$ are given weights proportional to $f/g,$ represented by the radii of the dots.}
 \end{center}
\end{figure}

\bigskip 

In its simplest form, importance sampling is based on the observation that if $f, g$ and $h$ have compatible supports, then
$$\mathcal{I}_f[h]=\int_Eh(x)f(x)\, dx=\int_E h(x)\frac{f(x)}{g(x)}g(x) \, dx=\mathcal{I}_g\biggl[h\frac{f}{g}\biggr].$$
This identity suggests that $\mathcal{I}_f[h]$ may be approximated using classical Monte Carlo integration by sampling the density $g$ and considering $hf/g$ as the test function, leading to the following algorithm:

\begin{algorithm}[H]
\caption{Importance Sampling\index{importance sampling}}
\label{alg:IS}
\begin{algorithmic}[1]
 \STATEx{ \textbf{Input:} Target distribution $f,$ proposal distribution $g,$ test function $h,$ sample size $N.$}
\STATE Sample $X^{(1)}, \ldots, X^{(N)} \stackrel{\text{i.i.d.}}{\sim} g.$
\STATE Set $$w(X^{(n)}):=\frac{f(X^{(n)})}{g(X^{(n)})}.$$
 \STATEx{ \textbf{Output:} Estimator  $\mathcal{I}_f^{\text{\tiny IS}}[h] := \frac{1}{N}  \sum_{n=1}^{N} w(X^{(n)})  h(X^{(n)})\approx\mathcal{I}_f[h] .$  }
\end{algorithmic}
\end{algorithm}
Note that the importance sampling estimator $\mathcal{I}_f^{\text{ \tiny IS}}[h]$ depends on the choice of proposal $g$, but we do not include that dependence in our notation. In addition to studying the importance sampling estimator $\mathcal{I}_f^{\text{ \tiny IS}}[h],$ we will study an autonormalized implementation motivated by the two following drawbacks of Algorithm \ref{alg:IS}. 

First, note that in order to compute the weights
$$w(X^{(n)}):=\frac{f(X^{(n)})}{g(X^{(n)})}$$
we need to be able to evaluate the ratio $f/g$, which may in practice involve evaluating both the target and the proposal with their appropriate normalizing constants. 

Second, note that
$$\mathcal{I}_g[w]=\int w(x)g(x)\, dx=\int\frac{f(x)}{g(x)}g(x) \, dx=1,$$
but in general
$$\frac{1}{N}  \sum_{n=1}^{N}w(X^{(n)}) \neq 1.$$
In other words, importance sampling does not 
exactly estimate 
the constant test function $h \equiv 1.$
These two observations motivate the use of \emph{autonormalized} importance sampling.\index{importance sampling!autonormalized}

\begin{algorithm}[H]
\caption{Autonormalized Importance Sampling}
\label{alg:AIS}
\begin{algorithmic}[1]
 \STATEx{ \textbf{Input:} Target $f,$ proposal $g$, test function $h,$ sample size $N.$}
\STATE Sample $X^{(1)}, \ldots, X^{(N)} \stackrel{\text{i.i.d.}}{\sim} g.$
\STATE  Set $$w^*(X^{(n)}):=\frac{w(X^{(n)})}{\sum_{l=1}^Nw(X^{(l)})}.$$
 \STATEx{ \textbf{Output:} Autonormalized estimator  $\mathcal{I}_f^{\text{\tiny AIS}}[h]:= \sum_{n=1}^{N} w^*(X^{(n)})  h(X^{(n)})\approx  \mathcal{I}_f[h].$  }
\end{algorithmic}
\end{algorithm}
The quantities $w^*(X^{(n)})$ are known as \emph{normalized} weights because they add up to $1$. Notice that $w$ appears in
both
the numerator and the denominator of $w^*$, and so $w^*$ can be evaluated as long as $w(x)=\frac{f(x)}{g(x)}$ is known up to a multiplicative constant. Thus, autonormalized importance sampling can be implemented as long as $f$ and $g$ are known up to a normalizing constant, which greatly extends the applicability of the method. As we will see, the advantages afforded by the autonormalized estimator come at the cost of introducing a bias.

\subsection{Error Bounds for Importance Sampling\index{importance sampling!error bounds}}
In this subsection we analyze the importance sampling estimator $\mathcal{I}_f^{\text{\tiny IS}}[h].$ The presentation is focused on the variance reduction motivation for importance sampling; 
thus, we analyze the error for a given test function and study how to choose the proposal to reduce the variance of the classical Monte Carlo estimator.

The following result is parallel to Theorem \ref{thm2.1}. As in classical Monte Carlo, the importance sampling estimator is unbiased and its mean squared error is of order $1/N.$ Now, however, the mean squared error depends on the proposal 
distribution $g.$
\begin{theorem}[Importance Sampling Error] \label{thm2.2} 
The following holds:
\begin{enumerate}
\item $\mathbb{E} \bigl[\mathcal{I}_f^{\emph{\tiny IS}}[h] \bigr]=\mathcal{I}_f[h];$
\item $ \mathbb{V} \bigl[\mathcal{I}_f^{\emph{\tiny IS}}[h]\bigr]=\frac{1}{N}\mathbb{V}_{X \sim g} \bigl[h(X)w(X) \bigr].$
\end{enumerate}
\end{theorem}
\begin{proof}
This is a corollary of Theorem \ref{thm2.1} since
$\mathcal{I}_f[h]=\mathcal{I}_g[hw]$ and $\mathcal{I}_f^{\text{\tiny IS}}[h]=\mathcal{I}_g^{\text{\tiny MC}}[hw].$ \hfill $\square$
\end{proof}

Theorem \ref{thm2.2} suggests that in order to approximate $\mathcal{I}_f[h]$ for given target $f$ and  test function $h$ one should choose the proposal $g$ that minimizes
$$\mathbb{V}_{X \sim g}[h(X)w(X)].$$
The next result is interesting for historical reasons and to develop intuition. As 
discussed below, 
its practical significance is rather limited.

\begin{theorem}[Importance Sampling Optimal Proposal\index{importance sampling!optimal proposal}] \label{thm2.3}
The proposal $g$ that minimizes the variance of $\mathcal{I}_f^{\text{\tiny IS}}[h]$ is 
$$g^*(x)=\frac{{|h(x)|}f(x)}{\int_E {|h(t)|}f(t)dt}.$$
\end{theorem}

\begin{proof}
We have that
\begin{align*}
N\mathbb{V}\bigl[\mathcal{I}_f^{\text{\tiny IS}}[h]\bigr]&=\mathbb{V}_g\biggl[h(X)\frac{f(X)}{g(X)}\biggr] \\
&=\mathbb{E}_g\biggl[\Bigl(h(X)\frac{f(X)}{g(X)}\Bigr)^2\biggr]-\mathbb{E}_g\biggl[h(X)\frac{f(X)}{g(X)}\biggr]^2.
\end{align*}
Since $\mathbb{E}_g \Bigl[h(X)\frac{f(X)}{g(X)} \Bigr]^2=\mathcal{I}_f[h]^2$ does not depend on $g$, we only need to minimize the first term in the right-hand side. Note that:
\\(i) Plugging $g=g^*$ into 
this
term gives
\begin{align*}
\mathbb{E}_{g^*}\biggl[\Bigl(h(X)\frac{f(X)}{g^*(X)}\Bigr)^2\biggr]&=\int \frac{h(x)^2f(x)^2}{g^*(x)} \, dx \\
&=\int \frac{h(x)^2f(x)^2}{|h(x)|f(x)} \, dx\int|h(t)|f(t) \, dt=\Bigl(\int|h(x)|f(x) \, dx\Bigr)^2.
\end{align*}
(ii) Applying Jensen's inequality we have that, for any $g$,
$$\mathbb{E}_g\biggl[\Bigl(\frac{h(X)f(X)}{g(X)}\Bigr)^2\biggr] \geq \mathbb{E}_g\biggl[\frac{|h(X)|f(X)}{g(X)}\biggr]^2=\Bigl(\int|h(x)|f(x) \, dx\Bigr)^2.$$
Combining (i) and (ii) gives the result. \hfill $\square$
\end{proof}

It is important to note that, in order to use the optimal proposal $g^*$ for importance sampling, we need to be able to evaluate $g^*$; however, evaluating $g^*$ involves computing an integral similar to the one we originally set out to compute! 
In fact, if $h$ is positive, then the denominator in the definition of $g^*$ equals $\mathcal{I}_f[h]$. In that case,
the proof of Theorem \ref{thm2.3} implies that $\mathcal{I}_f^{\text{\tiny IS}}[h]$ has zero variance, and a direct calculation using the definition of $\mathcal{I}_f^{\text{\tiny IS}}[h]$ shows that  $\mathcal{I}_f^{\text{\tiny IS}}[h] =\mathcal{I}_f[h].$ 
While these observations point to the limited  practical relevance of Theorem \ref{thm2.3}, 
a useful takeaway is that importance sampling can be super-efficient: it is possible to obtain a better 
(lower variance) 
estimator using samples from a proposal distribution 
than from the target.
To further demonstrate this point, note that when $g=g^*$,
\begin{align*}
N\mathbb{V} \bigl[\mathcal{I}_f^{\text{\tiny IS}}[h] \bigr]&=\mathbb{E}_f \big[|h(X)| \big]^2-\mathcal{I}_f[h]^2 \\
&\leq \mathbb{E}_f \big[h(X)^2\big]-\mathbb{E}_f \big[h(X)\big]^2=\mathbb{V}_f \big[h(X)\big]=N\mathbb{V} \big[\mathcal{I}_f^{\text{\tiny MC}}[h]\big].
\end{align*}
Another useful takeaway from Theorem~\ref{thm2.3} is 
that samples from the proposal should concentrate on regions where $|h|$ is large, as illustrated in the following example. 

\begin{example}[Computing a Gaussian Tail Probability\index{rare event}] \label{egIS}
Let $p=\mathbb{P}(X>4)$ where $X \sim \mathcal{N}(0,1)$ with corresponding p.d.f. $f$.  Using classical Monte Carlo we can approximate $p$ by 
$$\hat{p}=\frac{1}{N}\sum_{n=1}^N \textbf{1}_{\{X^{(n)}>4\}}, \quad \quad X^{(1)}, \ldots, X^{(N)}  \stackrel{\text{i.i.d.}}{\sim} \mathcal{N}(0,1).$$
However, $N$ would have to be very large to get an estimate that differs from 0.
An alternative is to generate a sample of \text{i.i.d.} $\rm{Exponential}(1)$ random variables translated to the right by $4,$ with corresponding p.d.f. $g$. The weights given to the samples would be determined by the weight function
$$w(x)=\frac{f(x)}{g(x)}=\frac{1}{\sqrt{2\pi}} \exp\Bigl(-\frac{1}{2}x^2+(x-4)\Bigr).$$
Figure \ref{egpic2.6} shows an experiment where sampling from $g$ rather than $f$ is advantageous. \hfill \qedhere
\begin{figure}[H]
\centering
\includegraphics[width=0.8\columnwidth]{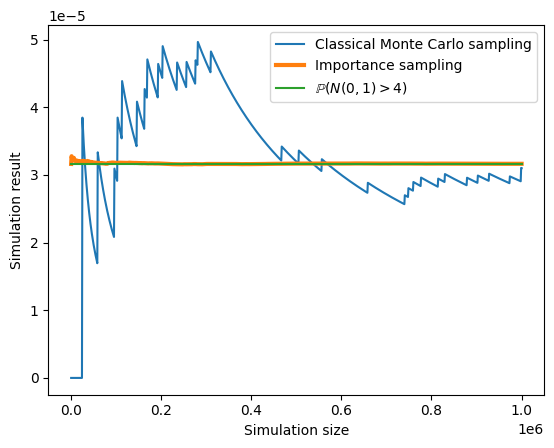}
\caption{Approximation of $p=\mathbb{P}(X>4)$ with $X\sim \Nc(0,1)$ using classical Monte Carlo and importance sampling with a shifted exponential. The green line shows the true value. The approximation with importance sampling is accurate and stable with $N=10^4$ samples, while classical Monte Carlo requires around $N=10^6$ samples to reach a similar level of accuracy.}
\label{egpic2.6}
\end{figure}
\end{example}

\subsection{Error Bounds for Autonormalized Importance Sampling}
Recall that autonormalized importance sampling can be used as long as the ratio $f/g$ can be evaluated up to a multiplicative constant. Here, we analyze autonormalized importance sampling in the following setting:
\begin{itemize}
\item Target $f(x) = \tilde{f}(x)/ c \, , $ where $f$ is 
the normalized version of $\tilde{f}.$
\vspace{-0.2cm}
\item Normalized proposal 
$g(x).$
\end{itemize}

Our analysis is motivated by applications where sampling $f$ directly may not be possible,  and where we can evaluate $\tilde{f}$ but the normalizing constant $c$ of $f$ is unknown. We are interested in estimating $\mathcal{I}_f[h]$ for many different test functions $h$ or for a test function $h$ that is unknown before designing the algorithm. Consequently, we will be concerned with bounding the \emph{worst-case error} over a \emph{family} of test functions. 
We  will work 
with a family of bounded functions, but similar results can be obtained for other classes of test functions. 

Autonormalized importance sampling is based on writing $\mathcal{I}_f[h]$ as a ratio of expectations with respect to $g$ and approximating both expectations with classical Monte Carlo. Precisely,
\begin{equation*}
    \mathcal{I}_f[h] = \\ \frac{\mathcal{I}_g[h(X)\tilde{w}(X)]}{\mathcal{I}_g[\tilde{w}(X)]}
    \\ \approx \frac{\mathcal{I}^{\text{\tiny MC}}_g[h\tilde{w}]}{\mathcal{I}^{\text{\tiny MC}}_g[\tilde{w}]} \\
    = \mathcal{I}^{\text{\tiny AIS}}_f[h],
\end{equation*} 
where $\tilde{w}(x) = \tilde{f}(x)/g(x)$. Note that the denominator $\mathcal{I}^{\text{\tiny MC}}_g[\tilde{w}]$ 
approximates 
the unknown normalizing constant $c$ of the density $f.$ The next result is parallel to Theorems \ref{thm2.1} and \ref{thm2.2}. Note that, in contrast to classical Monte Carlo and importance sampling, autonormalized importance sampling gives a \emph{biased} estimator, since the ratio of two unbiased estimators is in general biased.

\begin{theorem}[Autonormalized Importance Sampling Error]\label{thm:autonormalized} 
Let $h:E \to \mathbb{R}$ with $|h|_\infty :=\sup_{x \in E} |h(x)| \leq 1$. 
The following holds:
\begin{enumerate}
\item $\Bigl|\mathbb{E} \Big[\mathcal{I}^{\emph{\tiny AIS}}_f[h] - \mathcal{I}_f[h]\Big] \Bigr| \leq \frac{2}{N} \frac{\mathcal{I}_g[\tilde{w}^2]}{\mathcal{I}_g[\tilde{w}]^2}$;
\item $\mathbb{E} \Big[\bigl(\mathcal{I}^{\emph{\tiny AIS}}_f[h] - \mathcal{I}_f[h]\bigr)^2 \Big] \leq \frac{4}{N} \frac{\mathcal{I}_g [\tilde{w}^2]}{\mathcal{I}_g [\tilde{w}]^2}.$
\end{enumerate}
\end{theorem}

\begin{proof}
Note that 
\begin{align*}
    \mathcal{I}^{\text{\tiny AIS}}_f[h] - \mathcal{I}_f[h]  
    &= \mathcal{I}^{\text{\tiny AIS}}_f[h] - \frac{\mathcal{I}_g[\tilde{w}h]}{\mathcal{I}_g[\tilde{w}]}  \\
  &= \Bigl(\mathcal{I}_g[\tilde{w}] - \mathcal{I}^{\text{\tiny MC}}_g[\tilde{w}] \Bigr)\frac{\mathcal{I}^{\text{\tiny AIS}}_f[h]}{\mathcal{I}_g[\tilde{w}]} - \frac{\mathcal{I}_g[\tilde{w}h] - \mathcal{I}^{\text{\tiny MC}}_g[\tilde{w}h]}{\mathcal{I}_g[\tilde{w}]}.
\end{align*} 
This identity will be used to show the bounds for the bias and the mean squared error. We start with the latter:
\begin{align*}
    \mathbb{E} \Big[(\mathcal{I}^{\text{\tiny AIS}}_f[h] - \mathcal{I}_f[h])^2 \Big] 
     &\leq \frac{2}{\mathcal{I}_g[\tilde{w}]^2} \biggl( \mathbb{E} \Big[(\mathcal{I}^{\text{\tiny AIS}}_f[h])^2(\mathcal{I}_g[\tilde{w}] - \mathcal{I}^{\text{\tiny MC}}_g [\tilde{w}])^2\Big]  
     + \mathbb{E} \Big[(\mathcal{I}_g[\tilde{w}h] - \mathcal{I}^{\text{\tiny MC}}_g[\tilde{w}h])^2\Big] \biggr) \\
    &\leq  \frac{2}{\mathcal{I}_g[\tilde{w}]^2N} \Bigl(\mathbb{V}_g[\tilde{w}]+\mathbb{V}_g[\tilde{w}h]  \Bigr)  \\
   & \leq \frac{4}{N} \frac{\mathcal{I}_g[\tilde{w}^2]}{\mathcal{I}_g[\tilde{w}]^2}.
\end{align*}
We have used  (i) for any $a,b \in \R,$  $(a - b)^2 \leq 2(a^2 +b^2)$; (ii)
 Theorem \ref{thm2.1} for mean squared error of classical Monte Carlo;  (iii) that $h$ is bounded by $1;$ and (iv) the bound
$\mathbb{V}[\tilde{w}] \leq \mathbb{E}_g[\tilde{w}^2] = \mathcal{I}_g[\tilde{w}^2]$. 

Now we prove the result for the bias:
\begin{align*}
    \Bigl|\mathbb{E} \Big[\mathcal{I}^{\text{AIS}}_f[h] - \mathcal{I}_f[h]\Big] \Bigr|
     &= \frac{1}{\mathcal{I}_g[\tilde{w}]}  \Bigl|\mathbb{E}\Bigl[(\mathcal{I}^{\text{AIS}}_f[h] - \mathcal{I}_f[h])(\mathcal{I}_g[\tilde{w}] - \mathcal{I}^{\text{\tiny MC}}_g[\tilde{w}]) \Bigr]  \Bigr| \\
     &\leq  \frac{1}{\mathcal{I}_g[\tilde{w}]}  \Bigl(\mathbb{E} \Bigl[(\mathcal{I}^{\text{AIS}}_f[h] - \mathcal{I}_f[h])^2 \Bigr]\Bigl)^{1/2} \Bigl(\mathbb{E} \Bigl[(\mathcal{I}_g[\tilde{w}] - \mathcal{I}^{\text{\tiny MC}}_g[\tilde{w}])^2\Bigr] \Bigr)^{1/2}  \\
   & \leq  \frac{1}{\mathcal{I}_g[\tilde{w}]}   \biggl(\frac{4}{N} \frac{\mathcal{I}_g[\tilde{w}^2]}{\mathcal{I}_g[\tilde{w}]^2}\biggr)^{1/2} \biggl(\frac{\mathcal{I}_g[\tilde{w}^2]}{N}\biggr)^{1/2}\\
     &\leq \frac{2}{N} \frac{\mathcal{I}_g[\tilde{w}^2]}{\mathcal{I}_g[\tilde{w}]^2}.
\end{align*}
For the first equality we used that $\mathbb{E} \big[\mathcal{I}_g[\tilde{w}] - \mathcal{I}^{\text{\tiny MC}}_g[\tilde{w}]\big] = 0,$ for the first inequality Cauchy-Schwarz, and for the second inequality the bound for the mean squared error. \hfill $\square$
\end{proof}

\begin{remark} 
Theorem \ref{thm:autonormalized} shows that the bias and mean squared error of autonormalized importance sampling are of order $1/N$ over the class of uniformly bounded test functions. Recall that the mean squared error equals the sum of the squared bias and the variance, i.e.
\begin{equation*}
    \text{MSE} = \text{bias}^{2} + \text{variance}.
\end{equation*}
Consequently, 
the squared bias is of order $1/N^2,$
whereas the variance is of order $1/N.$ This implies that, for large sample size, the bias of autonormalized importance sampling has a 
relatively
modest contribution to the overall mean squared error. 
\end{remark}
\paragraph{Autonormalized Importance Sampling and Closeness Between Target and Proposal }
The upper bounds on the bias and mean squared error depend on
\begin{equation*}
    \zeta := \frac{\mathcal{I}_g[\tilde{w}^2]}{\mathcal{I}_g[\tilde{w}]^2}.
\end{equation*}
The value of $\zeta$ quantifies the variability of  the weights. It holds that $\zeta \ge 1$ with equality only if  $\tilde{w}$ is constant, i.e. if $f = g$. 
Larger values of $\zeta$ imply worse behavior of importance sampling over the class of bounded test functions. Thus, if no knowledge of the relationship between the test function $h$ and the target is available, it is important to choose a proposal that is close to the target, so that $\zeta$ is small. 

The quantity $\zeta$ is intimately related to the $\chi^{2}$-divergence\index{divergence!$\chi^2$} between the target and the proposal. The $\chi^{2}$-divergence is a way to quantify the closeness between two densities, defined by
\begin{equation*}
    \dchi(f \| g) := \int_{E} \biggl(\frac{f(x)}{g(x)} - 1 \biggr)^2g(x) \, dx.
\end{equation*}
Therefore, we have that  $\dchi (f \| g) = \zeta - 1$.
Hence, Theorem \ref{thm:autonormalized} implies that for autonormalized importance sampling to perform accurately over the class of bounded test functions, the $\chi^2$-divergence between target and proposal needs to be moderate. 

\paragraph{Effective Sample Size\index{effective sample size}}
The effective sample size defined by
\begin{equation}\label{eq:ESSimportancesampling}
    \text{ESS} := \frac{1}{\sum_{n = 1}^N w^*(X^{(n)})^2}, \quad \quad 
    w^*(X^{(n)}) := \frac{\tilde{w}(X^{(n)})}{\sum_{n = 1}^N \tilde{w}(X^{(n)})}
\end{equation}
is widely used by practitioners as a  diagnostic for the performance of (autonormalized) importance sampling.
Note that
\begin{itemize}
\item  $\text{ESS} = 1$ if there is $n \in \{1,\dots,N\}$ such that $w^*(X^{(n)}) = 1$ (in which case all other normalized weights are zero).
\item $\text{ESS} = N$ if $w^*(X^{(n)}) = \frac{1}{N} $ for all  $n \in \{1,\dots,N\}.$
\item $\text{ESS} \in [1,N]$ for any configuration of weights.
\end{itemize}
The effective sample size ESS is intimately related to $\zeta$ and $\dchi(f \| g).$
Indeed,
\begin{equation*}
    \frac{\text{ESS}}{N} = \frac{1}{N\sum w^*(X^{(n)})^2}
    = \frac{(\sum \tilde{w}(X^{(n)}))^2}{N \sum \tilde{w}(X^{(n)})^2}
    = \frac{ \biggl(\frac{\sum \tilde{w}(X^{(n)})}{N} \biggr)^2}{\frac{\sum \tilde{w}(X^{(n)})^2}{N}}
    = \frac{\mathcal{I}^{\text{\tiny MC}}_g[\tilde{w}]^2}{\mathcal{I}^{\text{\tiny MC}}_g[\tilde{w}^2]},
\end{equation*}
and so 
$$\frac{N}{ \text{ESS}} = \frac{ \mathcal{I}^{\text{\tiny MC}}_g[\tilde{w}^2]}{ \mathcal{I}^{\text{\tiny MC}}_g[\tilde{w}]^2}
\approx \zeta.$$
 It is important to note, however, that ESS does not include any test function $h$ in its definition. Therefore, the performance of (autonormalized) importance sampling with any given test function cannot be fully understood via the effective sample size.

\paragraph{Autonormalized Importance Sampling and Curse of Dimension}
Consider (as for rejection sampling) the i.i.d. setting
 \begin{align*}
 \tilde{f}_d(x) &= \tilde{f}(x_1) \times \dots \times \tilde{f}(x_d), \hspace{3mm} \tilde{f}:\mathbb{R} \longrightarrow \R,
\hspace{3mm} x = (x_1,\dots,x_d) \in \mathbb{R}^d, \\
g_d( x ) &= g(x_1) \times \dots \times g(x_d), \hspace{6mm} g:\mathbb{R} \longrightarrow \R,
 \end{align*}
 where the tildes stand for unnormalized densities. 
Then, 
\begin{equation*}
    \frac{\mathcal{I}_{g_d}[\tilde{w}^2_d]}{\mathcal{I}_{g_d}[\tilde{w}_d]^2} = \biggl(\frac{\mathcal{I}_{g}[\tilde{w}^2]}{\mathcal{I}_{g}[\tilde{w}]^2}\biggr)^d.
\end{equation*}
This identity along with Theorem \ref{thm:autonormalized} suggest that in order to have the same accuracy in importance sampling with proposal $g_d$ and target $f_d$ across a sequence of problems with increasing $d$, one needs to increase the number of samples exponentially with $d$.

\begin{mybox}[colback=white]{Pros and Cons}
As opposed to classical Monte Carlo, (autonormalized) importance sampling can be implemented when $f$ cannot be sampled from, and, in contrast to rejection sampling, no samples are thrown away. Importance sampling can have smaller variance than classical Monte Carlo integration, but finding an adequate proposal that gives variance reduction can be challenging. Autonormalized importance sampling can have smaller variance than both importance sampling and classical Monte Carlo, but is biased. On the other hand, importance sampling requires evaluating the normalized target and proposal distributions, but is unbiased.  If the target and the proposal are far apart (which is typically the case in high-dimensional problems) the weights have a large variance. Then, the effective sample size will be small, and there will be test functions for which importance sampling performs poorly.
\end{mybox}

\section{Variance Reduction Techniques}\label{sec:variancereduction}
We have seen in Theorem \ref{thm2.3} that importance sampling with an appropriate choice of proposal distribution  ---adapted to both the target and the test function--- can have smaller variance than classical Monte Carlo. In this section, we describe two other variance reduction techniques: antithetic sampling and control variates. The general problem that we consider is again the estimation of $\mathcal{I}_f[h] := \mathbb{E}_{X \sim f}[h(X)]$ 
using $N$ queries 
of the p.d.f. $f.$ Recall that the classical Monte Carlo estimator $\mathcal{I}_f^{\text{\tiny MC}}[h] = \frac{1}{N} \sum^{N}_{n=1} h(X^{(n)}) $ is unbiased and has variance $\sigma^2/N$ with $\sigma^2:=\mathbb{V}_{X \sim f}[h(X)].$
	\subsection{Antithetic Sampling\index{antithetic sampling}}
	This method requires the target distribution to be symmetric with respect to some center $c,$ so that $f(x)=f(2c-x)$ holds for all $x.$ We let $N$ be an even integer and draw $N/2$ samples from $f;$ symmetry is used to define the remaining $N/2$ samples. 
		\begin{algorithm}[H]
			\caption{Antithetic Sampling}
			\label{alg:anti_samp}
			\begin{algorithmic}[1]
			 \STATEx{ \textbf{Input:} Target $f$ which is symmetric with respect to $c,$ even sample size $N.$}
			\STATE Sample $X^{(1)}, \ldots, X^{(\frac{N}{2})} \stackrel{\text{i.i.d.}}{\sim} f.$
			\STATE Compute symmetrized samples $\tilde X^{(n)} := 2c-X^{(n)}, \, \, n = 1, \ldots, N/2.$
			  \STATEx{\textbf{Output:} Antithetic estimator $\mathcal{I}_f^{\text{ \tiny AS}}[h] :=\frac{1}{N} \sum_{n \leq\frac{N}{2}} \Bigl(h(X^{(n)})+h(\tilde X^{(n)})\Bigr) \approx \mathcal{I}_f[h].$}
			\end{algorithmic}
		\end{algorithm}
The estimator $\mathcal{I}_f^{\text{ \tiny AS}}[h]$ is unbiased and has variance
		\begin{align*}
			\V\big[    \mathcal{I}_f^{\text{\tiny AS}}[h]  \big]
			&=\frac{N/2}{N^2}  \V\bigl[h(X)+h(\tilde X)\bigr]\\
			&=\frac{1}{2N} \Bigl(  \V[h(X)]+\V[h(\tilde X)]+2\operatorname{Cov} \bigl(h(X),h(\tilde X)\bigr) \Bigr)\\
			&=\frac{\sigma^2}{N}(1+\rho),
		\end{align*}
	where $\tilde X = 2c - X,$ $\sigma^2 = \V[h(X)] = \V[h(\tilde X)],$ and $\rho= \operatorname{Corr} \bigl(h(X),h(\tilde X)\bigr)$. Note that when $h$ is linear, we have $\rho=-1$ and hence zero variance. Indeed, for any affine function $h(X)=AX+b$, the estimator $\mathcal{I}_f^{\text{\tiny AS}}[h] $ is constant and agrees with the  expectation of interest. In the general case where $-1\leq\rho\leq 1$, we have $$\V\big[    \mathcal{I}_f^{\text{ \tiny AS}}[h]  \big]\leq \frac{2 \sigma^2}{N}= 2 \V\big[  \mathcal{I}_f^{\text{ \tiny MC}}[h] \big],$$ i.e. in the worst case antithetic estimation has twice the variance of classical Monte Carlo.
	
	 In order to analyze the variance for 
     a
     general test function $h,$ we write $h = h_o + h_e,$  where $$h_o(x):=\frac{h(x)-h(\tilde x)}{2}, \quad \quad h_e(x):=\frac{h(x)+h(\tilde x)}{2},$$ are respectively the odd and even components of $h$ with respect to the density $f,$ and $\tilde{x} = 2c - x.$
	Note that $\Expect[h_e(X)]= \Expect[h(X)]$ and $\Expect[h_o(X)]=0$.   Furthermore, $h_e$ and $h_o$ are uncorrelated, since 
		\begin{align*}
	\operatorname{Cov} \bigl(h_e(X),h_o(X) \bigr)
		&=\Expect\left[ \biggl(\frac{h(X)+h(\tilde X)}{2}-\Expect[h]\biggr) \biggl(\frac{h(X)-h(\tilde X)}{2} \biggr)\right ]\\
		&=\frac{1}{4}\Expect\left [ \Bigl(h(X)+h(\tilde X)-2\Expect[h]\Bigr) \Bigl(h(X)-h(\tilde X) \Bigr)\right ]\\
		&=\frac{1}{4}\Expect\left [ h(X)^2-h(\tilde X)^2-2\Expect[h] \bigl(h(X)-h(\tilde X)\bigr)\right ]\\
		&=0.
		\end{align*}
	As a result, $\sigma^2=\sigma_e^2+\sigma_o^2,$ where $\sigma_e^2 :=\V [h_e(X)]$ and $\sigma_o^2 :=\V [h_o(X)]$. Rewriting the variance of our estimators, we have
	$$\V\big[  \mathcal{I}_f^{\text{ \tiny MC}}[h] \big]=\frac{1}{N}\sigma^2=\frac{\sigma_e^2+\sigma_o^2}{N} $$
	and
	\[  \V\big[  \mathcal{I}_f^{\text{ \tiny AS} }[h] \big]   =\frac{1}{2N} \V\bigl[ h(X)+h(\tilde X)\bigr]=\frac{1}{2N}\V\bigl[2h_e(X)\bigr]=\frac{2\sigma_e^2}{N}.\]
	Observe that antithetic sampling eliminates $\sigma_o^2$ but amplifies $\sigma_e^2$. In particular, any affine function is ``odd'' with respect to a symmetric density, which explains again our previous findings. Generally, antithetic sampling performs well when the ``odd'' part of the function has larger variance than the ``even'' part.

		\begin{example}[Integral Approximation with Antithetic Sampling]\label{example:anthithetic}
		Suppose we want to approximate an integral 
		$$\int_0^1 h(x) \, dx = \Expect_{X\sim \text{Unif}(0,1)} [h(X)].$$
		We can apply Algorithm \ref{alg:anti_samp} since the uniform distribution is symmetric with respect to the center $c = \frac{1}{2}$. For instance, we can first draw $50$ samples 
		$X^{(1)},X^{(2)},\ldots,X^{(50)}$
		 uniformly from $[0,1]$ and estimate the integral by $\frac{1}{100}\sum_{n=1}^{50}\left(h(X^{(n)})+h(1-X^{(n)})\right)$. \hfill \qedhere
		\end{example}

	\subsection{Control Variates\index{control variates}}
	If we have access to an approximation $\hat h$ of the test function $h$ along with its true mean $\mu := \mathbb{E}_{X \sim f}\bigl[\hat h(X) \bigr],$ we can apply control variates.
		\begin{algorithm}[H]
		\caption{Control Variates}
		\label{alg:control_var}
			\begin{algorithmic}[1]
			 \STATEx{ \textbf{Input:} Target $f,$ test function $h$ and approximation $\hat h$, $\mu:=\Expect \bigl[\hat h(X) \bigr]$, sample size $N.$}
			\STATE Sample $X^{(1)}, \ldots, X^{(N)} \stackrel{\text{i.i.d.}}{\sim} f.$
			 \STATEx{ \textbf{Output:} Estimator $\mathcal{I}_f^{\text{ \tiny CV}}[h]:=\mu+\frac{1}{N}\sum_{n\leq N} \bigl(h(X^{(n)})-\hat h(X^{(n)})\bigr)   \approx \mathcal{I}_f[h] $.}
			\end{algorithmic}
		\end{algorithm}
	Note that the trivial approximation $\hat h = 0$ leads to the classical Monte Carlo estimator $\mathcal{I}_f^{\text{ \tiny CV}}[h]\equiv \mathcal{I}_f^{\text{\tiny MC}}[h]$. On the other hand, the oracle approximation $\hat h=h$ gives $\mathcal{I}_f^{\text{\tiny CV}}[h]=\mathcal{I}_f[h],$ which is unsurprising since control variates assumes access to the true mean of 
    the approximation $\hat h = h$. In general, $\mathcal{I}_f^{\text{\tiny CV}}[h]$ is an unbiased estimator with variance $\V\bigl[  \mathcal{I}_f^{\text{ \tiny CV}}[h]\bigr]=\frac{1}{N}\V\bigl[ h(X)-\hat h(X)\bigr]$. Therefore, the performance of control variate estimation is determined by the quality of the approximation $\hat h \approx h$.
		\begin{example}[Integral Approximation with Control Variates]\label{example:controlvariates}
		Suppose, as in Example \ref{example:anthithetic}, that we wish to approximate 
		$$\int_0^1 h(x) \, dx = \Expect_{X\sim \text{Unif}(0,1)} [h(X)],$$
		where $h$ is a smooth test function which does not admit a closed-form primitive. We can use a polynomial approximation $\hat h$ (e.g. Taylor expansion) of $h$,  compute $\mu = \int_0^1 \hat h(x) \, dx$ exactly, and use Algorithm \ref{alg:control_var} by drawing uniform samples from the unit interval. \hfill \qedhere
		\end{example}

\section{Discussion and Bibliography}\label{sec:bibliochapter2}
The book chapter \cite{hammersley1964percolation} is a classic reference on the Monte Carlo method and variance reduction techniques. Modern textbooks include \cite{asmussen2007stochastic,owenbook,glasserman2004monte,kroese2013handbook,liu2008monte,robert2013monte,graham2013stochastic,barbu2020monte,newman1999monte,robert2010introducing,voss2013introduction}. Some of these books emphasize specific applications, such as finance \cite{glasserman2004monte,wang2012monte,jackel2002monte}, statistics \cite{asmussen2007stochastic,robert2013monte,owenbook,voss2013introduction}, scientific computing \cite{liu2008monte}, statistical physics \cite{newman1999monte}, stochastic simulation \cite{graham2013stochastic}, or computer vision, machine learning, and artificial intelligence \cite{barbu2020monte}. 
For a textbook on quantum Monte Carlo methods with applications in chemistry, we refer to \cite{gubernatis2016quantum}. The textbook \cite{robert2010introducing} introduces Monte Carlo methods through examples coded in R, while \cite{kroese2013handbook} includes numerous worked examples using MATLAB. Finally, \cite{caflisch1998monte} provides an accessible overview of Monte Carlo methods aimed at applied mathematicians. 

In Section \ref{ssec:classicalMC} we invoked the central limit theorem to derive an asymptotic confidence interval for classical Monte Carlo integration. Non-asymptotic analyses are also possible, and we refer to \cite[Chapter 3]{graham2013stochastic} for further details. Empirical process theory provides another perspective on Monte Carlo methods, see for instance \cite[Chapter 8]{vershynin2018high} for a gentle introduction. 

The lecture notes \cite{anderson2014monte} compare Monte Carlo and importance sampling through examples. Importance sampling was developed as a variance reduction technique in the early 
1950s
\cite{kahn1953methods,kahn1955use}. Theorem \ref{thm2.3}, which shows how to optimally choose the proposal density for a given test function $h$ and target density $f$, was already shown in \cite{kahn1953methods}. For recent methodological developments, see  \cite{li2013two,owen2000safe,liu1998sequential,tan2004likelihood,kawai2017adaptive}. The book \cite{chopin2020introduction} contains a modern presentation of importance sampling in a general framework.
A review of importance sampling, from the perspective of filtering and sequential importance resampling, can be found in \cite{agapiou2017importance}; the proofs in this chapter closely follow the presentation in that paper. For a review of adaptive importance sampling techniques, see \cite{bugallo2017adaptive}. These adaptive methods often combine samples from multiple proposals using suitable weighting procedures. Importance sampling methods that use multiple proposals are systematized in \cite{elvira2019generalized}. 

The key role of the second moment of the weight function $w$ has long been 
recognized
\cite{liu1996metropolized,pitt1999filtering}, and it is known to be  asymptotically linked to the effective sample size \cite{kong1992note,kong1994sequential,liu1996metropolized}. The
problem of choosing a proposal distribution that leads to a small second moment
has been widely studied, and we refer to \cite{liang2007stochastic} and references therein.  Similar to \cite{chen2005another}, Theorem \ref{thm:autonormalized} quantifies the estimation error in terms of the $\chi^2$-divergence between target and proposal; recent complementary analysis of importance sampling in \cite{CP15} utilizes the Kullback-Leibler divergence. Necessary sample size results for importance sampling in terms of several divergences between target and proposal were established in \cite{sanz2018importance,sanz2021bayesian}.  A review of useful distances between probability measures can be found in \cite{gibbs2002choosing}. Notions of effective sample size based on discrepancy measures are studied in \cite{martino2017effective}. The paper \cite{elvira2022rethinking} overviews several notions of effective sample size for importance sampling algorithms and proposes new ones, while \cite{agarwal2022principled}
introduces a sequential stopping rule that terminates simulation when the variability in estimation is small. 

Variance reduction techniques are covered in most standard textbooks on Monte Carlo methods, see for instance \cite[Chapter 5]{asmussen2007stochastic} and \cite[Chapter 4]{robert2013monte}. The monograph \cite{graham2013stochastic} considers more advanced variance reduction techniques. In particular, we refer to  \cite[Part III]{graham2013stochastic} for variance reduction techniques for simulation of stochastic differential equations.  

\chapter{Metropolis--Hastings}
\label{chap:MCMC}

This chapter introduces the important idea of using Markov chains for sampling, which underlies a variety of Markov Chain Monte Carlo (MCMC) \index{Markov Chain Monte Carlo (MCMC)} algorithms studied in this and subsequent chapters. These methods draw samples from a Markov chain that has the desired target as invariant distribution. 
Here, we focus on the Metropolis--Hastings \index{Metropolis--Hastings} algorithm, a general framework
to build a Markov chain with this property by combining draws from a proposal Markov kernel with a probabilistic accept/reject mechanism. 
We refer to Appendix \ref{chap:markovchains} for all necessary background on Markov chains, including the concepts of Markov kernel and invariant distribution. 

The Metropolis--Hastings algorithm yields correlated draws $\{X^{(n)}\}_{n=1}^N$
that are only approximately distributed like the target. 
By design, these draws form a Markov chain that satisfies detailed balance with respect to the target, which implies that the target is an invariant distribution (Theorem \ref{thm:detailedimpliesgeneral}).  
Similar to 
the Monte Carlo integration methods 
described 
in Chapter \ref{chap:MCintegration}, 
the draws $\{X^{(n)}\}_{n=1}^N$ 
can be used to approximate integrals of the form 
$$
\mathcal{I}_f[h]=\int_E h(x)f(x) \, dx,
$$
where the target 
p.d.f. $f$ 
is 
supported on
$E\subseteq \R^d$ and $h:E\rightarrow \R$ is a real-valued test function.\footnote{The methodology extends to 
discrete target distributions $f$ and vector-valued test functions.}
To that end, one can define an estimator 
$$
\mathcal{I}_f^{\text{\tiny MH}}[h]: = \frac{1}{N}\sum_{n=1}^N h(X^{(n)}) \approx \mathcal{I}_f[h]. 
$$
As in classical Monte Carlo, draws are assigned uniform weights $1/N$; however,
in contrast to the i.i.d. setting, the draws are correlated and only approximately distributed according to $f$.
Standard Markov chain theory, reviewed in Appendix \ref{chap:markovchains}, guarantees that if the chain $\{X^{(n)}\}_{n=1}^N$ is ergodic, then the estimator $\mathcal{I}_f^{\text{\tiny MH}}[h]$ is asymptotically unbiased (Theorem \ref{thm:LLN}), and if the chain is geometrically ergodic,\index{ergodicity!geometric} then $\mathcal{I}_f^{\text{\tiny MH}}[h]$
additionally
satisfies a central limit theorem  (Theorem \ref{thm:CLT}).
The Metropolis--Hastings algorithm is 
useful when the target distribution is intractable and cannot be sampled directly. Importance sampling 
overcomes the intractability of the target 
by sampling from a proposal distribution and assigning an importance weight to each draw. In contrast, Metropolis--Hastings 
overcomes the intractability of the target by sampling from a Markov kernel and using a 
probabilistic accept/reject mechanism. 
A key advantage of this approach is that all draws receive equal weights, thereby avoiding the weight degeneracy that often arises when using importance sampling in high-dimensional settings. 
Moreover, the use of a proposal Markov kernel offers greater flexibility than specifying a single proposal distribution.

This chapter is organized as follows. Section \ref{ssec:MHA} introduces the Metropolis--Hastings algorithm. Two choices of proposal Markov kernel are discussed in Section \ref{ssec:proposalDistn}, which lead to the independence sampler and the random walk Metropolis algorithm. Section \ref{ssec:implementation} discusses implementation issues and convergence diagnostics. These diagnostics address an important caveat of the Metropolis--Hastings algorithm: it is often difficult to determine whether the Markov chain has converged, and, relatedly, whether 
the sample size $N$ is sufficient to achieve a desired level of accuracy.
Section \ref{sec:bibliochapter3} closes with bibliographical remarks.

\section{Metropolis--Hastings Algorithm}\label{ssec:MHA}
The Metropolis--Hastings \index{Metropolis--Hastings} algorithm provides a flexible framework for Markov chain sampling.  
The algorithm 
leverages a probabilistic accept/reject mechanism to 
transform
a user-chosen proposal Markov kernel \index{Markov kernel!proposal} $q$ into a Markov kernel $\pmh$ that satisfies detailed balance \index{detailed balance} with respect to the target $f.$ 
Since detailed balance implies general balance (Theorem \ref{thm:detailedimpliesgeneral}), the Markov kernel $\pmh$ leaves the target distribution 
invariant. Markov chain sampling from $\pmh$ only requires 
the ability to sample from the proposal kernel and to evaluate an acceptance probability.
Recalling the definition of Markov kernel 
from
Appendix \ref{ssec:markovdefinitions}, for fixed $x,$ the proposal $q(x,\cdotp)$ is a p.m.f. if $E$ is discrete and a p.d.f. if $E$ is continuous. For each $x\in E,$ $q(x,z)$ 
represents 
the probability (or probability density in the continuous case) 
of proposing a move from $x$ to $z$.

The algorithm proceeds as follows. Given the $n$-th draw $X^{(n)},$ we sample a proposed move $Z^*\sim q(X^{(n)},\cdotp)$. 
We accept the move, which means setting $X^{(n+1)}=Z^*,$ with 
probability $a(X^{(n)},Z^*)$. If the 
proposed $Z^*$ is rejected, we instead set 
$X^{(n+1)}=X^{(n)}$. The probability of accepting a proposed move from $x\in E$ to $z \in E$ is defined to be 
$$
a(x,z):=\min\biggl\{1,\frac{f(z)}{f(x)}\frac{q(z,x)}{q(x,z)} \biggr\}.
$$
We will show in Theorem \ref{thm:detailedbalanceMCMC} that this choice of acceptance probability \index{Metropolis--Hastings!acceptance probability} 
yields a kernel $\pmh$ that satisfies
detailed balance with respect to $f.$  The algorithm is summarized below.

\begin{algorithm}[H]
  \caption{Metropolis--Hastings Algorithm\label{algo:MH}}
  \begin{algorithmic}[1]
  \STATEx{ \textbf{Input:} Target $f$, initial distribution $\pi_0$, proposal Markov kernel $q(x,z),$ sample size $N.$ }
    \STATEx{ Initial draw: Sample $X^{(0)}\sim \pi_0.$ }
    \STATEx{ Subsequent draws: For $n=0,\ldots ,N-1$ do:}
  \STATE{\textbf{Proposal step:} Sample $ Z^*\sim q(X^{(n)},\cdotp)$.}
  \STATE{\textbf{Accept/reject step:} Update $$X^{(n+1)}=\begin{cases}
    Z^* \qquad \quad \text{ with probability } a(X^{(n)},Z^*), \\
    X^{(n)} \quad \quad \,   \text{ with probability } 1-a(X^{(n)},Z^*).
                    \end{cases} $$}
  \STATEx{\textbf{Output:} Sample $\{X^{(n)}\}_{n=1}^N.$}
    \end{algorithmic}
\end{algorithm}

As previously mentioned, given a test function $h:E \to \R,$ we can use the output $\{X^{(n)}\}_{n=1}^N$ of the Metropolis--Hastings algorithm to approximate
\begin{equation}\label{eq:MHestimator}
    \mathcal{I}_f[h] = \int_E h(x)f(x)\, dx \approx\frac{1}{N}\sum_{n=1}^N h(X^{(n)}) = :\mathcal{I}_f^{\text{\tiny MH}}[h].
\end{equation}

\begin{remark}\label{rem:commentsMH}
\begin{enumerate}
	\item To implement the Metropolis--Hastings algorithm we 
    must be able
    to sample 
    from
    $q(x,\cdotp)$ for  $x\in E$ and evaluate the acceptance probability $a(x,z)$ for $x,z\in E$. A priori, the only necessary condition on the proposal is that $$E\subseteq \bigcup_{x\in E} \text{support}\bigl(q(x,\cdotp)\bigr).$$
	Importantly, the target $f$ appears in the acceptance probability only as a ratio, and therefore the algorithm can be implemented even if $f$ is only known up to a normalizing constant.
	\item If $q(x,z)=q(z,x)$ for all $x,z \in E,$ then $a(x,z)=\min\Bigl\{1,\frac{f(z)}{f(x)}\Bigr\}$. Moves to regions of higher target density are always accepted, while moves to regions of lower but non-zero target density are accepted with positive probability in order to ensure exploration of the state space $E$. If $q$ is not symmetric, moves for which $q(z,x)\geq q(x,z)$ are favored, as they are more likely to be undone by chance.
	\item The accept/reject step can be implemented as follows:
	\begin{itemize}
	    \item Draw $U^*\sim$ \emph{Unif}$(0,1)$ independently from $Z^*$ and all other randomness.
	    \item Update 
	    \begin{align*}
	    X^{(n+1)}=\begin{cases}
    Z^* \quad \text{if }\quad U^*<a(X^{(n)},Z^*),\\
    X^{(n)} \text{ if }\quad a(X^{(n)},Z^*)<U^*. 
    \end{cases}
	    \end{align*}
	\end{itemize}
\end{enumerate}
\end{remark}

\begin{mybox}[colback=white]{Pros and Cons} The Metropolis--Hastings algorithm can be implemented without 
knowledge of
the normalizing constant of the target, and is extremely flexible due to the freedom in the choice of proposal kernel $q.$ It 
offers important advantages over all the 
algorithms covered in Chapters \ref{chapter1} and \ref{chap:MCintegration}.
In contrast to rejection sampling, no samples are 
discarded. In contrast
to importance sampling, MCMC samples are given equal weights, avoiding weight 
collapse. In contrast to ABC, MCMC methods are, in many cases, provably accurate in the large-$N$ limit.
The Metropolis--Hastings algorithm also has 
some 
caveats. First, it 
is challenging (often impossible) to determine if the chain is close to stationarity. 
Second, unlike rejection and importance sampling,
MCMC yields samples that are typically positively correlated, which, as established in Lemma \ref{lemmaMarkovchainerror}, 
results in 
larger variance estimates than 
independent or negatively correlated samples. Additionally, 
it is often challenging to choose an appropriate proposal kernel, and the ergodic behavior of the algorithm depends critically on that choice.
In particular, effective proposals often require careful tuning
by trial and error (the Gibbs sampler, studied in Chapter \ref{chap:gibbs}, is an exception). Finally, the Metropolis--Hastings algorithm 
is not naturally parallelizable.
\end{mybox}

\bigskip
By construction, the draws $\{X^{(n)} \}_{n=1}^N$ 
produced
by the Metropolis--Hastings algorithm form a Markov chain. 
Notice that in  Algorithm \ref{algo:MH} we have adopted the standard convention, described in Remark \ref{rem:independence}, of omitting independence requirements that are clear from context. In particular, in the proposal step, 
$Z^* \sim q(X^{(n)},\cdot)$ depends on $X^{(n)}$ but is drawn independently 
of all other randomness; similarly, it is implicitly assumed that the randomness 
used
in the accept/reject step is independent from all other
sources of
randomness, 
as made explicit in the implementation in 
Remark \ref{rem:commentsMH}. We denote by $\pmh(x,z)$ the Markov kernel of the time homogeneous chain $\{X^{(n)} \}_{n=1}^N$. Intuitively, $\pmh(x,z)$ 
represents the probability (or probability density) of finding the  $(n+1)$-th draw $X^{(n)}$ at location $z\in E$ given that the $n$-th draw $X^{(n)}$ was at $x\in E.$ 
Notice, however, that $\pmh$ is not absolutely continuous with respect to Lebesgue measure when $E$ is continuous: 
the probability of remaining in the current state $x$ is strictly positive, since this occurs whenever a proposed move is rejected.

\begin{lemma}
The Metropolis--Hastings Markov kernel \index{Metropolis--Hastings!Markov kernel} is given by 
$$\pmh(x,z)=q(x,z)a(x,z)+\delta_x(z) r(x),$$
 where
 \begin{align*}
r(x)=\begin{cases}
    \sum_{z\in E}q(x,z)\bigl(1-a(x,z)\bigr)\quad  &\text{ if } E\, \, \text{ is discrete,} \\
    \int_{E}q(x,z)\bigl(1-a(x,z)\bigr)\, dz \quad &\text{ if } E\,\, \text{ is continuous,}
                    \end{cases} 
\end{align*}                    
  and $\delta_x(z)$ denotes a Dirac mass at $x.$         
\end{lemma}

\begin{proof} We only prove the discrete case.
A transition from $x$ to a new state $z \neq x$ occurs if $z$ is proposed and accepted, which occurs with probability $q(x,z)a(x,z)$. 
A transition from $x$ to itself
may occur in two ways:
\begin{enumerate}
	\item Propose $x$ and accept it, which occurs
    with probability $q(x,x)a(x,x)$.
	\item Propose any $z\in E$ and reject it, which 
    occurs
    with probability $$r(x)=\sum_{z\in E}q(x,z)\bigl(1-a(x,z)\bigr).$$
\end{enumerate}
Combining these contributions yields
\begin{equation*}
 \pmh(x,z)=q(x,z)a(x,z)+\delta_x(z) r(x).   \tag*{\qedhere}
\end{equation*}
\end{proof}
As 
part of the 
proof of the previous lemma, 
we have shown that if $x\neq z,$ 
then
$$
\pmh(x,z)=q(x,z)a(x,z).
$$
Using this identity, we now establish the detailed balance of $\pmh$ with respect to $f$:

\begin{theorem}[Detailed Balance of Metropolis--Hastings\index{Metropolis--Hastings!detailed balance}]\label{thm:detailedbalanceMCMC}
The Metropolis--Hastings kernel $\pmh$ satisfies detailed balance with respect to $f.$
\end{theorem}

\begin{proof}
We need to show that 
$$
f(x)\pmh(x,z)=f(z)\pmh(z,x),\quad \forall x,z\in E.
$$
The case $x=z$ is immediate.
If $x\neq z$, then
\begin{align*}
f(x)\pmh(x,z)&=f(x)q(x,z)a(x,z)
=\min\Bigl\{f(x)q(x,z),f(z)q(z,x)\Bigr\}.
\end{align*}
The right-hand side is symmetric in $x$ and $z,$ 
which proves the result. 
\hfill $\square$
\end{proof}

Since detailed balance implies general balance, Theorem \ref{thm:detailedbalanceMCMC} implies that $f$ is an invariant distribution for $\pmh.$

\section{Proposal Kernels}\label{ssec:proposalDistn}
The choice of proposal kernel is crucial 
to the performance
of the Metropolis--Hastings algorithm. In this section, we study two types of proposal kernels that lead to the independence sampler and the random walk Metropolis algorithm. Other proposal kernels will be 
discussed
in subsequent chapters.

\subsection{Independence Sampler \index{independence sampler}}
An independence sampler is a Metropolis--Hastings algorithm in which the proposal Markov kernel contains no information about the current state of the chain. That is, 
for some distribution $g$ on $E$, 
$$\qind(x,z) = g(z),$$ 
so the proposal Markov kernel $\qind(x,z)$ 
does not depend on 
the current state $x$ of the chain.

\begin{algorithm}[H]
  \caption{Independence Sampler \label{algo:independencesampler}}
  \begin{algorithmic}[1]
  \STATEx{\textbf{Input:} Target distribution $f$, initial distribution $\pi_0$, proposal $g(z),$ sample size $N.$}
  \STATE{Define the Markov kernel
  $$\qind(x,z) := g(z).$$}
  \STATE{Run Metropolis--Hastings (Algorithm \ref{algo:MH}) with inputs $f,$ $\pi_0,$ $\qind(x,z)$, $N.$}
  \vspace{2mm}
  \STATEx{\textbf{Output:} Sample $\{X^{(n)}\}_{n=1}^N$.}
 \end{algorithmic}
\end{algorithm}
Note that the acceptance probability of the independence sampler reduces to
\begin{align*}
a(x,z) = \min \biggl\{1, \frac{f(z)}{f(x)}  \frac{g(x)}{g(z)} \biggr  \}.
\end{align*}
By making analogy to importance sampling, we can 
define $w(x) := \frac{f(x)}{g(x)}$ 
and express the acceptance probability as follows:
\begin{align*}
a(x,z) = \min \biggl\{1, \frac{w(z)}{w(x)}  \biggr\}.
\end{align*}
This shows that moves $x \mapsto z$ for which $w(z)\ge w(x)$ are always accepted.

In practice, independence samplers often exhibit poor performance. Nonetheless, they remain a valuable component of the MCMC toolkit due to their well-understood theoretical properties.
To build
intuition and illustrate the theory, we compare the independence sampler with rejection sampling. 
In particular,
we establish a relation between the acceptance probability 
in
rejection sampling and the acceptance rate of the independence sampler at stationarity. 

\begin{theorem}[Acceptance Rate of Independence Sampler\index{independence sampler!acceptance rate}]\label{thm:acceptanceindependencesampler}
Suppose that the target $f$ and the proposal distribution $g$ 
for an independence sampler satisfy, 
for some $M \ge 1,$
\begin{equation}\label{eq:bound}
f(x) \leq M g(x), \quad \quad \forall x \in E.
\end{equation}
Then, 
at stationarity, the independence sampler has
acceptance rate 
at least
$1/M.$
\end{theorem}
\begin{proof}
Let $X\sim f$ and let $Z \sim g.$ We work in the continuous setting so that $\Prob\bigl( f(z) g(x) = f(x) g(z)\bigr) = 0.$   We need to show that 
\begin{align*}
\Expect\bigr[a(X,Z)\big] =\Expect\biggr[\min  \Bigl\{1, \frac{f(Z)}{f(X)}  \frac{g(X)}{g(Z)}  \Bigr\} \biggr]  \geq \frac{1}{M}.
\end{align*}
We express the minimum 
using
indicator functions  
\begin{align*}
\min \biggl\{1, \frac{f(Z)}{f(X)}  \frac{g(X)}{g(Z)}  \biggr\}
={\bf{1}}_{\bigr\{\frac{f(Z)}{f(X)}  \frac{g(X)}{g(Z)} > 1\bigr\}}
+   \frac{f(Z)}{f(X)}  \frac{g(X)}{g(Z)} {\bf{1}}_{\bigl\{\frac{f(Z)}{f(X)}  \frac{g(X)}{g(Z)} \leq 1\bigr\}}.
\end{align*}
Thus,
\begin{align*}
\Expect\bigr[a(X,Z)\big] = \int \int  {\bf{1}}_{\{ f(z) g(x) > f(x) g(z) \}}& f(x) g(z) \, dx dz \\
& \hspace{-0.5cm}+ \int \int   \frac{f(z)}{f(x)}  \frac{g(x)}{g(z)}  {\bf{1}}_{\{ f(z) g(x) \leq f(x) g(z) \}} f(x) g(z) \, dx dz. 
\end{align*}
Using symmetry and the bound \eqref{eq:bound}, we obtain
\begin{align*}
 \Expect\bigr[a(X,Z)\big] &= 2 \int \int  {\bf{1}}_{\{ f(z) g(x) > f(x) g(z) \}} f(x) g(z)\, dx dz \\
&\geq 2 \int \int  {\bf{1}}_{\{ f(z) g(x) > f(x) g(z) \}} f(x) \frac{ f(z)}{M}\, dx  dz \\
&= \frac{2}{M}   \Prob \Bigl( f(X_2) g(X_1) > f(X_1) g(X_2) \Bigr) \\
&= \frac{2}{M} \frac{1}{2} = \frac{1}{M},
\end{align*}
where  $X_1, X_2 \stackrel{\text{i.i.d.}}{\sim}  f.$  \hfill $\square$
\end{proof}

\begin{mybox}[colback=white]{Pros and Cons}
Theorem \ref{thm:acceptanceindependencesampler} shows that the acceptance rate of the independence sampler at stationarity is 
at least
the acceptance rate $1/M$ of rejection sampling established in Theorem \ref{th:rejectionsampling}. Moreover, 
unlike rejection sampling,
the independence sampler does not 
discard rejected draws. 
However, 
the independence sampler 
produces correlated draws and in practice is 
initialized far from stationarity 
---several iterations will be needed to approximately reach 
stationarity (see Theorem \ref{thm} below). 
Among Metropolis--Hastings algorithms, the independence sampler is  rarely used in practice and is mainly of theoretical importance. 
\end{mybox}

As 
in
rejection sampling, it is advisable to choose the independence sampler proposal $g$ to be close to the target $f$, while at the same time being easy to sample from. 
In the 
extreme case where $g=f$, the algorithm 
produces exact draws from $f$, the acceptance probability is 1, 
and the Markov chain 
is stationary from the outset.
However, 
choosing $g =f$ is clearly not useful in the interesting case where sampling $f$ directly is not possible. 

The following result establishes uniform ergodicity of the independence sampler. The notions of total variation distance \index{total variation distance} and uniform ergodicity \index{ergodicity!uniform} are reviewed in Definitions \ref{def:totalvariation} and \ref{def:geometricuniformergodicity}.

\begin{theorem}[Uniform Ergodicity of Independence Sampler\index{independence sampler!uniform ergodicity}]\label{thm}
Let $\pind$ be the Markov kernel of the independence sampler with proposal kernel $\qind(x,z) = g(z)$. Suppose that there is $M\ge1$ such that 
$f(x) \leq M g(x)$ for every  $x \in E.$
Then, for any $x \in E,$
\begin{align*}
\dtv \bigl(\pind^n(x, \cdot), f \bigr) \leq \Bigl(1 - \frac{1}{M} \Bigr)^n.
\end{align*}
\end{theorem}

\begin{proof}
We only provide a proof sketch. 
For
 $x \not = z$, 
\begin{align*}
\pind(x, z) 
= \qind(x,z) a(x,z)
= g(z) \min  \biggl\{1, \frac{f(z)}{f(x)}  \frac{g(x)}{g(z)}  \biggr\}
= \min  \biggl\{g(z), \frac{g(x) f(z)}{f(x)}  \biggr\}
\geq \frac{f(z)}{M}.
\end{align*}

The 
inequality $\pind(x, z) \geq \frac{f(z)}{M}$ also holds when  $x = z.$ 
This bound allows us to conclude the result with the same coupling argument used to show ergodicity of finite state space Markov chains in Theorem \ref{thm:ergodicity}.  A sketch is given below.

Let 
$\{B_n\}_{n=0}^\infty$ be a sequence of i.i.d. Bernoulli$(\frac{1}{M})$ random variables, independent of all other randomness, and define
$$W_{n+1} \sim \begin{cases}
s(W_n,\cdot)&{\text{if} } \,\, B_n=0,\\
r(W_n,\cdot)&{\text{if} } \,\, B_n=1,\\
\end{cases}$$
where $r(x,z) = f(z)$ and $$s(x,z)=\frac{\pind(x,z)-\frac{1}{M} f(z)}{1-\frac{1}{M}}.$$ The bound $\pind(x, z) \geq \frac{f(z)}{M}$ ensures that $s$ defines a Markov kernel. Moreover,  
the process
$\{W_n\}$ has transition kernel 
$\pind$.
The rest of the proof 
is identical to that of 
Theorem \ref{thm:ergodicity}. \hfill $\square$
\end{proof}

Notice the common underlying structure between Theorem \ref{thm} and Theorem \ref{thm:ergodicity}. In both cases, we use
a lower bound on the kernel $P$ 
namely that there
exist
$ \delta > 0$  and a probability measure $\nu$ such that $$P^m(x,A) \geq \delta \nu(A), \quad \quad \forall x \in E, \, \forall A \subseteq E.$$
Indeed, this minorization condition is equivalent to ergodicity of the chain. 
Whenever it holds, 
\begin{align*}
\dtv(P^n(x, \cdot), f) \leq (1-\delta)^{[\frac{n}{m}]}. 
\end{align*}

\begin{remark}
    It can be shown that if $\emph{ess inf} \frac{g(z)}{f(z)} = 0\, $, then the independence sampler is not even geometrically ergodic. We refer to the bibliography in Section \ref{sec:bibliochapter3} for additional details. 
\end{remark}

\paragraph{Scaling for the Independence Sampler \index{independence sampler!scaling}}
Suppose we 
wish
to use the independence sampler with a proposal distribution chosen from 
a
parametric family $\{ g_\theta : \theta \in \Theta \}$. For example, we may decide to use a normal distribution and seek to determine appropriate values for its mean and variance.
This type of problem is referred to as the
``scaling'' of Metropolis--Hastings methods.

Recall that we have shown that if $\frac{f(x)}{g(x)} \leq M$, then the independence sampler is uniformly ergodic with rate $1 - \frac{1}{M}.$ This suggests choosing the parameter $\theta$ 
so as to minimize $M_\theta$ subject to the constraint
$\frac{f(x)}{g_\theta(x)} \leq M_\theta$. Note that this is the same constrained optimization problem that arises in rejection sampling.
In many problems of interest, finding an upper bound for the ratio $f/g$ is infeasible.
In such cases, a practical approach is to
monitor the acceptance rate of the independence sampler 
across different values of $\theta$, and select 
the proposal $g_\theta$ that 
yields
the highest acceptance rate.

\subsection{Random Walk Metropolis \index{Random Walk Metropolis (RWM)}}
A  Random Walk Metropolis (RWM) algorithm is a Metropolis--Hastings method
in which the
proposal Markov kernel is chosen to be of the form
$$\qrwmh(x,z) = g(z-x)$$
for some distribution $g$.

\begin{algorithm}[H]
  \caption{Random Walk Metropolis (RWM) \label{algo:RWMH}}
  \begin{algorithmic}[1]
  \STATEx{\textbf{Input:} Target distribution $f$, initial distribution $\pi_0$, proposal distribution $g,$ sample size $N.$}
  \STATE{Define the Markov kernel
  $$\qrwmh(x,z) := g(z-x).$$}
  \STATE{Run Metropolis--Hastings (Algorithm \ref{algo:MH}) with inputs $f,$ $\pi_0,$ $\qrwmh(x,z)$, $N$.}
  \vspace{2mm}
  \STATEx{\textbf{Output:} Sample $\{X^{(n)}\}_{n=1}^N$.}
 \end{algorithmic}
\end{algorithm}

The name RWM comes from the fact that proposals are 
generated
according to a random walk: proposed moves $Z^*$ may be obtained by sampling $\xi^* \sim g$ independently of all other randomness and setting
\begin{align*}
Z^* = X^{(n)} + \xi^*.
\end{align*}
The acceptance probability is 
\begin{align*}
a(x,z) = \min \biggl\{1, \frac{f(z)}{f(x)}  \frac{g(x-z)}{g(z-x)}  \biggr\}.
\end{align*}
The original paper by Metropolis et al.
proposed RWM with $g$ symmetric about $0$, 
which simplifies the expression of the acceptance probability to
\begin{align*}
a(x,z) = \min \biggl\{1, \frac{f(z)}{f(x)} \biggr\}.
\end{align*}
Common choices 
for the
distribution $g$ include the normal distribution,
Student's $t$-distribution, and 
uniform (in some bounded subset, when $E$ is unbounded).

\begin{remark}
If $E = \mathbb{R}^d$ and $\qrwmh(x,z) = g(x-z) = g(z-x)$ for every $x,z \in E$, then the RWM kernel is not uniformly ergodic for any $f$.
However, if the target $f$ is log-concave in the tails and symmetric,
and $g > 0$ is continuous, then the RWM kernel is geometrically ergodic. If $f$ is not symmetric, the 
same
conclusion still holds 
provided that
additionally $g(x) \leq \ell e^{-\alpha |x|}$, where $\ell > 0$ and $\alpha$ is 
a constant quantifying the concavity
of $\log(f)$. We refer to the bibliography in Section \ref{sec:bibliochapter3} for additional details.
\end{remark}

\begin{figure}[!htb]
    \centering
    \includegraphics[width=1\linewidth]{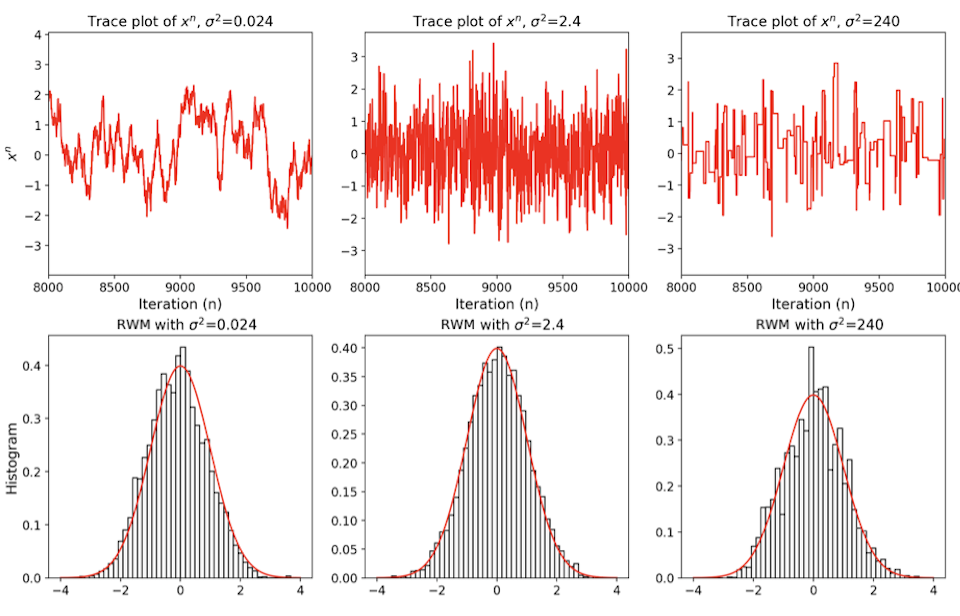}
\caption{\label{figure:RWMH}RWM with standard Gaussian target $f = \Nc(0,1)$ and three proposals $g_\sigma = \Nc(0,\sigma^2)$ with $\sigma^2 = 0.024, 2.4, 240.$ The choice $\sigma^2 = 2.4$ leads to good mixing of the chain, whereas $\sigma^2 = 0.024$ leads to 
poor
exploration and $\sigma^2 = 240$ leads to many rejections.}
\end{figure}

\paragraph{Scaling for RWM \index{Random Walk Metropolis (RWM)!scaling}}
While for the independence sampler it is advantageous to maximize the expected acceptance probability, this is \textit{not} the case for RWM. 
To build intuition, consider the
case where $f = \mathcal{N}(0,1)$ and $g_\sigma = \Nc(0, \sigma^2)$, as in Figure \ref{figure:RWMH}. Then,

\begin{itemize}
\item 
Small $\sigma^2$ leads to high acceptance rate but poor exploration of the state space.
\item
Large $\sigma^2$ leads to good exploration but low acceptance rate.
\end{itemize}
There are some widely applicable guidelines for tuning proposals.

\begin{enumerate}
\item 
First, in the simple case of a normal target with a 
normally
distributed random walk proposal, the optimal choice of variance $\sigma^2$ can be determined by minimizing the integrated autocorrelation time 
of the chain
(see the definition in Lemma \ref{lemmaMarkovchainerror}).
The optimal value is approximately $\sigma^2=2.4$, 
with corresponding acceptance rate
\begin{align*}
\alpha = \frac{2}{\pi} \arctan\Bigl(\frac{2}{\sigma}\Bigr) = 0.44.
\end{align*}
This has led to the practical recommendation
of aiming for an acceptance rate around $0.5$ in low-dimensional problems.
\item
Second, analysis of RWM in the setting
\begin{align*}
f_d( x) &= f(x_1) \times \dotsb \times f(x_d), \quad x = (x_1, \ldots, x_d) \in \mathbb{R}^d, \\
g_{d}(x) &= \Nc(0, \sigma_d^2 I_d),
\end{align*}
reveals two important insights.

First, in the large-$d$ asymptotic, $\sigma_d^2$ should be scaled as $1/d$; this, in turn, implies that the convergence time of the algorithm scales like $\mathcal{O}(d)$.

Second, in a precise sense and under suitable assumptions, the optimal acceptance probability for large $d$ is roughly $0.234.$ This result has motivated the general recommendation of aiming for an acceptance rate around $ 1/4$ in high-dimensional problems. We refer to the bibliography in Section \ref{sec:bibliochapter3} for additional details. 
\end{enumerate}

\section{Implementation of Markov Chain Monte Carlo Methods\index{Markov Chain Monte Carlo (MCMC)!implementation issues}}\label{ssec:implementation}

\subsection{Initialization Bias and Burn-In \index{burn-in}}
MCMC chains are 
typically
initialized outside stationarity (if we could sample the initial draw from the target
we would not be doing MCMC!).  
To reduce the initialization bias 
induced by
the starting value, 
the first $M$ draws can be discarded. Estimation is then based on the states visited after time $M$:
$$\mathcal{I}_f[h] \approx \frac{1}{N-M} \sum_{n=M+1}^N h(X^{(n)}).$$
The initial phase up to time $M$ is 
referred to as the \emph{transient phase} or \emph{burn-in period}.
How should the length of the burn-in period be chosen?
A first step 
would be to 
examine the output of the chain 
visually.
While this approach is crude, it is quick and inexpensive.
However, this method should be complemented with more systematic ones.

  \FloatBarrier
    \begin{figure}[htp]
      \centering
      \includegraphics[width=0.75\columnwidth]{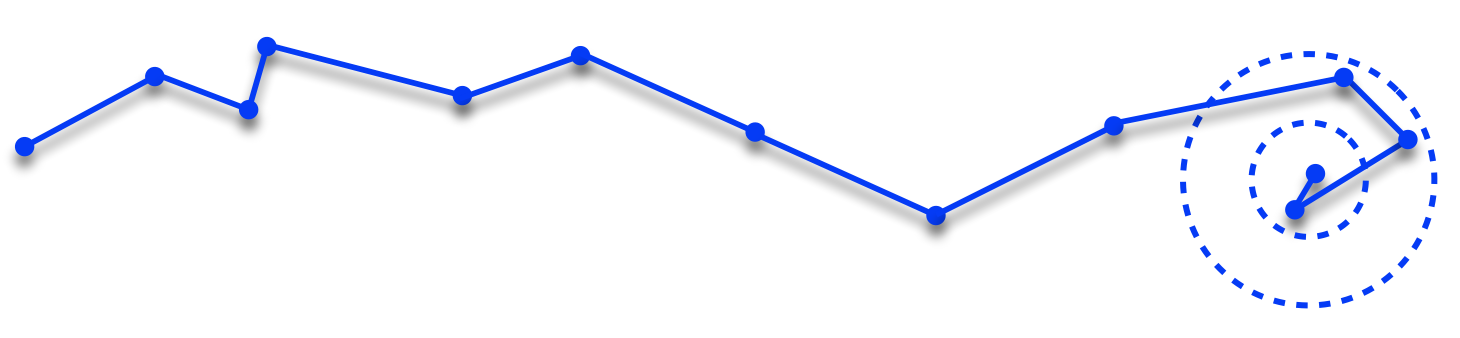}
      \caption{The MCMC chain is typically started outside stationarity, and it may take several iterations to reach a region of high target density.}
      \label{fig:MCMC_chain}
    \end{figure}
    \FloatBarrier

\subsection{Effective Sample Size \index{effective sample size}}
Recall from Chapter \ref{chap:MCintegration} that the effective sample size in \eqref{eq:ESSimportancesampling} provides a convenient 
way to monitor 
the performance of (autonormalized) importance sampling. In the context of MCMC, 
one can monitor an estimator for $\mathcal{I}_f[h] = \Expect_{X\sim f}\bigl[h(X)] $ using the effective sample size 
\begin{equation}\label{eq:essMCMC}
 \text{ESS} := \frac{N}{ 1 + 2 \sum_{k=1}^\infty \rho_k} , 
\end{equation}
where $\rho_k$ is the 
lag $k$ autocorrelation\index{autocorrelation}; see Definition \ref{ref:autocovariance}. Thus, when the samples are positively correlated, ESS is smaller than the number $N$ of MCMC iterations.
The definition \eqref{eq:essMCMC} 
is motivated by Theorem \ref{thm:CLT}, which shows that the asymptotic variance of ergodic averages based on Markov chain sampling is determined by the integrated autocorrelation time 
 $\tau^2:= \sigma^2 \Big[ 1 + 2\sum_{k=1}^\infty \rho_k \Bigr]$, where $\sigma^2 = \mathbb{V}_{X\sim f} \bigl[h(X)\bigr].$ Consequently, for large $N$ it holds that $\mathbb{V}\bigl[\mathcal{I}_f^{\text{\tiny MH}}[h]\bigr] \approx \sigma^2/ \text{ESS}$ for the Metropolis--Hastings estimator \eqref{eq:MHestimator}, while  $\mathbb{V} \bigl[\mathcal{I}_f^{\text{\tiny MC}}[h] \bigr] = \sigma^2/N$ for the classical Monte Carlo estimator in Algorithm \ref{alg:MCI}.

However, the ESS \eqref{eq:essMCMC} 
has two important limitations.
First,  while the ESS in \eqref{eq:ESSimportancesampling} is easy to compute from importance sampling output, approximating \eqref{eq:essMCMC} from MCMC output can be challenging (in fact, several software packages have been developed for this purpose). Second, the ESS in \eqref{eq:essMCMC} is motivated, as explained above, by its connection with the asymptotic variance of ergodic averages, and it does not reflect the number of iterations required for the chain to (approximately) reach its stationary distribution.  
For these reasons, the ESS in \eqref{eq:essMCMC} is less commonly used than its importance sampling counterpart. This motivates us to consider in the next subsection other approaches to assess the convergence and monitor the performance of MCMC algorithms.

\subsection{Assessing Convergence \index{Markov Chain Monte Carlo (MCMC)!assessing convergence}}
Assessing convergence is a central and often fiercely debated topic in the MCMC literature.
Part of the difficulty lies in the fact that there exist multiple notions of convergence.
First, 
there is the issue of whether the distribution 
of the chain is close to the stationary distribution. Second, there is the issue of how well the chain explores the state space. Generally, chains that do not explore the state space well (we say, mix slowly) tend to approach stationarity more slowly. 
Third, one may ask how accurately the ergodic averages approximate the desired integrals. Even when a chain is at stationarity, slow mixing can lead to highly correlated samples and therefore inaccurate estimates.
These issues are closely related but not identical, which can be a source of confusion.
For example, consider a chain with a bi-modal stationary distribution that has only explored one mode. In this case, the chain has not yet reached equilibrium. However, if one is estimating the integral of a function that takes the same value on both modes, the corresponding ergodic average may still converge quickly.

There are two general approaches to assess convergence to stationarity: convergence rate calculation and convergence diagnostics.

\begin{enumerate}
\item {\bf Convergence rate calculation}: 
While analyzing the independence sampler, we effectively performed a convergence rate calculation.
The advantage of this approach is, of course, that it is exact. 
However, such
calculations are often
difficult and will only apply to specific cases. Because they are exact methods, they can be overly pessimistic as they take into account any worst-case scenario even if it has a very small probability of occurring. At times they can be so pessimistic that the suggested burn-in period simply becomes impractical.
\item {\bf Convergence diagnostics\index{convergence diagnostics}}: This is the most widely used 
approach in practice.
Convergence diagnostics examine the output of the chain and attempt
to detect features that indicate a lack of convergence.
They are generally
easy to implement. However, while convergence diagnostics may indicate that convergence has not been reached, they are not able to guarantee convergence.
\end{enumerate}
    
   We will now examine some of the most commonly used diagnostics. 
    \paragraph{Geweke's Method \index{convergence diagnostics!Geweke}}
    Geweke proposed a 
   diagnostic
    based on spectral analysis of time series. Consider two subsequences from a single MCMC run;
    one consisting of states immediately after burn-in and 
    the other 
    consisting of states at the very end of the run. Suppose that the run is of length $N$ and we want to examine if a burn-in of $M$ steps is sufficient. 
    Consider 
    ergodic averages
    $$ \mathcal{I}_f^A[h]: = \frac{1}{N_A} \sum_{n= M+1}^{M+N_A} h(X^{(n)}) , \quad \quad  \mathcal{I}_f^B[h]: = \frac{1}{N_B} \sum_{n= N-N_B}^{N} h(X^{(n)}) $$
with $M + N_A < N - N_B.$ If the chain has converged, then the two ergodic averages should be similar. 
The extent to which they differ can be
quantified using an appropriate central limit theorem. 
Under stationarity, the variance of ergodic averages
is determined by the spectral density of the time series $\{h(X^{(n)}) \}_{n=1}^\infty,$ defined as 
$$S(\omega) = \frac{1}{2\pi} \sum_{k = -\infty}^ \infty \gamma_k \exp(ik \omega),$$
where $\gamma_k = \operatorname{Cov} \bigl( h(X^{(0}), h(X^{(k)}) \bigr)$  is the $k$-lag autocovariance of the time series (see Definition \ref{ref:autocovariance} in Appendix \ref{chap:markovchains}). 

If the ratios 
$N_A/N$  and $N_B/N$ are fixed with $(N_A + N_B)/(N-M) <1,$ then, under stationarity and as $N \to \infty,$
\begin{equation*}
\frac{\mathcal{I}_f^A[h]  - \mathcal{I}_f^B[h] }{\sqrt{ \frac{1}{N_A} \hat{S}_A(0) +   \frac{1}{N_B} \hat{S}_B(0)    }}  \to \Nc(0,1),
\end{equation*}
where $\hat{S}_A(0)$ and $\hat{S}_B(0)$ are 
estimates of the spectral density at zero (i.e. the asymptotic variance).
This result can be used to
test whether the means of the two 
subsequences
are statistically indistinguishable.
Geweke originally suggested 
using the values  
$N_A = N/10$ and $N_B = N/2.$

    \paragraph{Gelman and Rubin's Method \index{convergence diagnostics!Gelman and Rubin}}
    This method uses multiple chains initialized at over-dispersed starting points relative to the stationary distribution.
    The idea is that, in the presence of initialization bias,
    the chains may be concentrated around different modes, whereas under stationarity they should behave similarly. In particular, if a
    chain tends to become trapped in a single mode, its trajectory may exhibit stable local fluctuations despite not having converged. This phenomenon is known as \emph{meta-stability}, and is often associated with multi-modal target distributions.
    Gelman and Rubin's method draws on ideas from Analysis of Variance (ANOVA) and relies on approximate normality.
    We assume each of a total of $J$ chains are run for $2N$ iterations where the first $N$ iterations are classified as burn-in. 
   For each chain, compute the ergodic average based on the remaining $N$ iterations, and define the between-chain variance
    $$ B = \frac{1}{J-1} \sum_{j= 1}^J \Bigl( \mathcal{I}_f^j[h] - \bar{\mathcal{I}}_f^j  [h] \Bigr)^2,
    \quad
    \bar{\mathcal{I}}_f^j  [h] = \frac{1}{J}\sum_{j=1}^J \mathcal{I}_f^j[h],
    $$
    where $\bar{\mathcal{I}}_f^j  [h]$ is the average of the ergodic averages across chains. This quantity measures the variability between chains and tends to zero as $N\to\infty$ under convergence.
    
    Next, define the within-chain variance
    $$W = \frac{1}{J} \sum_{j = 1}^J \frac{1}{N-1}  \sum_{n= N+1}^{2N} \Bigl( h(X_j^{(n)}) - \mathcal{I}_f^j[h]    \Bigr)^2,$$
    where $X_j^{(n)}$ denotes the $n$-th draw of the $j$-th chain. This measures the average variability within the chains.
    
An estimate of the target variance is then obtained by combining the two quantities:
$$V = \frac{N-1}{N} W + B$$
The quantity $R=\sqrt{V/W}$ is called the \emph{potential scale reduction factor}. Typically, $R>1$ and $R \to 1$ as the chains approach stationarity.
Under approximate normality,
one can show that $R$ has an $F$-distribution and the hypothesis of $R = 1$ can be tested. As a rule of thumb, if the $0.975$-quantile is less than $1.2$, no lack of convergence has been detected.

    \paragraph{Raftery and Lewis' Method \index{convergence diagnostics!Raftery and Lewis}}
In this diagnostic, a posterior quantile $p$  and an acceptable tolerance $r$ for error in computing $p$ are specified. In addition, the user specifies the desired probability $\alpha$ of obtaining an estimate for $p$ within the prescribed error tolerance. Raftery and Lewis suggested a way to estimate the number $N$ of iterations and the burn-in period $M$ that are necessary to satisfy the prescribed conditions.

\section{Discussion and Bibliography}\label{sec:bibliochapter3}
The Metropolis--Hastings algorithm was introduced in \cite{metropolis1953equation} restricted to the case of Gaussian random walk proposals; the extension to user-chosen proposal kernel was 
developed
in \cite{hastings1970monte}. The paper \cite{haario2001adaptive} introduced an adaptive strategy, where the proposal kernel is updated along the process using the full information gathered so far. Historical overviews of the Metropolis--Hastings algorithm and related MCMC methods include  \cite{robert2011short,dunson2020hastings}. Many books  \cite{gilks1995markov,brooks2011handbook,robert2013monte} and survey articles \cite{brooks1998markov,tierney1994markov} are devoted to the Metropolis--Hastings algorithm; adaptive strategies are surveyed in \cite{andrieu2008tutorial,roberts2009examples}.

As shown in Theorem \ref{thm:detailedbalanceMCMC}, the Metropolis--Hastings algorithm is designed to ensure detailed balance with respect to the target; in other words, Markov chains generated using the Metropolis--Hastings algorithm are \emph{reversible}. The paper \cite{bierkens2016non} shows that the algorithm can be modified to generate \emph{non-reversible} chains that may have better properties in terms of mixing behavior or asymptotic variance.
Uniform ergodicity of the independence sampler and geometric ergodicity of RWM were established in \cite{mengersen1996rates}.
A comparison between rejection sampling, importance sampling, and the independence sampler can be found in \cite{liu1996metropolized}. The study of optimal scaling for Metropolis--Hastings algorithms,\index{Metropolis--Hastings!optimal scaling} which originates from \cite{roberts1997weak,roberts1996exponential} and is reviewed in \cite{roberts2001optimal}, is still an active area of research \cite{yang2019optimal}. The convergence diagnostics by Raftery, Geweke, and Gelman and Rubin were introduced in \cite{raftery1991many,geweke1991evaluating,gelman1992inference}. We refer to \cite{mykland1995regeneration,yu1998looking,vehtari2021rank,vats2021revisiting} for other important examples of convergence diagnostics, and to \cite{brooks1998convergence} for an early survey on this topic. 

An important feature of the Metropolis--Hastings algorithm is that it can be implemented without knowledge of the normalizing constant of the target. However, in many applications such as Bayesian inverse problems \cite{stuart2010inverse}, machine learning with tall data \cite{bardenet2017markov}, and genetics \cite{beaumont2003estimation} evaluating the unnormalized target density to compute the Metropolis--Hastings acceptance probability can be expensive. If a \emph{deterministic} cheap-to-evaluate approximation of the target is available, this issue can be alleviated using delayed-acceptance Metropolis--Hastings algorithms \cite{christen2005markov,efendiev2006preconditioning,efendiev2005efficient,sherlock2017adaptive}. The idea is similar to that behind the envelope rejection sampling algorithm studied in Chapter \ref{chapter1}: first, we do an accept/reject step with the cheap-to-evaluate target approximation; if the move is accepted, we do a second accept/reject step with the expensive-to-evaluate target. If a \emph{stochastic} unbiased estimator of the unnormalized target density is available, pseudo-marginal MCMC\index{Markov Chain Monte Carlo (MCMC)!pseudo-marginal} algorithms \cite{beaumont2003estimation,andrieu2009pseudo} can be employed. Particular instances of pseudo-marginal MCMC algorithms are particle MCMC \index{Markov Chain Monte Carlo (MCMC)!particle} algorithms \cite{andrieu2010particle} that rely on unbiased estimates computed via particle filters (see Chapter \ref{chap:particlefilters}). Particle MCMC algorithms are powerful algorithms for parameter estimation in hidden Markov models, see e.g. \cite{golightly2011bayesian}.

In this chapter, we introduced the Metropolis--Hastings algorithm on a countable or uncountable state space $E\subseteq \R^d.$ Reversible jump MCMC \index{Markov Chain Monte Carlo (MCMC)!reversible jump} algorithms \cite{green1995reversible} allow to sample posterior distributions on spaces of varying dimension; this important extension of the Metropolis--Hastings framework enables Bayesian inference in applications where the number of parameters in the model is unknown. An adaptive, delayed-rejection algorithm for reversible jump MCMC was introduced in \cite{green2001delayed}.  Another important extension of the Metropolis--Hastings framework concerns its formulation in general state spaces \cite{Tie,roberts2004general},\index{Metropolis--Hastings!general state space} which enables for instance function space sampling \index{Markov Chain Monte Carlo (MCMC)!function space} \cite{cotter2013mcmc}.
Finally, affine-invariant samplers, introduced in \cite{affine}, are ensemble MCMC methods that evolve a collection of samples, called \emph{walkers}, using proposal mechanisms that are invariant under affine transformations of the state space. This invariance makes them especially effective for target distributions with strong anisotropy or linear correlations, since their performance is less sensitive to rescaling, rotation, or stretching of the coordinates. 

A common limitation of MCMC methodology is 
that it is not easily parallelizable due to its sequential structure. 
We refer to \cite{calderhead2014general,glatt2022parallel,schwedes2021rao} for attempts to parallelize Metropolis--Hastings algorithms. These approaches often propose several moves in each iteration, an idea also found in \cite{liu2000multiple}. Finally, we remark that MCMC algorithms that rely on piecewise deterministic Markov processes \index{piecewise deterministic Markov processes} are studied for instance in \cite{fearnhead2018piecewise,bierkens2019zig,bouchard2018bouncy}.

\chapter{Gibbs Sampling}
\label{chap:gibbs}

The Gibbs sampler\index{Gibbs sampler} is a Markov chain Monte Carlo algorithm for multivariate target distributions. Gibbs samplers operate coordinate-wise: each new draw differs from the previous one in only a subset of variables, which are updated by sampling from the distribution of those variables conditioned 
on
all others.
This idea is extremely powerful when sampling conditional distributions is easier than sampling the joint distribution. In many problems where the conditionals are known, Gibbs samplers help avoid the curse of 
dimension 
by sampling these low-dimensional conditional distributions instead of directly sampling the joint target distribution. 

Gibbs samplers were conceived and developed independently of the Metropolis--Hastings algorithms studied in Chapter \ref{chap:MCMC}, and their design and analysis have a unique flavor. However,
 Gibbs samplers can be seen as a particular instance of the Metropolis--Hastings framework, where the proposal kernel is defined by the so-called \emph{full conditionals} of the target distribution. Remarkably, such a choice of proposal kernel implies that proposed moves are always accepted.  Hence, in contrast to other Metropolis--Hastings algorithms that rely on a step-size parameter to control the acceptance rate, Gibbs samplers require no tuning. On the other hand, the parameterization of the problem can severely impact the behavior of the algorithm. Specifically, Gibbs samplers may converge slowly if the variables are strongly correlated or, relatedly, when the target is poorly conditioned. Thus, care should be taken to choose a convenient parameterization of the variables of interest and to understand their dependencies before implementing the Gibbs sampler.

This chapter is organized as follows. Section \ref{sec:example} introduces the idea behind the Gibbs sampler through a simple example. The main algorithm is presented in Section \ref{sec:fullconditionalsandGibbs}, where we also provide the formal definition of full conditionals, show that Gibbs samplers belong to the Metropolis--Hastings family, and prove that choosing full conditionals as proposal kernels results in moves that are always accepted. In Section \ref{sec:convergenceGibbs} we study the convergence of Gibbs samplers in finite state spaces and for Gaussian target distributions, showing that the algorithm converges slowly if the variables are highly correlated. Section \ref{sec:Gibbsbibliography} closes with bibliographical remarks.

\section{Motivating Example}\label{sec:example}
Consider a Gaussian likelihood \index{likelihood} model
$$
f(y|\mu, \sigma) = \frac{1}{(2\pi\sigma^2)^{K/2}} \exp \biggl(-\frac{1}{2\sigma^2} \sum_{k = 1}^K (y^{(k)} - \mu)^2\biggr),
$$
where $y = \{y^{(k)}\}_{k = 1}^K$ is the observed data. 

We take a Bayesian approach and choose the following priors for the precision $\tau = \frac{1}{\sigma^2}$ and mean $\mu$:
\begin{equation*}
    \begin{aligned}
        & \tau \sim \text{Gamma}(\alpha, \beta),\\
        & \mu \sim \Nc(m, s^2).
    \end{aligned}
\end{equation*}
The posterior for $\mu, \tau$ is 
$$
f(\mu, \tau | y) \propto \exp \biggl(-\frac{1}{2s^2}(\mu - m)^2 \biggr)\exp\biggl(-\frac{\tau}{2} \sum_{k = 1}^K (y^{(k)} - \mu)^2\biggr) \exp\bigl(-\beta \tau \bigr)\tau^{\alpha - 1 + \frac{K}{2}}.
$$
There is no closed form expression for the normalizing constant. \index{normalizing constant} Note that:
\begin{enumerate}
    \item Conditional on $\tau$,\\
    $$ \mu | \tau \sim \Nc\biggl(\frac{K\bar{y}\tau + ms^{-2}}{K\tau + s^{-2}}, \frac{1}{K\tau + s^{-2}} \biggr),   $$
    where $\bar{y}$ is the sample average. 
    
    \item Conditional on $\mu$,\\
    $$
    \tau | \mu \sim \text{Gamma} \biggl(\alpha + \frac{K}{2}, \beta + \frac{1}{2} \sum_{k =1}^K (y^{(k)} - \mu)^2 \biggr).
    $$
\end{enumerate}
 Both conditionals are straightforward to sample from,
and it is hence natural to use them to define proposal kernels. 
Suppose that $(\mu^{(n)}, \tau^{(n)})$ are given. Inspired by the Metropolis--Hastings \index{Metropolis--Hastings} algorithm studied in Chapter \ref{chap:MCMC}, we generate
$(\mu^{(n+1)}, \tau^{(n+1)})$ as follows:
\begin{enumerate}
    \item Propose
    \begin{equation*}
        \begin{aligned}
             \mu^* & \sim \Nc\biggl(\frac{K\bar{y}\tau^{(n)} + ms^{-2}}{K\tau^{(n)} + s^{-2}}, \frac{1}{K\tau^{(n)} + s^{-2}} \biggr).
        \end{aligned}
    \end{equation*}

    \item Accept $\mu^{(n+1)} = \mu^*$ with probability  $a\bigl((\tau^{(n)}, \mu^{(n)}),(\tau^{(n)}, \mu^{*}) \bigr).$\\ Otherwise, set $\mu^{(n+1)} = \mu^{(n)}$. 
    
    \item Propose
    \begin{equation*}
        \begin{aligned}
             \tau^* & \sim \text{Gamma} \biggl(\alpha + \frac{K}{2}, \beta + \frac{1}{2} \sum_{k =1}^K (y^{(k)} - \mu^{(n+1)})^2 \biggr).
        \end{aligned}
    \end{equation*}

    \item Accept  $\tau^{(n+1)} = \tau^*$ with probability $a\bigl((\tau^{(n)}, \mu^{(n+1)}),(\tau^{*}, \mu^{(n+1)})\bigr).$\\ Otherwise set $\tau^{(n+1)} = \tau^{(n)}$. 
\end{enumerate}
The fact that the normalizing constant \index{normalizing constant} of the posterior \index{posterior} is unknown does not obstruct the implementation of this 
scheme.
Moreover, the theory we will develop in the next section will show that the acceptance probability in steps 2 and 4 is one. Therefore, 
the algorithm simplifies to:
            \begin{enumerate}
                \item Sample 
                \begin{equation*}
                         \mu^{(n+1)}  \sim \Nc\biggl(\frac{K\bar{y}\tau^{(n)} + ms^{-2}}{K\tau^{(n)} + s^{-2}}, \frac{1}{K\tau^{(n)} + s^{-2}} \biggr).
                \end{equation*}
                
                \item Sample 
                \begin{equation*}
                         \tau^{(n+1)}   \sim \text{Gamma} \biggl(\alpha + \frac{K}{2}, \beta + \frac{1}{2} \sum_{k =1}^K (y^{(k)} - \mu^{(n+1)})^2 \biggr).   
                \end{equation*}
            \end{enumerate}
We have derived a Gibbs sampler for the posterior distribution $f(\mu,\tau|y).$

\section{Full Conditionals \index{full conditional} and the Gibbs Sampler}\label{sec:fullconditionalsandGibbs}
We next define the concept of full conditionals. For clarity, we give the definition in both discrete and continuous cases. 

\begin{definition}[Full Conditionals]
\begin{enumerate}[i)]
    \item If $E$ is discrete and $f(x) = f(x_1, \ldots, x_d)$ is the p.m.f.~of a random vector $X = (X_1,\ldots, X_d)$ taking values in $E \times \cdots \times E = E^d,$ we denote 
    
   $$ f(x_{-j}) = \sum_{\xi\in E} f(x_1,\ldots, x_{j-1}, \xi, x_{j+1},\ldots, x_d),$$
    and define the $j$-th full conditional p.m.f.
    \begin{equation*}
        \begin{aligned}
          f_j(x_j|x_{-j})  =   f_{j}(x_{j} | x_{i}, i \neq j) & = \Prob(X_j = x_j | X_i = x_i, i \neq j)\\
            & = \frac{\Prob(X_j = x_j, X_i = x_i, i \neq j)}{\Prob(X_i = x_i, i \neq j)}\\
            & = \frac{f(x_1,\ldots,x_{j-1}, x_{j},x_{j+1}, \ldots x_{d})}{f(x_{-j})}.
        \end{aligned}
    \end{equation*}
    
    \item If $E$ is continuous and $f(x) = f(x_1,\ldots,x_d)$ is the p.d.f.~of a random vector $X = (X_1,\ldots, X_d)$ taking values in 
    $E \times \cdots \times E = E^d$, 
    we denote 
    $$
    f(x_{-j}) = \int_{E} f(x_1, \ldots, x_{j-1}, \xi, x_{j+1}, \ldots, x_d) \, d\xi,$$
    and define the $j$-th full conditional p.d.f.
    \begin{equation*}
    f_j(x_j|x_{-j}) =   f_{j} (x_j | x_i, i \neq j) = \frac{f(x_1,\ldots, x_{j-1}, x_{j}, x_{j+1}, \ldots, x_d)}{f(x_{-j})}.    \tag*{\qedhere}
    \end{equation*}
\end{enumerate} 
\end{definition}
The distribution defined by $f_j(x_j | x_i, i \neq j)$ is called the $j$-th full conditional distribution of $f$. To implement the Gibbs sampler, we will only need to know these distributions up to a normalizing constant, \index{normalizing constant} and so terms that do not depend on $x_j$ can be ignored. In particular, it is useful to note that full conditionals \index{full conditional} are proportional to the joint density.

\begin{example}[Full Conditionals of Multivariate Gaussians\label{ex:fullcondGau}]
Consider a multivariate \index{full conditional!Gaussian} Gaussian random vector $X = (X_1, \ldots, X_d) \sim \Nc( { \mu}, H^{-1}).$ Then, $X_j|X_{-j}\sim\mathcal{N}(\nu_j,H_{jj}^{-1}),$ with 
$$\nu_j=\mu_j-H_{jj}^{-1}\sum_{k\neq j}  H_{jk}(x_k-\mu_k).$$
To verify this, recall that a univariate normal random variable with mean $\alpha$ and precision $\beta$ has p.d.f. proportional to 
\begin{equation}\label{eq:meanprec}
 \exp \Bigl( -\frac12  x^2 \beta + x \beta  \alpha     \Bigr).
\end{equation}
Assume first that $\mu=0$. Since the full conditional is proportional to the joint density,
\begin{align}\label{eq:eq2meanprec}
f_j(x_j | x_{-j}) \propto \exp \Bigl( - \frac12 x ^T H  x \Bigr) 
\propto  \exp \Bigl( - \frac12   x_j^2 H_{jj}  - x_j \sum_{k\neq j} H_{jk} x_k  \Bigr).
\end{align} 
Thus $X_j |X_{-j}$ is univariate normal. Comparing the coefficients of the quadratic terms in Equations \eqref{eq:meanprec} and \eqref{eq:eq2meanprec} we see that the precision is $H_{jj}, $ and comparing the coefficients for the linear term we see that the mean is $\nu_j.$ 
The general case follows by applying the argument to $X-\mu$. 
\hfill \qedhere
\end{example}

The Gibbs sampler is based on sampling full conditionals. 
It is not immediately obvious that full conditionals ---unlike marginals--- can specify a distribution uniquely.
Necessary and sufficient consistency conditions for the full conditionals to uniquely specify a joint distribution are provided by the
Hammersley-Clifford theorem. Here, we present a simplified statement of this theorem in the two-dimensional case.
\begin{theorem}[Hammersley–Clifford Two-dimensional Case] 
If $\int \frac{f(x_2|x_1)}{f(x_1|x_2)} \, dx_2$ exists, then the joint density associated with $f(x_2|x_1)$ and $f(x_1|x_2)$ is given by $$f(x_1,x_2) = \frac{f(x_2 | x_1)}{\int \frac{f(x_2|x_1)}{f(x_1|x_2)} \,dx_2} .$$
\end{theorem}
\begin{proof}
Since $f(x_2|x_1)f(x_1) = f(x_1|x_2)f(x_2)$, 
$$\int \frac{f(x_2|x_1)}{f(x_1|x_2)}\, dx_2 = \int \frac{f(x_2)}{f(x_1)} \,dx_2 = \frac{1}{f(x_1)}, $$
and so the claim follows if the integral above exists. \hfill $\square$
\end{proof}

We are now ready to introduce the Gibbs sampler.

\begin{algorithm}[H]
  \caption{Gibbs Sampler\label{algo:Gibbssampler}\index{Gibbs sampler}}
  \begin{algorithmic}[1]
  \STATEx{ \textbf{Input:} target $f(x_1, \ldots, x_d)$, initialization $X^{(0)} = \bigl(X_1^{(0)}, \ldots, X_d^{(0)}\bigr)$, sample size $N.$ }
    \STATEx{For $n=1,\ldots ,N$ do:}
  \STATE {Choose $j \in \{1,\ldots, d\}$ (randomly or sequentially).}
  \STATE  {Sample $X_j^{(n)} \sim f_j \bigl(x_j | X_i^{(n-1)}, i \neq j\bigr)$ and set $$X^{(n)} = \Bigl(X_1^{(n-1)}, \ldots, X_{j-1}^{(n-1)},  X_{j}^{(n)},  X_{j+1}^{(n-1)}, \ldots, X_d^{(n-1)}\Bigr).$$}
  \STATEx{\textbf{Output:} Sample $\{X^{(n)} \}_{n=1}^N.$}
    \end{algorithmic}
\end{algorithm}

\FloatBarrier
\begin{mybox}[colback=white]{Pros and Cons}
An important advantage of the Gibbs sampler is that 
it does not require tuning parameters.
It is
particularly useful in problems where the target distribution cannot be sampled directly, but its full conditionals are easy to sample from. This setting arises, for instance, in  hierarchical Bayesian models, mixture models, and missing data problems. Updating only one variable at a time makes the 
method scalable to high-dimensional settings. However, 
convergence can
be slow when the individual variables are strongly correlated. Moreover, as with other methods based on Markov chain sampling, the Gibbs sampler is not parallelizable. 
\end{mybox}
\FloatBarrier

There are two main ways to improve the performance of the Gibbs sampler when the variables are 
highly
correlated.

\begin{enumerate}
\item Reparameterize the distribution so that the new variables are less correlated.
\item Update groups of strongly correlated variables jointly; this is referred to as \emph{blocking}.
For example, if $X_1$ and $X_2$ are highly correlated, one can sample:
\begin{itemize}
\item $X_1, X_2 | X_3.$
\item $X_3 | X_1, X_2$.
\end{itemize}

\end{enumerate}

We conclude this section by showing that the Gibbs sampler is a Metropolis--Hastings algorithm. A minor difference from the Metropolis--Hastings framework presented in Algorithm \ref{algo:MH} is that the Gibbs sampler uses a different proposal Markov kernel for updating each variable.

\begin{theorem}[Gibbs Sampler is a Metropolis--Hastings Algorithm]
The Gibbs sampler is a Metropolis--Hastings algorithm that uses full conditionals \index{full conditional} as proposals. Proposed moves are always accepted. 
\end{theorem}
\begin{proof}
Let $x = (x_1, \ldots, x_d)$ and  $z = (z_1, \ldots, z_d)$ with $x_i = z_i$ for $i \neq j$. We then have
\begin{align*}
        a(x, z) & = \min \biggl\{1, \frac{f(z)}{f(x)} \frac{f_j (x_j | z_i, i \neq j)}{f_j (z_j | x_i, i \neq j)} \biggr\} \\
        & = \min \biggl\{1, \frac{f(z)}{f(x)} \frac{f(x)/ f(z_{-j})}{f(z)/f(x_{-j})}   \biggr\} \quad (\text{by definition of full conditional})\\
        & = \min \biggl\{1, \frac{f(z)}{f(x)} \frac{f(x)/ f(x_{-j})}{f(z)/f(x_{-j})}   \biggr\} \quad  \bigl(f(z_{-j}) = f(x_{-j}) \ \text{ because }\ x_i = z_i \, \forall i \neq j \bigr) \\
        & = 1.   \tag*{\qedhere}
\end{align*}
\end{proof}

Since the Gibbs sampler \index{Gibbs sampler!as Metropolis--Hastings algorithm} is a particular example of a Metropolis--Hastings algorithm, its transition kernel satisfies detailed balance with respect to $f$. We verify this directly in the following proposition.
\begin{proposition}\label{prop:gibbs_detailed_balance}
Each transition of the Gibbs sampler satisfies detailed balance with respect to the target distribution $f$. 
\end{proposition}
\begin{proof}
Denote by $\pgs$ the 
transition kernel of the Gibbs sampler and let
\begin{equation*}
    \begin{aligned}
        & x  = (x_1, \ldots, x_d),\\
        & 
        z  = (x_1, \ldots,x_{j-1}, \xi ,x_{j+1},\ldots,  x_d). 
    \end{aligned}
\end{equation*}
Then,
\begin{align*}
        f(x) \pgs(x, z) & = f(x)\frac{f(z)}{f(x_{-j})} = f(z)\frac{f(x)}{f(x_{-j})} \\
        & = f(z) f_{j} (x_j | x_i, i \neq j) = f(z)\pgs(z, x).   \tag*{\qedhere}
\end{align*}
\end{proof}

\begin{example}[2D Slice Sampler\index{slice sampler}]\label{ex:slice}
 When rejection sampling was introduced in Chapter ~\ref{chapter1}, we  noted that a density $f(x)$ with   $x \in \mathbb{R}$ can be represented as the marginal of a joint distribution on $(x,u)$ with density $f(x,u) = \textbf{1}_{\{ 0 < u < f(x)\} }.$ Figure \ref{fig:2dslice} illustrates the idea. 
The corresponding 
full conditionals are 
\begin{equation*}
    \begin{aligned}
    & f(x|u)  \propto \textbf{1}_{\{   0 < u < f(x)  \} }, \\
    & f(u|x)  \propto \textbf{1}_{\{     0 < u < f(x)  \} },
    \end{aligned}
 \end{equation*}
which are uniform and hence easy to sample from.  
We can therefore apply the Gibbs sampler to generate
samples from $f(x,u).$  
By discarding the auxiliary variable $u$ and retaining only the $x$ component,
we obtain samples from $f(x)$. This specific 
version of the
Gibbs sampler is 
known as the 2D slice sampler. Generalizations of this idea to high-dimensional settings are possible; see the bibliography in Section \ref{sec:Gibbsbibliography}.   \hfill \qedhere
\end{example}

\begin{figure}[h!]
    \centering
    \includegraphics[width=0.65\columnwidth]{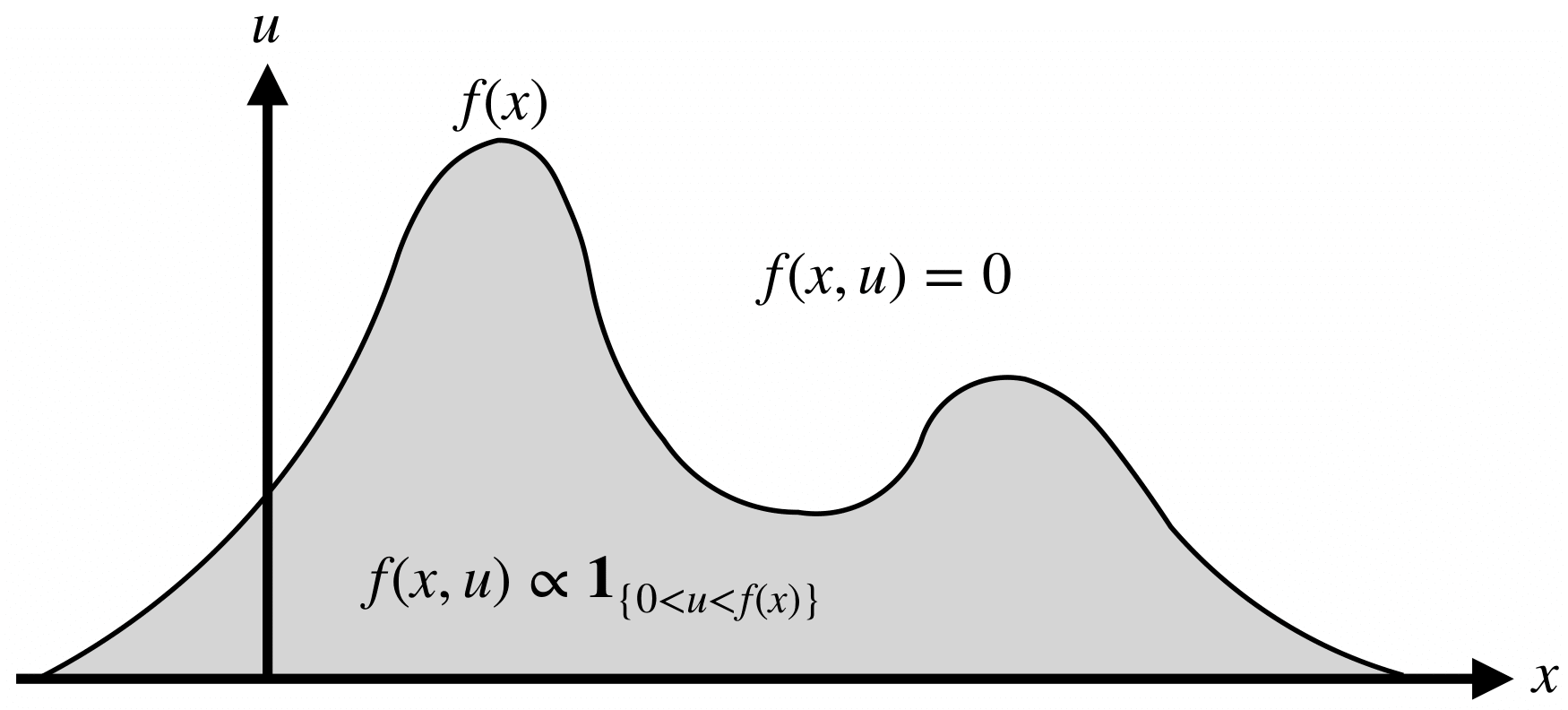}
    \caption{$f(x,u) = \textbf{1}_{\{ 0 < u < f(x)\} }.$}
    \label{fig:2dslice}
\end{figure}

\section{Convergence of the Gibbs Sampler\index{Gibbs sampler!convergence}}\label{sec:convergenceGibbs}
Since the Gibbs sampler satisfies detailed balance with respect to the target, the target is an invariant distribution. If the chain is additionally ergodic, then convergence to the target follows. The following lemma gives a sufficient condition for irreducibility. We refer to Definitions \ref{def:irreducible1} and \ref{def:irreducible2} in Appendix \ref{chap:markovchains} for the notion of irreducibility in countable and general state spaces. 
\begin{lemma}
Suppose that the target $f(x_1, \ldots, x_d)$ satisfies:
$$\underbrace{f_i (x_i)}_{\text{Marginal}} > 0  \quad \forall i \in \left\{ 1,\ldots,d\right\} \Longrightarrow f(x_1, \ldots, x_d) > 0.$$
Then, the Gibbs sampler is f-irreducible.
\end{lemma}
\begin{example}[Reducible Gibbs Sampler\label{example:gibbs_reducible_bad_case}]
Consider the following distribution in $\mathbb{R}^2:$
$$f(x_1,x_2)=\frac{1}{2} \textbf{1}_{\{  0 \le x_1 \le 1, 0 \le x_2\le 1\} } + \frac{1}{2} \textbf{1}_{\{  -1 \le x_1 \le 0,  -1 \le x_2 \le 0\} },$$
which is depicted in Figure \ref{fig:2stateunif}. The full conditionals are:
\begin{equation*}
    \begin{aligned}
    & f(x_1|x_2)=
\begin{cases}
\textbf{1}_{\{  x_1 \in [0,1]\} }  & x_2 \in [0,1], \\
\textbf{1}_{\{  x_1 \in [-1,0]\} }  & x_2 \in [-1,0], 
\end{cases} \\
  & f(x_2|x_1)=
\begin{cases}
\textbf{1}_{\{  x_2 \in [0,1]\} }  & x_1 \in [0,1], \\
\textbf{1}_{\{  x_2 \in [-1,0]\} }  & x_1 \in [-1,0]. 
\end{cases}
    \end{aligned}
   \end{equation*}
   \FloatBarrier
\begin{figure}[h!]
    \centering
    \includegraphics[width=0.4\columnwidth]{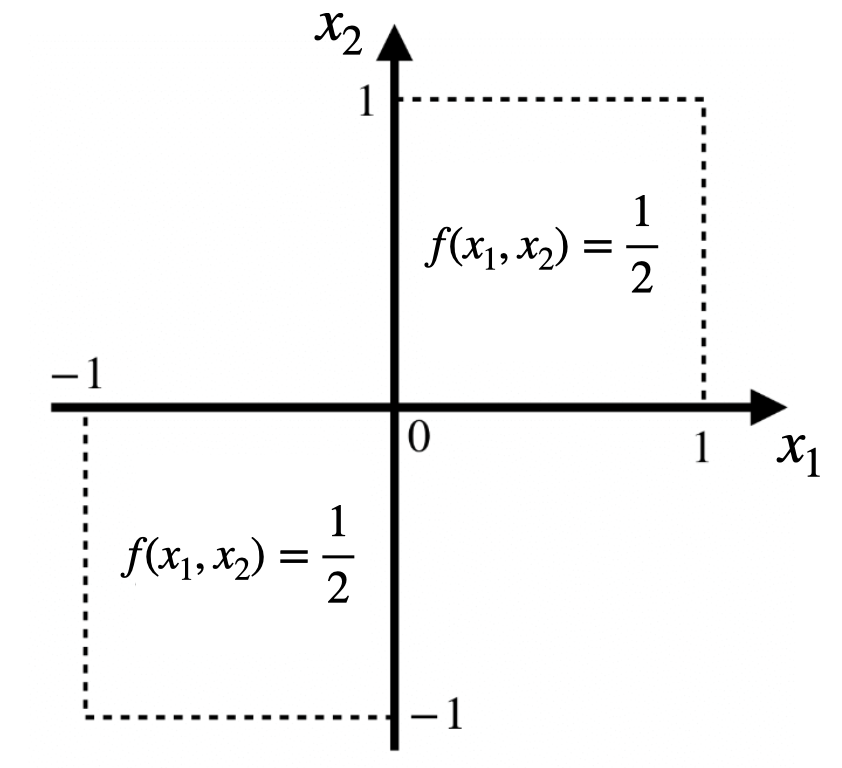}
    \caption{Target density for which the Gibbs sampler is reducible.}
    \label{fig:2stateunif}
\end{figure}
\FloatBarrier
The Gibbs sampler targeting this distribution
is not irreducible. For instance, if the chain is initialized in the upper-right square, it will never reach the lower-left square. 
Note that the condition of the previous lemma fails here:
\begin{align*}
    & f_1 (x_1)  = \frac{1}{2} \textbf{1}_{\{  -1 \le x_1 \le 1\} }, \\
    & f_2 (x_2)  = \frac{1}{2} \textbf{1}_{\{  -1 \le x_2 \le 1\} } ,
 \end{align*}
 \begin{equation*}
  f(-0.5, 0.5) = 0 \quad \text{but} \quad f_1(-0.5)>0, \quad f_2(0.5)>0.     \tag*{\qedhere}
 \end{equation*} 
\end{example}
\begin{remark}
It can be shown that Harris recurrence (see Definition \ref{def:harris}) holds if the transition kernel is absolutely continuous with respect to an appropriate dominating measure.
\end{remark}

\subsection{Convergence for Finite State Distributions}

In this and the following subsection we analyze the convergence of the Gibbs sampler. We restrict attention to the Deterministic Update Gibbs Sampler (DUGS), in which coordinates are updated sequentially.
\begin{algorithm}[H]
  \caption{Deterministic Update Gibbs Sampler (DUGS)  \label{algo:dugs}}
  \begin{algorithmic}[1]
  \STATEx{ \textbf{Input:}  Target $f(x_1, \ldots, x_d)$, initialization $X^{(0)} = \bigl(X_1^{(0)}, \ldots, X_d^{(0)}\bigr)$, sample size $N.$  }
    \FOR {$n=1,\ldots ,N$}
        \FOR {$j=1,\ldots,d$}
          \STATE{Sample $X_j^{(n)} 
            \sim f_j \bigl(X_j | X_1^{(n)},\ldots,X_{j-1}^{(n)},X_{j+1}^{(n-1)},\ldots,X_d^{(n-1)}\bigr).$}
          \STATE{Update the $j$-th coordinate by setting
          $$X^{(n)} = \Bigl(X^{(n)}_1, \ldots, X^{(n)}_{j-1},  X_j^{(n)},  X^{(n-1)}_{j+1}, \ldots, X^{(n-1)}_d \Bigr).$$}
        \ENDFOR
        \STATE{Record sample $X^{(n)}$.}
    \ENDFOR
  \STATEx{\textbf{Output:} Sample $\{X^{(n)}\}_{n=1}^N.$}
    \end{algorithmic}
\end{algorithm}

Convergence of the Gibbs sampler in a finite state space can be established using the results on ergodicity of Markov chains reviewed in Appendix \ref{chap:markovchains}. Proposition \ref{prop:gibbs_detailed_balance} ensures that the target is the invariant distribution of the chain; convergence to the target follows if we guarantee that the chain is irreducible and aperiodic.

\begin{theorem}[Convergence of Gibbs Sampler in Finite State Space\index{Gibbs sampler!convergence in finite state space}]
 Consider a target p.m.f. $f(x)$ on a finite state space $E^d = E \times \cdots \times E$ and assume that $f(x)>0$ for all $x\in E^d.$ Then, there is $\epsilon>0$ such that, for any $x \in E^d,$ 
\begin{equation*}
\dtv\bigl(\pdugs^n(x, \cdot), f\bigr) \le \bigl(1 - \epsilon)^n,
\end{equation*}
where $\pdugs$ denotes the DUGS Markov kernel defined by one sweep of sampling all full conditionals $f_j,$ $1\le j \le d.$  
\end{theorem}
\begin{proof}
 The assumption that $f(x) >0$ for all $x \in E^d$ implies that all full conditionals are positive. Consequently, each coordinate can be updated from any point in $E$ to any other point in $E$ with positive probability after sampling the corresponding full conditional. This immediately implies that the Markov kernel $\pdugs$ ---which samples sequentially over all full conditionals--- is bounded below by a positive constant. The result hence follows from Theorem \ref{thm:ergodicity} on ergodicity of Markov chains in finite state space.  \hfill $\square$
\end{proof}

\subsection{Convergence for Gaussian Distributions}
Now we consider continuous distributions, focusing on Gaussian targets $f = \Nc( \mu, \Sigma) = \Nc(\mu, H^{-1})$. We first find a compact way to write  the DUGS updates. To that end, we decompose the precision matrix $H= H_L + H_D + H_U$ into its lower-triangular, diagonal, and upper-triangular components. We have the following lemma.
	\begin{lemma}\label{gibbs_ar}
	Let $\{X^{(n)}\}_{n \geq 0}$ be the Gibbs output with target $f = \Nc( \mu, \Sigma) = \Nc(\mu, H^{-1})$ and initial state $X^{(0)}$. Then, $\{X^{(n)}\}_{n \geq 0}$ satisfies the following recurrence:
	    \begin{equation}
        X^{(n+1)}-\mu = -(H_L+H_D)^{-1}H_U(X^{(n)}-\mu)+(H_L+H_D)^{-1} U^{(n)},   \label{eqn:gibbs_ar}
        \end{equation}
	where $U^{(n)}\sim\mathcal{N}(0,H_D)$. 
	\end{lemma}
		\begin{proof}
    Denote the state at iteration $n$ by $x=(x_1,x_2,\ldots,x_d)$ and the state at iteration $n+1,$ after a Gibbs update of all coordinates, by $x'=(x'_1,x'_2,\ldots,x'_d)$. Using Example \ref{ex:fullcondGau}  we can write the Gibbs update on block $j$ as
    \[H_{jj}(x'_j-\nu_j)=u_j,  \]
    where $u_j\sim\mathcal{N}(0,H_{jj})$ with 
    \[\nu_j=\mu_j-H_{jj}^{-1}\sum_{i<j}H_{ji}(x_i'-\mu_i)-H_{jj}^{-1}\sum_{i>j}H_{ji}(x_i-\mu_i)\]
    since we are conditioning on $x'_1,x'_2,\ldots,x'_{j-1},x_{j+1},\ldots,x_d$. Expanding out $\nu_j$ gives 
    \[H_{jj}(x'_j-\mu_j)+
    [H_{j1};\ldots;H_{jj-1}]^T(x'_{1:j-1}-\mu_{1:j-1})+
    [H_{jj+1};\ldots;H_{jd}]^T(x_{j+1:d}-\mu_{j+1:d})=u_j,\]
   where $x_{1:j-1} = (x_1, \ldots, x_{j-1})$ and the other quantities are defined similarly.
    Therefore, a full Gibbs update can be written in vectorized form as
    \[(H_L+H_D)(x'-\mu)+H_U(x-\mu)=u,\]
    where $u\sim\mathcal{N} (0,H_D)$ as desired. \hfill $\square$
    \end{proof}
    
    The matrix $$B=-(H_L+H_D)^{-1}H_U$$ is known as the Gauss-Seidel companion matrix. The Gauss-Seidel \index{Gauss-Seidel} recurrence in \eqref{eqn:gibbs_ar} belongs to a larger family of iterative techniques to solve linear systems of equations via matrix splitting. It is clear that a sufficient condition for convergence is that
    the spectral radius $\rho(B)$ is strictly less than one. 
    Exercise~5.6 of the published version of this book investigates why this condition holds provided that the matrix $H$ is symmetric and positive definite.

As explored in Exercise~3.7 of the published version of this book, the $\chi^2$-divergence \index{divergence!$\chi^2$} between Gaussian distributions 
 $f=\mathcal{N}(\mu,\Sigma)$ and $\widetilde{f}=\mathcal{N}( \widetilde{\mu}, \widetilde{\Sigma})$ is given by
 \begin{equation}\label{eq:chisqdivergenceGaussians}
 \dchi(f \, \| \widetilde{f} \, )=\frac{|W|}{\sqrt{|2W-I|}}e^{u^T(2W-I)^{-1}\Sigma^{-1} u}
 \end{equation}
    with $W= \widetilde{\Sigma}\Sigma^{-1}$ and $u=\mu- \widetilde{\mu}$.

	\begin{theorem}[Convergence of Gibbs Sampler for Gaussian Targets\index{Gibbs sampler!convergence for Gaussian targets}]\label{thm:gibbs_cvg_rate}
	Let the target distribution of the Deterministic Update Gibbs Sampler (DUGS) be $f = \mathcal{N}(\mu,H^{-1}) = \mathcal{N}(\mu,\Sigma) $, and let $\pdugs$ denote the Markov kernel defined by one full sweep of the DUGS scheme. Then,
	\[\mathbb{E}_f \Bigl[ \dchi \bigl( \pdugs^n(x^{(0)}, \cdot)   \|f \bigr) \Bigr] = \mathcal{O} \bigl(\rho(B)^{2n} \bigr),\]
	where $B=-(H_L+H_D)^{-1}H_U.$
	\end{theorem}
    \begin{proof}
    We rewrite the recurrence relation of Gibbs updates derived in \eqref{eqn:gibbs_ar} as 
    \[X^{(n)}-\mu=B(X^{(n-1)}-\mu)+Z^{(n)},\]
    where  $Z^{(n)} \stackrel{\text{i.i.d.}}{\sim}\mathcal{N}(0,\Sigma_z)$ for some $\Sigma_z$ that we will determine in what follows. Taking expectation on both sides gives 
    \[\mathbb{E} \bigl[X^{(n)}-\mu\bigr]= B\, \mathbb{E}\bigl[X^{(n-1)}-\mu\bigr].\]
    On the other hand, we know that $\mathcal{N}(\mu,\Sigma)$ is the invariant distribution, so
    $\text{Cov}\bigl[X\bigr]=\text{Cov} \bigl[BX\bigr]+\text{Cov} \bigl[Z \bigr],$
    where $X$ has the target distribution and $Z$ has the noise distribution in the recurrence equation. This implies that
    $\Sigma_z=\Sigma-B\Sigma B^T.$ 
    Taking the covariance on both sides of the original recurrence relation yields
    \begin{align*}
    \text{Cov} \bigl[ X^{(n)} \bigr]
    & = \text{Cov}  \bigl[BX^{(n-1)}\bigr]+\Sigma_z\\
    & =B \, \text{Cov}\bigl[X^{(n-1)} \bigr] B^T+\Sigma-B\Sigma B^T\\
    & =\Sigma+B\Bigl(\text{Cov} \bigl[X^{(n-1)} \bigr] -\Sigma\Bigr)B^T.
    \end{align*}
    
    Let $\mu_n$ and $\Sigma_n$ be the mean and covariance of the conditional distribution $X^{(n)}|x^{(0)}$. Then,
    \begin{align*}
     \mu_n &=\mu+B^n(x^{(0)}-\mu), \\
    \Sigma_n  &=\Sigma+B^n(0-\Sigma)(B^T)^n=\Sigma-B^n\Sigma(B^T)^n,
    \end{align*}
    since $\Sigma_0=0$. 
     Using Equation \eqref{eq:chisqdivergenceGaussians}, we have
    \[\dchi \bigl(\pdugs^n(x^{(0)},\cdot )  \|f \bigr)=\frac{|W|}{\sqrt{|2W-I|}}e^{u^T \Sigma_n^{-1}  (2W-I)^{-1}u}-1,\]
    with 
        \begin{align*}
        W &= \Sigma \Sigma_n^{-1},   \\
        u &=B^n(x^{(0)}-\mu).
        \end{align*}
    Observe that the only $x^{(0)}$-dependence lies in the exponent. Integrating over $x^{(0)}\sim f$ gives 
     \begin{align*}
        \mathbb{E}_f  \Bigl[ \dchi\bigl( & \pdugs^n(x^{(0)}, \cdot)\|f \bigr)\Bigr]  + 1
        =\frac{|W|}{\sqrt{|2W-I|}}     
        \int
        e^{(B^n(x-\mu))^T  \Sigma_n^{-1} (2W-I)^{-1}B^n(x-\mu)}
        f(x)  \, 
        dx
        \\
        & =  \frac{|W|}{\sqrt{|2W-I|}}     
        \int
        e^{(x-\mu)^T (B^T)^n \Sigma_n^{-1} (2W-I)^{-1}B^n(x-\mu)}
        \times \frac{e^{ - \frac{1}{2}(x-\mu)^T \Sigma^{-1}(x-\mu)}}{(2\pi)^{\frac{d}{2}}|\Sigma|^{\frac{1}{2}}}  \, 
        dx
        \\ 
         &= \frac{|W|}{\sqrt{|2W-I|}}      
        \frac{1}{(2\pi)^{\frac{d}{2}}|\Sigma|^{\frac{1}{2}}}
        \int e^{-\frac{1}{2}(x-\mu)^T U (x-\mu)} \, dx
     \\
        &= \frac{|W|}{\sqrt{|2W-I| |\Sigma| |U|}},
        \end{align*}
        where $U:= \Sigma^{-1} - 2 (B^T)^n \Sigma_n^{-1} (2W - I )^{-1} B^n$ can be shown to be positive definite, so the Gaussian integral evaluates $(2\pi)^{d/2} |U|^{-1/2}.$ 
   Now note that we have 
  \begin{align*}
      \|B^n\| &\sim \rho(B)^n\, ,\\
      \Sigma_n &= \Sigma + \mathcal{O}(\| B^n\|^2), 
  \end{align*}
  and hence 
   $$W = \Sigma \Sigma_n^{-1} = I + \mathcal{O}( \|B^n\|^2).$$
   It follows that 
   \begin{align*}
        |W| &= 1 + \mathcal{O}(\|B^n\|^2), \\
        |2W-I| &= 1+ \mathcal{O}(\|B^n\|^2),\\
        |U| &= |\Sigma^{-1}| \bigl( 1 + \mathcal{O}(\|B^n\|^2)\bigr)\, .
   \end{align*}
   Putting everything together, 
\begin{align*}
  \mathbb{E}_f  \Bigl[ \dchi\bigl( & \pdugs^n(x^{(0)}, \cdot)\|f \bigr)\Bigr]  =  \frac{|W|}{\sqrt{|2W-I| |\Sigma| |U|}} - 1 = \mathcal{O}(\|B^n\|^2) = \mathcal{O}(\rho(B)^{2n}),
\end{align*}
   completing the proof.
 \hfill $\square$
    \end{proof}

\begin{figure}[!htb]
    \centering
\minipage{1\textwidth}
  \includegraphics[width=\linewidth]{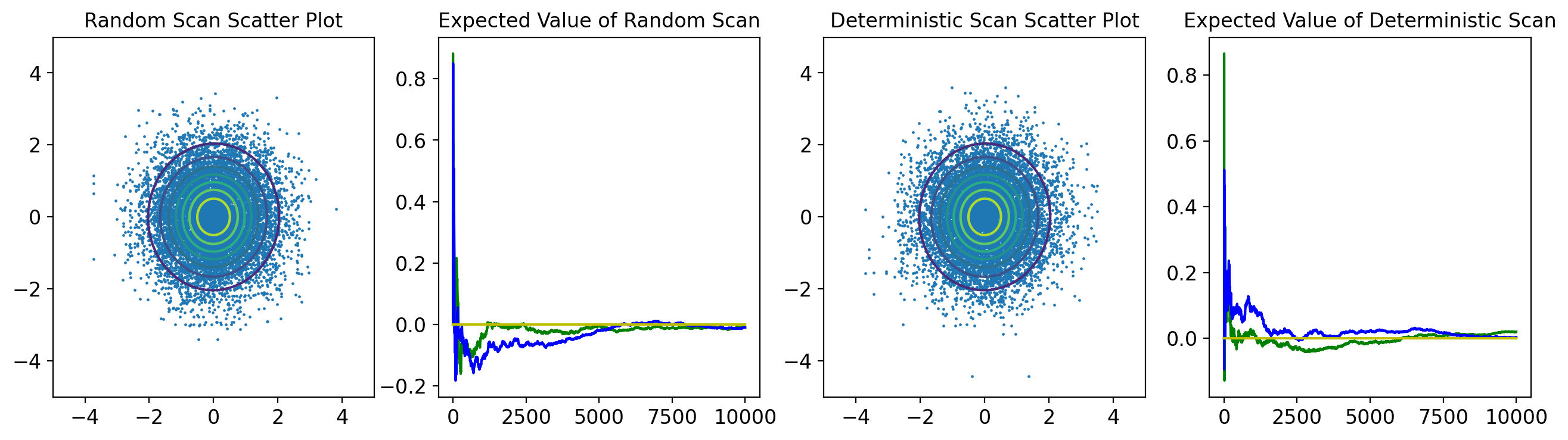}
\endminipage\hfill
\minipage{1\textwidth}%
  \includegraphics[width=\linewidth]{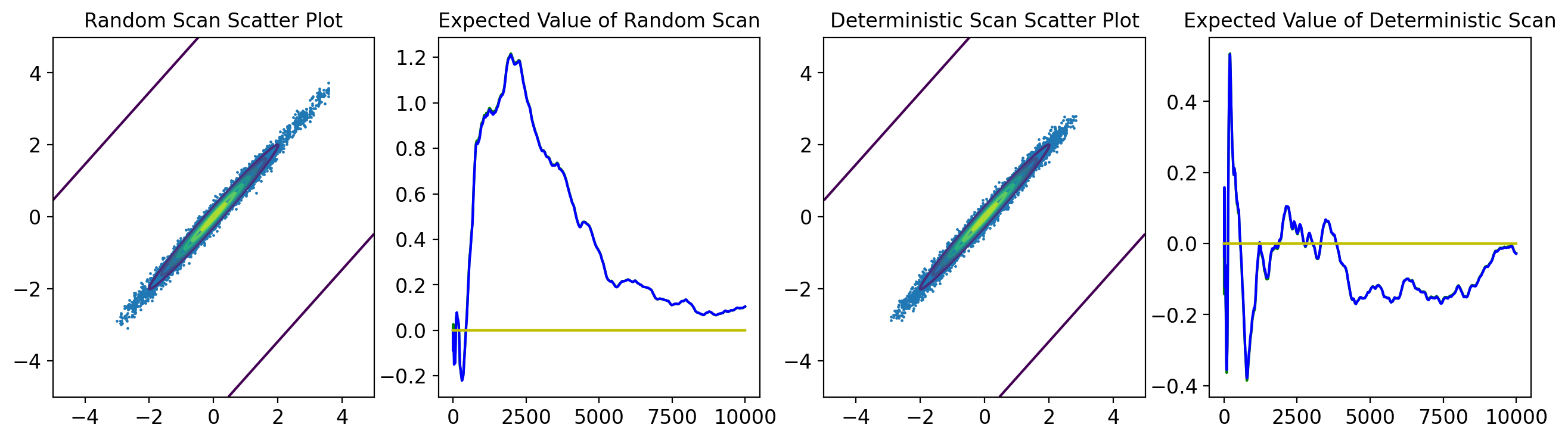}
\endminipage\hfill
\caption{Gibbs samplers for a zero-mean bivariate Gaussian distribution with covariance matrix $[1, \alpha; \alpha, 1]$. Top row: $\alpha = 0.01.$ Bottom row: $\alpha = 0.99$.}
\end{figure}

\section{Discussion and Bibliography}\label{sec:Gibbsbibliography}
The Gibbs sampler was introduced in \cite{geman1987stochastic}. The paper \cite{gelfand1990sampling} further popularized the approach, making apparent its usefulness for Bayesian inference.
 The derivation of full conditionals for Gaussians in Example \ref{ex:fullcondGau}  can be found for instance in \cite{rue2005gaussian}. Example \ref{ex:slice} briefly introduced the slice sampler for univariate targets. We refer to \cite{neal2003slice} for multivariate extensions and to \cite{murray2010elliptical} for slice samplers for function space sampling.
 A comparison between random and deterministic updates can be found in \cite{amit1991comparing}. Simple conditions on the target to guarantee convergence were established in \cite{roberts1994simple}, and geometric ergodicity of the Gibbs sampler was further studied in \cite{roberts1994geometric}. The proof of convergence for Gaussian targets is taken from \cite{roberts1997updating}, which also contains fundamental results for the understanding of blocking and other variants of the Gibbs sampler. Gibbs samplers for hierarchical models have been widely studied, see e.g. \cite{papaspiliopoulos2008stability} and references therein. Adaptive Gibbs samplers are investigated in \cite{chimisov2018adapting}. Exercise~5.4 of the published version of this book, adapted from \cite{sanz2023inverse}, investigates Gibbs samplers for graphical and Ising models. Readers interested in additional background on Ising models are referred to \cite{newman1999monte}.

\chapter{Langevin Monte Carlo}
\label{chap:diffusions}

This chapter introduces Langevin Monte Carlo\index{Langevin!Monte Carlo} sampling algorithms 
, which 
share many similarities with gradient descent optimization methods. Gradient descent minimizes an objective function $V: \R^d \to \R$ by taking steps in the direction of 
$-\nabla V.$ Similarly,
Langevin Monte Carlo 
generates
samples from a target distribution $f \propto \exp(-V)$ by moving, on average, in the direction of $-\nabla V.$ In the same way that gradient descent can quickly find a local minimum of $V,$ Langevin algorithms can quickly find and explore a mode of the target distribution. Consequently, 
gradient descent is effective for globally optimizing convex objective functions, where local minima are also global minima. Analogously, Langevin Monte Carlo is effective for sampling from log-concave distributions, where the target has a single dominant mode and no separated modes that must be discovered.
In contrast,
for non-convex objectives, gradient descent can struggle to escape local minima, and for multi-modal \index{distribution!multi-modal} targets, Langevin Monte Carlo may require many iterations to find and explore each mode.

 We will study two Langevin sampling algorithms: the Unadjusted Langevin Algorithm (ULA) and the Metropolis Adjusted Langevin Algorithm (MALA). Our presentation of ULA highlights 
its analogy with gradient descent.
 In a Gaussian setting, we will show that ULA is biased: the limit distribution of the ULA Markov chain does not coincide with the target.
 Nonetheless, this bias can be reduced by using a smaller step-size. \index{step-size} MALA provides a natural 
 way to remove this bias by incorporating
 a Metropolis--Hastings accept/reject mechanism. 
 
In addition to the optimization perspective 
emphasized 
in Sections \ref{sec:ULA} and \ref{sec:MALA}, 
Section \ref{sec:LangevinSDE} provides an alternative perspective by deriving ULA as
the Euler-Maruyama discretization of the Langevin stochastic differential equation. This perspective predates the optimization 
viewpoint and motivates the name of the algorithms. As we shall see, the Langevin equation has the key property of leaving the target 
distribution invariant. However,
this property is generally lost after time-discretization, which explains the bias of ULA.
From this perspective,
 the Metropolis--Hastings accept/reject step in MALA 
 corrects 
 the Euler-Maruyama discretization error
and ensures convergence to the target. The chapter concludes in Section \ref{sec:literatureLangevin}  with bibliographical remarks.

\section{Unadjusted Langevin Algorithm}\label{sec:ULA}
A simple but effective approach to 
minimizing
an objective function $V: \R^d \to \R$ is 
gradient descent. Starting
from an initial guess $x_0 \in \R^d,$ we recursively define
\begin{equation}\label{eq:gradientdescent}
    x_{n+1} = x_n - \epsilon \nabla V(x_n), \qquad n \ge 0.
\end{equation}
The update $x_n \mapsto x_{n+1}$ 
moves in the direction of greatest descent, namely the direction $-\nabla V(x_n),$ thereby decreasing $V$ whenever the step-size \index{step-size} $\epsilon>0$ is sufficiently small.
Under regularity and convexity assumptions on $V$ and an appropriate choice of step-size, 
 the deterministic iterates $\{x_n\}_{n=1}^\infty$ converge to the minimizer of $V.$

 \begin{figure}[h!]
    \centering
    \includegraphics[width=0.55\textwidth]{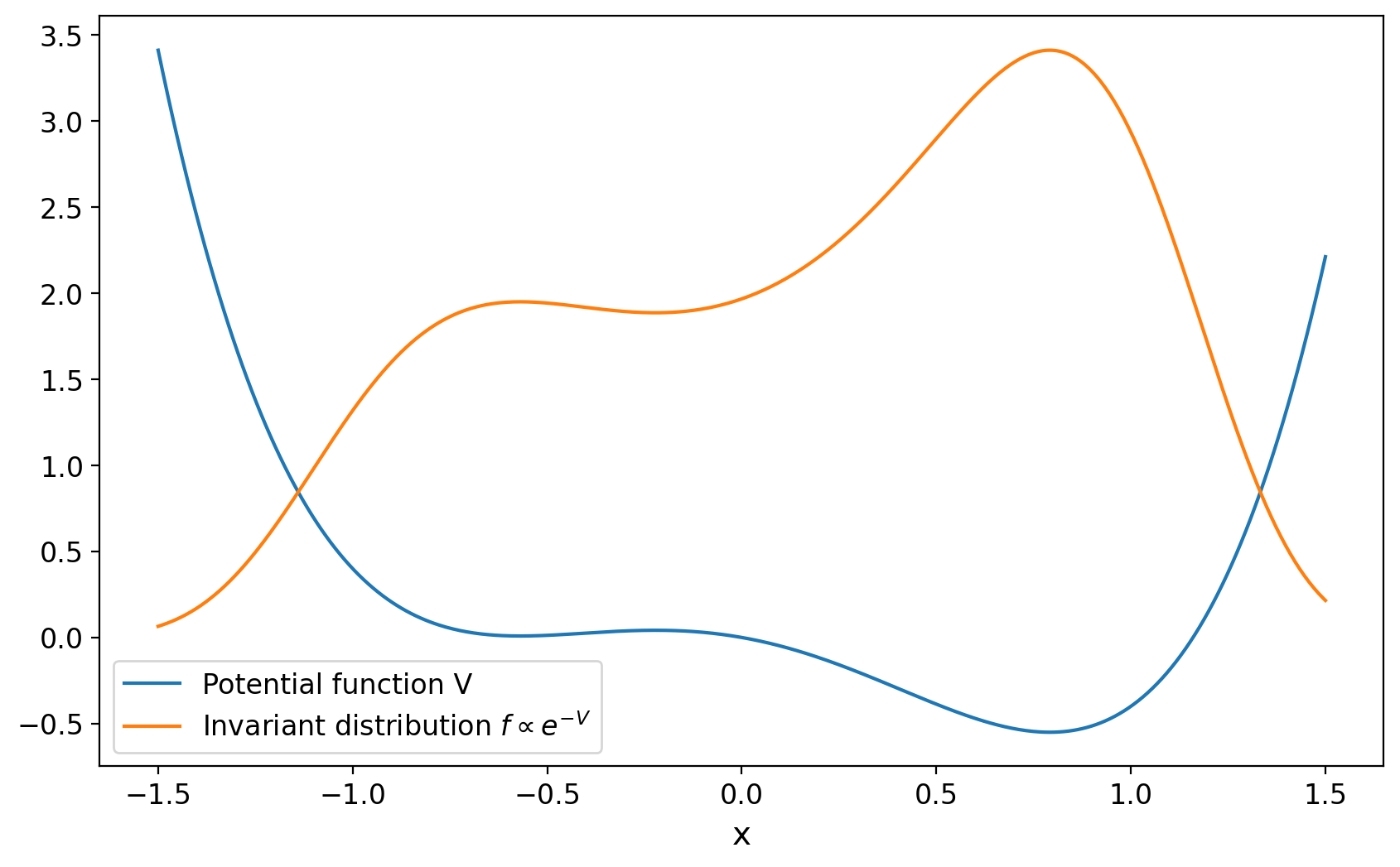}
    \caption{Objective function $V(x) = x^4 -x^2 -0.4x$ and corresponding target density $f(x) \propto \exp\bigl( -V(x) \big)$. Low values of $V$ correspond to large values of $f.$}
    \label{fig:diffusioninvariant}
\end{figure}

The Unadjusted Langevin Algorithm (ULA)\index{Unadjusted Langevin Algorithm (ULA)} 
uses a stochastic analogue
of the deterministic update \eqref{eq:gradientdescent} to generate approximate samples from a target distribution $f(x) \propto \exp\bigl(-V(x)\bigr),$ see Figure \ref{fig:diffusioninvariant}. Starting from a random variable $X^{(0)},$ 
the  ULA iterates $\{X^{(n)} \}_{n=1}^\infty$ satisfy
\begin{equation*}
    \Expect\bigl[X^{(n+1)} | X^{(n)} \bigr] = X^{(n)} - \epsilon \nabla V\bigl(X^{(n)}\bigr), \qquad n \ge 0.
\end{equation*}
Therefore, on average, the update $X^{(n)} \mapsto X^{(n+1)}$ 
moves in 
the direction of $-\nabla V\bigl(X^{(n)}\bigr),$ thereby decreasing $V$ and increasing the target density $f \propto \exp (-V),$ provided the step-size \index{step-size} $\epsilon>0$ is sufficiently small. The method is summarized in Algorithm \ref{algo:ULA}.

\begin{algorithm}[H]
  \caption{Unadjusted Langevin Algorithm (ULA) \label{algo:ULA}}
  \begin{algorithmic}[1]
  \STATEx{\textbf{Input:} Target $f(x) \propto \exp\bigl(-V(x)\bigr)$, initial distribution $\pi_0$, step-size \index{step-size} $\epsilon>0,$ sample size $N.$}
  \STATE{{\bf{Initial draw:}} Sample $X^{(0)} \sim \pi_0.$}
  \STATE{{\bf Subsequent draws:}} For $n = 0, 1, \ldots, N-1$ set 
  $$X^{(n+1)} = X^{(n)} - \epsilon \nabla V\bigl(X^{(n)}\bigr) + \sqrt{2 \epsilon} \xi^{(n)}, \qquad \xi^{(n)} \sim \Nc(0,I).$$
  \STATEx{\textbf{Output:} Sample $\{X^{(n)}\}_{n=1}^N.$}
 \end{algorithmic}
\end{algorithm}

\begin{mybox}[colback=white]{Pros and Cons}
ULA can be implemented without knowing the normalizing constant of the target distribution. Moreover, 
because the updates $X^{(n)} \mapsto X^{(n+1)}$ move, on average, in the
direction of $-\nabla V\bigl(X^{(n)}\bigr),$ ULA is able to quickly find and explore regions of high target density. Under convexity 
and regularity 
assumptions on $V$ and for sufficiently small step-size $\epsilon,$ \index{step-size} 
the ULA iterates are, after a burn-in period, 
approximately distributed according to the target. However, the algorithm is biased: for any fixed
$\epsilon,$ the distribution of the 
ULA iterates converges, as $n \to \infty$, to a limit distribution that does not coincide with the target.
Another limitation
of ULA is that 
moving in the direction of $-\nabla V \bigl(X^{(n)}\bigr)$ can be counterproductive 
when the target is multi-modal. \index{distribution!multi-modal} In that case, ULA can 
become stuck in one mode 
and fail to adequately explore the others.
Finally, ULA requires evaluating $\nabla V,$ which can be computationally expensive.
\end{mybox}

We now
analyze ULA for 
a
Gaussian target $f = \Nc(\mu, \Sigma).$  In this setting, 
similarly
to the Gibbs sampler in Chapter \ref{chap:gibbs},
the ULA iterates
satisfy a simple recursion. As in Chapter \ref{chap:gibbs}, we denote by $H= \Sigma^{-1}$ the precision matrix.

\begin{lemma}\label{lemma:recursionULA} Let $\{X^{(n)} \}_{n\ge 1}$ be the output of ULA with target $f = \Nc(\mu, \Sigma) =  \Nc(\mu, H^{-1})$ and initial state $X^{(0)}.$ 
It holds that
\begin{equation}\label{eq:ULArecursion}
    X^{(n+1)} - \mu = (I -\epsilon H)^{n+1}(X^{(0)} - \mu) + \sqrt{2\epsilon} \sum_{i = 0}^n (I - \epsilon H)^{n -i} \xi^{(i)}.
\end{equation}
Suppose further that $\pi_0 = \Nc(\mu_0, \sigma_0^2 I)$ for some $\mu_0 \in \R^d$ and $\sigma_0^2>0.$ Then, $X^{(n)} \sim \Nc(\mu_n, \Sigma_n),$ where 
\begin{align}
    \mu_n &= \mu + (I - \epsilon H)^n (\mu_0 - \mu), \label{eq:ULAmean}\\
    \Sigma_n &= \Bigl(H - \frac{\epsilon}{2}H^2\Bigr)^{-1} + (I - \epsilon H)^{2n} \Biggl(\sigma_0^2 I -\Bigl(H - \frac{\epsilon}{2}H^2\Bigr)^{-1}\Biggr). \label{eq:ULAcovariance}
\end{align}
\end{lemma}
\begin{proof}
    Since the target is Gaussian, we may take 
    $V(x) = \frac{1}{2}(x - \mu)^T H(x-\mu),$ 
    so that
    $\nabla V(x) = H(x-\mu).$ 
   Hence,
    \begin{align*}
        X^{(n+1)} - \mu &= X^{(n)} - \epsilon \nabla V(X^{(n)}) + \sqrt{2 \epsilon} \,  \xi^{(n)} - \mu \\
        & = (I - \epsilon H) (X^{(n)} - \mu) + \sqrt{2 \epsilon} \, \xi^{(n)} \\ 
        & = (I -\epsilon H)^{n+1}(X^{(0)} - \mu) + \sqrt{2\epsilon} \sum_{i = 0}^n (I - \epsilon H)^{n -i} \xi^{(i)},
    \end{align*}
 thus establishing \eqref{eq:ULArecursion}. Under the additional assumption that  $\pi_0 = \Nc(\mu_0, \sigma_0^2 I),$ 
 the representation \eqref{eq:ULArecursion} shows
 that each $X^{(n)}$ 
 is Gaussian.
 To compute 
 its mean, we use
 $\Expect[X^{(0)}] = \mu_0$ 
 and
 $\Expect[\xi^{(i)}] = 0$ for all 
 $i \ge 0$ to obtain
 \begin{align*}
    \mu_n - \mu &=  \Expect\Bigl[ (I -\epsilon H)^{n}(X^{(0)} - \mu) + \sqrt{2\epsilon} \sum_{i = 0}^{n-1} (I - \epsilon H)^{n -i -1} \xi^{(i)}  \Bigr] \\
     & = (I - \epsilon H)^n  (\mu_0 - \mu),
 \end{align*}
 which proves \eqref{eq:ULAmean}.
 To compute the covariance, note first that
\begin{align*}
    X^{(n)} - \mu_n &= (I -\epsilon H)^{n}(X^{(0)} - \mu) + \sqrt{2\epsilon} \sum_{i = 0}^{n-1} (I - \epsilon H)^{n -i -1} \xi^{(i)} - ( \mu_n - \mu) \\
    & = (I - \epsilon H)^n ( X^{(0)} - \mu_0 ) + \sqrt{2\epsilon} \sum_{i = 0}^{n-1} (I - \epsilon H)^{n -i -1} \xi^{(i)}.
\end{align*}
Using
\begin{align*}
    \Expect[(X^{(0)} - \mu_0)(X^{(0)} - \mu_0)^T] = \sigma_0^2 I,
    \qquad 
    \Expect[ \xi^{(i)} (\xi^{(i)})^T] = I,
\end{align*}
together with the independence of $X^{(0)}$ and each of the $\xi^{(i)},$ we obtain 
 \begin{align*}
     \Sigma_n &= \Expect\bigl[(X^{(n)} - \mu_n)(X^{(n)} - \mu_n)^T \bigr] \\
     & = \sigma_0^2 (I - \epsilon H)^{2n} + 2 \epsilon \sum_{i = 0}^{n-1} (I - \epsilon H)^{2i} \\
     & = \Bigl(H - \frac{\epsilon}{2}H^2\Bigr)^{-1} + (I - \epsilon H)^{2n} \biggl(\sigma_0^2 I - \Bigl(H - \frac{\epsilon}{2}H^2\Bigr)^{-1}\biggr),
 \end{align*}
where in the last line we used the formula for the partial sum of the geometric series 
$\sum_{i=0}^{n-1} A^i = (I - A)^{-1} - A^n(I-A)^{-1}.$ 
We have hence established 
\eqref{eq:ULAcovariance}, completing the proof. \hfill $\square$
\end{proof}

The key takeaway from Lemma \ref{lemma:recursionULA} is that, for sufficiently small fixed $\epsilon$, we have 
\begin{align*}
    \mu_n &\xrightarrow{ n \to \infty } \mu, \\
    \Sigma_n &\xrightarrow{ n \to \infty } \Sigma_\epsilon:= \Bigl(H - \frac{\epsilon}{2}H^2\Bigr)^{-1}.
\end{align*}
 Thus, the means $\mu_n$ of the ULA iterates converge to the mean $\mu$ of the target, but the covariances $\Sigma_n$ converge to $\Sigma_\epsilon \neq \Sigma. $ 
 
The proof of the next theorem 
 shows that, in Wasserstein distance, the distributions $\pi_n := \Nc(\mu_n,\Sigma_n)$ of ULA iterates converge exponentially fast  to $f_\epsilon:= \Nc(\mu, \Sigma_\epsilon)$ and that $f_\epsilon$ is order $\epsilon$ apart from $f$. 
We recall that, if $\pi = \Nc(\mu, \Sigma)$ and
$\widetilde{\pi} = \Nc(\widetilde{\mu}, \widetilde{\Sigma})$ are Gaussian
distributions with commuting covariance matrices, then their (order $2$) Wasserstein distance is
\begin{equation*}
    W_2(\pi, \widetilde{\pi})^2
    =
    \| \mu - \widetilde{\mu} \|^2
    +
    \| \Sigma^{1/2} - \widetilde{\Sigma}^{1/2} \|_F^2.
\end{equation*}
Here $\| A \|_F = \sqrt{\text{Tr}(A^T A)}$ denotes the Frobenius norm. We 
write
$\lambda_{\min}(H)$ and $\lambda_{\max}(H)$
for
the smallest and largest eigenvalues of $H.$

\begin{theorem}[Convergence of ULA for Gaussian Targets]\label{thm:ULAtheorem}
    Let $\{X^{(n)} \}_{n\ge 1}$ be the output of ULA with target $f = \Nc(\mu, \Sigma) = \Nc(\mu, H^{-1})$ and initial distribution $\pi_0 = \Nc(\mu_0, \sigma_0^2 I).$ Let $\pi_n$ denote the distribution of $X^{(n)}.$  Suppose that $\sigma_0^2 \le \lambda_{\max}(H)^{-1}.$ Then, there is a constant $c$, 
    which depends on 
    $\mu_0, \Sigma_0, \mu,$ and $\Sigma$ such that
    \begin{equation*}
        W_2(\pi_n, f) \le c\bigl(1 - \epsilon \lambda_{\min}(H)\bigr)^n + \frac{\epsilon}{4} \sqrt{\emph{Tr}(H)} + \mathcal{O}(\epsilon^2), \qquad n =1, 2, \ldots
    \end{equation*}
\end{theorem}
\begin{proof}
    By the triangle inequality, 
    $$
    W_2(\pi_n, f) \le W_2(\pi_n, f_\epsilon) + W_2(f_\epsilon, f).
    $$ 
    The first term in the right-hand side represents the distance between the distribution of $X^{(n)}$ and the
    limit 
    distribution of ULA, and the second 
    term represents the distance between 
    the limit distribution of ULA 
    and the target. We 
    will bound 
    both terms in turn. 
    We denote by $c$ a 
    constant that depends on $\mu_0, \Sigma_0, \mu,$ and $\Sigma$ 
    and which 
    may change from line to line.

    To bound $W_2(\pi_n, f_\epsilon),$ we need to control the deviation between the means and covariances of $\pi_n = \Nc(\mu_n,\Sigma_n)$ and $f_\epsilon = \Nc(\mu, \Sigma_\epsilon).$ For the means,
    \begin{align*}
        \| \mu_n - \mu \|^2 & = \| (I - \epsilon H)^n (\mu_0 - \mu) \|^2 \\ 
                            & \le \max\{ |1- \epsilon \lambda_{\min}(H)|, |1- \epsilon \lambda_{\max}(H)|\}^{2n} \| \mu_0 - \mu \|^2 \\
                            & =c  \bigl(1- \epsilon \lambda_{\min}(H)\bigr)^{2n},
    \end{align*}
where we used the Cauchy-Schwarz inequality and that, for all $\epsilon$ sufficiently small,
$\max\{ |1- \epsilon \lambda_{\min}(H)|, |1- \epsilon \lambda_{\max}(H)|\}  = |1- \epsilon \lambda_{\min}(H)|.$

For the covariances, let $Q \text{diag}(h_1, \ldots, h_d) Q^T$ be the eigendecomposition of $H,$ where $\{h_i\}_{i=1}^d$ are the eigenvalues of $H$ and $Q$ is orthonormal. Then, 
\begin{align*}
    \Sigma_\epsilon &= \Bigl(H - \frac{\epsilon}{2}H^2\Bigr)^{-1} = Q \text{diag} \bigl(\alpha_{1}, \ldots, \alpha_d \bigr) Q^T,  \\
    \Sigma_n & = Q \text{diag} \bigl(\alpha_1 - \beta_1, \ldots, \alpha_d - \beta_d\bigr) Q^T,
\end{align*}
where  $\alpha_i  = \Bigl(h_i - \frac{\epsilon}{2} h_i^2\Bigr)^{-1}$ and $\beta_i = (1 - \epsilon h_i)^{2n}(\alpha_i - \sigma_0^2). $
Notice that, for all 
sufficiently small $\epsilon,$ we 
have that 
$\frac{\epsilon}{2} h_i < 1,$ and 
hence that 
$\alpha_i>0.$ Furthermore, the assumption 
$\sigma_0^2 \le \lambda_{\max}(H)^{-1}$
ensures that $\sigma_0^2 \le \alpha_i,$ and hence that $\beta_i \ge 0.$  Consequently, for all $\epsilon$ sufficiently small,
\begin{align*}
    \| \Sigma_n^{1/2} - \Sigma_\epsilon^{1/2} \|_F^2 &= \sum_{i = 1}^d \Bigl( (\alpha_i - \beta_i)^{1/2} - \alpha_i^{1/2} \Bigr)^2 
     = \sum_{i=1}^d \Bigl(\alpha_i - \beta_i + \alpha_i - 2\sqrt{\alpha_i(\alpha_i - \beta_i} )  \Bigr)\\
    & \le \sum_{i=1}^d \Bigl(\alpha_i - \beta_i + \alpha_i - 2\sqrt{(\alpha_i - \beta_i)(\alpha_i - \beta_i} ) \Bigr) = \sum_{i=1}^d \beta_i  \\
    &= \sum_{i=1}^d (1 - \epsilon h_i)^{2n}(\alpha_i - \sigma_0^2) 
     \le c \bigl(1 - \epsilon \lambda_{\min}(H)\bigr)^{2n}, 
\end{align*}
where for the last inequality we use that $\text{Tr}(\Sigma_\epsilon - \Sigma_0) \to \text{Tr}(\Sigma - \Sigma_0)$ as $\epsilon \to 0,$ which implies that $\text{Tr}(\Sigma_\epsilon - \Sigma_0)$ is uniformly bounded for all sufficiently small $\epsilon.$
Combining the bounds for the means and the covariances
yields
\begin{align}\label{eq:bound1ULA}
\begin{split}
    W_2(\pi_n, f_\epsilon)^2 &\le c \bigl(1 - \epsilon \lambda_{\min}(H)\bigr)^{2n}. 
\end{split}
\end{align}

Next, we bound 
$W_2(f_\epsilon,f).$
Since $f_\epsilon$ and $f$ have the same mean, the Wasserstein distance depends only on the covariance matrices. Using the same eigendecomposition as above, we obtain
\begin{align*}
  W_2(f_\epsilon, f)^2 &=    \sum_{i=1}^d \frac{1}{h_i} \Biggl( 1 - \frac{1}{\sqrt{1 - \frac{\epsilon}{2}  h_i} } \Biggr)^2 
   = \sum_{i=1}^d \frac{1}{h_i} \phi\Bigl(\frac{\epsilon}{2}  h_i\Bigr), 
\end{align*}
where $\phi(z) : = \Bigl( 1 - \frac{1}{\sqrt{1-z}}\Bigr)^2.$ By Taylor expansion around zero, one can show that, for small $z,$ $\phi(z) = \frac{z^2}{4} + \mathcal{O}(z^3)$. Hence, for small $\epsilon,$
\begin{align}\label{eq:bound2ULA}
        W_2(f_\epsilon, f)^2 \le \frac{\epsilon^2}{16} \sum_{i=1}^d h_i  +  \mathcal{O}(\epsilon^3) = \frac{\epsilon^2}{16} \text{Tr}(H) + \mathcal{O}(\epsilon^3).
    \end{align}
Taking square roots in \eqref{eq:bound1ULA} and \eqref{eq:bound2ULA} and using the triangle inequality completes the proof. $\hfill \square$
\end{proof}

\paragraph{Choice of Step-Size\index{step-size}}
The proof of Theorem \ref{thm:ULAtheorem} decomposes the Wasserstein distance between $\pi_n$ and $f$ into two terms. The first term, which corresponds to the distance between $\pi_n$ and $f_\epsilon,$ decreases exponentially fast with the number $n$ of iterations. The second term, which corresponds to the distance between $f_\epsilon$ and $f,$ 
represents a
systematic bias of order $\epsilon.$ 
This decomposition reveals a trade-off in the choice of step-size. A
smaller step-size \index{step-size} reduces the bias, but slows down the
convergence of $\pi_n$ to the biased 
limit 
$f_\epsilon.$ Thus, 
it is recommended to take the largest value of $\epsilon$ for which the error in 
the approximation $f \approx f_\epsilon$ 
can be tolerated. 
While not considered here, time-varying step-sizes can also be used; in particular, a 
common approach is to take 
bold large 
steps 
during the initial iterations to promote
fast exploration, and then gradually decrease the step-size to reduce the 
bias.

\section{Metropolis Adjusted Langevin Algorithm \index{Metropolis Adjusted Langevin Algorithm (MALA)}}\label{sec:MALA}
In the previous section, 
we have shown that the limit distribution 
of the Markov chain $\{X^{(n)}\}_{n=1}^N$ 
generated
by ULA does not 
agree with the target, even in a simple Gaussian setting. 
In other words, 
the target $f \propto \exp(-V)$ is not an invariant distribution for the Markov kernel 
$$\qlgv(x,\cdot) = \Nc \bigl(x - \epsilon \nabla V(x),2 \epsilon I \bigr),
$$
which defines the conditional distribution of
$X^{(n+1)}$ given that $X^{(n)} = x$. However, the proof of Theorem \ref{thm:ULAtheorem} suggests that, for small step-size \index{step-size} $\epsilon$, the 
limit 
distribution $f_\epsilon$ of ULA is 
only
order $\epsilon$ 
away 
from $f,$ so that the kernel $\qlgv$ can \emph{approximately} preserve the target
if $\epsilon$ is small. 
This motivates the idea of using $\qlgv$ as a proposal Markov kernel within the Metropolis--Hastings algorithm studied in Chapter \ref{chap:MCMC}. The key advantage of this choice of proposal is that 
since $\qlgv$ approximately preserves the target, the 
accept/reject mechanism in the Metropolis--Hastings framework only needs to account for a small correction to ensure convergence to the target.
The resulting method, known as the Metropolis Adjusted Langevin Algorithm (MALA), is summarized below.

\begin{algorithm}[H]
  \caption{Metropolis Adjusted Langevin Algorithm (MALA)\label{algo:MALA}}
  \begin{algorithmic}[1]
  \STATEx{\textbf{Input:} Target $f(x) \propto \exp\bigl(-V(x)\bigr)$, initial distribution $\pi_0$, step-size \index{step-size} $\epsilon>0,$ sample size $N.$}
  \STATE{Define the Markov kernel
  $$\qlgv(x,\cdot) := \Nc \bigl(x - \epsilon \nabla V(x),2 \epsilon I \bigr).$$}
  \STATE{Run Metropolis--Hastings (Algorithm \ref{algo:MH}) with inputs $f$, $\pi_0,$  $\qlgv$, $N.$}
  \vspace{2mm}
  \STATEx{\textbf{Output:} Sample $\{X^{(n)}\}_{n=1}^N.$ }
 \end{algorithmic}
\end{algorithm}

\begin{mybox}[colback=white]{Pros and Cons}
MALA removes the 
asymptotic
bias of ULA at the price of introducing a Metropolis--Hastings accept/reject step.
MALA shares many of the 
pros and cons 
of ULA. 
They are both 
well suited to 
sampling from
log-concave targets but can struggle to explore multi-modal distributions. \index{distribution!multi-modal} Like ULA, MALA requires evaluating $\nabla V,$ which can be computationally
expensive. Compared 
with the
RWM proposals from Chapter \ref{chap:MCMC}, MALA can scale better to high-dimensional problems, particularly for log-concave targets. However, 
the step-size parameter $\epsilon$ can be more delicate to tune than for RWM.
Existing theory shows that the step-size should be smaller in the transient than in the stationary phase of the chain. 
\end{mybox}

\paragraph{Choice of Step-Size\index{step-size}}
As we saw in Section \ref{sec:ULA}, the 
limit 
distribution of ULA depends on the choice of step-size. The ability of ULA to produce samples that are approximately distributed according to the target hinges on 
the
convexity of $V$ and on choosing a small step-size. \index{step-size} On the other hand, MALA corrects the bias of ULA via a Metropolis--Hastings accept/reject mechanism. For MALA, using a larger step-size leads to more rejections but faster exploration of the state space. As for ULA, time-varying step-sizes can be used. 

\paragraph{Scaling} \index{Metropolis Adjusted Langevin Algorithm (MALA)!scaling} Classical theoretical results for MALA developed in the same large-$d$ scaling setting 
that we considered for 
RWM in Chapter \ref{chap:MCMC} 
reveal 
two key insights. 
\begin{enumerate}
	\item The step-size $\epsilon$ should 
    decrease 
    with 
    dimension
    $d$ at rate $d^{-1/3}$ in order to 
    keep the acceptance rate bounded
    away from zero and one. Consequently, the number of iterations needed to get close to equilibrium is $\mathcal{O}(d^{1/3})$, a significant improvement over $\mathcal{O}(d)$ for RWM. 
    Although this depends on the problem, in many settings the cost per iteration is $\mathcal{O}(d)$ for both methods. Thus, the
    overall complexity 
    is 
    $\mathcal{O}(d^{4/3})$
    for MALA,
    compared 
    with $\mathcal O(d^2)$ for RWM.
	\item The optimal acceptance rate for MALA is 
    approximately $0.574$,
    compared 
    with
    the considerably lower 
    rate of $0.234$
    for RWM. 
\end{enumerate}
Additional results concerning the scalability of ULA and MALA under various convexity assumptions on the potential $V$ are discussed in the bibliography Section \ref{sec:literatureLangevin}.
A heuristic reason for the improvement of MALA over RWM is that MALA, unlike RWM, incorporates
information about
the target distribution 
into
the proposal kernel through the gradient $\nabla V$. 

\section{Langevin Equation and its Numerical Solution \index{Langevin!equation}}\label{sec:LangevinSDE}
In Section \ref{sec:ULA} we introduced ULA by analogy 
with
gradient descent\index{gradient descent}. \index{optimization} Another 
insightful 
interpretation of ULA is 
to view it as a 
discretization of the Langevin stochastic differential equation. \index{stochastic differential equation} In fact, this perspective predates the optimization one and explains the name of the algorithm. Here, we summarize some key ideas without delving into the 
important technical issues 
that are needed for a 
rigorous treatment of this topic.  

Given a potential $V: \R^d \to \R,$ the (overdamped) Langevin equation is given by  
\begin{equation}\label{eq:diffusion}
dX(t) = - \nabla V\big(X(t)\big)\, dt + \sqrt{2}\,dW(t), \ \  0 \leq t < \infty,
\end{equation}
where $W$ denotes a standard Brownian motion in $\R^d.$ The solution 
of 
\eqref{eq:diffusion} is a continuous-time stochastic process $\{ X(t)\}_{t\ge 0}.$ Given a final time $T>0,$ a \emph{trajectory} of this process is a random function 
\begin{align*}
[0,T] &\to \R^d \\
t &\mapsto X(t).
\end{align*}
Figure \ref{fig:diffusion} shows 
a single trajectory of the Langevin dynamics.
\begin{figure}[h!]
    \centering
    \includegraphics[width=0.75\textwidth]{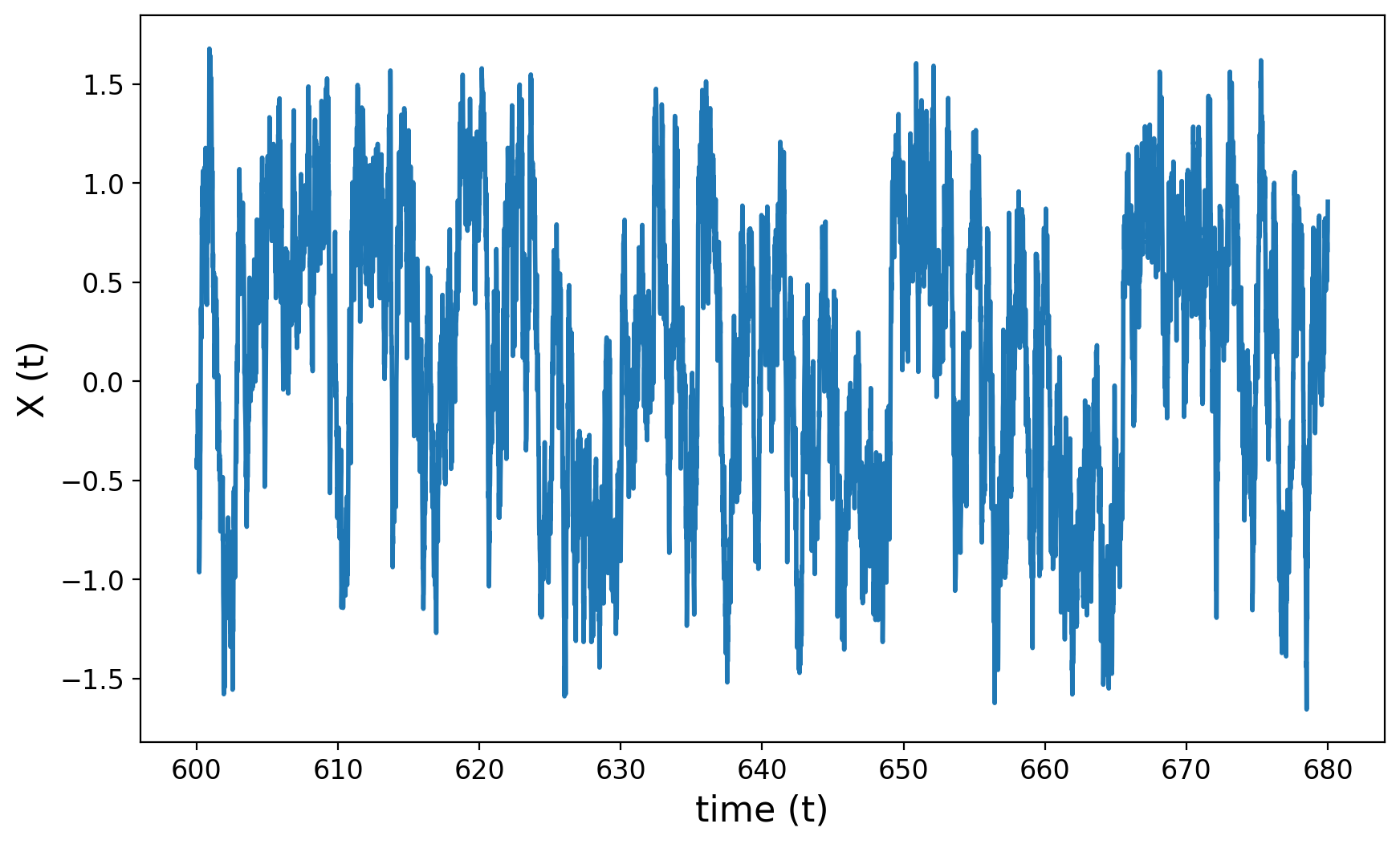}
    \caption{Random trajectory of Langevin equation with the potential $V(x)= x^4 - x^2-0.4x$ shown in Figure \ref{fig:diffusioninvariant}. The process spends most of its time around the modes of the invariant distribution $f \propto \exp( - V)$. }
    \label{fig:diffusion}
\end{figure}

The p.d.f. $\pi_t(x) = \pi(x,t)$ of $X(t)$ defined by \eqref{eq:diffusion} satisfies the Fokker-Planck equation 
\begin{align*}
   \frac{\partial \pi}{\partial t} &= \text{div} \bigl(\nabla V \pi + \nabla \pi \bigr), \\
   \pi(x,0) &= \pi_0(x).
\end{align*}

Notice that for $f \propto \exp(-V),$ 
we have
$\nabla V = - \nabla \log f = \frac{- \nabla f}{f},$ which implies that 
\begin{equation*}
    \text{div} \bigl(\nabla V f + \nabla f \bigr) = 0.
\end{equation*}
Thus, $f$ is an invariant distribution for the Fokker-Planck equation \index{Fokker-Planck equation}: if
$X(0) \sim f$, then  $X(t) \sim f$ for all $t \ge 0.$ 
Furthermore, the key property that makes the Langevin equation useful for
sampling is that, under mild asymptotic growth conditions on the potential $V,$ it defines an ergodic process. This means that, for large $t$,
the distribution of $X(t)$ is approximately $f,$
regardless of the initial distribution $\pi_0$ of $X(0)$. An illustration is given in Figure \ref{fig:histograms}.

\begin{figure}
    \centering
    \includegraphics[width=0.75\textwidth]{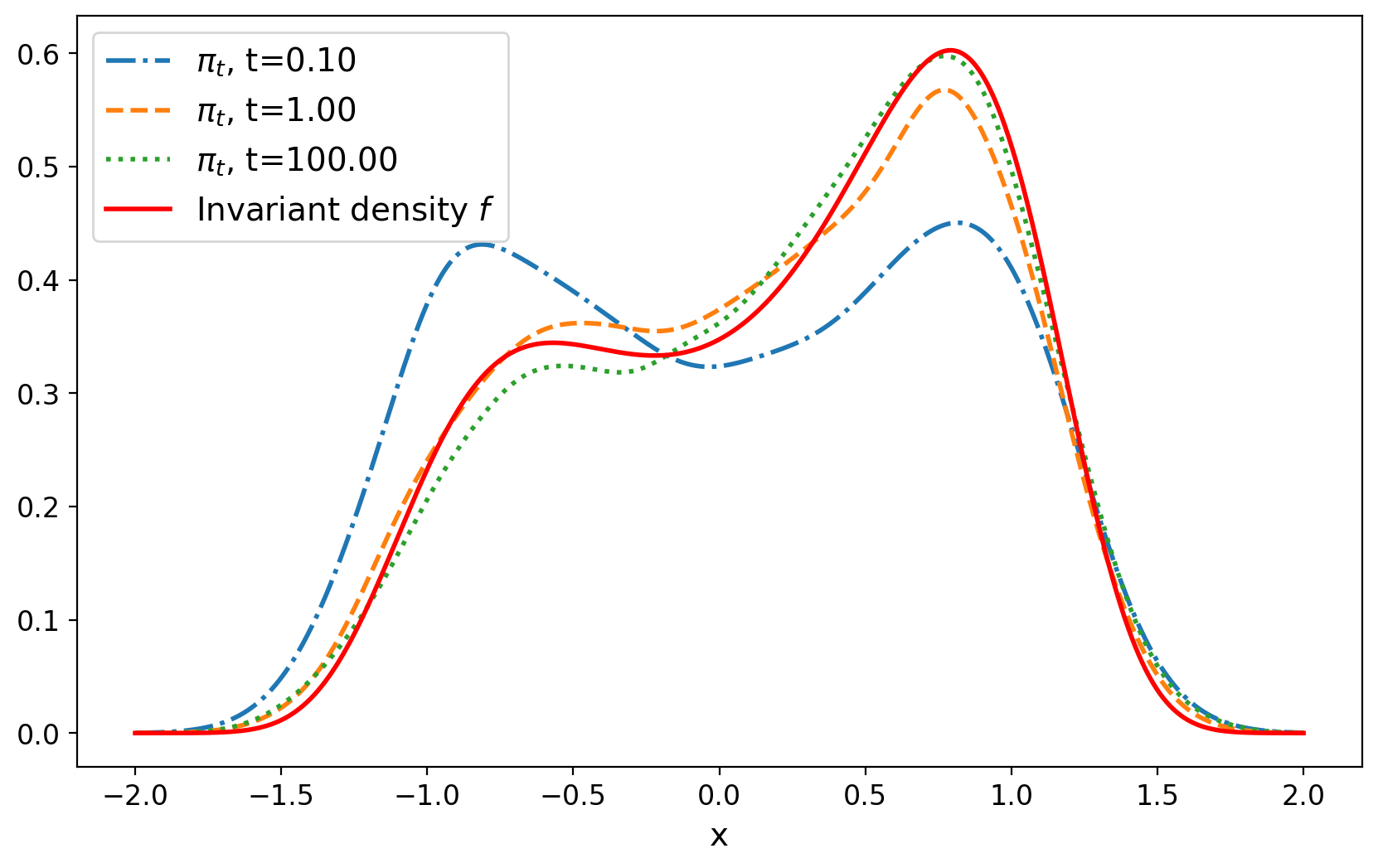}
   \vspace{-2ex} \caption{ Histograms of the distribution $\pi_t$ of $X(t)$ for Langevin equation with the potential $V(x)= x^4 - x^2-0.4x$ in Figure \ref{fig:diffusioninvariant}. For large $t,$ $\pi_t$ is close to the target density $f \propto \exp(-V).$}
    \label{fig:histograms}
\end{figure}

Although any given trajectory may spend 
long periods
of time hovering around some local minima of $V$, ergodicity implies that,
if $T$ is sufficiently large
and $A \subset E = \R^d,$ the proportion of time that the trajectory spends in $A$ 
is approximately 
$\int_A f(x) \, dx$. 
This
sample path ergodicity \index{ergodicity!sample path} 
suggests
a natural way to estimate expectations with respect to $f:$ 
\begin{equation}\label{eq:samplepathergodicityLangevin}
    I_{f}[h] := \int_E h(x) f(x) \, dx \approx \frac{1}{T} \int_0^T h \bigl(X(t)\bigr) \, dt,
\end{equation}
where $h$ is a given test function.  The accuracy of this approximation  ---in which a spatial average of over $E$ is computed via a temporal average along a trajectory of length $T$--- increases as $T$ becomes larger. 

In practice, 
the
Langevin equation \eqref{eq:diffusion} can 
be solved analytically only in simple cases, e.g. for 
quadratic potentials $V$
that lead to
linear stochastic differential equations. Outside these cases, the Euler-Maruyama \index{Euler-Maruyama} method is a standard approach to find numerical solutions.  Let  $0 < t_1 < \cdots < t_N = T$ with $t_n-t_{n-1} =\epsilon, \, n = 1, \ldots, N$. Applied to the Langevin equation \eqref{eq:diffusion}, Euler-Maruyama approximates $X(t_n)$ with $X^{(n)},$ where 
\begin{equation*}
X^{(n+1)} = X^{(n)} - \epsilon \nabla V \bigl(X^{(n)} \big) + \sqrt{2 \epsilon} \, \xi^{(n)}, \qquad \xi^{(n)} \stackrel{\text{i.i.d.}}{\sim} \Nc(0,I).
\end{equation*}
Thus,
ULA 
is precisely the
Euler-Maruyama discretization of 
the
Langevin equation. In analogy with \eqref{eq:samplepathergodicityLangevin}, one can then approximate expectations with respect to the target using the time-discretized process:
\begin{equation*}
    I_{f}[h] := \int_E h(x) f(x) \, dx \approx \frac{1}{N} \sum_{n=1}^N h\bigl(X^{(n)}\bigr).
\end{equation*}
An important observation is that 
the Euler-Maruyama discretization does not generally
preserve the ergodic property of 
the
Langevin equation, i.e. the invariant distribution 
of the
chain $\{ X^{(n)} \}_{n \geq 1}$ 
is generally not 
$f$. However, 
the Markov kernel 
$$
\qlgv(x,\cdot) = \Nc\bigl(x+\epsilon \nabla \log f(x),2 \epsilon I\bigr),
$$
of the ULA Markov chain can be used as a
proposal kernel \index{Markov kernel!proposal} within a Metropolis--Hastings algorithm. 
Combining this proposal with the standard Metropolis--Hastings 
acceptance probability 
$$a(x,z):=\min\biggl\{1,\frac{f(z)}{f(x)}\frac{\qlgv(z,x)}{\qlgv(x,z)} \biggr\}$$
defines MALA. 
From this perspective,
the Metropolis--Hastings accept/reject step in MALA 
corrects the bias
introduced by Euler-Maruyama discretization of 
the
Langevin equation.

\section{Discussion and Bibliography}\label{sec:literatureLangevin}
The idea of using Langevin dynamics  for sampling a given target distribution can be found in
\cite{besag1994comments} and \cite{grenander1994representations}, but it was only systematically developed in 
\cite{roberts1996exponential} where MALA was first proposed and analyzed. This chapter emphasized the optimization viewpoint on Langevin Monte Carlo, as
considered for instance in 
\cite{trillos2023optimization,trillos2018bayesian,wibisono2018sampling,durmus2019analysis}. These works establish convergence guarantees for Langevin equation and discretizations thereof under convexity assumptions 
analogous
to those 
used in the analysis 
of gradient descent optimization \cite{boyd2004convex}.

The convergence and optimal scaling of the step-size for ULA and MALA \index{step-size} have been widely studied, see for instance \cite{roberts2001optimal,christensen2005scaling,durmus2017nonasymptotic,durmus2019high,dwivedi2019log,li2021sqrt,chewi2021optimal} and references therein. In particular, these works show that for both algorithms the step-size must be small when $d$ is large, with scaling $d^{-1/3}, d^{-1/2}$ or $d^{-1}$ under different assumptions on the target and initial distributions. The paper \cite{chen2024convergence} uncovered a \emph{de-localization of bias} phenomenon by which ULA achieves dimension-independence for approximating expectations $\mathcal{I}_f[h]$ of test functions $h$ that only depend on a fixed number of variables. This phenomenon, which has been subsequently studied in \cite{lacker2025hierarchical}, explains the advantage of ULA over MALA for approximating low-dimensional marginals in 
high-dimensional problems.

As in optimization,
\emph{preconditioning} 
can speed up the convergence of Langevin Monte Carlo algorithms.
Preconditioning 
aims to improve
the conditioning of the target and  may be implemented, for instance, by leveraging an ensemble of particles to build an approximation of the inverse Hessian of the potential 
\cite{leimkuhler2018ensemble,garbuno2020interacting,chada2021iterative}. Relatedly, many algorithms leverage ideas from Riemannian geometry to adapt Langevin dynamics to the geometry induced by the target distribution \cite{girolami2011riemann,patterson2013stochastic,trillos2018bayesian}. 
Another sampling
method inspired by optimization is the proximal sampler, in which a Gibbs sampler is run in an extended state space \cite{chen2022improved}. 

An important 
caveat of Langevin Monte Carlo algorithms is 
that they require evaluating gradients,
which in many applications is computationally expensive. For this reason, there has been significant interest in studying Langevin Monte Carlo algorithms that use approximate gradients \cite{dalalyan2019user}.
In this direction, Stochastic Gradient Langevin Dynamics (SGLD) \index{Stochastic Gradient Langevin Dynamics (SGLD)} \cite{welling2011bayesian} provides an important extension 
of ULA and MALA. 
Consider 
Bayesian inference for a parameter $\theta$ 
based on
i.i.d. data $\{y_i\}_{i=1}^K$, giving
rise to 
the target posterior distribution
$$
f(\theta | y_1,\ldots,y_K) \propto f(\theta) \prod_{i=1}^K f(y_i | \theta).
$$
Evaluating the
gradient of the
log-likelihood 
requires summing
over $K$ terms, which is expensive for large $K$. 
SGLD addresses this cost by subsampling $k \ll K $ data points uniformly at random at each iteration. This leads to the update
$$\theta^{(n+1)} = \theta^{(n)} + \epsilon_n \biggl\{\nabla \log\bigl(f(\theta^{(n)} ) \bigr) + \frac{K}{k} \sum_{i=1}^k \nabla \log \bigl(f(y_i^{(n)} | \theta^{(n)}) \bigr) \biggr\} + \sqrt{2\epsilon_n} \, \xi^{(n)}.
$$
Here $y_i^{(n)}$ denotes the $i$-th data point 
among the $k$ observations subsampled
during the $n$-th iteration, and $\xi^{(n)} \sim \Nc(0,I)$. 
The
step-size sequence
$\eps_n$ is chosen 
so that
\begin{align*}
    \sum_{n=1}^{\infty} \epsilon_n = \infty,
    \qquad
    \sum_{n=1}^{\infty} \epsilon_n^2  < \infty.
\end{align*}
These conditions ensure that it is possible to sample from the chain \textit{without} the need for accepting and rejecting. 
For example,
$\epsilon_n = 1/n$ satisfies 
these requirements.
Another approach to 
designing
gradient-free Langevin Monte Carlo algorithms is to rely on an ensemble of interacting particles, coupled through their empirical covariance,
to approximate gradients and precondition the dynamics \cite{leimkuhler2018ensemble,garbuno2020interacting,chada2021iterative}.

Section \ref{sec:LangevinSDE} provided a brief and informal introduction to 
the
Langevin equation. For 
a
rigorous but accessible 
introduction
to stochastic differential equations,\index{stochastic differential equation} we refer to the textbooks \cite{oksendal2013stochastic,evans2012introduction,mao2007stochastic}. For more details on the Langevin equation
and its use 
in
sampling, we refer to \cite{pavliotis2014stochastic}. Numerical methods for stochastic differential equations
are covered in \cite{kloeden2013numerical,higham2021introduction}.

From a theoretical perspective, several works \cite{lelievre2013optimal,rey2015irreversible,christie2025speeding} have shown that adding
a non-reversible drift term 
to the
Langevin equation can accelerate
convergence to the stationary distribution in continuous time. However, 
these non-reversible 
diffusions
can be more expensive to discretize \cite{trillos2018bayesian,trillos2023optimization}, making 
their algorithmic benefits less clear.
A general recipe for building preconditioned non-reversible Langevin-type diffusions can be found in \cite{ma2015complete}, while \cite{ma2017stochastic} discusses Langevin methods for sampling in hidden Markov models.

\chapter{Annealing Strategies}
\label{chap:annealing}

This chapter presents annealing strategies for optimization  \index{optimization} and sampling. \index{sampling}
In metallurgy, annealing 
refers to the process of heating or cooling a 
material to alter its physical properties. In the same spirit, the annealing strategies in this chapter alter a given objective function or target distribution to facilitate optimization 
or
sampling. For optimization, we will \emph{cool} the objective 
so that it becomes increasingly concentrated
around its global optima; for sampling, we will \emph{heat} the target to make it flatter and hence easier to sample from.

The first algorithm we will study is \emph{simulated annealing}, a probabilistic technique for global optimization.
Simulated annealing is particularly 
useful
when 
one seeks an adequate approximation to a global optimum rather a highly accurate approximation of a local optimum.
The main idea behind the algorithm is to sample 
from
a sequence of target distributions that gradually become more 
concentrated
around the global 
optimum. 
Sampling is typically performed 
using
a Metropolis--Hastings algorithm, 
and 
simulated annealing can thus be used for both discrete and continuous optimization
problems. In particular, simulated annealing does not require 
derivatives of the objective, and consequently has been useful for solving and has been applied to solve 
challenging large-scale discrete optimization problems.

While 
in
optimization we alter our objective so that it becomes more peaked, 
in
sampling we flatten the target distribution to facilitate exploration. The main motivation 
is that sampling from multi-modal \index{distribution!multi-modal} distributions can be challenging when regions of high probability are separated by regions of low probability. This is especially problematic for Markov chain Monte Carlo algorithms based on local moves.
 
We will study two strategies 
that
leverage auxiliary heated target distributions to speed-up exploration: simulated tempering and parallel tempering. \emph{Simulated tempering} runs a single chain in an augmented state space; at each iteration, 
the temperature is updated by sampling an auxiliary variable. 
In contrast, \emph{parallel tempering} ---also known as replica exchange Markov chain Monte Carlo--- runs several chains at different temperatures in parallel; at each iteration, a Metropolis--Hastings step is used to potentially swap the states of two of the chains. Augmenting the target to speed-up mixing and running multiple 
interacting chains
are two important ideas that underpin many sampling algorithms and convergence diagnostics.

This chapter is organized as follows. Section \ref{sec:annealingoptimization} introduces annealing strategies for optimization, leading to the simulated annealing algorithm. Section \ref{sec:annealingsampling}  describes 
annealing
strategies for sampling, 
focusing
on simulated tempering and parallel tempering. Section \ref{sec:annealingbibliography} closes with bibliographical remarks.

\section{Annealing Strategies for Optimization \index{annealing strategies!optimization}}\label{sec:annealingoptimization}
In this section, we first consider sampling algorithms for finding the mode of a given target distribution. We then introduce 
simulated annealing as an optimization algorithm for minimizing a given objective function.

\subsection{Finding the Mode of a Distribution} \label{ssec:711}
Let $f$ be the distribution of interest. The mode of $f$ is the set $ \mathcal{M}=\{\xi: f(\xi) \ge f(x)\,\, \text{for all } x\in E \}$ of \emph{global} maxima of $f$.\footnote{We are mainly interested in cases where the mode comprises a single element $\mathcal{M} = \{\xi\}$, in which case we abuse our terminology and refer to the element $\xi$ rather than the set $\{\xi\}$ as the mode.} 
A naive way to approximate the mode is to sample from $f$ and return the draw with the largest density:

\FloatBarrier
\begin{algorithm}
  \begin{algorithmic}[1]
    \STATE{ Generate $X^{(n)}\sim f,$ \quad $1 \le n \le N.$ }
    \STATE{ Output $\text{argmax}_{1 \le n \le N} f(X^{(n)}).$
    }
  \end{algorithmic}
\end{algorithm}
\FloatBarrier
This approach is 
generally inefficient.
Since our goal is to locate the mode, we would prefer to generate more samples near the mode than would typically be obtained by sampling directly from $f.$
This suggests replacing $f$ by a new distribution with the same mode, but with more mass concentrated around it.
One way to 
do this is to define
\begin{equation}
f_\beta(x) \propto  f(x)^\beta
\end{equation}
where $\beta>0$ is large.
The 
distribution
$f_\beta$ is more peaked than $f$, but
has the same mode. 
Throughout this section, we tacitly assume that $\int f(x)^\beta \, dx <\infty,$ so that $f(x)^\beta $ can be normalized into a probability density.
\begin{example}[Gaussian Target]
Consider a Gaussian density $f = \Nc(\mu, \sigma^2)$ with mode $\mu.$ In this case 
$$f_\beta(x) \propto \exp\biggl(- \frac{(x-\mu)^2}{2\sigma^2/\beta} \biggr),$$
which shows that $f_\beta(x)$ is 
the
Gaussian density $\Nc(\mu, \sigma^2/\beta)$. Thus, $f_\beta$ has the same mode as $f$, but its variance is smaller by a factor of $\beta$.
As $\beta$ increases, $f_\beta$ becomes increasingly concentrated around $\mu.$
\hfill \qedhere
\end{example}

\begin{example}[Interpretation of $\beta$ as Inverse Temperature]
    Consider, as in Chapter \ref{chap:diffusions},
    an unnormalized target density of the form $f(x) \propto \exp \bigl( -V(x)\bigr),$ where $V: \R^d \to \R$ is a given potential function. Then, $f_\beta(x) \propto \exp \bigl( - \beta V(x)\bigr).$ Following the ideas in Chapter \ref{chap:diffusions}, the Langevin equation
\begin{equation}\label{eq:diffusionannealing}
dX(t) = - \nabla V\big(X(t)\big)\, dt + \sqrt{2\beta^{-1}}\,dW(t), \ \  0 \leq t < \infty,
\end{equation}
has $f_\beta$ as 
its
invariant distribution under mild assumptions on $V$. 
The parameter $\beta$ controls the size of the noise term. When
$\beta^{-1}$ is large, the stochastic forcing from $dW(t)$ has a stronger
effect on the trajectory. When $\beta$ is large, the noise term
$\sqrt{2\beta^{-1}}\,dW(t)$ is small, and the Langevin dynamics becomes
closer to the deterministic gradient system  $\dot{x} = -\nabla V(x)$.

In physical applications, $\beta$ is interpreted as the
\emph{inverse temperature}\index{inverse temperature} of the system.
Equivalently, $\beta^{-1}$ plays the role of temperature. Physically,
temperature measures the amount of microscopic agitation in the system.
When the temperature is low, particles have smaller typical kinetic
energy, and the deterministic drift $-\nabla V$ has a stronger influence
on the motion. The dynamics therefore tends to move downhill in the
potential landscape. When the temperature is high, particles have larger
typical kinetic energy, and random thermal fluctuations play a larger
role. These fluctuations make the dynamics less deterministic and allow it to
explore the state space more broadly.

    \hfill \qedhere
\end{example}

We can use any of the sampling algorithms 
considered 
in previous chapters to sample 
from
$f_\beta.$ However, sampling
from
$f_\beta$ typically becomes 
harder for larger $\beta.$ 
To see this, consider the
RWM Algorithm \ref{algo:RWMH} with a proposal distribution $g$ 
that is
symmetric 
about 
zero. The probability of accepting a
proposed
move from $X^{(n)}$ to $Z^*$ 
is 
\begin{equation}
\min \biggl\{ 1, \frac{f_\beta(Z^*)}{f_\beta(X^{(n)})}  \biggr\} = 
\min \biggl\{ 1,  \biggl( \frac{f(Z^*)}{f(X^{(n) })} \biggr)^\beta \biggr\} ,
\end{equation}
which does not depend on the (usually unknown) 
normalizing constant
\index{normalizing constant} of $f_\beta.$ 
For large values of $\beta$, however, the chain becomes increasingly unlikely to accept moves to states with lower density. In the limit $\beta\to\infty,$
the probability of accepting a 
proposal
$Z^*$ becomes $1$ if $f(Z^*)> f(X^{(n)}) $ and $0$ otherwise. For this reason, a typical scenario 
in the large $\beta$ regime
is that 
the chain becomes trapped at a \emph{local} maximum of $f$, rather than reaching a mode (global maximum).

\begin{example}[Discrete Target and Convergence to Local Maxima]\label{Ex:RWHMannealing}
Consider the p.m.f. $f(x)$  supported  on $\{1, 2, \ldots, 5\},$ defined by
\begin{align*}
f(x) =
\begin{cases}
0.4 \quad \quad &x=2, \\
0.3 \quad \quad &x=4, \\
0.1 \quad \quad &x=1,3,5.
\end{cases}
\end{align*}
and shown in Figure \ref{fig:annealingpmf}. Clearly the mode is $x=2.$  Suppose
we sample 
from
$f_\beta(x) \propto f(x)^\beta$ 
using the random walk proposal
$$Z^* = X^{(n)} + \xi^*,$$
where 
\begin{align*}
\Prob( \xi^* = \pm 1) &= 0.5 \quad \text{if}\,\, X^{(n)} \in \{2,3,4\}, \\
\Prob( \xi^* =  - 1) &= 1 \quad ~~~\text{if} \,\,X^{(n)} = 5, \\
\Prob( \xi^* =   +  1) &= 1 \quad ~~~ \text{if} \,\, X^{(n)} = 1.
\end{align*}
As $\beta \to \infty,$
the probability of accepting a move from $4$ to $3$ converges to $0$ since $f(4) >f(3).$ 
Since the
chain can only move from $4$ to $2$ via $3,$ it cannot escape the local
maximum 
at $x=4$ 
in this limiting regime. \hfill \qedhere

\begin{figure}
    \centering
    \includegraphics[width=0.75\linewidth]{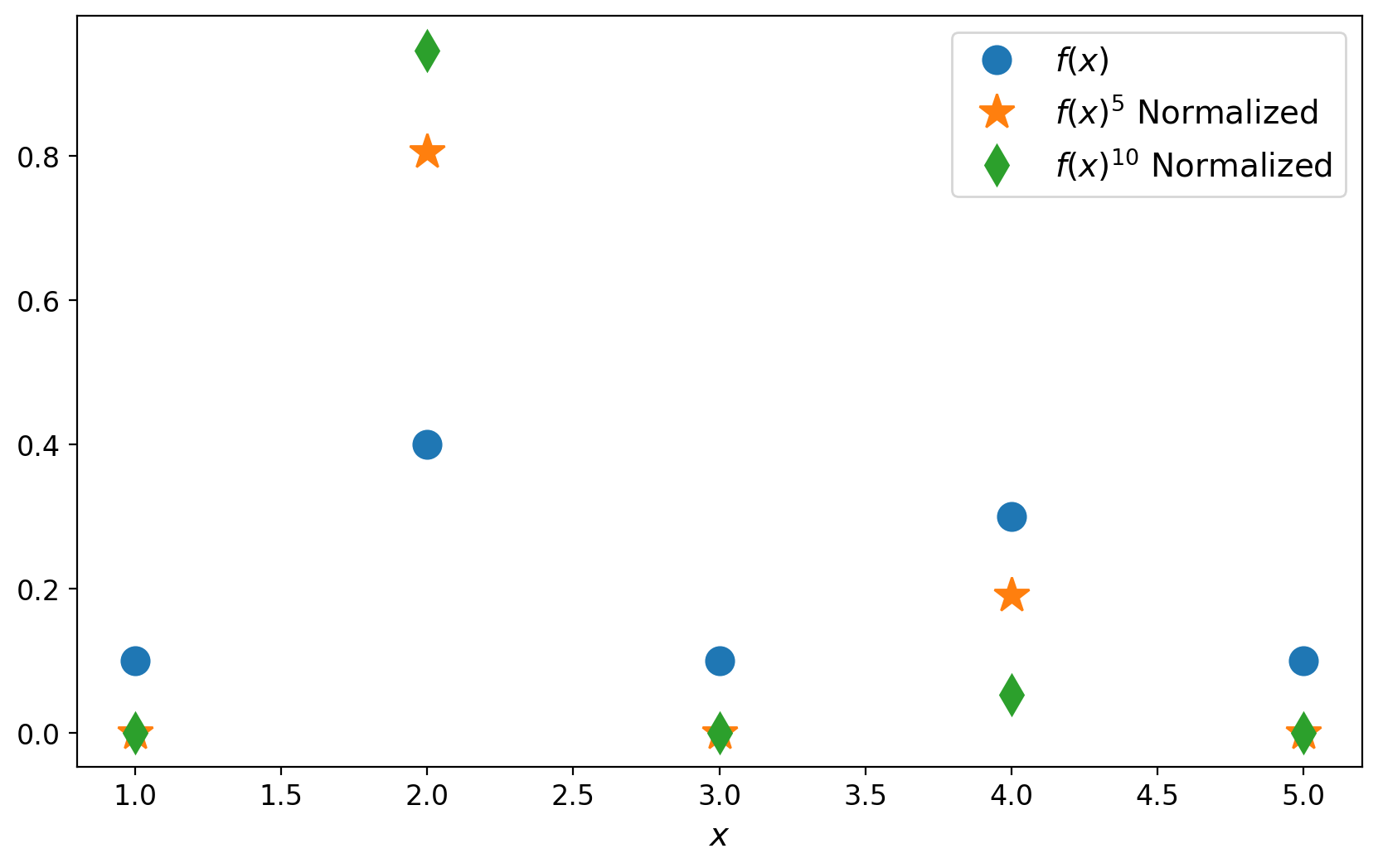}
    \caption{\label{fig:annealingpmf} An illustration of the target p.m.f. in Example \ref{Ex:RWHMannealing} along with $f_\beta(x) \propto f(x)^\beta$ with $\beta=5,10.$}
\end{figure}
\end{example}

The 
discussion above shows that for large $\beta$,
the distribution $f_\beta$ is
more
concentrated around the mode, but it is also harder
to sample from. This motivates considering a time-changing sequence of target distributions $f_{\beta_n}$, where $(\beta_n)_{n=0}^\infty$ is an increasing sequence of inverse temperatures. 
This is the key idea behind the simulated annealing algorithm introduced next.

\subsection{Minimizing an Objective Function \index{objective function}}\label{ssec:712}
We now turn to the problem of finding the global minimum of an objective function $V: E\to \R.$ Equivalently,\footnote{This equivalence 
implicitly assumes
that $\int \exp\bigl(-V(x) \bigr) \, dx <\infty,$ so that $\exp\bigl(-V(x) \bigr) $ can be normalized.} we seek to find the mode of the distribution 
$$f(x) \propto \exp \bigl( -V(x) \bigr).$$
As in Subsection \ref{ssec:711}, we can raise $f$ to different powers to obtain 
$$f_{\beta_n} (x) \propto \exp \bigl( - \beta_n V(x) \bigr),$$
where $\{\beta_n\}_{n=1}^N$ is a given sequence of inverse temperatures. The simulated annealing \index{simulated annealing} algorithm proceeds as follows:

  \FloatBarrier
\begin{algorithm}
\caption{Simulated Annealing}
\begin{algorithmic}[1]
\STATEx {\bf Input}: Function to minimize $V$, proposal Markov kernel $q(x,z),$ initialization $X^{(0)},$ sample size $N,$ sequence of inverse temperatures $\{\beta_n \}_{n=1}^{N}.$
\STATE For $n= 0, 1, \ldots, N-1$ do:
\STATE Draw $Z^* \sim q( X^{(n)} , \cdot).$
\State  Update
$$X^{(n+1)}=\begin{cases}
    Z^* \quad \quad \text{ with probability }   a :=  \min \biggl\{ 1, \frac{\exp \left( - \beta_{n+1} V(Z^*) \right)}{\exp \left(-\beta_{n+1} V(X^{(n)})\right)}  \frac{q(Z^*,X^{(n)})  }{q(X^{(n)}, Z^*) }  \biggr\} , \\
    X^{(n)} \quad \,   ~\text{with probability}  \, 1-a.
                    \end{cases} $$
\STATE {\bf Output}: Sample $\{X^{(n)} \}_{n=1}^N$ that can be used to minimize $V$ by  $\text{argmin}_{1 \le n \le N} V(X^{(n)}).$ 
\end{algorithmic}
\end{algorithm}
  \FloatBarrier
  
  Simulated annealing
  is a highly versatile optimization algorithm that can be applied 
  to both
  discrete and continuous optimization without requiring 
  first-
  or second-order derivatives of the objective.
 Notice that the accept/reject mechanism agrees with that in the Metropolis--Hastings framework in Algorithm \ref{algo:MH}, 
 with $f_{\beta_{n+1}}$ used as the target distribution when proposing the update from $X^{(n)}$ to $X^{(n+1)}$.

As in the Metropolis--Hastings \index{Metropolis--Hastings} framework, there is freedom in the choice of proposal Markov kernel. Since a key feature of simulated annealing is 
that is can be implemented without derivative information, random walk proposal kernels have traditionally been favored over Langevin proposal kernels considered in Chapter \ref{chap:diffusions}, which require evaluating $\nabla V$. 
If a RWM \index{Random Walk Metropolis (RWM)} with symmetric proposal is used, 
that is,
$Z^* = X^{(n)} + \xi^*$ 
where
the distribution of $\xi^*$ 
is symmetric about the origin,
then the acceptance probability simplifies to 
$$a(X^{(n)}, Z^*) = \min \biggl\{ 1, \exp\Bigl(-\beta_n\bigl(V(Z^*) - V(X^{(n)}) \bigr) \Bigr)\biggr\},$$
giving the standard implementation of simulated annealing. 

\paragraph{Choice of Inverse Temperatures}
The convergence analysis of simulated annealing is challenging. The algorithm usually converges to a local minimum, 
and whether it reaches a global minimum depends strongly
on the choice of inverse temperatures. \index{inverse temperature} 

The process of varying the temperature parameter is often called
\emph{tempering}\index{tempering}. In simulated annealing, this amounts
to choosing an increasing sequence of inverse temperatures
$\{\beta_n\}_{n=1}^N$. Two common inverse-temperature schedules are the
following:
\begin{itemize}
\item \textbf{Logarithmic tempering:} \index{tempering!logarithmic} For 
a given $\beta_0>0,$ set
$\beta_n := \frac{\log(1+n)}{\beta_0}, \, n \ge 1.$
This schedule increases the inverse temperature
slowly enough that global convergence can be established for certain special cases. Simulated annealing with logarithmic tempering behaves similarly to an exhaustive search, and so while it 
has strong
theoretical guarantees, it is typically too slow to
be of practical use.
\item \textbf{Geometric tempering:} \index{tempering!geometric} 
For a given $\beta_0>0$ and $\alpha > 1,$ set $\beta_n = \alpha^n \beta_0$.
This schedule increases the inverse temperature relatively quickly and is
therefore more common in practice, although it typically comes with fewer
theoretical guarantees.
\end{itemize}

\begin{mybox}[colback=white]{Pros and Cons}
Simulated annealing can be effective for discrete or continuous optimization problems in which the convexity or smoothness assumptions required by deterministic first- or second-order methods are not satisfied.
Importantly, simulated annealing does not require derivatives and can be applied in high-dimensional settings. 

The algorithm converges under mild assumptions when logarithmic tempering is used, but in practice it is usually implemented with geometric tempering, for which few theoretical guarantees are available. Additionally, the choice of the temperatures can have a significant impact on the behavior of the algorithm. 
\end{mybox}

\section{Annealing Strategies for Sampling \index{annealing strategies!sampling}}\label{sec:annealingsampling}
Sampling multi-modal target distributions that have regions of high probability separated by regions of low probability can be challenging. Annealing
strategies seek to improve sampling by introducing auxiliary \emph{tempered} distributions that bridge regions of low probability. \index{distribution!multi-modal}

\subsection{Simulated Tempering\index{simulated tempering}}
Let $1 = \beta_1 > \beta_2 > \cdots > \beta_K$ be a decreasing sequence of inverse temperatures, and define
$f_{\beta_k}(x) \propto  f(x)^{\beta_k}.$ Here $f = f_{\beta_1}$ is the target distribution of interest. For $k \ge 2$, 
the density $f_{\beta_k}$ is a \emph{flattened} version of $f$
and is therefore typically easier to sample from. 
Simulated tempering considers an augmented target distribution
$$f^{\text{\tiny ST}}(x,k) \propto c_k f_{\beta_k}(x)$$
over the variable of interest $x$ and the discrete auxiliary variable 
$k \in \{ 1,\dots, K\}$
which indexes the inverse temperatures. 
Since $\beta_1=1,$ samples
from the distribution of interest $f$ can 
then be obtained by retaining the $x$-component of the 
samples $(X^{(n)}, k^{(n)})$ 
for which
$k^{(n)} = 1.$

A simple  
Markov chain Monte Carlo algorithm for the 
augmented 
distribution $f^{\text{\tiny ST}}$ 
alternates between updating the state variable $x$ and updating the temperature index $k$. 
Moves for the 
temperature index
can be proposed using the simple kernel 
\begin{align*}
r(k, k+1) &= r(k,k-1) = \frac12, \quad \quad k \in \{2,\ldots, K-1\},\\
r(K,K-1) &= r(1,2) = 1.
\end{align*}
The simulated tempering algorithm then proceeds as follows:

  \FloatBarrier
\begin{algorithm}
\caption{Simulated Tempering}
\begin{algorithmic}[1]
\STATEx {\bf Input}: Target distribution $f,$ inverse temperatures $1 = \beta_1 > \cdots > \beta_K$, constants $\{c_k\}_{k=1}^K,$
proposal Markov kernel $q(x,z),$ proposal Markov kernel $r(k,\ell),$ initialization $(X^{(0)}, k^{(0)}),$ sample size $N.$
\STATE For $n= 0, 1, \ldots, N-1$ do:
\STATE Draw $\ell^* \sim r( k^{(n)} , \cdot).$
\State  Update 
$$k^{(n+1)}=\begin{cases}
    \ell^* \quad \quad \text{ with probability } \,  a_{k,\ell}:=
      \min \biggl\{ 1, \frac{f^{\text{\tiny ST}}(X^{(n)},\ell^*) }{f^{\text{\tiny ST}}(X^{(n)},k^{(n)})}  \frac{r(\ell^*, k^{(n)})}{r(k^{(n)}, \ell^*)}    \biggr\} , \\
    k^{(n)} \quad \, \text{ with probability } \,  1 - a_{k,\ell}.
                    \end{cases} $$
\STATE Draw $Z^* \sim q( X^{(n)} , \cdot).$
\State  Update 
$$X^{(n+1)}=\begin{cases}
    Z^*  \quad \quad \text{ with probability } \,  a_{X,Z}:=
      \min \biggl\{ 1, \frac{f^{\text{\tiny ST}}(Z^*,k^{(n+1)})}{f^{\text{\tiny ST}}(X^{(n)},k^{(n+1)})}  \frac{q(Z^*, X^{(n)})}{q(X^{(n)}, Z^*)}    \biggr\} , \\
    X^{(n)} \quad \,  \text{ with probability } \,  1 - a_{X,Z}.\end{cases} $$
\STATE {\bf Output}: Sample $\{ (X^{(n)}, k^{(n)})  \}_{n=1}^N.$ 
\end{algorithmic}
\end{algorithm}
  \FloatBarrier

The idea is that, at lower inverse temperatures, the tempered distributions are flatter and the chain can move more freely across the state space. When the augmented chain returns to the temperature of interest $\beta_1=1$, the $x$-component may therefore have moved to a different region of the state space. This can improve mixing for the target distribution $f$.

\begin{mybox}[colback=white]{Pros and Cons}
Simulated tempering faces two 
main implementation challenges.
The first is the choice of the constants $c_k$, which control how often the augmented chain visits each temperature level.
A common goal is to choose them so that the chain spends roughly the same amount of time at each temperature.
If all the $f_{\beta_k}$ can be normalized, then this objective can be achieved by taking the constants $c_k$ to be equal.
In practice, however, the required
normalizing constants are typically unknown and need to be estimated using an alternative adaptive algorithm (e.g. Wang-Landau\index{Wang-Landau}). The second implementation issue is 
the choice of the
inverse temperatures. \index{inverse temperature} There is a trade-off between using temperature levels that are close enough to allow moves between neighboring temperatures to be accepted with reasonable probability, and using few enough levels that the algorithm remains computationally efficient.
There 
exists
a vast literature on this issue, notably in physics. 
\end{mybox}

\subsection{Parallel Tempering \index{parallel tempering}}
Parallel tempering also uses a decreasing sequence of inverse temperatures $1 = \beta_1 > \beta_2 > \cdots > \beta_K.$ Unlike simulated tempering, however, 
it runs $K$ chains in parallel. The $k$-th chain targets the tempered distribution $f_{\beta_k},$ and we set $\beta_1 = 1$ so that $f_{\beta_1} = f.$ 
The chains  are allowed to communicate through a swapping mechanism. Intuitively, chains at  lower inverse temperatures target flatter distributions and can therefore explore the state space more freely. Occasional swaps allow the chain at the target temperature $\beta_1=1$ to benefit from this improved exploration. More precisely, after updating each of the $K$ chains, we choose a pair of chains and propose to swap their states. This proposed swap is then accepted or rejected in a way that preserves the correct joint invariant distribution.
We denote by $X^{(n)}_k$ the $n$-th sample from the $k$-th chain. The algorithm outputs the chain $\{X_1^{(n)}\}_{n=1}^N$ corresponding to temperature $\beta_1=1.$

  \FloatBarrier
\begin{algorithm}
\caption{Parallel Tempering}
\begin{algorithmic}[1]
\STATEx {\bf Input}: Target distribution $f,$ inverse temperatures 
$1= \beta_1 > \cdots > \beta_K,$
proposal Markov kernels $\{q_k(x,z)\}_{k=1}^K,$  initializations $\{X_k^{(0)}\}_{k=1}^K,$ sample size $N.$
\STATE For $n= 0, 1, \ldots, N-1$ do:
\STATE For $k = 1, \ldots, K$: generate $\tilde{X}_k^{(n+1)}$ by doing a Metropolis--Hastings step (including accept/reject) with current state $X_k^{(n)},$  proposal kernel $q_k,$ and target $f_{\beta_k}.$ 
\STATE Choose $\ell, m \in \{1, \ldots, K\}$ with $\ell \neq m$ uniformly at random. 
\STATE For $k \notin \{ \ell, m\}$ set $X_k^{(n+1)} = \tilde{X}_k^{(n+1)}.$
\State  Attempt a swap of states between the $\ell$-th and the $m$-th chains:
$$ \bigl(X_\ell^{(n+1)} , X_m^{(n+1)}\bigr)=
\begin{cases}
    \bigl(\tilde{X}_m^{(n+1)} , \tilde{X}_\ell^{(n+1)}\bigr) \quad \quad   \text{ with probability }\,  a_{\ell,m},   \\
    \bigl(\tilde{X}_\ell^{(n+1)} , \tilde{X}_m^{(n+1)}\bigr) \quad \quad  \text{ with probability }\,  1 - a_{\ell,m},
                    \end{cases} $$
                    where
                    $$a_{\ell,m}
                    := \min \Biggl\{ 1, \frac{f_{\beta_\ell}\bigl(\tilde{X}_m^{(n+1)}\bigr) f_{\beta_m}\bigl(\tilde{X}_\ell^{(n+1)}\bigr) }{f_{\beta_m}\bigl(\tilde{X}_m^{(n+1)}\bigr) f_{\beta_\ell}\bigl(\tilde{X}_\ell^{(n+1)}\bigr)}   \Biggr\}.$$
\STATE {\bf Output}: Sample $\{ X_1^{(n)}\}_{n=1}^N.$ 
\end{algorithmic}
\end{algorithm}
  \FloatBarrier

Parallel tempering requires proposal Markov kernels $\{q_k\}_{k=1}^K$ to update each chain. \index{Markov kernel!proposal} A popular choice is to use the Langevin proposals described in Chapter \ref{chap:diffusions}. 
Specifically,
for the $k$-th chain, one can propose
\begin{equation*}
Z_k^* = X_k^{(n)} + \epsilon_k  \nabla \log f\big(X_k^{(n)}\big) + \sqrt{2\epsilon_k \beta_k^{-1}}\xi_k^*, 
\end{equation*}
where $\xi_k^* \sim \Nc(0,I)$ and $\epsilon_k>0$ is a given step-size. Notice that large $\beta_k^{-1}$ results in large random moves that 
facilitate exploration of the state space.
Such moves may still be accepted with high probability, 
because the corresponding chain targets
the flattened distribution $f_{\beta_k}.$

\begin{mybox}[colback=white]{Pros and Cons}
In contrast to simulated tempering, parallel tempering does not require estimating the normalizing constants $c_k$. While parallel tempering involves running $K$ chains in parallel, doing so is not a significant obstacle with modern parallel computing. The
choice of the number of chains $K$ and the inverse temperatures $\beta_k$ can have a substantial effect on performance. This choice 
has been widely studied, and there is some evidence that geometric scaling of the temperatures may be optimal under mild assumptions. 
\end{mybox}

\section{Discussion and Bibliography}\label{sec:annealingbibliography}
The simulated annealing \index{simulated annealing} algorithm was introduced in \cite{kirkpatrick1983optimization}. Geman and Geman \cite{geman1987stochastic} conjectured 
convergence under logarithmic tempering \index{tempering!logarithmic} in finite state spaces and Gidas \cite{gidas1985nonstationary} gave a rigorous proof.
Hajek \cite{hajek1988cooling} later  
showed that choosing the inverse temperatures \index{inverse temperature} via logarithmic tempering is a sufficient and necessary condition for global convergence. 
Early refinements were surveyed in 
\cite{van1987simulated,szu1987fast,aarts1988simulated,aarts1985statistical}. For more recent expositions, we refer to \cite{henderson2003theory,lou2016massively}. The presentation in this chapter follows \cite[Chapter 8]{johansen2010monte}. In particular, Example \ref{Ex:RWHMannealing} is taken from that reference.

Simulated and parallel tempering were \index{simulated tempering} \index{parallel tempering}introduced in \cite{swendsen1986replica,marinari1992simulated}, and further developed in \cite{geyer1991computing,hukushima1996exchange}. For a review of parallel tempering, we refer to \cite{earl2005parallel}. The question of how to choose the temperatures has been widely studied, and recent theory suggests that a geometric sequence of temperature ratios is optimal under mild assumptions \cite{dupuis2022analysis}. Another important question is how often 
swaps between chains should be proposed.
In a continuous time setting, recent theory indicates the potential advantage of frequent swaps, formalized 
through
the infinite swapping limit for parallel tempering \cite{dupuis2012importance}.

Augmenting the target to 
improve mixing, and running
multiple chains that communicate with 
one another,
are two important ideas that 
appear throughout modern sampling methodology and convergence diagnostics. Related themes arise, for example, in the Box-Muller sampling method in Chapter \ref{chapter1}, Gelman and Rubin's diagnostic in Chapter \ref{chap:MCMC}, the slice sampler in Chapter \ref{chap:gibbs}, and Hamiltonian Monte Carlo in Chapter \ref{chap:HMC}.
Likewise, introducing a sequence of auxiliary target densities 
that gradually
reach a desired target is a powerful idea that underpins annealed importance sampling \cite{neal2001annealed} and many sequential Monte Carlo algorithms; 
see particle filters in Chapter \ref{chap:particlefilters} and references therein. 
Algorithms that rely on tempered transitions for sampling multi-modal distributions are discussed in \cite{neal1996sampling}. \index{distribution!multi-modal}
The recent work \cite{han2025polynomial}
establishes polynomial complexity for sampling multi-modal distributions via sequential Monte Carlo with a geometric annealing schedule. Annealed importance sampling combined with Langevin Monte Carlo is studied in \cite{agarwal2026convergence}.

\chapter{Hamiltonian Monte Carlo}
\label{chap:HMC}

This chapter focuses on MCMC methods in which proposals are generated using Hamiltonian dynamics. \index{Hamiltonian!dynamics} The key advantage of Hamiltonian Monte Carlo (HMC) \index{Hamiltonian Monte Carlo (HMC)} is that it permits large moves while maintaining a high acceptance rate. As a result, HMC 
avoids 
the local behavior characteristic of the Metropolis--Hastings algorithms considered 
in 
previous chapters, where successive samples 
typically 
remain confined to a small neighborhood of the current state. 

Throughout this chapter, we write the target density as
$$f(q) \propto \exp \bigl(-V(q)\bigr), \qquad q \in \R^d.$$
Rather than directly sampling from $f$, 
we will sample 
from an \emph{augmented target} $f^H$ 
of the form
\begin{align}\label{eq:deffH}
\begin{split}
f^H(q,p) &\coloneqq \frac{1}{c}\exp\bigl(-V(q)\bigr)\exp\bigr(-K(p)\bigr)\\
&= \frac{1}{c}\exp\bigl(-H(q,p)\bigr), \qquad q,p \in \R^d,
\end{split}
\end{align}
where $c$ is a normalizing constant, $K: \R^d \to \R$ satisfies $\int_{\R^d}\exp\bigl(-K(p)\bigr)\,dp < \infty,$ and 
the function $H: \R^{2d} \to \R$ given by $$H(q,p) := V(q)+K(p)$$ is called the Hamiltonian. 
This function determines the Hamiltonian dynamics used to generate proposals in the augmented space $\R^{2d}.$ 
Notice that $f$ is the marginal density of $f^H$ with respect to $q$. Thus, samples from $f$ can be obtained by sampling from $f^H$ and discarding the auxiliary variable $p$.

At first sight, it
 is not immediately obvious why sampling from $f^H,$ rather than directly sampling from $f,$ should be advantageous. Indeed, augmenting the target
 may seem  counterintuitive, 
 since $f^H$
 is defined on a higher-dimensional space than $f$. The central insight, however, is that augmenting the target makes it possible to leverage Hamiltonian dynamics to design proposals that traverse the state space in large steps without sacrificing acceptance probability or the overall mixing efficiency of the chain. 

The notation and terminology used in this chapter are rooted in physics, 
which is
the origin of HMC. We now briefly summarize the physical interpretation 
of the objects introduced above:

\begin{itemize}
	\item 
    $q \in \R^d$ are the variables of interest, physically interpreted as position \index{position} variables.
    \item  $p \in \R^d$ are auxiliary variables, interpreted as momentum \index{momentum}  variables.
	\item $V$ is the negative log-target density,
    interpreted as potential energy.\index{potential energy}
	\item $K$ specifies the distribution of the auxiliary 
    variables. 
    In HMC the choice of $K$ is part of the algorithmic design. A standard choice is $K(p) = \frac{1}{2}p^T M^{-1}p$ 
    where $M \in \R^{d \times d}$ is positive definite. This choice gives the momentum variables a Gaussian distribution.
    The function $K$ is physically interpreted as kinetic energy. \index{kinetic energy}
	\item The measure $\mu^H$ with Lebesgue density $f^H$ is 
    known as 
    the canonical measure \index{canonical measure} associated with the Hamiltonian $H.$ It corresponds to the Gibbs measure \index{Gibbs measure} 
    governing the joint distribution
    of $(q,p),$ as described in Section \ref{sec:motivatingapplications}. The normalizing constant $c$ is referred to as the partition function, and is typically unknown. 
\end{itemize}

This chapter is organized as follows. Hamiltonian dynamics are reviewed in Section \ref{sec:hamiltondynamics}, emphasizing three important properties for designing sampling algorithms. The HMC algorithm and an idealized version of it are introduced in Section \ref{sec:samplingalgorithm}. We show that these algorithms are exact, in the sense that
they leave the target distribution invariant.
Finally, 
Section \ref{sec:biblioHMC} contains bibliographical remarks, including references that address the geometric ergodicity of HMC. 

\section{Hamiltonian Dynamics \index{Hamiltonian!dynamics}}\label{sec:hamiltondynamics}
Given a Hamiltonian \index{Hamiltonian} $H,$  the corresponding Hamilton equations are
\begin{equation}\label{eq:hamilton}
\begin{alignedat}{3}
&\frac{dq_i}{dt} = \frac{\partial H}{\partial p_i},  &&\qquad \quad \quad    1 \leq i \leq d,\\
&\frac{dp_i}{dt} = -\frac{\partial H}{\partial q_i}, &&\qquad \quad \quad   1 \leq i \leq d.
\end{alignedat}
\end{equation}
  We vectorize equations \eqref{eq:hamilton} as follows:
    \begin{equation*}
    \frac{dx}{dt} 
   \  =J^{-1}\nabla H(x), \quad \quad J:=
    \begin{bmatrix}
      0 & -I_{d\times d}\\
      I_{d\times d} & 0
    \end{bmatrix}  \in \R^{2d \times 2d},
    \quad x:=  \begin{bmatrix} q \\p 
   \end{bmatrix} \in \R^{2d},
  \end{equation*}
  where 
  \begin{equation*}
  \frac{d}{dt} \begin{bmatrix} q \\p 
   \end{bmatrix} =\left[\frac{dq_1}{dt},\ldots,\frac{dq_d}{dt},\frac{dp_1}{dt},\ldots,\frac{dp_d}{dt}\right]^T, \quad \quad \nabla H = \left[\frac{\partial H}{\partial q_1}, \ldots, \frac{\partial H}{\partial q_d},\frac{\partial H}{\partial p_1},\ldots,\frac{\partial H}{\partial p_d}\right]^T  .
  \end{equation*}
   The matrix $J$ plays a central role in the Hamiltonian formalism. 
   It is skew-symmetric, meaning that $J^T = -J$. Consequently:
  \begin{enumerate}
    \item $J^{-1}$ is also skew-symmetric, since $J^{-1}=-(J^T)^{-1}=-(J^{-1})^T$.
    \item The corresponding quadratic form is zero, since $
    x^TJx=(x^TJx)^T=x^TJ^Tx=-x^TJx.
    $
  \end{enumerate}
  We will use these two properties of $J$ in what follows. 

\bigskip

Hamilton equations \eqref{eq:hamilton} define a flow on phase space. For $t \ge 0,$ let $\phi_t: \R^{2d} \to \R^{2d}$ denote the
time-$t$ flow map generated by Hamilton equations.
Then, for $x \in \R^{2d}$, $\phi_t(x) := x(t)$ where $x(t)$ 
is the solution at time $t$ of Hamilton equations with initial condition $x \in \R^{2d}:$
$$ \frac{dx}{dt}=J^{-1}\nabla H (x),
    \qquad x(0)=x.$$ 
    In particular,  $\phi_0(x) = x(0) = x,$ so $\phi_0$ is the identity map.
Moreover, for $s,t \ge 0,$ the semigroup property holds:  
$$\phi_{s+t}(x) = \phi_s \circ \phi_t(x).$$
Indeed, the solution at time $t+s$ can be obtained by first evolving the system up to time $t,$ yielding $x(t) = \phi_t(x),$ and then evolving for an additional time $s$ from this intermediate state, giving $x(t+s) = \phi_s \bigl( x(t)\bigr). $ Finally, we note that under the standard assumptions ensuring existence and uniqueness of solutions, the flow of a Hamiltonian system is invertible, with $\phi_t^{-1} = \phi_{-t},$ a fact we will use repeatedly in this chapter. 

The notion of a flow map is useful for analyzing the qualitative behavior of numerical integrators, which can be viewed as approximations of the exact flow. As we shall see, to ensure that HMC algorithms preserve the augmented target distribution, it is essential to employ numerical methods that retain some 
key properties of Hamiltonian flows. 

\begin{example}[Harmonic Oscillator]
	Let $d = 1$, $V(q) = \frac{q^2}{2}$, and $K(p) = \frac{p^2}{2}$. Then,  Hamilton equations are
	\begin{align*}
	&\frac{dq}{dt} = p,\\
	&\frac{dp}{dt} = -q.
	\end{align*}
	The general solution has the form
	\begin{align*}
	&q(t) = r\cos(a+t),\\
	&p(t) = -r\sin(a+t),
	\end{align*}
    for constants $a \in \R$ and $r \ge 0$.
    Thus,  the flow map $\phi_t$ is a clockwise rotation in the 
    $q$-$p$
    plane. \hfill \qedhere
	\end{example}

Hamiltonian dynamics 
dovetail nicely with MCMC because 
they satisfy three key properties: conservation of the Hamiltonian, symplecticity, and reversibility. We next describe these properties and explain how they support the design of MCMC algorithms.

\subsection{Conservation of Hamiltonian \index{Hamiltonian!conservation}}
  By 
  the chain rule and the skew-symmetry of $J$, along a solution trajectory $\{x(t)\}_{t \ge 0}$ of Hamilton equations we 
  have
    \begin{align*}
    \frac{d}{dt} H\bigl(x(t)\bigr) = \nabla H\bigl(x(t)\bigr)^T\frac{d}{dt}x(t) 
   = \nabla H\bigl(x(t)\bigr)^TJ^{-1}\nabla H \bigl(x(t)\bigr)=0, \quad \quad x = \begin{bmatrix} q \\p 
   \end{bmatrix}.
    \end{align*}
    Thus, the Hamiltonian $H$ is conserved along 
    trajectories of 
    Hamilton 
    equations. 
    Equivalently, if $\phi_t$ denotes the Hamiltonian flow, then 
    for all $t \ge 0$ and $x \in \R^{2d},$
    $$H\bigl(\phi_t(x)\bigr)=H(x).$$  
	Recalling the definition of the augmented target $f^H$ in \eqref{eq:deffH}, we note that $f^H$ depends on $x = (q,p)$ only through $H(x)$. 
   Therefore, conservation of $H$ implies that Hamiltonian trajectories remain on the level sets of $f^H$.
	\begin{figure}[h]
		\centering
		\includegraphics[width=0.5\textwidth]{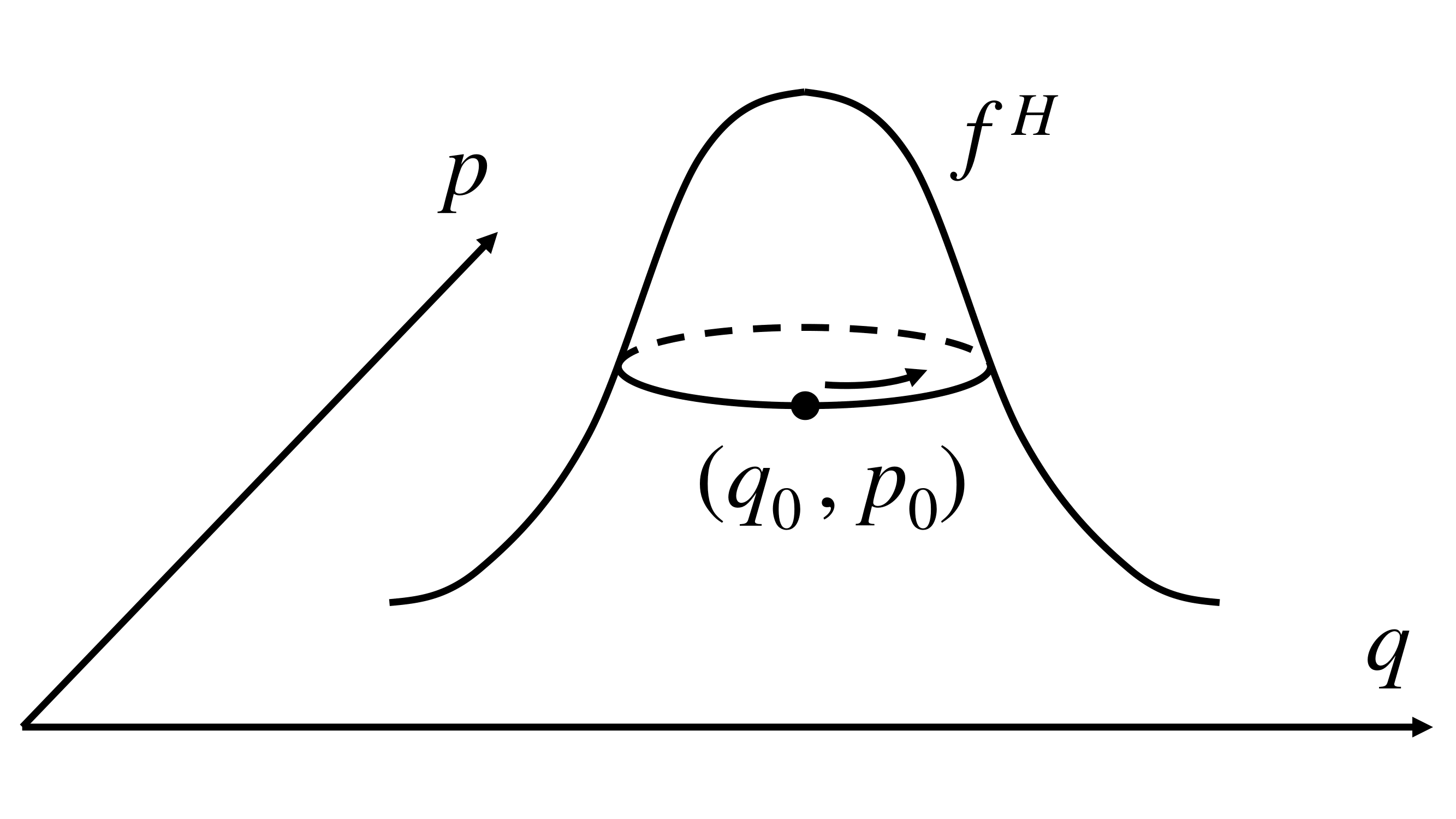}
		\caption{Movement along a level set of $f^H$ starting from $(q_0,p_0)$.}
		\label{fig:hmc_fig}
	\end{figure}	
	This property, together with the preservation of the canonical measure described next, 
    allows Hamiltonian dynamics to generate proposals that can move far from the current state while still having a high probability of being accepted.
	
    \subsection{Symplecticity \index{symplecticity} and Preservation of the Canonical Measure \index{canonical measure!preservation}}
  The Hamiltonian flow map \(\phi_t\) advances solutions of Hamilton
equations by time \(t\). This flow preserves volume in phase space. More precisely, if
\(D\subseteq \R^{2d}\) is Lebesgue-measurable and
\(\operatorname{Vol}(D)\) denotes its \(2d\)-dimensional Lebesgue measure,
then $\operatorname{Vol}(\phi_t(D))=\operatorname{Vol}(D)$.
  This is a direct consequence of the symplecticity of Hamiltonian dynamics. We next recall the definition of symplecticity and prove the symplecticity of Hamiltonian dynamics using a theorem 
 of
  Poincaré.
  \begin{definition}[Symplectic Maps]
  A linear map $L:\mathbb{R}^{2d}\to\mathbb{R}^{2d}$ is symplectic if $L^TJL=J$.  A continuously differentiable map $\psi: \R^{2d} \to\mathbb{R}^{2d}$ is symplectic if its Jacobian is symplectic everywhere. That is, if 
  \begin{equation*}
    \psi'(x)^TJ\psi'(x)=J, \quad \quad \forall x \in \R^{2d}. \tag*{\qedhere}
  \end{equation*}
  \end{definition}
  \begin{remark}
  In $\R^2$ symplecticity of linear maps 
  is equivalent to 
  area preservation of parallelograms; in $\R^{2d}$ it is equivalent to conservation of the oriented areas of the projections of parallelograms into the coordinate planes $(q_i,p_i)$. 
  \end{remark}
  \begin{theorem}[Hamiltonian Flows are Symplectic]\label{thm:poincare}
  Let $H$ be a $\mathcal{C}^2$ Hamiltonian. Then the corresponding flow map $\phi_t$ is symplectic for all $t \ge 0$.
  \end{theorem}
  \begin{proof}
  Let $\phi(x,t) := \phi_t(x)$.
    We need to show that, for any $t \ge 0$ and $x\in \R^{2d},$
    \begin{equation}\label{eq:idenity}
    \frac{\partial}{\partial x}\phi(x,t)^TJ\frac{\partial}{\partial x}\phi(x,t)  = J.
    \end{equation}
  At $t = 0$, this identity holds since $\phi(x,0)$ is the identity map on $\mathbb{R}^{2d}$.  Therefore, it suffices to show that the left-hand side of \eqref{eq:idenity} is constant in time. To that end, recall that 
  $\frac{\partial}{\partial t}\phi(x,t)=J^{-1}\nabla H\bigl(\phi(x,t)\bigr).$ Therefore, using that $H$ is twice continuously differentiable, we have
  \begin{align*}
  \frac{\partial^2}{\partial t \partial x} \phi(x,t)
  &=\frac{\partial^2}{\partial x \partial t} \phi(x,t)\\
  &=\frac{\partial}{\partial x} J^{-1}\nabla H \bigl(\phi(x,t)\bigr)\\
  & =J^{-1}\nabla^2H \bigl(\phi(x,t)\bigr)\frac{\partial}{\partial x}\phi(x,t),
  \end{align*}
  where $\nabla^2$ denotes the Hessian. Using this identity and the skew-symmetry of $J,$ we compute
  \begin{align*}
  &\frac{\partial}{\partial t} \biggl(\frac{\partial}{\partial x}\phi(x,t)^TJ\frac{\partial}{\partial x}\phi(x,t)  \biggr)\\ 
  &=\frac{\partial^2}{\partial t\partial x}\phi(x,t)^TJ\frac{\partial}{\partial x}\phi(x,t)+
  \frac{\partial}{\partial x}\phi(x,t)^TJ\frac{\partial}{\partial t\partial x}\phi(x,t)\\
  &=\frac{\partial}{\partial x}\phi(x,t)^T\nabla^2H \bigl(\phi(x,t)\bigr)J^{-T}J\frac{\partial}{\partial x}\phi(x,t)+
  \frac{\partial}{\partial x}\phi(x,t)^TJJ^{-1}\nabla^2H\bigl(\phi(x,t)\bigr)\frac{\partial}{\partial x}\phi(x,t)\\
  &=-\frac{\partial}{\partial x}\phi(x,t)^T\nabla^2H\bigl(\phi(x,t)\bigr)\frac{\partial}{\partial x}\phi(x,t)+\frac{\partial}{\partial x}\phi(x,t)^T\nabla^2H \bigl(\phi(x,t)\bigr)\frac{\partial}{\partial x}\phi(x,t)\\
  &= 0. 
  \end{align*}
  Thus the expression in \eqref{eq:idenity} is constant in time, and since it holds at $t = 0,$ it hold for all $t \ge 0.$ This completes the proof. \hfill $\square$
  \end{proof}

	Taking determinants on both sides of the identity
		$$(\phi_t')^T J \phi_t' =J$$
		 shows that the Jacobian determinant of $\phi_t$ has absolute value equal to one. Combining this property and the conservation of the Hamiltonian yields the following important result:

    \begin{theorem}[Preservation of Canonical Measure]\label{preserve canonical measure}
    The 
    Hamiltonian
    flow $\phi_t$ preserves the canonical measure associated with $H$.
    \end{theorem}
    \begin{proof}
    For any Lebesgue-measurable set $D$, using the change of variables $x=\phi_t(y)$, or equivalently $y= \phi_{-t}(x)$, we obtain
      \begin{align*}
      \mu^H\bigl(\phi_t^{-1}(D)\bigr)
      &=\int_{\phi_{-t}(D)}\frac{1}{c}e^{-H(x)} \, dx
      =\int_D\frac{1}{c}e^{-H(\phi_{-t}(x))}\left|\phi_{-t}'(x)\right| \, dx\\
      &=\int_D\frac{1}{c}e^{-H(x)} \, dx
      =\mu^H(D). 
      \end{align*}
      Here we used conservation of the Hamiltonian,
\(H(\phi_{-t}(x))=H(x)\), and volume preservation,
\(|\phi_{-t}'(x)|=1\). \hfill $\square$
    \end{proof}
    
    In the language of probability theory and 
    (optimal) transport, 
    Theorem \ref{preserve canonical measure} shows 
    that the \emph{push-forward} of $\mu^H$ under the map $\phi_t$ 
    is again 
    $\mu^H$.

		\subsection{Time Reversibility\index{time reversibility}}
  The term \emph{time reversible} originates in classical mechanics.
  As a simple example, consider a pendulum of mass $m$ moving without friction.  
  Suppose the pendulum is released from $q(0)=q_0$ with momentum $p(0)=mv_0.$ Since there is no friction, energy is conserved. When the pendulum later returns to the same position $q_0$, its velocity has the opposite sign, namely $-v_0$. Equivalently, its momentum is $-mv_0$.
   If we reverse the momentum at that time and let the system evolve again, the pendulum retraces its previous path backward in time. This is the basic idea of time reversibility. We now formalize this notion for dynamical systems.
  
  \begin{definition}[Involution and Reversibility\index{reversibility}]
  A linear map $S$ is an involution \index{involution} if 
  $S^2 = S\circ S$ is the identity map.
  A bijection $\phi$ is called reversible with respect to an involution $S$ if $S\circ\phi=\phi^{-1}\circ S$.
  A differential equation is called reversible with respect to an involution $S$  if, for every fixed $t$, its flow
  map
  $\phi_t$ is reversible with respect to $S$. That is,   $ S\circ \phi_t = \phi_t^{-1}\circ S.$
  \end{definition}
  The following result establishes a convenient characterization of time reversibility for differential equations:
  \begin{theorem}[Characterization of Reversibility]\label{thm:reversibility}
  The
  differential equation $\frac{dx}{dt}=u(x)$ is time reversible with respect to the involution $S$ if and only if 
  $S\circ u=-u\circ S$.
  \end{theorem}
  \begin{proof}
By definition, the flow map $\phi_t$ of the differential equation satisfies 
\begin{equation*}
    \frac{d}{dt}  \phi_t(x) = u \bigl(\phi_t(x) \bigr).
\end{equation*}
In particular, since $\phi_0$ is the identity map, evaluating at $t=0$ we obtain
\begin{equation}\label{eq:auxrevers}
      \frac{d}{dt} \phi_t(x) \Bigl|_{t=0} = u(x). 
\end{equation}

Suppose first that the differential equation is time reversible, so that 
\begin{equation}\label{eq:revaux2}
    S \circ \phi_t = \phi_{-t} \circ S.
\end{equation}
Differentiating the left-hand side of \eqref{eq:revaux2} with respect to $t$ and evaluating at $t=0$ yields
\begin{align*}
    \frac{d}{dt} S\bigl( \phi_t(x)\bigr) \Big|_{t=0} = S \Bigl( \, \frac{d}{dt} \phi_t(x) \Big|_{t=0} \Bigr)  = S\bigl(u(x) \bigr),
\end{align*}
where we used linearity of $S$ and \eqref{eq:auxrevers}. 
Differentiating the right-hand side of \eqref{eq:revaux2} with respect to $t$ and evaluating at $t=0,$ we obtain by the chain rule that
\begin{equation*}
    \frac{d}{dt} \phi_{-t} (Sx) \Big|_{t=0} = - u(Sx).
\end{equation*}
Hence, we have shown that $S\circ u = -u \circ S,$ as desired.

Conversely, suppose that $S\circ u = -u \circ S.$ Define $y(t) = S \bigl( \phi_t(x)\bigr).$
Then 
\begin{equation*}
    \frac{dy}{dt} = S \Bigl( \frac{d}{dt} \phi_t(x) \Bigr) = S \bigl( u(\phi_t(x))\bigr) = - u\bigl(y(t) \bigr). 
\end{equation*}
Thus $y(t)$ solves the reversed system $\frac{dy}{dt} = -u(y)$ with initial condition $y(0) = Sx.$ Therefore, $y(t) = \phi_{-t}(Sx),$ which implies $S\bigl(\phi_t(x)\bigr) = \phi_{-t}(Sx),$ completing the proof. 
  \hfill $\square$
  \end{proof}

The following result, whose proof is the subject of Exercise~8.1 of the published version of this book, is a direct consequence of Theorem \ref{thm:reversibility}.  
  
  \begin{corollary}[Reversibility of Hamiltonian Systems]\label{corollaryeven}
 Consider a Hamiltonian system with
  Hamiltonian $H(q,p) = V(q) + K(p),$ where 
  $K$ 
  is an even function. Then the system is reversible with respect to the momentum flip involution $S(q,p) = (q, -p)$.
    \end{corollary}

  \begin{example}[Gaussian Momentum Variables and Reversibility]\label{ex:kinetic}
  The kinetic energy $K(p)=\frac{1}{2}  p^TM^{-1}p$ is an even function, since
  $$K(-p)= \frac{1}{2} (-p)^TM^{-1}(-p)= \frac{1}{2}  p^TM^{-1}p=K(p).$$ 
  Therefore, Corollary~\ref{corollaryeven} implies that Hamiltonian systems with Hamiltonian $H(q,p) = V(q) + \frac{1}{2}  p^TM^{-1}p$
  are
  time reversible with respect to 
  the momentum flip $S(p,q)=(q,-p)$.
  Notice that with this choice of kinetic energy,
the marginal of $f^H(q,p) \propto \exp\bigl(-H(q,p) \bigr)$ over the momentum variables is Gaussian $\Nc(0,M).$ 
  \hfill \qedhere
  \end{example}

\section{From Hamiltonian Dynamics to a Sampling Algorithm}\label{sec:samplingalgorithm}

\subsection{Idealized Hamiltonian Monte Carlo Algorithm}
   For the following algorithm, we consider a Hamiltonian of the form 
  $H(q,p)=V(q) + \frac{1}{2} p^TM^{-1}p$ as in Example \ref{ex:kinetic}. Our goal is to construct a Markov kernel that leaves the canonical measure $\mu^H$ associated with $H$ invariant. Recall that the target of interest is the marginal distribution in $q,$ with density proportional to  $e^{-V}$; the Gaussian momentum variable $p$ serves only as an auxiliary variable. 
  As usual, we assume access to the potential $V(q)$ 
  and its gradient, 
  but not to the normalizing constant of 
  either the target or the augmented target.
  \begin{algorithm}[H]
  \caption{Idealized, Not-Implementable Hamiltonian Monte Carlo\label{algo:Exact HMC}}
  \begin{algorithmic}[1]
    \STATEx{\textbf{Input:} Target distribution $f,$ positive definite matrix $M,$ 
    initialization $q^{(0)},$ duration parameter \index{duration parameter} $\lambda$, flow map $\phi_\lambda,$ sample size $N.$}
    \STATEx For $n =0, \ldots, N-1$ do:
    \STATE{\textbf{Momentum refreshment:}
      Sample $p^*\sim\mathcal{N}(0,M).$
    }
    \STATE{\textbf{State update:}
    $(q^{(n+1)},p^{(n+1)})=\phi_\lambda(q^{(n)},p^*).$}
    \STATEx{\textbf{Output:} Sample $\{q^{(n)}\}_{n=1}^N$.}
  \end{algorithmic}
\end{algorithm}
The idealized HMC \index{Hamiltonian Monte Carlo (HMC)!idealized} algorithm is not implementable in practice because it assumes that the Hamiltonian flow map $\phi_\lambda$ can be evaluated exactly. 
Nevertheless, this idealized algorithm is conceptually useful because it clarifies the role of Hamiltonian dynamics. If the exact flow map were
available, 
then
no accept/reject 
correction
would be 
needed: the resulting Markov kernel would preserve the canonical measure exactly. This is proved in Theorem~\ref{th:exactHMC} below.

In practice, 
the flow map
$\phi_\lambda$ is approximated using 
a numerical integrator. Since numerical integrators 
generally do not preserve the 
Hamiltonian exactly,
an accept/reject step is introduced to correct for 
this approximation error. 
This will be made precise in Algorithm \ref{algo:Numerical HMC} and Theorem \ref{th:approximateHMC}.

\begin{theorem}\label{th:exactHMC}
Let $\mu^H$ be the canonical measure associated with the Hamiltonian $H(q,p)=\frac{1}{2}p^TM^{-1}p+V(q)$. 
The 
Markov kernel implicitly defined by Algorithm \ref{algo:Exact HMC} leaves $\mu^H$ invariant. 
\end{theorem}
\begin{proof}
By the definition of the Hamiltonian, under $\mu^H$ the momentum variables  are $\Nc(0,M)$ and the position and momentum variables are independent. Therefore $\mu^H$ is clearly invariant under the momentum refreshment step. 
On the other hand, Theorem \ref{preserve canonical measure} shows that the 
flow map $\phi_\lambda$ preserves the canonical measure associated with $H$. 
This implies that step $2$ also preserves the canonical measure $\mu^H$. 
\hfill $\square$
\end{proof}

\subsection{Hamiltonian Monte Carlo Algorithm\index{Hamiltonian Monte Carlo (HMC)}}
Time reversibility, conservation of the Hamiltonian, and symplecticity are fundamental properties of Hamiltonian systems.
In most problems of interest, however, 
Hamilton 
equations cannot be solved in closed form and must therefore be 
approximated numerically.
A numerical integrator 
may preserve 
some of these properties \emph{exactly}.
For example, 
the well-known 
leapfrog\index{leapfrog} method is both reversible and symplectic. Unfortunately, 
no numerical integrator can, in general, be simultaneously 
symplectic and exactly 
conserve the Hamiltonian.  

\begin{example}[The Leapfrog Integrator]
	The leapfrog method applied to Hamilton equations with 
    $K(p) \coloneqq \sum_{i=1}^d \frac{p_i^2}{2m_i}$ 
    and 
    step-size \index{step-size} $\epsilon>0$ is
	\begin{align*}
	p_i\left(t+\frac{\varepsilon}{2}\right) &= p_i(t)-\frac{\varepsilon}{2} \, \frac{\partial V}{\partial q_i}\bigl(q(t)\bigr) \, ,  &\text{(half momentum step)}\\
	q_i(t+\varepsilon) &= q_i(t)+\varepsilon \, \frac{p_i\left(t+\frac{\varepsilon}{2}\right)}{m_i} \, , &\text{(full position step)}\\
	p_i(t+\varepsilon) &= p_i\left(t+\frac{\varepsilon}{2}\right)-\frac{\varepsilon}{2} \, \frac{\partial V}{\partial q_i} \bigl(q(t+\varepsilon)\bigr) \, . &\text{(half momentum step)}
	\end{align*}
	It 
    is time reversible and 
    preserves 
    volume exactly. To 
    discretize Hamilton equations for a time interval of length $\lambda,$ we may take $L$ leapfrog  steps with $\lambda = \epsilon L. $  
    \hfill \qedhere
 \end{example}
 
  \begin{figure}[htp]
    \centering
    \includegraphics[width=1\columnwidth]{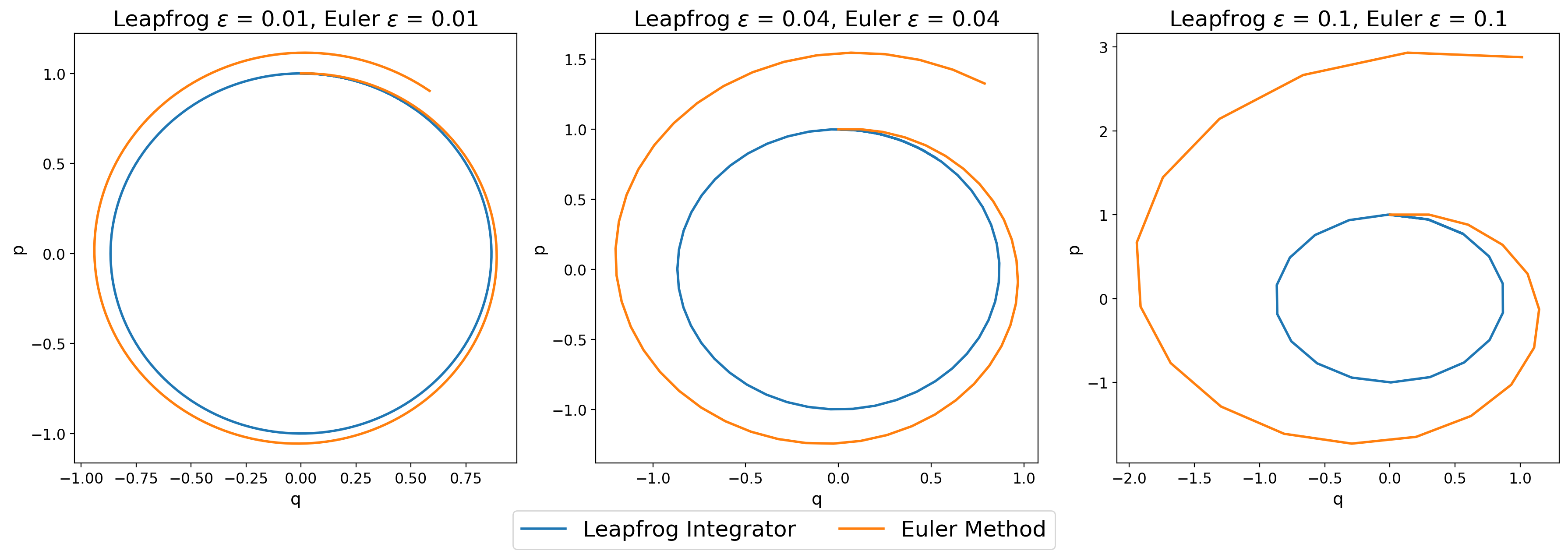}
    \caption{Simulation of 
    Hamilton 
    equations with $H(q,p) = \frac{3p^2}{2}+\frac{4q^2}{2}$ using leapfrog and the forward Euler scheme $p(t+\varepsilon) = p(t)+\varepsilon \frac{dp}{dt}, q(t+\varepsilon) = q(t)+\varepsilon \frac{dq}{dt}$.
    The trajectories
    produced by the Euler method quickly drift away from the corresponding level set of the Hamiltonian.  This reflects the fact that the forward Euler scheme is not symplectic, and is therefore unsuitable for use in HMC.} 
    \label{fig:HMC_toy}
  \end{figure}

Let $\Psi_\lambda$ be a symplectic \index{symplecticity} numerical integrator that is time reversible \index{time reversibility} with respect to the momentum flip involution and 
approximates the Hamiltonian flow map $\phi_\lambda$.
In HMC, proposals are generated using 
$\Psi_\lambda$. Since $\Psi_\lambda$ does not, in general, preserve the Hamiltonian exactly, a Metropolis--Hastings accept/reject step is used to correct for the resulting discretization error.
This is analogous to MALA (Algorithm \ref{algo:MALA} in Chapter \ref{chap:diffusions}), where a Metropolis--Hastings accept/reject step is used to correct for the bias introduced 
by discretizing the Langevin equation with the Euler-Maruyama scheme.
\begin{algorithm}[H]
  \caption{Hamiltonian Monte Carlo (HMC) \label{algo:Numerical HMC}}
  \begin{algorithmic}[1]
    \STATEx{\textbf{Input:}
    Target distribution $f,$ positive definite matrix $M,$ 
    initialization $q^{(0)}$, duration parameter $\lambda$, numerical integrator $\Psi_\lambda,$ sample size $N.$}
    \STATEx For $n =0, \ldots, N-1$ do:
    \STATE{\textbf{Momentum refreshment:}
      Sample $p^{\text{\tiny NEW}}\sim\mathcal{N}(0,M)$.
    }
    \STATE{\textbf{Proposal step:} Set
    $(q^*,p^*)=\Psi_\lambda\bigl(q^{(n)},p^{\text{\tiny NEW}}\bigr)$.}
    \STATE{\textbf{Accept/reject step:} Set 
    $$\bigl(q^{(n+1)},p^{(n+1)}\bigr) = \begin{cases}
    (q^*,p^*)  \qquad \,\,  \text{with probability} \,\, a:= \min\left\{1,e^{H(q^{(n)},p^{\text{\tiny NEW}}\bigr)-H(q^*,p^*)}\right\}, \\
     \bigl(q^{(n)}, -p^{\text{\tiny NEW}}\bigr) \, \text{with probability} \,\, 1- a.
    \end{cases}$$
     }
    \STATEx{\textbf{Output:} Sample $\{q^{(n)}\}_{n=1}^N$.}
  \end{algorithmic}
\end{algorithm}

\begin{mybox}[colback=white]{Pros and Cons}
By leveraging Hamiltonian dynamics in an augmented state space, HMC 
overcomes the local 
behavior characteristic of RWMH and MALA proposals. As a result, it typically scales more effectively to high-dimensional problems (see below for further discussion). However, HMC is often more difficult to implement and 
tune
than either RWMH or MALA. The No-U-Turn Sampler, discussed in the bibliography in Section \ref{sec:biblioHMC}, addresses this issue by 
enabling automatic tuning of HMC. 
Like MALA, HMC requires 
evaluating the 
gradient of $V,$ which can be computationally expensive.
\end{mybox}

\paragraph*{Scaling of HMC \index{Hamiltonian Monte Carlo (HMC)!scaling}}
Classical theoretical results for HMC developed in the same high-dimensional setting in which we discussed RWMH in Chapter \ref{chap:MCMC} and MALA in Chapter \ref{chap:diffusions} show that
the leapfrog step-size \index{step-size} should scale as $\mathcal{O}(d^{-1/4})$. Consequently, HMC requires $\mathcal{O}(d^{1/4})$ steps to traverse the state space, compared with $\mathcal{O}(d^{1/3})$ for MALA and $\mathcal{O}(d)$ for RWMH. 
These analyses further indicate that the optimal acceptance rate for HMC is $0.651$, compared with $0.574$ for MALA and $0.234$ for RWMH. 

\bigskip

Our goal in the remainder of this section is to understand why the 
acceptance probability used in HMC
ensures that
$f^H$ is an invariant distribution. 
To this end, we will rely on the time reversibility of the numerical integrator and its volume preservation. Without volume preservation, the correct Metropolis--Hastings acceptance probability would contain an additional Jacobian determinant term.
The main result of this chapter is the following: 

\begin{theorem}[HMC Leaves Target Invariant]\label{th:approximateHMC}
Suppose that $\Psi_\lambda$ is symplectic and reversible with respect to the momentum flip involution. Then, the  Markov kernel implicitly defined by HMC leaves invariant the canonical measure associated with $H$.
\end{theorem}
\begin{proof}
Throughout the proof, we denote by $S$ the momentum flip involution.  First, notice that the momentum refreshment step preserves the canonical measure $\mu^H$. It therefore suffices to show that $\mu^H$ is invariant under steps $2$ and $3$. 
The corresponding Markov kernel is given by 
$$\pi_x(dz)=a(x)\delta_{\Psi_\lambda(x)}(z)dz+ \bigl(1-a(x)\bigr)\delta_{S(x)}(z)dz,$$
where $\delta_x(z)$ denotes a Dirac mass at $x$ and 
$$a(x) :=\min \Bigl\{1,e^{H(x)-H(\Psi_\lambda(x))} \Bigr\}.$$
Let $A\subseteq \R^{2d}$ be a Lebesgue-measurable set. We need to show that $$\int_{\mathbb{R}^{2d}}\pi_x(A)\mu^H(dx)=\int_A\mu^H(dx).$$ Expanding the left-hand side yields 
\begin{align*}
& \int_{\R^{2d}} \pi_x(A)\mu^H(dx)
=\int_{\R^{2d}}a(x){\bf{1}}_A \bigl(\Psi_\lambda(x)\bigr)f^H(x) \, dx+\int_{\R^{2d}} \bigl(1-a(x)\bigr){\bf{1}}_A \bigl(S(x)\bigr)f^H(x) \, dx\\
&=\int_{\R^{2d}} {\bf{1}}_A \bigl(S(x)\bigr)f^H(x) \, dx +\int_{\R^{2d}}a(x){\bf{1}}_A \bigl(\Psi_\lambda(x)\bigr)f^H(x) \, dx-\int_{\R^{2d}}a(x){\bf{1}}_A \bigl(S(x)\bigr)f^H(x)\, dx.
\end{align*}
The first term on
the right-hand side equals $\int_A\mu^H(dx)$,
since $S$ preserves the canonical measure. 
It remains to show that the last two terms cancel. 
For this, we perform a change of variables in the third term. Using  that $S\circ S = id,$ and that
both  $\Psi_\lambda$ and $S$ have unit Jacobian determinant due to the symplecticity of $\Psi_\lambda$ and reversibility of $S,$ we obtain
\begin{align*}
\int_{\R^{2d}}a(x){\bf{1}}_A \bigl(S(x)\bigr)f^H(x) \, dx
=&\int_{\R^{2d}}a \bigl(S\circ\Psi_\lambda(x)\bigr){\bf{1}}_A \bigl(\Psi_\lambda(x)\bigr)f^H \bigl(\Psi_\lambda(x)\bigr) \, dx.
\end{align*}
Next, using the relation $\Psi_\lambda\circ S\circ\Psi_\lambda= S$ (time reversibility with respect to $S$) together with $H\circ S = H,$ we deduce that
\begin{align*}
a\bigl(S\circ\Psi_\lambda(x)\bigr)
=&\min \Bigl\{1,e^{H \left(S\circ\Psi_\lambda(x)\right)-H\left(\Psi_\lambda\circ S\circ\Psi_\lambda(x) \right)} \Bigr\}\\
=&\min \Bigl\{1,e^{H \left(S\circ\Psi_\lambda(x) \right)-H\left(S(x)\right)}  \Bigr\}\\
=&\min \Bigl\{1,e^{H\left(\Psi_\lambda(x)\right)-H(x)}  \Bigr\}\\
=&\frac{f^H(x)}{f^H \bigl(\Psi_\lambda(x)\bigr)}a(x).
\end{align*}
Substituting this identity back into the integral shows that
$$\int_{\R^{2d}}a(x){\bf{1}}_A \bigl(S(x)\bigr)f^H(x)\, dx=\int_{\R^{2d}}a(x){\bf{1}}_A \bigl(\Psi_\lambda(x)\bigr)f^H(x) \, dx,$$
as desired. \hfill $\square$
\end{proof}

\section{Discussion and Bibliography}\label{sec:biblioHMC}
Hamiltonian Monte Carlo (HMC) was introduced in the physics literature under the name Hybrid Monte Carlo \index{Hybrid Monte Carlo (HMC)}  \cite{duane1987hybrid}. It was later brought into the statistics community through \cite{neal2012bayesian}.
An accessible introduction is given in
\cite{neal2011mcmc}; see also \cite{betancourt2017conceptual}. Theorem \ref{thm:poincare}, which establishes that Hamiltonian flows are symplectic, 
goes back to Poincaré
\cite{poincare1899methodes}. 
Additional background on Hamiltonian dynamics and their numerical approximation can be found in \cite{sanz2018numerical,hairer2006geometric}, while numerical integrators tailored to use in HMC are reviewed and compared in \cite{bou2018geometric,calvo2019hmc}. 

In this chapter, we have seen that the Metropolis--Hastings framework can be generalized through the use of an involution map. In the context of HMC, this involution is typically taken to be the momentum flip.
The use of involutions within the Metropolis--Hastings algorithm was proposed in \cite{andrieu2020general,glatt2021mixing}, and has been leveraged to establish mixing rates for HMC in \cite{glatt2021mixing}. The analysis in \cite{glatt2021mixing}, which relies on the weak Harris approach developed in \cite{hairer2014spectral}, also extends to implementations of HMC in infinite-dimensional Hilbert spaces \cite{beskos2011hybrid}.  The paper \cite{bou2018coupling} relies instead on a novel coupling technique to analyze the convergence of HMC algorithms. 
Further convergence results can be found in \cite{livingstone2019geometric,durmus2017convergence}, while optimal scaling in high dimensions is studied in \cite{beskos2013optimal}. 

Tuning HMC algorithms can be challenging in practice. The No-U-Turn Sampler (NUTS) algorithm \index{No-U-Turn (NUTS)} \cite{hoffman2014no} addresses this issue by providing an adaptive mechanism to select both the integration time $\lambda$ and the step-size \index{step-size} $\epsilon.$

\chapter{Sequential Monte Carlo}
\label{chap:particlefilters}
This chapter is devoted to Monte Carlo algorithms that leverage weighted samples to approximate a sequence of target distributions. Sequential Monte Carlo \index{sequential Monte Carlo} algorithms combine three main ingredients: (i) a Markov kernel 
for propagating
a collection of particles; (ii) a weighting mechanism akin to importance sampling; and (iii) a resampling \index{resampling} scheme 
for mitigating
weight degeneracy. For a given problem, there is often 
significant 
flexibility in how
these components are specified and implemented, which contributes to making 
sequential Monte Carlo an extremely rich family of algorithms with deep theoretical foundations.

We focus
on Bayesian filtering, \index{filtering} where the targets 
represent 
the conditional 
laws 
of a time-evolving hidden state given noisy observations. In this context, we introduce two algorithms: the bootstrap particle filter and the optimal particle filter. 
These algorithms use
different Markov kernels to propagate particles
and, consequently, 
different weighting mechanisms.
We emphasize, however, 
that Bayesian filtering is only one natural setting in which to introduce the main ideas of sequential Monte Carlo. The methodology applies much more broadly.
In some applications, the sequence of targets is not specified 
directly
by the problem at hand,
but instead arises by artificially introducing 
a sequence of tempered targets, similar to the annealing strategies in Chapter \ref{chap:annealing}. For instance, in applications to rare event sampling, \index{rare event} the sequence of targets may be chosen to gradually concentrate around the event of interest.
 
As discussed in Chapter \ref{chap:MCintegration}, the main 
caveat 
of sampling algorithms that rely on weighting mechanisms is that the variance of the weights 
increases as the target and proposal distributions become further apart, which can cause weight degeneracy in high-dimensional problems. 
Indeed, sequential Monte Carlo methods are 
extremely powerful in low and moderate dimensions, but they often perform poorly in high dimensions. 
Resampling schemes are partly motivated by this degeneracy. They discard particles with negligible weights and replicate particles with large weights. However, 
resampling cannot fully 
resolve 
this fundamental limitation of algorithms that rely on importance weights. 

This chapter is structured as follows. Section \ref{sec:HMM} 
overviews 
the Bayesian formulation of the filtering problem and introduces the sequence of target distributions 
that
we seek to approximate. Section \ref{sec:bootstrapparticlefilter} introduces the bootstrap particle filter and establishes its convergence in the large-particle limit using a numerical analysis argument reminiscent of the Lax equivalence theorem.
Section \ref{sec:optimalparticlefilter} introduces 
the optimal particle filter, which partly alleviates
the weight degeneracy of 
the bootstrap particle filter.
Section \ref{sec:PFbibliography} closes with bibliographical remarks.

\section{Bayesian Filtering}\label{sec:HMM}
Hidden Markov models \index{hidden Markov model} arise in applications where a latent or ``hidden'' Markov process $\{X_j\}_{j=0}^\infty$ needs to be estimated based on an \emph{observed} process $\{Y_j\}_{j=1}^\infty.$ For concreteness, we assume that the hidden and observed processes are specified as follows:
\begin{alignat*}{4}
&\text{(Initial condition)} \qquad \qquad X_0 &&\sim f_0, \\
&\text{(Dynamics model)} \quad \qquad X_{j+1} &&= a(X_j)  + \xi_{j+1}, \qquad \qquad  &&&\xi_{j+1} \sim \Nc(0,\Sigma), \quad j = 0, 1, \ldots\\
&\text{(Observation model)} \qquad \, \, Y_{j+1} &&= b(X_{j+1}) +  \eta_{j+1}, \,\, &&&\eta_{j+1} \sim \Nc(0,\Gamma), \quad j = 0, 1, \ldots
\end{alignat*}
Here 
$X_0,$ $\{\xi_j\},$ and $\{\eta_j\}$ are 
all independent; see Figure \ref{fig:HMMdependence} for a representation of the resulting dependence structure of the hidden and observed processes. 
The initial distribution $f_0,$ the dynamics map $a: \R^d \to \R^d,$ the observation map $b: \R^d \to \R^k,$ and the positive definite covariances $\Sigma$ and $\Gamma$ are all assumed to be known.
We assume 
Gaussianity of $\xi_j$ and $\eta_j$ 
for pedagogical reasons, as it leads to concrete and familiar expressions for the Markov kernel and likelihood function that are implicitly defined, respectively, by the dynamics and observation models.

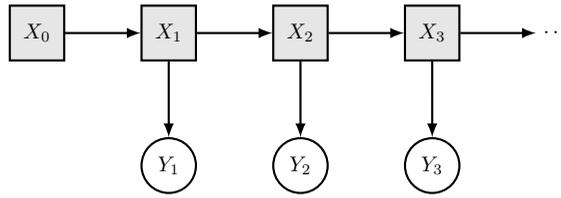
\begin{figure}\label{fig:HMMdependence}
\centering
\begin{tikzpicture}[scale = 0.8, every node/.style={scale=0.8}]
\tikzstyle{circ}=[circle, minimum size = 9mm, thick, draw =black!80, node distance = 10mm, black]
\tikzstyle{rect}=[rectangle, minimum size = 9mm, thick, draw =black!80, node distance = 10mm, black]
\tikzstyle{connect}=[-latex, thick, black]
\tikzstyle{box}=[rectangle, draw=black!100]
  \node[rect,fill=black!10] (X0) {$X_0$};
  \node[rect,fill=black!10] (X1) [right=of X0] {$X_1$};
  \node[rect,fill=black!10] (X2) [right=of X1] {$X_2$};  
  \node[rect,fill=black!10] (X3) [right=of X2] {$X_3$};
  \node[] (X4) [right=of X3] {$\nc \cdots$};
  \node[circ] (Y1) [below=of X1] {$Y_1$};
  \node[circ] (Y2) [right=of Y1,below=of X2] {$Y_2$};
  \node[circ] (Y3) [right=of Y2,below=of X3] {$Y_3$};
  \path (X0) edge [connect] (X1);
  \path (X1) edge [connect] (X2);
  \path (X2) edge [connect] (X3);
   \path (X3) edge [connect] (X4);
  \path (X1) edge [connect] (Y1);
  \path (X2) edge [connect] (Y2);
  \path (X3) edge [connect] (Y3);
\end{tikzpicture}
\caption{Graphical representation of the assumed dependence structure.}
\end{figure}

We are interested in estimating $X_j$ 
from
the observations $Y_{1:j}:= \{Y_i\}_{i=1}^j$ available at time $j.$ In the Bayesian framework, inference is carried out 
via 
the posterior 
distribution of $X_j$ given $Y_{1:j}$. This posterior is called the \emph{filtering distribution} and is denoted by $f_j$. We will study two sequential Monte Carlo algorithms 
 for approximating
 $f_j$: the bootstrap particle filter \index{particle filter!bootstrap} and the optimal particle filter.\index{particle filter!optimal} An important feature of these algorithms is that they are \emph{sequential}: at time $j+1,$ they use a particle approximation of $f_j$ together with the new observation $Y_{j+1}$ to construct a particle approximation of $f_{j+1}$.
Each algorithm relies on a different conceptual decomposition of the filtering step:

     \paragraph{Bootstrap Decomposition} In the bootstrap particle filter, particles are 
     first
     propagated through the dynamics model. Bayes' formula is then
     applied to assimilate the new observation. This corresponds to the following decomposition of the filtering step:
    $$X_j|Y_{1:j} \longrightarrow X_{j+1}|Y_{1:j} \longrightarrow X_{j+1}|Y_{1:j+1}.$$

    \paragraph{Optimal Decomposition} In the optimal particle filter, Bayes' formula is applied 
    before propagating through the dynamics. This corresponds to the  decomposition of the filtering step:
    $$X_j|Y_{1:j} \longrightarrow X_j|Y_{1:j+1}\longrightarrow X_{j+1}|Y_{1:j+1}.$$

 Throughout this chapter, we often consider distributions of the form
\begin{equation*}
    \pi(x) = \frac{1}{N} \sum^{N}_{n=1} \delta \bigl(x-X^{(n)}\bigr),
\end{equation*}
where $\{X^{(n)} \}_{n=1}^N$ is a sample and $\delta(x-X^{(n)})$ denotes a Dirac mass at $X^{(n)}$. These \emph{empirical measures} will be informally treated as p.d.f.s. For sufficiently smooth test functions $h,$ they satisfy
$$ \mathcal{I}_{\pi}[h] = \Expect_\pi [h] = \frac{1}{N} \sum_{n=1}^N h \bigl(X^{(n)} \bigr). $$

Given normalized weights $\{w^{(n)} \}_{n=1}^N,$ we likewise consider distributions of the form
$$
\pi(x) =  \sum^{N}_{n=1} w^{(n)} \delta \bigl(x-X^{(n)}\bigr).
$$
These \emph{weighted empirical measures} satisfy that, for sufficiently smooth test functions $h,$
\begin{equation*}
   \mathcal{I}_{\pi}[h] = \Expect_{\pi} [h] = \sum_{n=1}^N w^{(n)} h \bigl(X^{(n)} \bigr).  
\end{equation*}

Finally, for 
a positive definite
matrix $A$ and vector $v,$ we denote by  $| v |_A^2 : = v^T A^{-1} v$ the Mahalanobis norm.  This notation helps streamline the expression of multivariate Gaussian densities. 

\section{Bootstrap Particle Filter\index{particle filter!bootstrap}}\label{sec:bootstrapparticlefilter}
In the bootstrap particle filter, we propagate particles through the dynamics model, and then apply Bayes' formula to weight each particle according to the likelihood function implied by the observation model. At each time step, we resample new particles to avoid weight degeneracy. 

  \FloatBarrier
\begin{algorithm}[H]
  \caption{Bootstrap Particle Filter\label{algo:bootparticlefilteralgo}}
  \begin{algorithmic}[1]
  \STATEx{\textbf{Input:} Initial distribution $f_0^N = f_0$, observations $Y_{1:J}$, sample size $N.$}
  \STATEx For $j =0, \ldots, J-1$ do for $1 \le n \le N$:
  \STATE{\textbf{Resampling:} Draw $X_j^{(n)} \stackrel{\text{i.i.d.}}{\sim} f_j^N.$}
  \STATE{\textbf{Particle propagation:} Propagate each sample through the dynamics model: 
  $$\hat{X}_{j+1}^{(n)} = a \bigl(X_j^{(n)} \bigr) + \xi_{j+1}^{(n)}, \quad \quad \xi_{j+1}^{(n)} \stackrel{\text{i.i.d.}}{\sim} \Nc(0, \Sigma).$$ }
  \STATE{\textbf{Particle weighting:} Weight each sample with the observation likelihood model: 
  $$ \tilde{w}_{j+1}^{(n)} = \exp{ \Big( -\frac{1}{2} \bigl| {Y_{j+1} - b \bigl(\hat{X}_{j+1}^{(n)} \bigr) \bigr|_{\Gamma}^2} \Big) }.$$ }
  \STATEx{Normalize the weights: $$ {w_{j+1}^{(n)}} = \frac{\tilde{w}_{j+1}^{(n)}}{\sum_{n=1}^N\tilde{w}_{j+1}^{(n)}} .$$ }
  \STATE{ Set $$ f_{j+1}^N (x) = \sum_{n=1}^N w_{j+1}^{(n)} \delta \bigl(x - \hat{X}_{j+1}^{(n)}\bigr).$$ }
  \STATEx{\textbf{Output:} Approximation $f_J^N$ of the filtering distribution $f_J$.}
 \end{algorithmic}
\end{algorithm}
  \FloatBarrier
  
  \begin{mybox}[colback=white]{Pros and Cons}
The bootstrap particle filter is a sequential Monte Carlo algorithm 
for approximating
a time-evolving sequence of target distributions. It can be applied in general hidden Markov models, and it is provably convergent in the large $N$ limit. However, in high-dimensional settings the weights typically have 
large variance,
leading to a small effective sample size and poor performance. 
\end{mybox}

Our goal now is to show that, for large $N,$  $f^N_J$ is close to $f_J$. To perform the analysis we first introduce three operators acting on p.d.f.s.  Recall that throughout this chapter we use the term p.d.f. broadly to also include probability measures defined by (possibly weighted) sums of Diracs.

\subsection{Prediction, Weighting, and Sampling Operators}\label{ssec:algorithm}
The first operator we introduce, which we call the prediction operator, describes how p.d.f.s propagate through the Markov kernel specified by the dynamics model. We refer to Definition \ref{def:nsteptransition} and Lemma \ref{lem:nsteptransition} for related background. 
\begin{definition}[Prediction Operator\index{prediction operator}]
    The \emph{prediction} operator $\mathcal{P}$ acting on a p.d.f. $\pi$ is given by 
    $$ (\mathcal{P}\pi)(x) = \int_{\mathbb{R}^d} p(\tilde{x}, x) \pi(\tilde{x}) \, d\tilde{x}, $$
    where $p$ is the Markov kernel associated with the dynamics model, namely
    \begin{equation}\label{eq:kernel}
 p(x, \tilde{x}) = \frac{1}{\sqrt{(2\pi)^d \det{\Sigma}}} \exp{ \Big( -\frac{1}{2} \bigl| \tilde{x} - a(x) \bigr|_{\Sigma}^2 \Big) }.    
    \end{equation}
\end{definition}

The second operator we introduce, which we call the weighting operator, takes as input a p.d.f. $\pi$ and produces as output the posterior p.d.f. obtained by taking $\pi$ as the prior and using the likelihood function specified by the observation model.
\begin{definition}[Weighting Operator\index{weighting operator}]
    The \emph{weighting} operator $\mathcal{W}_j$ acting on a p.d.f. $\pi$ is given by
    $$ (\mathcal{W}_j\pi) (x) \propto \exp{ \Big( -\frac{1}{2} | Y_{j+1} - b(x) |^2_{\Gamma} \Big)} \pi(x). $$
\end{definition}
Denoting by $\hat{f}_{j+1} := \mathcal{P}f_j$ the distribution of $X_{j+1} | Y_{1:j},$ we thus have that $f_{j+1} = \mathcal{W}_j \hat{f}_{j+1}$. Consequently, we can write the filtering step as 
\begin{equation}\label{eq:truefiltersuccinct}
    f_{j+1} = \mathcal{W}_j \mathcal{P} f_j.
\end{equation}
    The operator $\mathcal{P}$ does not depend on $j$ because our dynamics model is time-homogeneous. However, the operator $\mathcal{W}_j$ depends on the observed data $Y_{j+1},$ and thus on the time index $j.$
\begin{remark}
	The operator $\mathcal{W}_j$ acting on a p.d.f. of the form
    $$ \pi(x) = \frac{1}{N} \sum_{n=1}^N \delta \bigl(x-X^{(n)}\bigr) $$
	gives
	$$ (\mathcal{W}_j\pi)(x) = \sum_{n=1}^N {w^{(n)}} \delta \bigl(x - X^{(n)}\bigr), $$
	where 
	$$ \tilde{w}^{(n)} = \exp \Bigl( -\frac{1}{2} \bigl| Y_{j+1} - b(X^{(n)}) \bigr|^2_{\Gamma} \Bigr) $$
	and $ {w^{(n)}}$ are the normalized weights.
\end{remark}

The third operator we introduce, which we call the sampling operator, takes as input a p.d.f. $\pi$ and outputs the empirical distribution
obtained by drawing $N$ samples from $\pi.$
\begin{definition}[Sampling Operator\index{sampling operator}]
    The \emph{sampling} operator $\mathcal{S}^N$ acts on a p.d.f. $\pi$ by drawing $N$ samples from $\pi$ and producing a Dirac approximation:
    $$ (\mathcal{S}^N \pi)(x) = \frac{1}{N} \sum_{n=1}^{N}{\delta \bigl(x-X^{(n)}\bigr) }, \hspace{3mm} X^{(n)}
    \stackrel{\emph{i.i.d.}}{\sim} \pi.$$
\end{definition}
    Notice that $\mathcal{S}^N \pi$ is a random p.d.f. 
    due to sampling from $\pi$. With the three operators we just defined, we can succinctly write the bootstrap particle filter as 
    \begin{equation}\label{eq:bootstrapsuccinct}
        f_{j+1}^{N} = \mathcal{W}_j\mathcal{S}^N\mathcal{P} f_j^N .
    \end{equation}

In the bootstrap particle filter Algorithm \ref{algo:bootparticlefilteralgo}, the resampling, particle propagation, and particle weighting steps can be viewed as the result of applying the sampling, prediction, and weighting operators.     
To see the agreement between the recursion \eqref{eq:bootstrapsuccinct} and the recursion provided by Algorithm \ref{algo:bootparticlefilteralgo},   
notice that performing prediction and then sampling is equivalent to sampling first and then performing prediction.

\subsection{Bounds for Prediction, Weighting, and Sampling}\label{ssec:bounds}
We now seek to obtain bounds for the prediction, weighting, and sampling operators. 
We will use the following distance between random probability p.d.f.s: \index{distance between random probability distributions}
\begin{equation}\label{eq:distancerandompdf}
    \mathsf{d}(\pi, \pi') = \sup_{|h|_\infty \le 1}{\sqrt{\mathbb{E}\bigl[(\mathbb{E}_{\pi}[h] - \mathbb{E}_{\pi'}[h])^2\bigr]}}.
\end{equation}
It can be verified that \eqref{eq:distancerandompdf} indeed defines a valid distance in the space of random probability p.d.f.s, so that in particular it satisfies the triangle inequality.
In what follows, the randomness in the distributions we consider will arise from sampling, and the outer expectation is with respect to such randomness.

We first show that the prediction operator is a contraction in this metric.
\begin{lemma}\label{lemma:prediction} 
It holds that $$ \mathsf{d}(\mathcal{P}\pi, \mathcal{P}\pi') \leq \mathsf{d}(\pi, \pi').$$ 
\end{lemma}
\begin{proof}
Let $|h|_\infty \le 1$ and define $$h^{\dagger}(x) = \int_{\mathbb{R}^d} p(x,\tilde{x})h(\tilde{x}) \, d\tilde{x},$$ where recall that $p$ denotes the transition kernel \eqref{eq:kernel} of the dynamics model. Note that
 $$|h^{\dagger}(x)|\leq \int_{\mathbb{R}^d} p(x,\tilde{x}) \, d\tilde{x}=1$$ 
 and 
\begin{align*}
    \mathbb{E}_{\pi}[h^{\dagger}] =& \int_{\mathbb{R}^d} h^{\dagger}(x)\pi(x) \, dx\\
    =& \int_{\mathbb{R}^d} \Big( \int_{\mathbb{R}^d}p(x,\tilde{x})h(\tilde{x}) \, d\tilde{x} \Big) \pi(x) \, dx\\
    =& \int_{\mathbb{R}^d} \Big(\int_{\mathbb{R}^d} p(x,\tilde{x})\pi(x) \, dx \Big) h(\tilde{x}) \, d\tilde{x}\\
    =& \int_{\mathbb{R}^d} (\mathcal{P}\pi)(\tilde{x})h(\tilde{x}) \, d\tilde{x}, 
\end{align*}
and so
$$ \mathbb{E}_{\pi}[h^{\dagger}]=\mathbb{E}_{\mathcal{P}\pi}[h].$$
Finally 
\begin{align*}
    \mathsf{d}(\mathcal{P}\pi, \mathcal{P}\pi') &= \sup_{|h|_\infty \le 1}{\sqrt{\mathbb{E}\bigl[(\mathbb{E}_{\mathcal{P}\pi}[h] - \mathbb{E}_{\mathcal{P}\pi'}[h])^2\bigr]}}\\
    &\leq \sup_{|h^{\dagger}| \leq 1}{\sqrt{\mathbb{E}\bigl[(\mathbb{E}_{\pi}[h^{\dagger}] - \mathbb{E}_{\pi'}[h^{\dagger}])^2\bigr]}}\\
    &= \mathsf{d}(\pi,\pi'),
\end{align*}
as desired. \hfill $\square$
\end{proof}

We next establish a bound for the weighting operator under the following assumption:
\begin{assumption} \label{assumption}
There exists $\kappa \in (0,1)$ such that, for all $x \in \mathbb{R}^d$ and $j \in \{ 0,1,\ldots,J-1\},$ 
 $$ \kappa \leq w_{j+1}(x) := \exp\Bigl(-\frac{1}{2} |Y_{j+1}- b{(x)}|_\Gamma^2\Bigr) \leq \kappa^{-1}.$$
 \end{assumption}
While this strong assumption can be significantly relaxed, it will allow us to streamline the theory. In particular, notice that the proof of the following result is nearly identical to that of Theorem \ref{thm:autonormalized} for autonormalized importance sampling. \index{importance sampling!autonormalized} 
\begin{lemma}\label{lemma:analysis}
Under Assumption \ref{assumption}, 
$$ \mathsf{d}( \mathcal{W}_j \pi, \mathcal{W}_j\pi') \leq \frac{2}{\kappa^2}\mathsf{d}(\pi,\pi').$$
\end{lemma}
\begin{proof}
Let $h$ with $|h|_\infty \le 1$, then 
\begin{align*}
    \mathbb{E}_{\mathcal{W}\pi}[h]-\mathbb{E}_{\mathcal{W}\pi'}[h] =& \frac{\mathbb{E}_{\pi}[hw]}{\mathbb{E}_{\pi}[w]}-\frac{\mathbb{E}_{\pi'}[hw]}{\mathbb{E}_{\pi'}[w]}\\
    =& \frac{\mathbb{E}_{\pi}[hw]}{\mathbb{E}_{\pi}[w]}-\frac{\mathbb{E}_{\pi'}[hw]}{\mathbb{E}_{\pi}[w]}+\frac{\mathbb{E}_{\pi'}[hw]}{\mathbb{E}_{\pi}[w]}-\frac{\mathbb{E}_{\pi'}[hw]}{\mathbb{E}_{\pi'}[w]}\\
    =& \frac{1}{\kappa} \biggl( \frac{\mathbb{E}_{\pi}[h\kappa w]-\mathbb{E}_{\pi'}[h\kappa w]}{\mathbb{E}_{\pi}[w]} +
    \frac{\mathbb{E}_{\pi'}[hw]}{\mathbb{E}_{\pi'}[w]}
    \frac{\mathbb{E}_{\pi'}[\kappa w]-\mathbb{E}_{\pi}[\kappa w]}{\mathbb{E}_{\pi}[w]} \biggr).
\end{align*}
Since $|h|_\infty \le 1,$ we have
$$\left|\mathbb{E}_{\mathcal{W}\pi'} [h] \right|=\left|\frac{\mathbb{E}_{\pi'}[hw]}{\mathbb{E}_{\pi'}[w]} \right| \leq 1.$$
In addition, by Assumption \ref{assumption}, we have $\mathbb{E}_\pi[w] \ge \kappa.$
Therefore, 
\begin{align*}
    \left| \mathbb{E}_{\mathcal{W}\pi}[h]-\mathbb{E}_{\mathcal{W}\pi'}[h] \right| 
    \leq \frac{1}{\kappa^2} 
    \biggl(
    \left| \mathbb{E}_{\pi}[h\kappa w]-\mathbb{E}_{\pi'}[h\kappa w] \right| + 
    \left| \mathbb{E}_{\pi'}[\kappa w]-\mathbb{E}_{\pi}[\kappa w] \right| 
   \biggr),
\end{align*}
and so
\begin{align*}
\mathbb{E}\Big[ \bigl(\mathbb{E}_{\mathcal{W}\pi}[h]-\mathbb{E}_{\mathcal{W}\pi'}[h] \bigr)^2 \Big] & \leq    \frac{2}{\kappa^4} \Biggl( \mathbb{E}\Bigl[\bigl(\mathbb{E}_{\pi}[h\kappa w]-\mathbb{E}_{\pi'}[h\kappa w] \bigr)^2 \Bigr] 
+ \mathbb{E} \Bigl[ \bigl( \mathbb{E}_{\pi'}[\kappa w]-\mathbb{E}_{\pi}[\kappa w]  \bigr)^2 \Bigr] \Biggr) \\
& \leq \frac{4}{\kappa^4} \mathsf{d}(\pi,\pi')^2,
\end{align*}
where in the last line we have used that $\left| \kappa w \right | \leq 1.$ Taking square roots and the supremum over $|h|_\infty \le 1$ completes the proof. \hfill $\square$
\end{proof}

Lastly, the sampling operator can be bounded using Theorem \ref{thm2.1} for classical Monte Carlo integration. 
\begin{lemma}\label{lemma:sampling}
For any distribution $\pi$, $$ \mathsf{d}(\pi,\mathcal{S}^N\pi) \leq \frac{1}{\sqrt{N}}.$$
\end{lemma}
\begin{proof}
For any $h$ with $|h|_\infty \le 1$, $$ \mathbb{E}_{\mathcal{S}^{{\tiny N}} \pi}[h]=\frac{1}{N}\sum_{n=1}^N h \bigl(X^{(n)}\bigr),$$ where $ X^{(n)} \stackrel{\text{i.i.d.}}{\sim}  \pi$. Therefore, using Theorem \ref{thm2.1} for classical Monte Carlo integration,
\begin{align*}
    \mathbb{E}\Big[ \bigl(\mathbb{E}_{\mathcal{S}^N \pi}[h]-\mathbb{E}_{\pi}[h]\bigr)^2 \Big] &= \mathbb{E}\Big[ \bigl(\mathcal{I}_{\pi}^{\text{\tiny MC}}[h]-\mathcal{I}_{\pi}[h]\bigr)^2 \Big]
    = \frac{\mathbb{V}_{\pi}[h]}{N} \leq  \frac{1}{N}.
\end{align*}
Taking square roots and the supremum over $|h|_\infty \le 1$ gives the result. \hfill $\square$
\end{proof}

\subsection{Error of Bootstrap Particle Filter\index{bootstrap particle filter!error bounds}}
Recall from equations \eqref{eq:truefiltersuccinct} and \eqref{eq:bootstrapsuccinct} that the filtering step and the bootstrap particle filter can be succinctly written as
\begin{align}\label{eq:repetition}
    f_{j+1} &= \mathcal{W}_j \mathcal{P} f_j, \\
    f_{j+1}^{N} &= \mathcal{W}_j\mathcal{S}^N\mathcal{P} f_j^N.
\end{align}
With the bounds we derived for prediction, weighting, and sampling operators in Lemmas \ref{lemma:prediction}, \ref{lemma:analysis}, and \ref{lemma:sampling}, the following theorem controls the error of the bootstrap particle filter. The proof resembles the classical Lax equivalence framework in numerical analysis, where convergence is established by ensuring consistency and stability. 
\begin{theorem}[Error of Bootstrap Particle Filter]
Under Assumption \ref{assumption}, there exists a constant $c=c(J,\kappa)$ such that, for $j=1,\ldots,J$, 
$$\mathsf{d}(f_{j},f_{j}^{N}) \leq \frac{c}{\sqrt{N}}.$$
\end{theorem}
\begin{proof}
Let $e_{j}=\mathsf{d}(f_{j},f_{j}^N)$, then using \eqref{eq:repetition} and the triangle inequality for the distance $\mathsf{d}$ defined in \eqref{eq:distancerandompdf}, we have
\begin{align*}
    e_{j+1} =  \mathsf{d}(f_{j+1},f_{j+1}^N) &=  \mathsf{d}(\mathcal{W}_j \mathcal{P}f_j,\mathcal{W}_j\mathcal{S}^N\mathcal{P}f_j^N)\\
    &\leq  \mathsf{d}(\mathcal{W}_j \mathcal{P}f_j,\mathcal{W}_j \mathcal{P}f_j^N)+\mathsf{d}(\mathcal{W}_j \mathcal{P}f_j^N,\mathcal{W}_j\mathcal{S}^N\mathcal{P}f_j^N).
\end{align*}
The first term represents the growth of error over one time-iteration by propagating through the operator $\mathcal{W}_j \mathcal{P}$ which governs the true evolution of the filtering distribution, whereas the second term is a truncation error, i.e. the error incurred by the algorithm in one iteration. 
We then have
\begin{align*}
    e_{j+1} \leq & \frac{2}{\kappa^2} \Bigl( \mathsf{d}(\mathcal{P}f_j^N,\mathcal{P}f_j) +\mathsf{d}(\mathcal{P}f_j^N,\mathcal{S}^N\mathcal{P}f_j^N) \Bigr)\\
    \leq &\frac{2}{\kappa^2} \Big(e_j+\frac{1}{\sqrt{N}}\Big),
\end{align*}
where in the first inequality we used Lemma \ref{lemma:analysis} for the weighting operator and in the second inequality we used Lemmas \ref{lemma:prediction} and \ref{lemma:sampling} for the prediction and sampling operators.
Letting $\lambda =\frac{2}{\kappa^2}$, induction gives
$$e_j \leq \lambda^j e_0 +\frac{\lambda}{\sqrt{N}}\frac{1-\lambda^j}{1-\lambda},$$ 
where $e_0=0$ since $f_0^N=f_0$.  Now, noting that the above expression is monotonically increasing in $j$ completes the proof by setting $c=\frac{\lambda(1-\lambda^J)}{1-\lambda}.$ \hfill $\square$
\end{proof}

\section{Optimal Particle Filter\index{particle filter!optimal}}\label{sec:optimalparticlefilter}
The bootstrap particle filter propagates particles through the dynamics and then assigns them weights according to the likelihood defined by the observation model. In contrast, 
the optimal particle filter stems from a different conceptual decomposition of the filtering step, where we first apply Bayes’ formula and then propagate particles through a Markov kernel which depends on the new observation.
\subsection{Derivation of the Optimal Filter Decomposition}
  The following calculation illustrates the decomposition of the filtering step that the optimal particle filter approximates via particles. We let $\Prob(\cdot |\cdot)$ denote conditional density and blur the distinction between random variables and arguments of their p.d.f.s. With these conventions, we have: 
\begin{align*}
    \Prob(X_{j+1}|Y_{1:j+1}) & =\int_{\mathbb{R}^d}\Prob(X_j,X_{j+1}|Y_{1:j+1})\, dX_{j}\\
    &=\int_{\mathbb{R}^d}\Prob(X_{j+1}|X_j,Y_{1:j+1})\Prob(X_{j}|Y_{1:j+1}) \, dX_{j}\\
    &=\int_{\mathbb{R}^d}\Prob(X_{j+1}|X_j,Y_{j+1})\Prob(X_{j}|Y_{1:j+1}) \, dX_{j}\\
    &=\int_{\mathbb{R}^d}\Prob(X_{j+1}|X_j,Y_{j+1})\frac{\Prob(Y_{j+1}|X_{j},Y_{1:j})}{\Prob(Y_{j+1}|Y_{1:j})}\Prob(X_j|Y_{1:j}) \, dX_{j}\\
    & = \int_{\mathbb{R}^d}\Prob(X_{j+1}|X_j,Y_{j+1})\frac{\Prob(Y_{j+1}|X_{j})}{\Prob(Y_{j+1}|Y_{1:j})}\Prob(X_j|Y_{1:j}) \, dX_{j}\\
    & = \mathcal{P}_j^{\text{\tiny OPF}}\mathcal{W}_j^{\text{\tiny OPF}}\Prob(X_j|Y_{1:j}),
\end{align*}
where in the last line we have defined 
\begin{align*}
\mathcal{P}_j^{\text{\tiny OPF}} \pi(x_{j+1}) &=\int_{\mathbb{R}^d}\Prob(x_{j+1}|X_j,Y_{j+1})\pi(X_j) \, dX_j,\\
 \mathcal{W}_j^{\text{\tiny OPF}}\pi(x_j) &\propto \Prob(Y_{j+1}|x_j) \pi(x_j).
\end{align*} 
Since $\Prob(X_{j}|Y_{1:j})$ is the filtering distribution $f_j,$ we 
have derived the following decomposition of the filtering step:
$$ f_{j+1}=\mathcal{P}_j^{\text{\tiny OPF}}\mathcal{W}_j^{\text{\tiny OPF}}f_j.$$ 
Notice that, in contrast to the prediction operator $\mathcal{P}$ used by the bootstrap particle filter, the prediction-type operator $\mathcal{P}_j^{\text{\tiny OPF}}$ utilizes the observation $Y_{j+1}.$ Intuitively, propagating particles with a Markov kernel that incorporates the observations can facilitate the subsequent weighting of the particles. This is the idea behind the optimal particle filter.

\subsection{Implementation}
In general, it is not possible to implement the optimal particle filter in the fully nonlinear setting due to two computational bottlenecks:
\begin{itemize}
\item 
There may not be a 
closed-form expression 
for evaluating the likelihood $\Prob(Y_{j+1}|X_{j}),$ making  the computation of the particle weights infeasible. 
\item It may not be possible to sample from the Markov kernel $\Prob(X_{j+1}|X_j,Y_{j+1}),$ making the propagation of particles infeasible.
\end{itemize}
However, when the observation function $b(\cdot)$ is linear, i.e. $b(\cdot) = H\cdot$ for some matrix $H,$ both bottlenecks may be overcome.
We thus consider the following setting, which arises in many applications:
\begin{alignat*}{3}
X_{j+1} &= a(X_j) + \xi_{j+1}, \quad \quad \quad  &&\xi_{j+1} \sim  \Nc(0, \Sigma),  \\
Y_{j+1} &= HX_{j+1} + \eta_{j+1}, \quad \quad \quad &&\eta_{j+1} \sim \mathcal{N}(0, \Gamma),
\end{alignat*}
where $X_0,$ $\{\xi_j\},$ and $\{\eta_j\}$ are all assumed to be independent. 
First, note that 
\begin{align*}
Y_{j+1} =  H a(X_j)+H\xi_{j+1}+\eta_{j+1},
\end{align*} 
   which implies that
    $\Prob(Y_{j+1}|X_j) =  \mathcal{N} \bigl(H a(X_j),H\Sigma H^{T}+\Gamma \bigr).$
Second, note that
 $$ \begin{aligned} \Prob(X_{j+1}|X_j,Y_{j+1}) & \propto  \exp\biggl(-\frac{1}{2} \bigl|Y_{j+1}-HX_{j+1} \bigr|^2_{\Gamma}-\frac{1}{2} \bigl|X_{j+1}- a(X_j) \bigr|^2_{\Sigma}\biggr),
    \end{aligned}$$
  which implies via completion of the square that $\Prob(X_{j+1}|X_j,Y_{j+1}) = \Nc(m_{j+1},C),$ with 
\begin{align*}
	m_{j+1} & = (I-KH) a(X_j) + KY_{j+1},\\
	C & = (I-KH)\Sigma,\\
	K & = \Sigma H^T S^{-1},\\
	S & = H \Sigma H^T + \Gamma.
\end{align*}
These formulas are similar to those arising in the Kalman filter,\index{Kalman filter} whose derivation is explored in Exercise~9.2 of the published version of this book. We are now ready to present a practical implementation of the optimal particle filter with linear observations.

\begin{algorithm}[H]
  \caption{Optimal Particle Filter\label{algo:optimalparticlefilteralgo}}
  \begin{algorithmic}[1]
  \STATEx{\textbf{Input:} Initial distribution $f_0$, observations $Y_{1:J},$ sample size $N.$}
  \STATEx{\textbf{Initial sampling:} Sample $ X_{0}^{(1)}, \ldots, X_{0}^{(N)}  \stackrel{\text{i.i.d.}}{\sim}  f_0$ so that $f_0^N= \mathcal{S}^N f_0.$}
  \STATEx{\textbf{Subsequent sampling:} 
  \STATEx For $j =0, \ldots, J-1$, do for $1 \le n \le N$:}
  \STATE{\textbf{Particle propagation:} Set $$ \hat{X}_{j+1}^{(n)}= (I-KH) a \bigl(X_j^{(n)}\bigr)+KY_{j+1}+\zeta_{j+1}^{(n)}, \quad \quad \zeta_{j+1}^{(n)} \stackrel{\text{i.i.d.}}{\sim}  \Nc(0,C).$$ }
  \STATE{\textbf{Particle weighting:} Set 
  $$\tilde{w}_{j+1}^{(n)} = \exp\Bigl(-\frac{1}{2} \bigl|Y_{j+1}-H a \bigl(X_{j}^{(n)}\bigr) \bigr|^2_{S}\Bigr).$$}
  \STATEx{ Normalize the weights: $$ {w_{j+1}^{(n)}} = \frac{\tilde{w}_{j+1}^{(n)}}{\sum_{\ell=1}^N\tilde{w}_{j+1}^{(\ell)}} .$$ }
  \STATE{\textbf{Resampling:} Draw $X_{j+1}^{(n)} \stackrel{\text{i.i.d.}}{\sim} \sum_{\ell=1}^N w_{j+1}^{(\ell)} \delta \bigl(x - \hat{X}_{j+1}^{(\ell)} \bigr).$}
  \STATE{Set $$ f_{j+1}^N (x)=\frac{1}{N}\sum_{n=1}^N \delta \bigl(x - X_{j+1}^{(n)}\bigr).$$ }
  \vspace{2mm}
  \STATEx{\textbf{Output:} Approximation $f_J^N$ of the filtering distribution $f_J$.}
 \end{algorithmic}
\end{algorithm}

  \begin{mybox}[colback=white]{Pros and Cons}
Similar to the bootstrap particle filter, the optimal particle filter \index{particle filter!optimal} is convergent in the large $N$ limit, but also suffers from a small effective sample size in high-dimensional settings. 
The variance of the weights is smaller, however, for the optimal particle filter than for the bootstrap particle filter. It is important to notice that, in contrast to the bootstrap filter, the optimal filter relies on Gaussian assumptions on the noise in the dynamics and observation models, and on linearity of the observations. Thus, it can only be deployed in a smaller class of problems. 
\end{mybox}

\section{Discussion and Bibliography}\label{sec:PFbibliography}
Sequential Monte Carlo methods \index{sequential Monte Carlo} are overviewed from an algorithmic viewpoint
in \cite{doucet2001introduction,doucet2000sequential} 
and from a mathematical perspective in \cite{del2004feynman,chopin2020introduction}. The term ``bootstrap'' particle filter was coined in the seminal work \cite{gordon1993novel} in reference to the bootstrap \cite{efron1979bootstrap}, a standard statistical procedure that also involves sampling with replacement. The optimal particle filter is discussed, and further references given,
in \cite{doucet2000sequential}; see section IID. While the term ``optimal'' particle filter is common in the literature,  the sense in which this filter is optimal is indeed rather weak. We refer to \cite[Chapter 10]{chopin2020introduction}, which utilizes instead the term ``guided'' particle filter, for further discussion on this topic; in particular, \cite[Theorem 10.1]{chopin2020introduction} formalizes the criterion of \emph{local optimality} and the ensuing discussion provides compelling criticisms of this notion of optimality. 

The proof of convergence of the bootstrap particle filter presented here
originates in \cite{rebeschini2015can} and can also be found in \cite{sanz2023inverse}. We refer to \cite[Chapter 12]{sanz2023inverse} for a similar proof of convergence for a Gaussianized optimal particle filter, and to \cite{johansen2008note} for a convergence analysis of the optimal particle filter. 
Throughout much of this chapter, we considered for simplicity the case of Gaussian additive noise and linear observation operator.
More general convergence and stability results for sequential Monte Carlo methods can be found, for instance, in \cite{del2004feynman} and \cite[Chapter 11]{chopin2020introduction}.

For problems in which the dynamics model is nonlinear and the state space of the hidden process is low dimensional, particle filters are enormously successful.
Extending them to
the high-dimensional problems that arise,
for example, in geophysical applications, provides a major challenge, since in those settings the particle weights typically concentrate on one, or a small number, of particles ---see for instance 
\cite{bickel2008sharp, snyder2008obstacles, Snyder2011,agapiou2017importance}. 
Weight collapse is intimately related to the fact that, as discussed in Chapter \ref{chap:MCintegration}, importance sampling weights have large variance when the target and the proposal distributions are far apart. 
As noted in  \cite{snyderberngtsson,agapiou2017importance}, an advantage of the optimal particle filter over the bootstrap particle filter is that the optimal filter utilizes a proposal that is closer to the target, which helps alleviate weight collapse over each iteration of the filter. Resampling techniques to turn weighted samples into unweighted ones are essential to the success of particle filters over multiple iterations. We refer to \cite[Chapter 9]{chopin2020introduction} for an in-depth exploration and comparison of resampling schemes in sequential Monte Carlo.

An important algorithmic idea in sequential Monte Carlo, as in many other algorithms studied in these notes, is 
the introduction of
auxiliary variables to accelerate computation \cite{pitt1999filtering,johansen2008note}.
In addition, exploiting decay of correlations through \emph{localization} is often needed when implementing particle filters for online estimation of discretizations of spatial fields. 
A review of local particle filters can be found in \cite{farchi2018comparison}. The paper  \cite{rebeschini2015can} investigates, from a theoretical viewpoint, whether localization can help to beat the curse of dimension. Attempts to bridge particle filters with ensemble Kalman filters to alleviate the curse of dimension include \cite{FK13,stordal2011bridging}, and the relation between the collapse of ensemble and particle methods is investigated in the paper \cite{morzfeld2017collapse}, which also emphasizes the importance of localization. Further theoretical investigations of localization for ensemble Kalman methods include \cite{al2024non,al2023covariance}. We close by pointing out that sequential Monte Carlo methods can be combined in several ways with algorithms studied in previous chapters, and we refer to \cite{andrieu2010particle} for a representative, important example.

\chapter{Variational Inference  and  Expectation Maximization}
\label{chap:optimization}

This chapter focuses on two important algorithmic frameworks: Variational Inference (VI) \index{Variational Inference (VI)} and Expectation Maximization (EM). \index{Expectation Maximization (EM)} In contrast to the methods studied in previous chapters, VI and EM are deterministic algorithms that do not rely on sampling. The goal of VI is to find a tractable \emph{variational distribution} that is close to a given intractable target distribution, so that expectations and other quantities of interest under the target distribution can be approximated using the variational distribution. The goal of EM is to optimize objective functions that can be simplified  by introducing auxiliary variables. This chapter emphasizes the unifying ideas underlying VI and EM, as well as how these techniques can be combined with, or used as alternatives to, Monte Carlo methods. 

VI is gaining popularity as an alternative to MCMC, and it is useful to highlight the conceptual differences between the two approaches. MCMC \index{Markov Chain Monte Carlo (MCMC)} algorithms approximate the target distribution by generating samples from a Markov chain whose limit distribution is the desired target. In practice, however, it can be challenging to determine how long the chain must be run before it reaches stationarity, and this can depend heavily on the choice of proposal kernel. MCMC methods also tend to be computationally intensive. VI, on the other hand, is often much faster, but it does not enjoy the same convergence guarantees: it seeks only a density from a prescribed family that is ``close enough'' to the target. If the target cannot be well approximated within this family (which is often the case),  then VI will not be accurate. Thus, VI is particularly well suited, for instance, to Bayesian inference with large data-sets or to problems where one wants to explore many models quickly; MCMC is better suited to smaller data-sets and scenarios where one is willing to pay a higher computational cost for greater accuracy.

The main idea behind EM is to introduce auxiliary variables to simplify the optimization of an intractable objective function. Each EM iteration consists of two steps: the E-step and the M-step. In the E-step, we compute an expectation in order to construct a lower bound on the objective function. In the M-step, we maximize this lower bound to obtain a new optimization iterate. 
If the expectation in the E-step is intractable, Monte Carlo methods can be used to approximate it.  
EM is widely used, for instance, to compute maximum likelihood estimators in applications where the likelihood function can be simplified by introducing latent variables. The role of EM in maximum likelihood estimation is analogous to the role of Gibbs samplers in posterior sampling; indeed, EM can be used to initialize Gibbs samplers for Bayesian inference.\index{Gibbs sampler}

This chapter is organized as follows. Section \ref{sec:VI} introduces VI, emphasizing an efficient coordinate-wise implementation akin to that of the Gibbs sampler. 
Section \ref{sec:EMalgorithm} introduces the EM algorithm, deriving it using the same ideas that underpin VI. Section \ref{sec:discussion} closes with bibliographical remarks.

\section{Variational Inference \index{Variational Inference (VI)}}\label{sec:VI}
The goal of Variational Inference (VI) is to find a tractable  distribution $g^*$ that is close to a given target distribution $f$, so that quantities of interest under the target can be computed efficiently using the variational distribution $g^*.$
For instance, once $g^*$ has been found,  expected values with respect to the target can be approximated as follows: 
$$ \mathcal{I}_f[h] = \int h(x) f(x)\, dx \approx \int h(x) g^*(x)\, dx = \mathcal{I}_{g^*}[h],$$
which is useful when computing $\mathcal{I}_{g^*}[h]$ is cheaper than computing  $\mathcal{I}_f[h].$

The variational distribution $g^*$ is defined as the (numerical) solution of an optimization problem. More precisely, one specifies a family $\mathscr{D}$ of tractable distributions and sets
\begin{equation}\label{eq:optimization}
g^* = \arg \, \min_{g \in \mathscr{D}} \, \dkl(g \| f) \, ,
\end{equation}
where $\dkl$ denotes the Kullback-Leibler (KL) divergence\index{divergence!Kullback-Leibler} between p.d.f.s. Thus, $g^*$  is the p.d.f. in $\mathscr{D}$
that is closest to the target in KL-divergence, as illustrated in Figure \ref{fig:variational_distribution}. The family $\mathscr{D}$ of admissible distributions must be chosen carefully, since it determines both the accuracy and the computational cost of the methodology. We introduce the KL-divergence in Subsection \ref{sec:KLandELBO}, and in Subsection \ref{sec:meanfield} we discuss the choice of the family $\mathscr{D}$ along with a Coordinate Ascent Variational Inference (CAVI) algorithm for solving the optimization problem \eqref{eq:optimization}. 

\begin{figure}[h]
    \centering
    \includegraphics[width=0.425\textwidth]{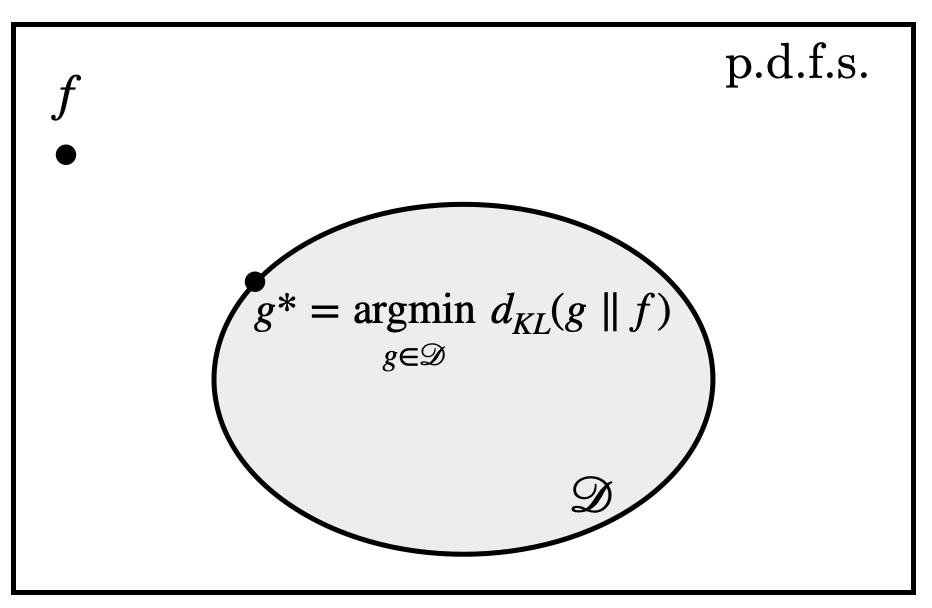}
    \caption{Illustration of the variational distribution $g^*$.}
    \label{fig:variational_distribution}
\end{figure}

\subsection{Kullback-Leibler Divergence and Evidence Lower-Bound}\label{sec:KLandELBO}
In this section, we motivate the choice of the KL-divergence \index{divergence!Kullback-Leibler} to quantify closeness between p.d.f.s  in \eqref{eq:optimization}. Recall that the KL-divergence between p.d.f.s $g$ and $f$ is given by 
\begin{equation*}
 \dkl(g \| f) = \mathbb{E}_{g}\left[\log \left( \frac{g}{f} \right) \right],
\end{equation*}
provided that $g$ is absolutely continuous with respect to $f;$  otherwise $\dkl(g \| f)  = \infty.$ Note that the KL-divergence satisfies $\dkl(g \|f ) \ge 0,$ with equality iff $g = f.$ However, the KL-divergence is not symmetric, does not satisfy the triangle inequality, and may take value infinity. Therefore, it is not a distance on the space of p.d.f.s. 

 An important motivation for working with the KL-divergence is that the optimization \index{optimization} problem of finding the p.d.f. that minimizes the KL-divergence to the target can be reformulated as an equivalent optimization problem that does not require the target to be normalized. Indeed, we will show  that minimizing the KL-divergence is equivalent to maximizing the Evidence Lower-Bound (ELBO), \index{Evidence Lower-Bound (ELBO)} which we now introduce. 

\begin{definition}[Evidence Lower-Bound]
The ELBO of a p.d.f. $g$ with respect to an unnormalized p.d.f. $\tilde{f}$ is defined by
\begin{align*}
      \ELBO(g) :=
  \mathbb{E}_g[\log \tilde{f}] -
  \mathbb{E}_g[\log g].  \tag*{\qedhere}
\end{align*}
\end{definition}
The next result  shows that the KL-divergence is equal, up to an additive constant, to the negative ELBO. Consequently, minimizing the KL-divergence is equivalent to maximizing the ELBO. 

\begin{theorem}[Relationship Between KL-Divergence and ELBO]\label{th:KLELBO}
Let $f(x) = \tilde{f}(x)/c$ be a p.d.f. and let $\ELBO(g)$ denote the ELBO of a p.d.f. $g$ with respect to the unnormalized p.d.f. $\tilde{f}.$ Then  
\begin{align*}
  \dkl (g \| f )  = -\ELBO(g) + \log c.
\end{align*}
\end{theorem}

\begin{proof}
By direct calculation,
\begin{align*}
\dkl (g \| f)  &= \mathbb{E}_g \left[\log \left( \frac{g}{f} \right)\right] \\
&= \mathbb{E}_{g}\left[\log g \right] - \mathbb{E}_{g}\left[\log f\right] \\[.5em]
&= \mathbb{E}_{g}\left[\log g \right] - \mathbb{E}_{g}\left[\log \left( \frac{\tilde{f}}{c} \right) \right] \\
&= \mathbb{E}_{g}\left[\log g \right] - \mathbb{E}_{g}\left[\log \tilde{f} \right] + \mathbb{E}_{g}\left[\log c\right] \\
&=   -\ELBO(g) + \log c\,.  \tag*{\qedhere}
\end{align*} 
\end{proof}

Theorem \ref{th:KLELBO} implies that the optimization problem \eqref{eq:optimization}  can be reformulated as the problem of maximizing the ELBO. The following corollary shows that the ELBO provides a lower-bound on the (log)-evidence, which motivates its name. We will discuss the interpretation of the normalizing constant \index{normalizing constant} $c$ as the model evidence \index{evidence} in Example \ref{ex:bayesianinterpretation} below. 

\begin{corollary}[Lower Bound Property of ELBO]\label{corollary:ELBO}
In the setting of Theorem \ref{th:KLELBO}, it holds that
$$ \ELBO(g) \leq \log c.$$
\end{corollary}

\begin{proof}
Since the KL-divergence is non-negative and $\dkl(g \| f ) + \ELBO(g) = \log c,$  the result follows.  \hfill $\square$
\end{proof}

\begin{example}[VI in Bayesian Statistics]\label{ex:bayesianinterpretation}
In Bayesian statistics, inference about a parameter $\theta$ given data $y$ is based on the posterior distribution $$f( \theta | y) = \frac{1}{f(y)} f(y | \theta) f(\theta),$$  which combines the likelihood \index{likelihood}  $f(y | \theta)$ with a prior \index{prior} distribution $f(\theta)$.  As discussed in Section \ref{sec:motivatingapplications}, computing posterior \index{posterior} expectations can be expensive when the dimension of $\theta$ is large. In this case, computing the normalizing constant $f(y),$ called the \emph{evidence}, \index{evidence} can also be challenging. It is therefore natural to view $$\tilde{f}(\theta): = f(y|\theta) f(\theta)$$ as an unnormalized target distribution with unknown normalizing constant $c:= f(y).$ 

Corollary \ref{corollary:ELBO} then shows that $\ELBO(g) \le \log f(y).$ This inequality explains the name \emph{evidence lower-bound}: $\ELBO(g)$ gives a lower bound on the log-evidence $\log f(y)$. In addition, this inequality motivates a crude but inexpensive model selection 
criterion:
since a large evidence suggests good model fit, one may choose the model that gives the highest ELBO.

In this Bayesian context, it is insightful to rewrite the ELBO \index{Evidence Lower-Bound (ELBO)} as follows:
\begin{align*}
\ELBO(g) &= \mathbb{E}_{g}\left[\log f(y | \theta)\right] + \mathbb{E}_{g}\left[\log f(\theta)\right] -  \mathbb{E}_{g}\left[\log g(\theta) \right] \\
&= \mathbb{E}_{g}\left[\log f(y | \theta)\right] + \mathbb{E}_{g}\left[\log \left( \frac{f(\theta)}{g(\theta)} \right)\right] \\
&= \mathbb{E}_{g}\left[\log f(y | \theta)\right] - \dkl(g \| f) .
\end{align*}
The first term in the last displayed equation is the expected log-likelihood, whereas the second is the negative KL-divergence \index{divergence!Kullback-Leibler} between the variational density $g$ and the prior. From this formulation, we see that the optimal $g$ finds a compromise between maximizing the expected log-likelihood and staying close to the prior distribution. The first term promotes agreement with the data, while the second term promotes agreement with prior beliefs.  \hfill \qedhere
\end{example}

\subsection{Mean-Field Approximation and Coordinate Ascent Variational Inference}\label{sec:meanfield}
This section discusses the choice of the family $\mathscr{D}$ of admissible distributions in \eqref{eq:optimization}. Ideally, the family $\mathscr{D}$ should be chosen so that:
\begin{enumerate}[(i)]
\item there is a $g \in \mathscr{D}$ that accurately approximates the target;
\item expectations with the variational distribution $g^*$ can be computed efficiently. Thus, $\mathscr{D}$ should contain tractable distributions; 
\item the optimization problem \eqref{eq:optimization} can be solved efficiently. 
\end{enumerate}

 In practice, a compromise needs to be made: enlarging the family $\mathscr{D}$ potentially helps with approximation accuracy, but it typically leads to a harder optimization problem. A popular choice of $\mathscr{D}$ is the mean-field family, which sacrifices the ability to accurately represent target distributions with highly correlated coordinates, but leads to efficient optimization and inference algorithms. 
 
\begin{definition}[Mean-Field Family\index{mean-field family}]
A p.d.f. $g(x)$ with $x = (x_1,\dots,x_d) \in \mathbb{R}^d$ is said to belong to the mean-field family if it can be written in the form
$$
  g(x) = \prod_{i=1}^{d} g_i(x_i) \: ,
$$
where $g_i$ is a one-dimensional p.d.f. over the $i$-th coordinate $x_i$. 
\end{definition}

In other words, the mean-field assumption requires that the variational density 
factorizes as the product of one-dimensionality densities.
We restrict attention to the case where $g$ factorizes as a product of one-dimensional distributions, although more general factorizations over larger groups of coordinates are also possible. 

\FloatBarrier
\begin{mybox}[colback=white]{Pros and Cons}
As we shall see, solving the optimization problem \eqref{eq:optimization} is relatively straightforward when using the mean-field family.
While restrictive, the  mean-field family \index{mean-field family} still affords some flexibility, as the marginals $g_i$  are left unconstrained. If required, we can impose  a specific parametric  family on these marginals as well.
The main caveat of the mean-field family is that it cannot capture the dependence structure of the different variables,  and estimates computed with a mean-field distribution may be inaccurate when the target has highly correlated variables. In that case, marginal variances can be severely over- or underestimated.
\end{mybox}
\FloatBarrier

\begin{remark}
The mean-field assumption has its roots in statistical physics, where ``mean-field'' refers to simplifying complex high-dimensional stochastic models by considering simpler ones that ignore second-order effects. While we focus on the mean-field variational family
due to its popularity and relative simplicity, it is possible to go beyond this restrictive assumption by incorporating dependencies between variables or by adding new latent variables. Both
of these approaches reduce the bias of VI, but can substantially increase
the difficulty of the associated optimization \index{optimization} problem \eqref{eq:optimization}.
\end{remark}

Under the mean-field approximation, each $g_i$ can be optimized individually, and this approach is known as Coordinate Ascent Variational Inference (CAVI). \index{Coordinate Ascent Variational Inference (CAVI)} We next prove a theorem that will lead to this technique. Let us set some notation. We denote by $\bti$ the $i$-th component of $\bt \in \mathbb{R}^d$. We denote by $\btmi \in \R^{d-1}$ the vector obtained by removing the $i$-th component of $\bt.$ To highlight the dependence of $f(x)$ on $x_i$ and $x_{-i},$ we write $f(x_i,x_{-i})$ abusing notation. Finally, $g_{-i}$ will denote a p.d.f. over the variables $x_{-i}.$

\begin{theorem}[Coordinate-Wise ELBO Maximization]\label{theoremCAVI}
Let $g(x)=\prod_{i=1}^{d}g_{i}(x_{i})$ and let $g_{-i}(x_{-i}) = \prod_{j \neq i} g_j(x_j)$ for fixed $g_j,$ $j \neq i$.
Then, the $g_{i}$ that maximizes
$\ELBO(g)$ is given by
\begin{equation}\label{eq:Lqistar}
 g_i^*(\bti) \propto
  \exp \Bigl(
    \EEqmi{ 
      \log \tilde{f}(\bti, \btmi)
      }
  \Bigr).
\end{equation}
\end{theorem}
\begin{proof}
Throughout this proof,  we denote by $C$ a constant that does not depend on $g_i$ and may change from line to line.
We will show that 
$$\ELBO(g) = - \dkl( g_i \| g_i^*) + C,$$
where $g_i^*$ is given by \eqref{eq:Lqistar}. This implies the desired result, since the  $g_i$ that minimizes the KL-divergence on the right-hand side is clearly $g_i = g_i^*.$

First, we show that 
$$
 \ELBO(g)  = \EEqi{ \EEqmi{ \log \tilde{f}(\bti, \btmi)  } - \log g_i (\bti) } + C. 
$$
Indeed, 
\begin{align*} \ELBO(g) & = \EEq{ \log \tilde{f}(\bt)}  - \EEq{ \log g(\bt) } \\
& = \EEq{ \log \tilde{f}(\bti, \btmi) } - \sum_{j} \EEqj{\log g_j(\btj)} \\ 
& = \EEqi{ \EEqmi{ \log \tilde{f}( \bti, \btmi) | \bti }} - \EEqi{\log g_i(\bti) } + C,\end{align*}
and the first term on the right-hand side can be simplified using the fact that
\begin{align*}
\EEqmi{ \log \tilde{f}(\bti, \btmi) | \bti }
&= \int_{\R^{d-1}} \log \tilde{f}( \bti, \btmi) \cdot g(\btmi|\bti) \, d\btmi \\
& = \int_{\R^{d-1}} \log \tilde{f}( \bti, \btmi) \cdot g(\btmi) \,  d\btmi \\
& = \EEqmi{ \log \tilde{f}(\bti, \btmi)  } .
\end{align*}
Therefore,
\begin{align*}
\ELBO(g) &= 
\EEqi{ \EEqmi{ \log \tilde{f}(\bti, \btmi)  } - \log g_i (\bti) }  + C \\
&= \EEqi{ \log \Bigl(\exp\bigl( \EEqmi{ \log \tilde{f}(\bti, \btmi) } \bigr) \Bigr) - \log g_i (\bti) }  + C \\ 
& =  -  \EEqi{  \log \left( \frac{ g_i (\bti)  }{ \exp \bigl( \EEqmi{ \log \tilde{f}( \bti, \btmi) } \bigr) } \right)  }  + C \\
& = - \dkl( g_i \| g_i^*) + C,
\end{align*} 
and the proof is complete. \hfill $\square$
\end{proof}

Theorem \ref{theoremCAVI} gives us an expression for the optimizer of the ELBO with respect to $g_i,$ holding the other factors fixed. This procedure can be repeated iteratively, optimizing each factor while holding the others fixed. The resulting algorithm is known as Coordinate Ascent Variational Inference (CAVI): \index{Coordinate Ascent Variational Inference (CAVI)}

\FloatBarrier
\begin{algorithm}
\caption{Coordinate Ascent Variational Inference (CAVI)}
\begin{algorithmic}[1]
\STATE {\bf Input}: Unnormalized target density $\tilde{f}$.
\STATE {\bf Initialization}: Choose initial variational densities $g_i$ for each $i= 1,\dots,d$. 
\STATE While the ELBO/KL-divergence has not converged, do: 
\STATE Update, for each $i=1,\dots,d,$
$$
\begin{aligned}
  g_i(\bti) \propto
  \exp \Bigl(
    \EEqmi{ 
      \log \tilde{f} ( \bti, \btmi)
      }
  \Bigr).
\end{aligned}
$$
\STATE {\bf Output}: $g(\bt) = \prod_{i = 1}^d  g_i(\bti) \approx f(x).$ 
\end{algorithmic}
\end{algorithm}
\FloatBarrier

In each update, CAVI \index{Coordinate Ascent Variational Inference (CAVI)} increases the ELBO \index{Evidence Lower-Bound (ELBO)} or, equivalently, lowers the KL-divergence \index{divergence!Kullback-Leibler} between the current variational distribution $g$ and the target $f$.
Hence, under suitable assumptions, one can expect convergence to a local minimizer of the KL-divergence. The ELBO \index{Evidence Lower-Bound (ELBO)} is not necessarily concave, and therefore there is, in general, no guarantee of reaching a global maximizer of the ELBO or, equivalently, a global minimizer of the KL-divergence.

\FloatBarrier
\begin{mybox}[colback=white]{Pros and Cons}
 Each CAVI update lowers the KL-divergence, which facilitates the introduction of stopping criteria for the algorithm. If the factors $g_i$ are assumed to belong to the exponential family, computing each factor update is simplified significantly. 
  However, since CAVI relies on the mean-field assumption, it inherits all of its disadvantages. In addition, for some target distributions 
 it may be computationally expensive to compute the updates of the variational factors. In such cases, other optimization methods (e.g. stochastic optimization and natural gradient methods) can be used to maximize the ELBO instead of CAVI. 
\end{mybox}
\FloatBarrier

\begin{remark}
CAVI shares a similar structure with the Gibbs sampler.\index{Gibbs sampler} To see why, notice that
$$ g_i^*(\bti) \propto
  \exp \Bigl(
    \EEqmi{ 
      \log f(\bti, \btmi)
      }
  \Bigr)
  \propto 
  \exp \Bigl(
    \EEqmi{ 
      \log f(\bti| \btmi)
      }
  \Bigr) .$$ 
In Gibbs sampling, we update each coordinate by sampling from the corresponding full conditional. CAVI uses full conditionals as well, with each density $g_i$ set to be proportional to the exponential of the expected logarithm of the full conditional. 
\end{remark}

\section{Expectation Maximization Algorithm}\label{sec:EMalgorithm}
In its most basic form, the Expectation Maximization (EM) \index{Expectation Maximization (EM)} algorithm is a deterministic optimization \index{optimization} method. 
EM is widely used to compute maximum likelihood estimators \index{maximum likelihood estimator} in problems where introducing a latent variable makes the likelihood function easier to optimize.

Suppose for concreteness that we want to find the maximum likelihood estimator 
$$ \widehat{\theta} = \arg \max_\theta f(y|\theta)$$ 
of a parameter $\theta$ given observed data $y.$ Assume that $f(y|\theta)$ is the marginal of a complete likelihood $f(y,z|\theta),$ which includes the data $y$ and an auxiliary variable $z.$ Then, for any p.d.f. $g$ with compatible support, Jensen's inequality gives
\begin{align*}
    \log f(y| \theta) & = \log \int f(y,z | \theta) \, dz \\
    & = \log \int \frac{f(y,z| \theta)}{g(z)} g(z) \, dz \\ 
    & \ge \int \log \left( \frac{f(y,z|\theta)}{g(z)} \right) g(z) \, dz \\
    & = \mathbb{E}_g[ \log f(y,z|\theta)] - \mathbb{E}_g[ \log(g)] = \ELBO(g, \theta).
\end{align*}
Here we view the ELBO \index{Evidence Lower-Bound (ELBO)} as a function of the auxiliary distribution $g$ and the parameter $\theta$ in the density $f(y,z |\theta).$
Now, notice that  
\begin{align*}
    \ELBO(g, \theta) &= \int \log \left( \frac{f(y,z|\theta)}{g(z)} \right) g(z) \, dz \\
    & = \int \log \left( \frac{f(z|y,\theta)}{g(z)} \right) g(z) \, dz + \int \log \bigl( f(y|\theta) \bigr) g(z) \, dz \\
    & = - \dkl \bigl( g \| f( \cdot | y, \theta) \bigr) + \log  f(y| \theta).
\end{align*}
In other words, 
\begin{equation}\label{eq:mainequation}
    \log f(y| \theta) = \ELBO(g,\theta) + \dkl \bigl( g \| f( \cdot | y, \theta) \bigr). 
\end{equation}
The above derivations are direct consequences of Theorem \ref{th:KLELBO} and Corollary \ref{corollary:ELBO} with target $f(z) \equiv f(z |y,\theta),$  unnormalized density $\tilde{f}(z) \equiv f(y,z |\theta),$ and normalizing constant $c \equiv f(y|\theta).$ \index{normalizing constant}

Equation \eqref{eq:mainequation} motivates an iterative approach to maximize the log-likelihood $\log f (y | \theta),$ which is equivalent to maximizing the likelihood $f(y | \theta)$ since the logarithm is an increasing function.
Given the current iterate $\theta_\ell,$ we obtain the next iterate  $\theta_{\ell+1}$ in two steps, maximizing in turn the two components of the ELBO; this approach is illustrated in Figure \ref{fig:ELBO}.

\begin{figure}[h]
    \centering
    \includegraphics[width=0.5\textwidth]{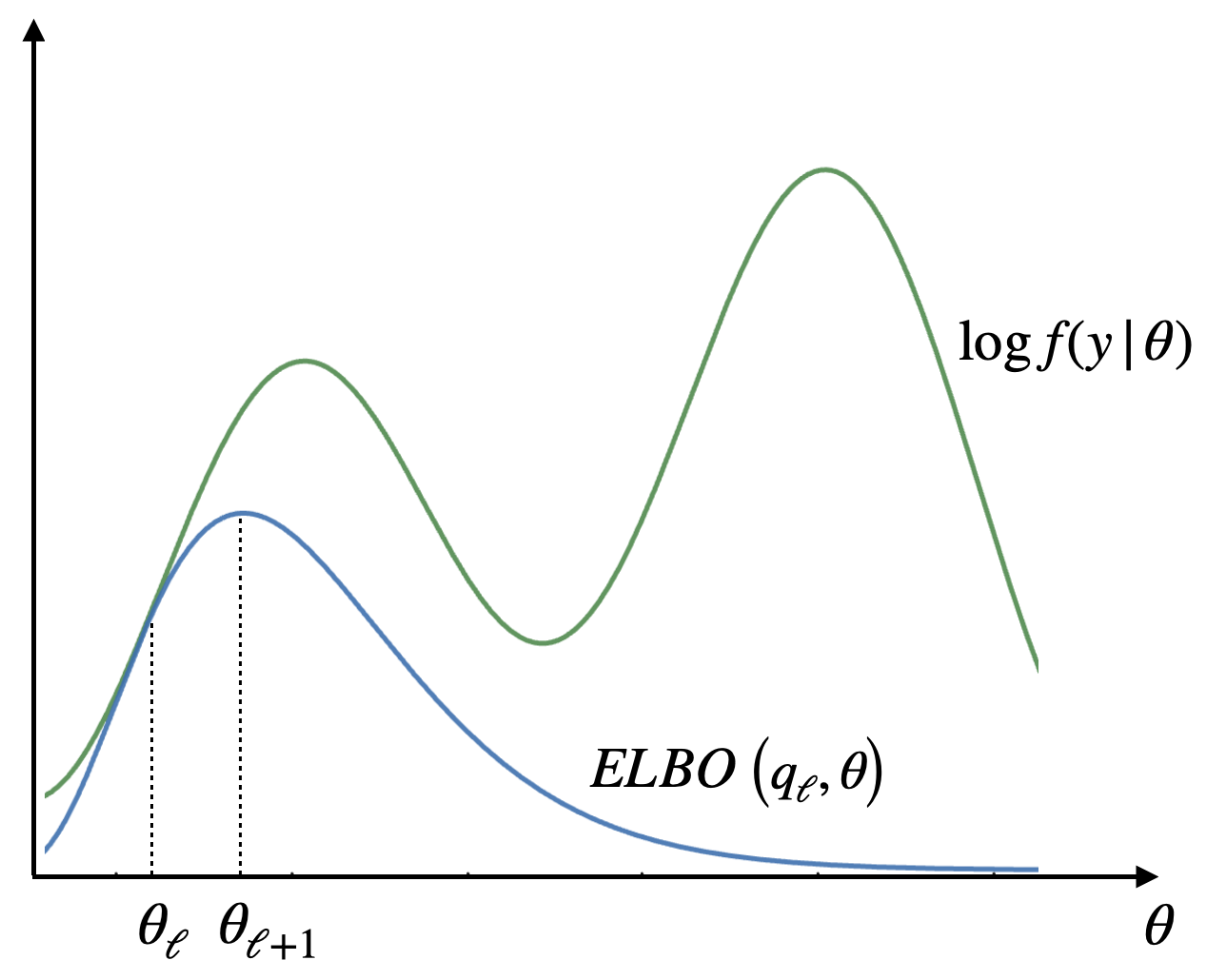}
    \caption{One iteration of the EM algorithm.}
    \label{fig:ELBO}
\end{figure}

\begin{enumerate}
\item First,  we find $g_\ell(z)$ by maximizing $\ELBO(g, \theta_\ell)$  over p.d.f.s. $g.$  From \eqref{eq:mainequation}, 
$$\log f (y | \theta_\ell) = \ELBO(g, \theta_\ell) + \dkl \bigl( g \| f(z|y,\theta_\ell)\bigr),$$
and it follows that maximizing $\ELBO(g, \theta_\ell)$ and minimizing $\dkl \bigl( g \| f(z|y,\theta_\ell)\bigr)$ over $g$ are equivalent. Since the KL-divergence is minimized when both arguments agree, we obtain $g_\ell(z) = f(z|y,\theta_\ell).$
\item Second, we  find $\theta_{\ell + 1}$ by maximizing  $\ELBO(g_\ell, \theta)$ over $\theta.$ Note that 
\begin{equation}
\ELBO(g_\ell, \theta) = \int \log f(y, z | \theta) f(z|y, \theta_\ell) \, dz + \text{const},
\end{equation}
where $\text{const}$ is independent of $\theta.$ Hence, the quantity to maximize is the expected value of the joint log-density $\log  f (y, z | \theta)$ with respect to $g_\ell(z) = f(z|y, \theta_\ell).$
\end{enumerate}

Combining these two steps gives the EM algorithm, summarized below.

 \begin{algorithm}
\caption{\label{alg:EM} Expectation Maximization (EM)}
\begin{algorithmic}[1]
\STATE {\bf Input}:  Initialization $\theta_0,$ number of iterations $L.$
\STATE For $\ell = 0, 1, \ldots, L-1$ perform the following expectation and maximization steps:
\STATE {\bf E-Step}: Compute 
\begin{align*}
\Expect_{Z \sim  f(z|y, \theta_\ell)} \Bigl[ \log f(y, Z | \theta) \Bigr]= \int \log f(y, z | \theta) f(z|y, \theta_\ell) \, dz. 
\end{align*}
\vspace{-0.5cm}
\STATE{{{\bf M-Step}}}: Compute 
\begin{equation*}
\theta_{\ell+1} = \arg \max_\theta \Expect_{Z \sim  f(z|y, \theta_\ell)} \Bigl[ \log f(y, Z | \theta) \Bigr].
\end{equation*}
\STATE{\bf Output}: Parameter $\theta_L.$
\end{algorithmic}
\end{algorithm}
\FloatBarrier 

The following result shows that the EM algorithm \index{Expectation Maximization (EM)} has the desirable property that the likelihood function increases monotonically along the iterates $\theta_\ell.$ 

\begin{theorem}[Monotonic Increase of Likelihood along EM Iterates\index{Expectation Maximization (EM)!monotonic increase of likelihood}]\label{th:EMmonotone}
Let $\{\theta_\ell \}_{\ell = 0}^L$ be the iterates of Algorithm \ref{alg:EM}. Then, for $0 \leq \ell \leq L-1,$ 
\begin{equation}
\log f(y | \theta_\ell) \le \log f (y | \theta_{\ell + 1}).
\end{equation}
\end{theorem}
\begin{proof}
Let $g_\ell(z) = f(z|y, \theta_\ell).$  Using the log-likelihood characterization in \eqref{eq:mainequation}, we have 
\begin{align*}
\log f (y | \theta_{\ell + 1}) &=\ELBO(g_\ell, \theta_{\ell + 1}) + \dkl \bigl( g_\ell \| f(z|y,\theta_{\ell + 1})\bigr) \\ 
&\ge \ELBO(g_\ell, \theta_{\ell}) +  \dkl \bigl( g_\ell \| f(z|y,\theta_{\ell + 1})\bigr) \\
&\ge \ELBO(g_\ell, \theta_{\ell})  + \dkl \bigl( g_\ell \| f(z|y,\theta_{\ell })\bigr) \\
& = \log f (y | \theta_{\ell}).
\end{align*}
The first inequality follows because $\theta_{\ell + 1}$ maximizes $\ELBO(g_\ell, \theta)$ over $\theta.$ The second inequality follows because $\dkl \bigl(g_\ell \| f(z|y, \theta_\ell) \bigr) = 0$ since $g_\ell = f(z|y, \theta_\ell),$ while $\dkl \bigl( g_\ell \| f(z|y,\theta_{\ell+1 })\bigr) \ge 0.$  \hfill $\square$
\end{proof}

\begin{remark}
\label{rem:em}
Under mild additional assumptions, it follows from Theorem \ref{th:EMmonotone} that the iterates $\theta_\ell$ of the EM algorithm converge, as $\ell \to \infty,$ to a local maximizer of the likelihood function. It is important to note, however, that the expectation in the E-step and the optimization in the M-step are often intractable. Monte Carlo algorithms may be employed to approximate the E-step, and optimization algorithms may be employed to approximate the M-step. Such approximations can lead to loss of monotonicity and convergence guarantees. 
\end{remark}

\FloatBarrier
\begin{mybox}[colback=white]{Pros and Cons}
The EM algorithm is a workhorse in computational mathematics and statistics. It is widely used to compute maximum likelihood estimators in mixture models, hierarchical models, hidden Markov models, and related settings. For some problems, the M-step can be solved in closed form, in which case the EM algorithm can be an efficient derivative-free optimization method. A caveat of the EM algorithm is that the choice of initialization is important, as the algorithm may converge to local maxima. In addition, for some problems computing the E- and M-steps can be computationally expensive, and convergence can be slow. 
\end{mybox}
\FloatBarrier

\begin{example}[Dempster et al. 1977] \label{ex:7.3}
Observations $y = (y_1,y_2,y_3,y_4)$ are drawn from the multinomial distribution 
\begin{align*}
\text{Multinomial} \left(K; \frac{1}{2}+\frac{\theta}{4} , \frac{1}{4}(1-\theta), \frac{1}{4} (1-\theta), \frac{\theta}{4} \right). 
\end{align*}
Estimation is simplified if $y_1$ is split into two cells $z = (z_1,z_2),$ so we create the augmented model 
\begin{align*}
(z_1,z_2,y_2,y_3,y_4) \sim  \text{Multinomial} \left(K; \frac{1}{2},\frac{\theta}{4} , \frac{1}{4}(1-\theta), \frac{1}{4} (1-\theta), \frac{\theta}{4} \right), 
\end{align*}
with $y_1=z_1+z_2$.  
Here, 
\begin{align*}
f(y,z|\theta)& \propto\theta^{z_2+y_4}(1-\theta) ^{y_2+y_3} , \\
f(y|\theta)& \propto (2+\theta)^{y_1} \theta^{y_4} (1-\theta)^{y_2+y_3}. 
\end{align*}
We have, up to constants,
\begin{align*} 
\Expect_{Z \sim  f(z|y, \theta_\ell)} \left[ \log f(y,Z|\theta) \right]
&=\Expect\left[ (Z_2 +y_4) \log \theta  +(y_2+y_3) \log (1-\theta) \right]\\
&=\left( \frac{\theta_\ell}{2+\theta_\ell} y_1 + y_4\right) \log \theta +(y_2+y_3) \log(1-\theta), 
\end{align*}
where we used that  
\begin{align*}
Z_2|y,\theta_\ell \sim \text{Binomial} \left( y_1, \frac{\frac{\theta_\ell}{4}}{ \frac{1}{2}+\frac{\theta_\ell}{4}} \right) =\text{Binomial} \left( y_1, \frac{\theta_\ell}{2+\theta_\ell} \right). 
\end{align*}
The expectation 
$\Expect_{Z \sim  f(z |y, \theta_\ell)} \left[ \log f(y,Z|\theta) \right]$ can then be maximized over $\theta$ to give 
\begin{align*}
\theta_{\ell+1}=\frac{ \frac{\theta_\ell y_1}{2+\theta_\ell}+y_4}{ \frac{\theta_\ell y_1}{2+\theta_\ell}+y_2+y_3+y_4}.   \tag*{\qedhere}
\end{align*}
\end{example}

\paragraph{Monte Carlo EM} 
In practice, computing the E-step may be computationally difficult, and the expectation appearing in the E-step can be approximated using Monte Carlo:
\begin{align*}
\Expect_{Z \sim  f(z|y, \theta_\ell)} \left[ \log f(y,Z|\theta) \right] \approx \frac{1}{N} \sum_{n=1}^N \log f(y,Z^{(n)}|\theta), 
\end{align*}
where 
\begin{align*}
 Z^{(1)}, \dots, Z^{(N)} \stackrel{\text{i.i.d.}}{\sim} f(z|y,\theta_\ell). 
\end{align*}

\begin{example}[Monte Carlo EM for Dempster et al. 1977]
In the setting of Example \ref{ex:7.3} there is no need to use Monte Carlo EM, since the expectation can be computed analytically using that $Z_2|y,\theta_\ell \sim \text{Binomial} \left(y_1, \frac{\theta_\ell}{2+\theta_\ell} \right)$. However, merely to illustrate the technique, we note that  Monte Carlo EM would set 
\begin{align*}
\theta_{\ell+1} =\frac{ \bar{z}_N+y_4}{\bar{z}_N+y_2+y_3+y_4},
\end{align*}
where 
\begin{align*}
\bar{z}_N= \frac{1}{N} \sum_{n=1}^N Z^{(n)}, 
\quad 
 Z^{(1)}, \dots, Z^{(N)} \stackrel{\text{i.i.d.}}{\sim} \text{Binomial}\left( y_1,\frac{\theta_\ell}{2+\theta_\ell} \right). \tag*{\qedhere}
\end{align*} 
\end{example}

\section{Discussion and Bibliography}\label{sec:discussion}
 We refer to \cite{bishop,Blei2017,bach2024machine} for accessible introductions to VI that have inspired the presentation in this chapter.
 For a more in-depth exploration of VI, see
 \cite{jordan1999introduction,wainwright2008graphical}. The paper \cite{hoffman2013stochastic} contains a modern approach to VI with applications to massive text data-sets.  Applications to graphical models, and a generalization of the mean-field family to allow for interactions, are discussed in \cite{jordan1999introduction}.
 While CAVI provides a simple framework for ELBO maximization under the mean-field family, computing the CAVI updates can be challenging for highly complex targets, in which case alternative optimization methods based on stochastic VI or autodifferentiation may be required \cite{hoffman2013stochastic,kucukelbir2017automatic}.
Beyond the mean-field family, other simple but practical choices of variational families include Gaussians \cite{sanz2023inverse} and Gaussian mixtures \cite{lambert2022variational}.

Our presentation has focused solely on VI techniques that seek to minimize the KL-divergence \cite{kullback1951}.
The works \cite{li2016renyi,hernandez2016black} consider VI using the broad family of Rényi-alpha divergences, which includes for instance the KL-divergence, the $\chi^2$-divergence, and the Hellinger distance \cite{gibbs2002choosing}. However, as emphasized in this chapter, the KL-divergence holds a special place in VI, since, among a wide class of divergences, it is the unique choice for which minimization of the objective does not require knowledge of the normalizing constant of the target \cite{chen2023gradient}. For sampling problems, the paper \cite{trillos2018bayesian} explores gradient flows arising from minimizing various objectives and shows that the choice of KL-divergence also leads to fundamental practical advantages. 

An overarching idea behind VI, the EM algorithm, and other popular computational methods such as variational autoencoders \cite{kingma2013auto}, is to minimize KL-divergence by maximizing the ELBO. The ELBO functional is typically non-convex, even for simple potentials; see, for instance, the discussion in \cite{agrawal2022variational}. Due to this lack of convexity, it is often challenging to develop
theory for VI outside restricted settings on the target distribution. For this reason, the study of statistical and optimization convergence guarantees for VI is still an active area of research, see, e.g., \cite{zhang2020convergence} for statistical convergence rates and  \cite{lambert2022variational} for a theoretical framework that relies on gradient flows. 

The EM algorithm was introduced in \cite{dempster1977maximum}, although particular instances had been previously considered on numerous occasions, see e.g. \cite{ceppellini1955estimation,hartley1958maximum}. Convergence of the EM algorithm was established in \cite{wu1983convergence}.  
More recently, \cite{bing2025convergence} established convergence and
near-minimax optimality guarantees for EM in multi-component Gaussian mixture
models under separation conditions, allowing unknown mixing probabilities and
covariance matrices. 
We refer to  \cite{mclachlan2007algorithm} for a textbook 
treatment of
the EM algorithm and to \cite{bilmes1998gentle,hastie2009elements,mackay2003information,bishop} for gentle introductions. EM is an example of a minorization-maximization method  \cite{hunter2004tutorial}, a broad family of algorithms with rich theoretical underpinnings. Several modifications to the original EM algorithm have been proposed to speed up its convergence. We refer to \cite{meng1993maximum,liu1994ecme,jamshidian1997acceleration,neal1998view} for some important methodological developments
and to \cite{meng1997algorithm} for a survey written on the 20th anniversary of the original paper \cite{dempster1977maximum}.

\appendix
\newpage \chapter{Markov Chain Theory}
\label{chap:markovchains}
This appendix contains an informal review of Markov chains. 
Its purpose is to give a brief and intuitive, rather than exhaustive and fully rigorous, overview of the key concepts and ideas needed to understand MCMC algorithms. The material is aimed at readers who have already studied Markov chains but may benefit from a quick refresher on certain specific concepts while reading this textbook. We emphasize that our presentation is not suitable for learning the material for the first time, nor is it meant to serve as a crash course on Markov chains.    

\section{Basic Definitions}\label{ssec:markovdefinitions}
\begin{definition}[Markov Chain\index{Markov chain}]
A collection  $\XX =\{X_n,\,n=0,1,2,\ldots\}$ of random variables taking values in a state space \index{state space} $E$ is called a \emph{Markov chain} if
$$\mathbb{P}(X_{n+1} \in A_{n+1} \,|\, X_n \in A_n,\ldots,\,X_0 \in A_0)=\mathbb{P}(X_{n+1}\in A_{n+1} \,|\, X_n \in A_n)$$
for all $n \geq 0$ and all measurable $A_{n+1}, A_n,\ldots, A_0 \subseteq  E.$
\end{definition}
The index $n$ is often interpreted as time, and the Markov property then states that the future (time $n+1$) is independent of the past (times $\leq n-1$) given the present (time $n$).

\begin{definition}[Time-Homogeneous Markov Chain;\index{Markov chain!time homogeneous} Markov Kernel]
A Markov chain $\XX$ is called \emph{time homogeneous} if its transition probabilities 
$$P(x,A) = \mathbb{P}(X_{n+1} \in A \,|\, X_n=x)$$
do not depend on $n$. $P(x,A)$ is called the transition kernel, or \emph{Markov kernel}, of $\XX.$
\end{definition}
Here and below, the expression $\mathbb{P}(X_{n+1} \in A \,|\, X_n=x)$ is only formal if $\Prob(X_n = x) =0,$ but it can be made rigorous using regular conditional probabilities. From now on, all Markov chains that we consider will be time homogeneous. In this textbook, we work mainly with transition kernels that are absolutely continuous with respect to either the Lebesgue measure or the counting measure; an exception is the Metropolis Hastings kernel studied in Chapter \ref{chap:MCMC}. This means that, for every $x,$ there is an associated p.d.f. (or p.m.f. for the discrete case) $p(x,z)$ such that 
$$P(x,A)=\int_{A}p(x,z) \,dz.$$
With a slight abuse of terminology, we will often call $p(x,z)$ a Markov kernel. \index{Markov kernel}

\begin{example}[Weather Markov Chain]\label{ex:weatherchain}
Suppose $X_n$ represents the weather on day $n,$  with
\[X_n = \begin{cases}
0& {\text {if sunny,}}\\
1&{\text{if cloudy,}}\\
2&{\text{if rainy.}}\\
\end{cases}
\]
\newline The state space is $E= \{0,1,2\}.$ We may specify time-homogeneous transition probabilities graphically as follows:
\begin{figure}[H]
\centering
\includegraphics[scale=0.4]{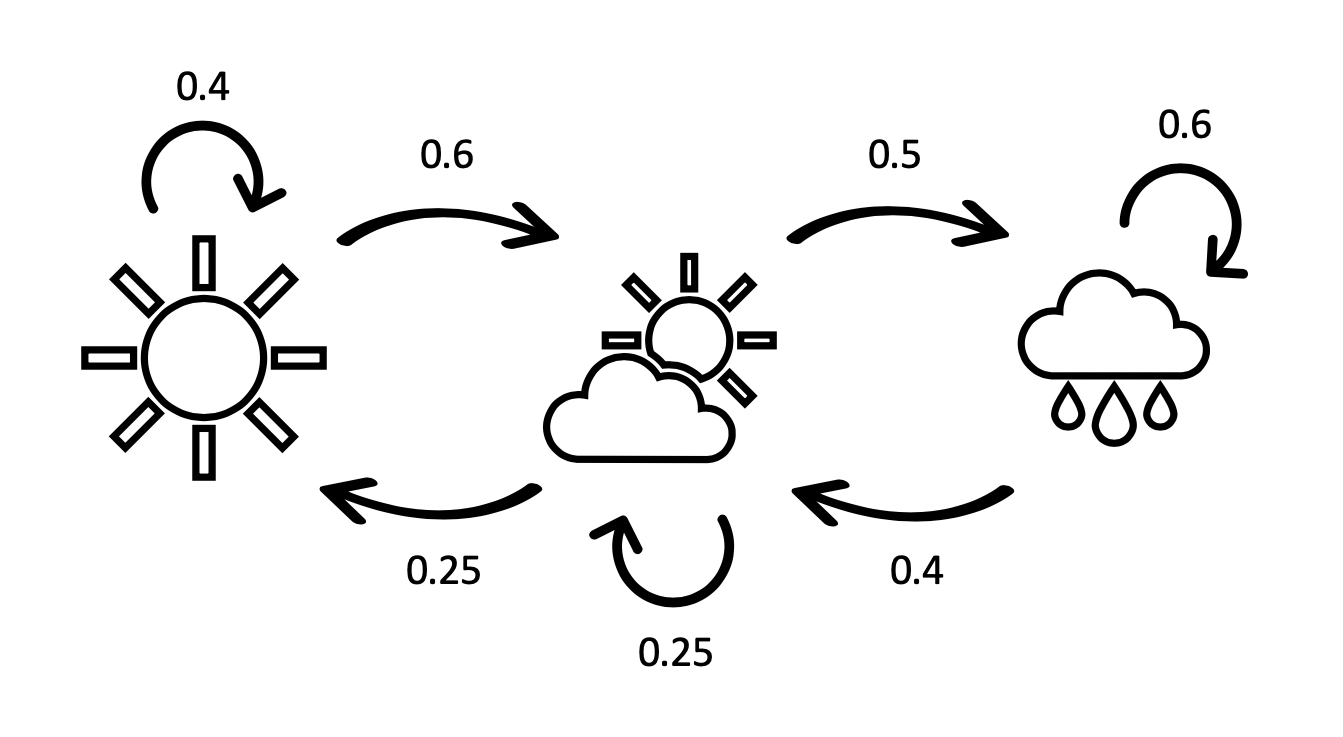}
\end{figure}
\hfill \qedhere
\end{example}

\begin{definition}[$n$-Step Transition Densities]\label{def:nsteptransition}
For $n \ge 1,$ the
\emph{$n$-step transition densities}
are defined by
\begin{equation*}
    \Prob(X_n \in A \,|\, X_0=x)=P^n(x,A)= \int_{A}p^n(x,z) \,dz. \tag*{\qedhere}
\end{equation*}
\end{definition}

If the state space is finite, say $E = \{ 0, 1, \ldots, s\}$ for some integer $s,$ then we can collect the transition probabilities in the  \emph{transition matrix}
$$P=[p_{ij}]_{i,j \in E},$$
where $p_{ij}=\mathbb{P}(X_{n+1}=j \,|\, X_n=i)$ for $(i,j) \in E \times E$.
Then $$p^n(i,j) = ( P^n )_{ij}.$$

\begin{lemma}\label{lem:nsteptransition}
If $X_0 \sim \pi_0$ and $\XX$ has $n$-step transition densities $p^n(x,z),$ then the density of $X_n$ is given by
$$\pi_n(x)=\int_{E}\pi_0(z)p^n(z,x) \,dz.$$
If the state space $E$ is finite,  then $\pi_n$ is the row probability vector given by $\pi_n=\pi_0P^n.$
\end{lemma}

\begin{example}[Weather Markov Chain: 2-Step Probabilities]\label{ex:weatherchain2}
The transition matrix of the weather chain is given by
$$P=\begin{bmatrix}
0.4&0.6&0\\
0.25&0.25&0.5\\
0&0.4&0.6
\end{bmatrix}.$$
Suppose day 0 is sunny. Then 
$$\pi_0=\begin{bmatrix}
1&0&0
\end{bmatrix}.$$
The distribution of the weather on day 2 is
$$\pi_2=\pi_0P^2= \begin{bmatrix}
1&0&0
\end{bmatrix}
P^2=\begin{bmatrix}
0.31&0.39&0.3
\end{bmatrix}.$$
Thus if day 0 is sunny, then day 2 has probability 0.31 of being sunny.  \hfill \qedhere
\end{example}

\section{Invariant Distributions: \index{invariant distribution} General Balance \index{general balance} and Detailed Balance\index{detailed balance}}\label{sec:generalanddetailbalance}
In this section, we discuss the notion of statistical equilibrium, which is central to the theory of Markov chains.
\begin{definition}[General Balance Equation; Invariant Distribution]
A Markov kernel $p(x,z)$ satisfies the \emph{general balance equation} with respect to $\pi$ if 
\begin{equation}\label{eq:GB}
	\pi(x) = \int_E \pi(z) p(z,x) \, dz.
\end{equation}
We then say that $\pi$ is an \emph{invariant distribution} for the Markov kernel $p(x,z).$
\end{definition}
If $E$ is discrete, the general balance equation is simply $\pi = \pi P.$

\begin{definition}[Ergodic Markov Chain]
A Markov chain $\XX$ is called \emph{ergodic} \index{Markov chain!ergodic} if there exists a distribution $\pi$ such that, for all $A\subseteq E$ and all initial distributions $\pi_0,$ it holds that
	\begin{equation*}
		\lim_{n\rightarrow\infty}\Prob(X_n\in A)=\pi(A).
	\end{equation*}
	We say that $\pi$ is the limit distribution of $\XX.$
	\end{definition}

\begin{lemma}
Let $\XX$ be an ergodic Markov chain with Markov kernel $p(x,z)$ and limit distribution $\pi.$ Then $\pi$ is an invariant distribution for $p(x,z).$ In particular, if $\XX$ is initialized at statistical equilibrium, meaning $X_0 \sim \pi$, then $X_n \sim \pi$ for all $n\ge 0.$ 
\end{lemma}
\begin{proof} 
We only give a sketch of the proof.
By the dominated convergence theorem,
\begin{align*}
	\lim_{n\rightarrow\infty}\pi_{n+1}(x) =&\lim_{n\rightarrow\infty} \int_E \pi_n(z)p(z,x)\, dz\\
	=&\int_E\lim_{n\rightarrow\infty}  \pi_n(z)p(z,x) \,dz,
\end{align*}
which shows that $\pi(x) = \int_E\pi(z)p(z,x) \, dz$. The final claim of the lemma follows by induction. \hfill \qedhere
\end{proof}

The general balance equation implies that the transition probabilities of $\XX$ preserve equilibrium. For this reason, we call a distribution $\pi$ that satisfies the general balance equation \eqref{eq:GB} an invariant distribution of the chain $\XX$. If $\XX$ is ergodic, we can use the general balance equation to show that $\pi$ is the equilibrium distribution of the chain. However, given a distribution $\pi$, it is often difficult to find a Markov kernel for which $\pi$ satisfies the general balance equation. It is actually easier to find a Markov kernel that satisfies a \emph{stronger} condition, known as  detailed balance.

\begin{definition}[Detailed Balance]
	We say that a transition density $p(x,z)$ satisfies \emph{detailed balance} with respect to $\pi$ if, for all $x,z\in E,$ 
		\[\pi(x)p(x,z) = \pi(z)p(z,x).\]
  \end{definition}

The Metropolis Hastings algorithm in Chapter \ref{chap:MCMC} provides a general recipe for defining a family of kernels that satisfy detailed balance with respect to a given distribution. The following result shows that detailed balance \index{detailed balance} implies general balance. \index{general balance}

\begin{theorem}[Detailed Balance Implies General Balance]
\label{thm:detailedimpliesgeneral}
Let $p(x,z)$ be a Markov kernel that satisfies detailed balance with respect to a distribution $\pi.$ Then $\pi$ is an invariant distribution for $p(x,z).$
\end{theorem}
\begin{proof}
Using detailed balance and the fact that $p(x,z)$ is a Markov kernel, so that it integrates to one over its second argument, gives
	\begin{equation*}
	\int_E \pi(z)p(z,x) \, dz = \int_E \pi(x)p(x,z) \, dz  =  \pi(x) \int_E p(x,z) \, dz =  \pi(x). \tag*{\qedhere}
	\end{equation*} 
\end{proof}

\begin{remark}
	Detailed balance is \emph{not} necessary for general balance. Detailed balance implies that the chain is time reversible: in equilibrium, the chain behaves the same whether it is run forward or backward in time. However, there are many ergodic Markov chains that are not time reversible.
\end{remark}

\section{Reducibility,\index{reducibility} Periodicity,\index{periodicity} and Transience\index{transience} in a Countable State Space}\label{ssec:reducibility}
A Markov chain taking values in a countable state space $E$ is ergodic if it is irreducible, aperiodic, and positive recurrent. We now review these concepts in the context of countable state spaces and then generalize them to general state spaces. 

\begin{definition}[Irreducible Markov Chain]\label{def:irreducible1}
A Markov chain is \emph{irreducible} if all states intercommunicate. That is, for all $i,j \in E$ there exists $n \geq 0$ such that 
\begin{equation*}
    \mathbb{P}(X_n=i \,|\, X_0=j)>0. \tag*{\qedhere}
\end{equation*}
\end{definition}

\begin{definition}[Recurrent Markov Chain]
A Markov chain is \emph{recurrent} if, for all $i \in E,$
\begin{equation*}
    \mathbb{P}(\XX\, \text{visits}\, i \, \text{again  eventually} \,|\, X_0=i)=1.  \tag*{\qedhere}
\end{equation*}
\end{definition}

\begin{definition}[Transient Markov Chain]
A Markov chain is \emph{transient} if, for all $i \in E,$
\begin{equation*}
    \mathbb{P}(\XX\, \text{visits}\, i \, \text{again  eventually} \,|\, X_0=i)<1.  \tag*{\qedhere}
\end{equation*}
\end{definition}

\begin{definition}[Positive Recurrent Markov Chain\index{Markov chain!positive recurrent}]
A Markov chain is \emph{positive recurrent} if $\mathbb{E}[T_{ii}]<\infty$ for all $i \in E,$ where $T_{ii}$ is the time of the first return to state $i$. 
\end{definition}

\begin{definition}[Aperiodic Markov Chain\index{Markov chain!aperiodic}]
A Markov chain is \emph{aperiodic} if all states have period one. The period of a state $i \in E$ is defined as
$$d(i):= \gcd\{n>0:p^n(i,i)>0\},$$
where $\gcd$ stands for greatest common divisor. 
\end{definition}

\begin{example}[Periodicity]
For this chain,
\begin{figure}[H]
\centering
\includegraphics[scale=0.5]{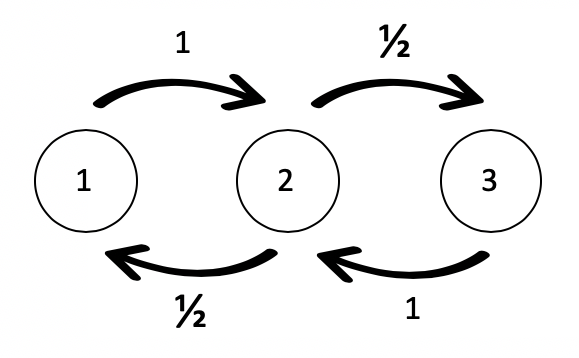}
\end{figure}
we can return to state 1 after 2, 4, 6, $\ldots$ steps. Hence, state 1 has period two.  \hfill \qedhere
\end{example}

\begin{example}[Aperiodic State]
For this chain,
\begin{figure}[H]
\centering
\includegraphics[scale=0.5]{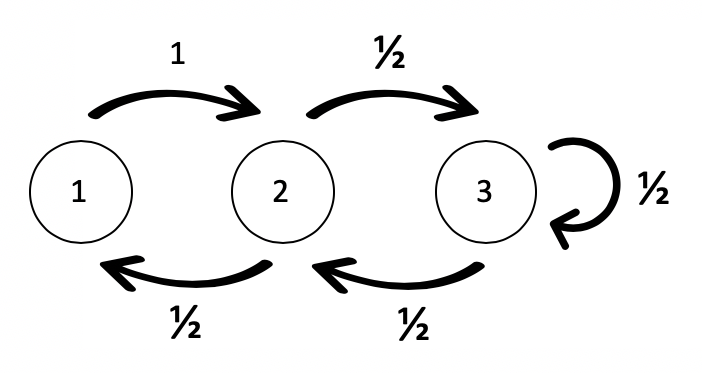}
\end{figure}
we can return to state 1 after 2, 4, 6, 7, $\ldots$ steps. Hence, state 1 has period one. \hfill \qedhere
\end{example}

\begin{example}[Periodic Chain]\label{ex:periodicchain}
This chain is periodic: 
\begin{figure}[H]
\centering
\includegraphics[scale=0.5]{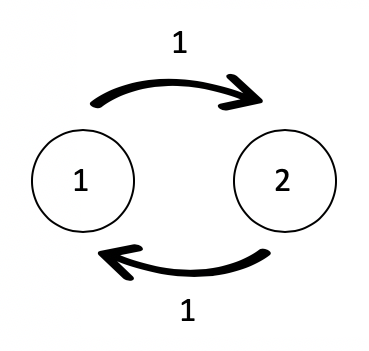}
\end{figure}
Note that
$$P=\begin{bmatrix}
0&1\\
1&0
\end{bmatrix}, \quad P^2=\begin{bmatrix}
1&0\\
0&1
\end{bmatrix},\quad P^3=\begin{bmatrix}
0&1\\
1&0
\end{bmatrix}, \ldots$$
and so the transition probabilities do not converge as $n\to \infty.$ Also, $P^n$ has some zero entries for all $n\ge 1.$ \hfill \qedhere
\end{example}

\begin{remark}
Recurrence and aperiodicity are chain properties for irreducible chains. If a Markov chain is irreducible, the whole state space is a single communicating class. Therefore, for an irreducible chain, if one state is recurrent, then all states are recurrent. Irreducibility ensures that the state space does not split into subsets such that the chain cannot move from one subset to another.
Positive recurrence ensures that the chain returns to states sufficiently often. In an irreducible positive recurrent chain, the chain returns to every state infinitely often, with finite expected return times. Periodicity causes the state space to split into cyclic classes that are visited by the chain in sequential order.
\end{remark}

\section{Ergodic Theorem in Finite State Space}\label{sec:ergfinite}
In this section, we study ergodicity \index{ergodicity} of Markov chains under the assumption that the state space $E$ is finite. In this setting, irreducibility and aperiodicity are enough to guarantee ergodicity.  We will use the following lemma.

\begin{lemma}\label{lem:regularity}
Let $P$ be the transition probability matrix of an irreducible, aperiodic Markov chain on a finite state space. Then there exists an integer $m$ such that, for all $n \geq m$, the matrix $P^n$ has strictly positive entries. 
\end{lemma}

Note that the conclusion of this lemma fails for the periodic Markov chain from Example \ref{ex:periodicchain}.
To establish ergodicity we will use a coupling argument that generalizes beyond the finite state space setting considered here. The coupling argument will be used repeatedly in these notes, and we will see that central to the argument is a lower bound on the Markov kernel, which, in the case of finite state space, follows from the previous lemma. The coupling argument uses the total variation distance between probability measures.

\begin{definition}[Total Variation Distance\index{total variation distance}]\label{def:totalvariation}
The \emph{total variation distance} between two probability measures $\nu_1$ and $\nu_2$ on $E$ is defined by
\begin{equation}\label{eq:tv}
\dtv(\nu_1,\nu_2)=\sup_{A \subseteq E} |\nu_1(A)-\nu_2(A)|. 
\end{equation}
\end{definition}

While \eqref{eq:tv}  is the standard definition of total variation distance, we will use the following characterization.
\begin{lemma}
It holds that
$$\dtv(\nu_1,\nu_2)=\frac{1}{2} \sup_{|h|_\infty \leq 1} \Bigl|\mathbb{E}_{X \sim \nu_1}[h(X)]-\mathbb{E}_{X \sim \nu_2}[h(X)]\Bigr|.$$ 
\end{lemma}

\begin{theorem}[Ergodicity of Markov Chains on Finite State Space]\label{thm:ergodicity}
Let $\XX=\{X_n\}_{n=0}^\infty$ be an irreducible and aperiodic Markov chain with finite state space $E$. Then there exist $\varepsilon \in (0,1)$ and a unique invariant distribution $\pi$ such that
\begin{equation}\label{eq:inequality}
\dtv(\pi_n,\pi) \leq (1-\varepsilon)^n,
\end{equation}
where $\pi_n$ denotes the law of $X_n$.
\end{theorem}
\begin{proof}
We will assume, for simplicity, that there exists $\varepsilon \in (0,1)$
such that
\begin{equation}\label{eq:auxergproof}
    p(i,j)\geq \varepsilon r(i,j), \qquad i,j\in E,
\end{equation}
where $p$ is the Markov kernel of the chain $\XX$ and $r$ is the uniform transition kernel on $E$. If this one-step minorization condition does not hold, the same argument can be applied to the $m$-step chain using Lemma \ref{lem:regularity}, with minor modifications.

Consider the map from $\mathcal{P}(E)$, the set of probability row vectors on $E$, into itself given by
\begin{align*}
\mathcal{P}(E) &\rightarrow \mathcal{P}(E),  \\
\nu &\mapsto \nu P,
\end{align*}
where $P$ is the transition matrix of the chain $\XX.$
Since this map is continuous and $\mathcal{P}(E)$ is compact and convex, Brouwer's fixed point theorem guarantees the existence of a fixed point. That is, there exists  $\pi \in \mathcal{P}(E)$ such that $\pi P= \pi$, showing the existence of an invariant distribution. 
In the rest of the proof, we show that if $\pi$ is an invariant distribution, then inequality \eqref{eq:inequality} holds.
This will also imply uniqueness of the invariant distribution.

We use a \emph{coupling} argument. \index{coupling argument} Let $\{B_n\}_{n=0}^\infty$ be a sequence of independent Bernoulli$(\varepsilon)$ random variables, independent of all other randomness. Given a random variable $W_0$, define, for $n\ge 0$,
\begin{equation}\label{eq:transition}
W_{n+1} \sim \begin{cases}
s(W_n,\cdot)&{\text{if} } \,\, B_n=0,\\
r(W_n,\cdot)&{\text{if} } \,\, B_n=1,\\
\end{cases}
\end{equation}
where 
$$s(i,j)=\frac{p(i,j)-\varepsilon r(i,j)}{1-\varepsilon}$$
and recall that $r$ denotes the uniform transition kernel. Note that
$$s(i,E)=\frac{p(i,E)-\varepsilon r(i,E)}{1-\varepsilon}=\frac{1-\varepsilon}{1-\varepsilon} =1$$ and, using \eqref{eq:auxergproof},
$$s(i,j)=\frac{p(i,j)-\varepsilon r(i,j)}{1-\varepsilon}\geq 0,$$
so $s$ is a Markov kernel. Moreover, 
\begin{align*}
	\mathbb{P}(W_{n+1}=j \,|\, W_n=i)
	&=\varepsilon\mathbb{P}(W_{n+1}=j \,|\, B_n=1,W_n=i)+(1-\varepsilon)\,\mathbb{P}(W_{n+1}=j \,|\, B_n=0,W_n=i)\\
	&=\varepsilon r(i,j)+(1-\varepsilon)s(i,j)\\
	&=p(i,j).
\end{align*}
Thus, $\{W_n\}$ has transition kernel $p$. 

Let $h:E\to\R$ with $|h|_\infty\leq1$, and let
$\tau=\min\{n\geq 0:B_n=1\}$.
Regardless of the distribution of $W_0$,
\begin{align*}
	\mathbb{E}[h(W_n)]
	&=\mathbb{E}[h(W_n) \,|\, \tau\geq n]\,\,\mathbb{P}(\tau\geq n) \,\, + \,\, \sum_{\ell=0}^{n-1}\mathbb{E}[h(W_n) \,|\, \tau=\ell]\mathbb{P}(\tau =\ell)\\
	&=\mathbb{E}[h(W_n) \,|\, \tau\geq n]  \,\, \mathbb{P}(\tau\geq n) \,\,  + \, \sum_{\ell=0}^{n-1}\mathbb{E}_{W_0\sim \text{Unif}(E)}[h(W_{n-\ell - 1})]\mathbb{P}(\tau =\ell),\\
\end{align*}
 where Unif$(E)$ denotes the uniform distribution on $E$. The second term is independent of the initial distribution, while
\[
\mathbb{P}(\tau\geq n)=(1-\varepsilon)^n.
\]

Now define two sequences of random variables as in  \eqref{eq:transition}:
let $\{Y_n\}_{n=0}^\infty$ be initialized with $\pi_0$ (the initial distribution of $\XX$), and let $\{Z_n\}_{n=0}^\infty$ be initialized with an invariant distribution $\pi$ (which we have shown that exists) so that $Z_n\sim\pi$ for all $n$.
Then,
\begin{align*}
	\dtv(\pi_n,\pi) 
    &=\frac{1}{2} \sup_{|h|_\infty\leq 1} 
    \Bigl|\mathbb{E}[h(Y_n)]-\mathbb{E}[h(Z_n)] \Bigr|  \\
    &\leq \frac{1}{2}\sup_{|h|_\infty\leq 1}
    \Bigl(2\mathbb{P}(\tau\geq n)\Bigr)\\
    &=(1-\varepsilon)^n. \tag*{\qedhere}
\end{align*}
\end{proof}

\section{Ergodic Theory in General State Space}\label{sec:erginfinite}
This section shows how to generalize the definitions of irreducibility, aperiodicity, and recurrence to general state spaces. We will also summarize some ergodicity results for general state spaces. 
\begin{definition}[$\phi$-Irreducible Markov Chain]\label{def:irreducible2}
	A Markov chain is called \emph{$\phi$-irreducible} if there exists a nonzero measure $\phi$ on $E$ such that, for all $A \subseteq E$ with $\phi(A)>0$ and for all $x \in E$, there exists a positive integer $n = n(x,A)$ such that $P^n(x,A)>0$.
\end{definition}

\begin{definition}[Aperiodic Markov Chain]
	A chain is \emph{aperiodic} if there do not exist $d\geq2$ and disjoint subsets $E_1, \dots , E_d \subseteq E$ such that 
	\begin{alignat*}{3}
		P(x,E_{i+1}) &=1,  \qquad  &&\forall   x \in E_i, \quad \quad i\in\{1,\dots,d-1\},\\
		P(x,E_{1}) &=1, \qquad  &&\forall   x \in E_d. 
	\end{alignat*}	
\end{definition}
	We remark that, if $E$ is discrete, this definition agrees with the one previously given.
The following result ensures that a Markov chain that is $\phi$-irreducible and aperiodic has a limit distribution.

\begin{theorem}[Limit Distribution for Almost All Initial States]
Let $\XX$ be an aperiodic, $\phi$-irreducible Markov chain with invariant probability distribution $\pi$. Then, for $\pi$-almost all $x \in E,$
\begin{equation*}
    \lim_{n\rightarrow\infty}\dtv \bigl(P^n(x,\cdot),\pi\bigr) = 0.
\end{equation*}
\end{theorem}

The above limit holds for almost all starting values $x \in E$. To make it hold for \emph{all} starting values $x\in E,$ and thus for all initial distributions $\pi_0,$ we need an additional property: \emph{Harris recurrence}.

\begin{definition}[Harris Recurrence\index{Harris recurrence}]\label{def:harris}
	A chain $\XX$ is \emph{Harris recurrent} if there exists a probability $\nu$ such that for all $A \subseteq E$ with $\nu(A)>0$ and all $x \in E$ we have 
	\begin{equation*}
		\Prob \bigl(X_n\in A\text{ for some }n>0 \, | \, X_0=x  \bigr)=1. \tag*{\qedhere}
	\end{equation*}
    A Markov chain $\XX$ is called \emph{positive Harris recurrent} if it is Harris recurrent and admits an invariant probability distribution $\pi$.
\end{definition}

\begin{theorem}[Limit Distribution for All Initial States]
Let $\XX$ be an aperiodic, positive Harris recurrent Markov chain with invariant distribution $\pi$. Then, for all $x \in E,$
\begin{equation*}
    \lim_{n\rightarrow\infty}\dtv \bigl(P^n(x,\cdot),\pi\bigr) = 0.
\end{equation*}
\end{theorem}

\begin{corollary}[Limit Distribution and Ergodicity\index{ergodicity}]
	Under the conditions of the previous theorem, for all $A \subseteq E$ and all initial distributions $\pi_0,$ it holds that
	\begin{equation*}
		\lim_{n\rightarrow\infty}\Prob(X_n\in A)=\pi(A).
	\end{equation*}
 Therefore, the chain is ergodic. 
\end{corollary}
\begin{proof}
    The result follows by the dominated convergence theorem, noting that $\pi_n(A)=\Prob(X_n \in A)=\int_E P^n(x,A)  \pi_0(x) \,dx.$  \hfill \qedhere
\end{proof}

\section{Sample Path Ergodicity}\label{sec:ergsamplepath}
Why do we care about ergodicity of Markov chains in a Monte Carlo simulation course? Markov chain Monte Carlo (MCMC) algorithms draw samples from a Markov chain
that has a prescribed target distribution $\pi = f$ as invariant distribution. 
 Ultimately, we are interested in understanding the \emph{sample path ergodicity} \index{ergodicity!sample path} of MCMC chains; that is, we want to determine whether ergodic averages
 $$\frac{1}{N}\sum_{n=1}^{N}h(X_n) $$
 computed from MCMC samples $\{X_n\}_{n=1}^N$ accurately approximate expectations $\mathcal{I}_f[h]$ with respect to the target. 
 These samples are \emph{not} independent, and consequently the efficiency of the ergodic average estimator depends on the correlations between them. 

First,  we have the following law of large numbers.
 \begin{theorem}[Law of Large Numbers]\label{thm:LLN}
	Let $h:E\rightarrow \R$ and let $\XX = \{X_n\}_{n=0}^\infty$ be an ergodic Markov chain with stationary distribution $\pi$. Then $\pi$-almost surely, 
	\begin{equation*}
		 \frac{1}{N}\sum_{n=1}^{N}h(X_n) \stackrel{N \rightarrow \infty}{\longrightarrow} \int_Eh(x)\pi(x) \, dx\equiv\mathcal{I}_\pi[h].
	\end{equation*}
\end{theorem}

We also have a central limit theorem, which requires geometric ergodicity of the chain.
\begin{definition}[Geometric and Uniform Ergodicity]\label{def:geometricuniformergodicity}
	An ergodic Markov chain with invariant distribution $\pi$ is \emph{geometrically ergodic}\index{ergodicity!geometric} if there exist a nonnegative function $M$ with $\mathbb{E}_\pi[M(X)]<\infty$ and a constant $0<r<1$ such that 
	\begin{equation*}
		\dtv(P^n(x,\cdot),\pi)\leq M(x)r^n
	\end{equation*}
	for all $x \in E$ and all $n \ge 0$. If the function $M$ is bounded above, then the chain is called \emph{uniformly ergodic}.\index{ergodicity!uniform} \hfill \qedhere
\end{definition}

	Note, as an example, that we have proved that any irreducible and aperiodic Markov chain on a finite state space is uniformly ergodic. The following result
    shows that a central limit theorem can be deduced from geometric ergodicity. In Chapter \ref{chap:MCMC}, we discuss the uniform and geometric ergodicity of several MCMC algorithms, thus guaranteeing, in particular, their sample path ergodicity.

\begin{theorem}[Central Limit Theorem]\label{thm:CLT}
	With the same notation as in the previous theorem, suppose further that $\XX$ is geometrically ergodic and that, for some $\epsilon>0,$
	\begin{equation*}
		\Expect_{Y\sim\pi}[|h(Y)|^{2+\epsilon}]<\infty.
	\end{equation*}
	Then \begin{equation*}
		\sqrt{N} \Bigl( \frac{1}{N}\sum_{n=1}^{N}h(X_n)-\mathcal{I}_\pi[h]\Bigr)\stackrel{\mathcal{D}}{\longrightarrow} \Nc(0,\tau^2),
	\end{equation*}
	where $\tau^2$ is the integrated autocorrelation time, \index{integrated autocorrelation time} defined below in Lemma \ref{lemmaMarkovchainerror}.
\end{theorem}

The rest of this appendix is devoted to defining and understanding the integrated autocorrelation time, which may be used to assess the relative efficiency of different MCMC algorithms. Unsurprisingly, the asymptotic variance $\tau^2$ in the central limit theorem above arises as the limiting variance of ergodic averages. We define
\begin{equation*}
	\frac{\tau_N^2}{N}:= \V\biggl[\frac{1}{N}\sum_{n=1}^N h(X_n)\biggr]
\end{equation*}
and note that $\tau_N^2$ depends on $h,$ as do $\sigma^2$ and $\tau$ below, but we omit this dependence from our notation. The quantity $\tau_N^2$ measures the amount of correlation among the variables $h(X_n),$ $1 \le n \le N$.

\begin{definition}[Autocovariance\index{autocovariance} and Autocorrelation\index{autocorrelation}]\label{ref:autocovariance}
	At stationarity, the \emph{autocovariance} of lag $k$ of the time series $\{h(X_n)\}_{n=0}^\infty$ is defined as
	\begin{align*}
		\gamma_k = \operatorname{Cov}\bigl(h(X_0),h(X_k)\bigr), \quad \quad k\geq 0.
	\end{align*}
	 The \emph{autocorrelation} of lag $k$ is 
	\begin{equation*}
		\rho_k=\frac{\gamma_k}{\sigma^2},\quad \quad k\geq 0, 
	\end{equation*}
    where
    \begin{equation*}
        \sigma^2 = \gamma_0 = \operatorname{Cov}\bigl(h(X_0),h(X_0)\bigr)=\V_{X_0\sim\pi}[h(X_0)].\tag*{\qedhere}
    \end{equation*}
\end{definition}

\begin{lemma}\label{lemmaMarkovchainerror}
	At stationarity,
	\begin{equation*} \tau_N^2=\sigma^2\left[1+2\sum_{k=1}^{N-1}\frac{N-k}{N}\rho_k\right]. 
	\end{equation*}
	As  $N \rightarrow \infty$, 
	\begin{equation*}
		\tau_N^2\rightarrow \tau^2:=\sigma^2\left[1+2\sum_{k=1}^{\infty}\rho_k\right],
	\end{equation*}
	where $\tau^2$ is called the integrated autocorrelation time. \index{integrated autocorrelation time}
\end{lemma}

\begin{proof}
A calculation shows that
	\begin{align*}
		\frac{\tau_{N}^2}{N} =& \V\left[\frac{1}{N}\sum_{n=1}^{N}h(X_n)\right]\\
		&= \frac{1}{N^2} \left[\sum_{n=1}^{N}\V[h(X_n)]+2\sum_{n=1}^{N-1}\sum_{k>n}\operatorname{Cov}\left(h(X_n),h(X_k)\right)\right]\\
		&= \frac{1}{N^2} \left[N\sigma^2+2\sum_{n=1}^{N-1}\sum_{k=1}^{N-n}\operatorname{Cov}\left(h(X_n),h(X_{n+k})\right)\right]\\
		&= \frac{\sigma^2}{N} \left[1+\frac{2}{N\sigma^2}\sum_{n=1}^{N-1}\sum_{k=1}^{N-n}\gamma_k\right]\\
		&= \frac{\sigma^2}{N} \left[1+\frac{2}{N}\sum_{n=1}^{N-1}\sum_{k=1}^{N-n}\rho_k\right]\\
		&= \frac{\sigma^2}{N} \left[1+\frac{2}{N}\sum_{k=1}^{N-1}\sum_{n=1}^{N-k}\rho_k\right]\\
		&= \frac{\sigma^2}{N} \left[1+2\sum_{k=1}^{N-1}\frac{N-k}{N}\rho_k\right]. \tag*{\qedhere}
	\end{align*}
\end{proof}

Note that if $X_n\stackrel{\text{i.i.d.}}{\sim}\pi$, then 
\begin{equation*}
	\V\left[\frac{1}{N}\sum_{n=1}^N h(X_n)\right] = \frac{\sigma^2}{N}.
\end{equation*}
Therefore, computing an ergodic average with positively autocorrelated samples leads to higher variance than computing it with an i.i.d. sample. An intuitive explanation is that positively correlated random variables contain redundant information and are therefore less informative than i.i.d. random variables. On the other hand, if the correlations are negative, ergodic averages may be more accurate than a direct Monte Carlo estimator with i.i.d. samples. 

\section{Discussion and Bibliography}
The material in this appendix can be found in any standard reference on Markov chains, and follows \cite{Robertsunpublished}. For gentle introductions to Markov chain theory, we refer to
\cite{norris1998markov,bremaud2013markov} and \cite[Chapter 1]{lawler2018introduction}. For a discussion of more advanced topics, we refer to
\cite{meyn2012markov,douc2018markov,nummelin2004general}.

\bibliographystyle{siam}
\bibliography{references}

\cleardoublepage

\printindex

\end{document}